%% file: main.tex
\documentclass[11pt]{article}
\usepackage[letterpaper,hmargin=0.8in,vmargin=0.8in]{geometry}
\usepackage{times}

\usepackage[ruled]{algorithm2e}
\usepackage{caption}
\usepackage{hyperref}
\usepackage{adjustbox,multirow}
\usepackage{verbatim,color}
\usepackage{amsmath,amsthm,amsfonts,latexsym,verbatim,threeparttable,amsfonts,rotating,pifont}

\input{packages_article}
\input{symbols_env}
\input{symbols}

\makeatletter
\DeclareRobustCommand*\cal{\@fontswitch\relax\mathcal}
\makeatother

\newtheorem{theorem}{Theorem}[section]
\newtheorem{corollary}[theorem]{Corollary}
\newtheorem{lemma}[theorem]{Lemma}

\theoremstyle{definition}
\newtheorem{definition}[theorem]{Definition}

\newtheorem{claim}[theorem]{Claim}
\newtheorem{notation}[theorem]{Notation}

\usepackage{color}
\usepackage{framed}

\newcommand{\commentAA}[1]{\textcolor{red}{{\sf (Ananya's  Note:} {\sl{#1}} {\sf EON)}}}

\begin{document}

\title{\bf Revisiting the Efficiency of Asynchronous Multi Party Computation Against General Adversaries}
\author{Ananya Appan\footnote{International Institute of Information Technology, Bangalore India.
 Email: {\tt{ananya.appan@iiitb.ac.in}}.} \and Anirudh Chandramouli\footnote{International Institute of Information Technology, Bangalore India.
 Email: {\tt{anirudh.c@iiitb.ac.in}}.}   \and  Ashish Choudhury\footnote{International Institute of Information Technology, Bangalore India.
 Email: {\tt{ashish.choudhury@iiitb.ac.in}}. } 
 }
\date{}
\maketitle
\vspace*{-0.7cm}
\begin{abstract}
In this paper, we design secure multi-party computation (MPC) protocols in the {\it asynchronous} communication setting with {\it optimal} resilience.
 Our protocols are secure against
  a computationally-unbounded {\it malicious} adversary, characterized by an {\it adversary structure} $\AdvStructure$,
  which enumerates all possible subsets of potentially corrupt parties. Our protocols incur a communication of
  $\Order(|\AdvStructure|^2)$ and $\Order(|\AdvStructure|)$ bits per multiplication for {\it perfect} and {\it statistical} security respectively. These are the {\it first} protocols with this communication complexity, as such protocols were known only in the {\it synchronous} communication setting (Hirt and Tschudi, ASIACRYPT 2013).  
\end{abstract}

\input{intro}

\input{prelim}

\input{perfect}

\input{stat}

\input{mpc}

 \bibliographystyle{plain}

\bibliography{main}

 \appendix
 
 \input{AppUC}

 \input{AppPerfectAMPCv2}

 \input{AppStatAMPCv2}

 \input{AppMPC}

\end{document}

%% file: packages_article.tex
\usepackage{enumitem}
\usepackage{latexsym}
\usepackage{mathtools}
\usepackage{multirow}
\usepackage{multicol}
\usepackage{mwe}
\usepackage{makecell}
\usepackage{pgfplots}
\usepackage[section]{placeins}
\usepackage{pifont}
\usepackage{ragged2e}
\usepackage{rotating}
\usepackage{subcaption}
\usepackage{tabu}
\usepackage{tabularx}
\usepackage{threeparttable}
\usepackage{tikz}
\usepackage[normalem]{ulem}
\usepackage{url}
\usepackage{verbatim}
\usepackage{wrapfig}
\usepackage{xcolor}
\usepackage{xspace}
\usepackage[framemethod=tikz]{mdframed}
\usepackage{algpseudocode}

%% file: symbols.tex
\newcommand{\SharingSpec}{\ensuremath{\mathbb{S}}}
\newcommand{\Q}{\ensuremath{\mathbb{Q}}}
\newcommand{\F}{\ensuremath{\mathbb{F}}}

\newcommand{\C}{\ensuremath{\mathcal{C}}}
\newcommand{\W}{\ensuremath{\mathcal{W}}}
\newcommand{\Order}{\ensuremath{\mathcal{O}}}
\newcommand{\R}{\ensuremath{\mathcal{SV}}}
\newcommand{\ckt}{\ensuremath{\mathsf{ckt}}}
\newcommand{\REAL}{\ensuremath{\mathsf{REAL}}}
\newcommand{\IDEAL}{\ensuremath{\mathsf{IDEAL}}}
\newcommand{\sid}{\ensuremath{\mathsf{sid}}}
\newcommand{\defined}{\ensuremath{\stackrel{def}{=}}}

\newcommand{\ShareSpec}{\ensuremath{\mathbb{S}}}

\newcommand{\View}{\ensuremath{\mathsf{View}}}
\newcommand{\ICSig}{\ensuremath{\mathsf{ICSig}}}
\newcommand{\Char}{\ensuremath{\mathsf{char}}}
\newcommand{\view}{\ensuremath{\mathsf{View}}}

\newcommand{\PartySet}{\ensuremath{\mathcal{P}}}
\newcommand{\AdvStruct}{\ensuremath{\mathcal{Z}}}
\newcommand{\Selected}{\ensuremath{\mathsf{Selected}}}
\newcommand{\Products}{\ensuremath{\mathsf{Summands}}}
\newcommand{\Waitlist}{\ensuremath{\mathcal{W}}}
\newcommand{\Discarded}{\ensuremath{\mathcal{GD}}}
\newcommand{\LocalDiscarded}{\ensuremath{\mathcal{LD}}}

\newcommand{\AdvStructure}{\ensuremath{\mathcal{Z}}}
\newcommand{\Sim}{\ensuremath{\mathcal{S}}}
\newcommand{\CoreSet}{\ensuremath{\mathcal{CS}}}
\newcommand{\AcceptedParties}{\ensuremath{\mathcal{G}}}
\newcommand{\Hon}{\ensuremath{\mathcal{H}}}

\newcommand{\GlobalProv}{\ensuremath{\mathcal{GP}}}

\newcommand{\SimAcast}{\ensuremath{\mathcal{S}_\mathsf{Acast}}}
\newcommand{\SimPVSS}{\ensuremath{\mathcal{S}_{\mathsf{PVSS}}}}
\newcommand{\SimPerAMPC}{\ensuremath{\mathcal{S}_{\mathsf{AMPC}}}}
\newcommand{\SimSVSS}{\ensuremath{\mathcal{S}_{\mathsf{SVSS}}}}
\newcommand{\SimPerTriples}{\ensuremath{\mathcal{S}_{\mathsf{PerTriples}}}}
\newcommand{\SimStatTriples}{\ensuremath{\mathcal{S}_{\mathsf{StatTriples}}}}

\newcommand{\errorAICP}{\ensuremath{\mathsf{\epsilon_{AICP}}}}

\newcommand{\D}{\ensuremath{\mathsf{D}}}

\newcommand{\OptMult}{\ensuremath{\Pi_\mathsf{OptMult}}}
\newcommand{\BasicMult}{\ensuremath{\Pi_\mathsf{BasicMult}}}

\newcommand{\Auth}{\ensuremath{\Pi_\mathsf{Auth}}}
\newcommand{\Reveal}{\ensuremath{\Pi_\mathsf{Reveal}}}

\newcommand{\PiPerVSS}{\ensuremath{\Pi_\mathsf{PVSS}}}
\newcommand{\PiStatVSS}{\ensuremath{\Pi_\mathsf{SVSS}}}
\newcommand{\PiPerDefaultShare}{\ensuremath{\Pi_\mathsf{PerDefSh}}}
\newcommand{\PiPerRecShare}{\ensuremath{\Pi_\mathsf{PerRecShare}}}
\newcommand{\PiPerRec}{\ensuremath{\Pi_\mathsf{PerRec}}}
\newcommand{\PiPerAMPC}{\ensuremath{\Pi_\mathsf{PerAMPC}}}
\newcommand{\MultLCE}{\ensuremath{\Pi_\mathsf{MultCI}}}
\newcommand{\MultGCE}{\ensuremath{\Pi_\mathsf{Mult}}}
\newcommand{\PiPerTriples}{\ensuremath{\Pi_\mathsf{PerTriples}}}
\newcommand{\PiAcast}{\ensuremath{\Pi_\mathsf{Acast}}}
\newcommand{\PiStatTriples}{\ensuremath{\Pi_\mathsf{StatTriples}}}
\newcommand{\RandMultLCE}{\ensuremath{\Pi_\mathsf{RandMultCI}}}
\newcommand{\PiMPC}{\ensuremath{\Pi_\mathsf{AMPC}}}

\newcommand{\OK}{\ensuremath{\mathsf{OK}}}
\newcommand{\NOK}{\ensuremath{\mathsf{NOK}}}
\newcommand{\Received}{\ensuremath{\mathsf{Received}}}
\newcommand{\reject}{\ensuremath{\mathsf{reject}}}
\newcommand{\inp}{\ensuremath{\mathsf{inp}}}
\newcommand{\coreset}{\ensuremath{\mathsf{coreset}}}
\newcommand{\out}{\ensuremath{\mathsf{out}}}
\newcommand{\Sender}{\ensuremath{\mathsf{sender}}}
\newcommand{\Acast}{\ensuremath{\mathsf{Acast}}}
\newcommand{\vote}{\ensuremath{\mathsf{vote}}}
\newcommand{\decide}{\ensuremath{\mathsf{decide}}}
\newcommand{\Dealer}{\ensuremath{\mathsf{dealer}}}
\newcommand{\Share}{\ensuremath{\mathsf{share}}}
\newcommand{\triples}{\ensuremath{\mathsf{triples}}}
\newcommand{\tripleshares}{\ensuremath{\mathsf{tripleshares}}}
\newcommand{\shares}{\ensuremath{\mathsf{shares}}}

\newcommand{\dist}{\ensuremath{\mathsf{dist}}}
\newcommand{\test}{\ensuremath{\mathsf{test}}}
\newcommand{\recShare}{\ensuremath{\mathsf{recShare}}}

\newcommand{\Ready}{\ensuremath{\mathsf{ready}}}

\newcommand{\Success}{\ensuremath{\mathsf{success}}}

\newcommand{\echo}{\ensuremath{\mathsf{echo}}}
\newcommand{\authPoly}{\ensuremath{\mathsf{authPoly}}}
\newcommand{\authPoint}{\ensuremath{\mathsf{authPoint}}}
\newcommand{\revealPoly}{\ensuremath{\mathsf{revealPoly}}}
\newcommand{\revealPoint}{\ensuremath{\mathsf{revealPoint}}}

\newcommand{\iter}{\ensuremath{\mathsf{iter}}}
\newcommand{\hop}{\ensuremath{\mathsf{hop}}}
\newcommand{\authCompleted}{\ensuremath{\mathsf{authCompleted}}}
\newcommand{\curr}{\ensuremath{\mathsf{curr}}}
\newcommand{\accept}{\ensuremath{\mathsf{accept}}}
\newcommand{\sima}{\ensuremath{\widetilde{a}}}
\newcommand{\simb}{\ensuremath{\widetilde{b}}}
\newcommand{\simc}{\ensuremath{\widetilde{c}}}
\newcommand{\simx}{\ensuremath{\widetilde{x}}}
\newcommand{\simy}{\ensuremath{\widetilde{y}}}
\newcommand{\simby}{\ensuremath{\widetilde{\mathbf{y}}}}

\newcommand{\flag}{\ensuremath{\mathsf{flag}}}

\newcommand{\Adv}{\ensuremath{\mathsf{Adv}}}
\newcommand{\Env}{\ensuremath{\mathsf{Env}}}

\newcommand{\Functionality}{\ensuremath{\mathcal{F}_{\mathsf{AMPC}}}}
\newcommand{\FAcast}{\ensuremath{\mathcal{F}_{\mathsf{Acast}}}}
\newcommand{\FABA}{\ensuremath{\mathcal{F}_{\mathsf{ABA}}}}
\newcommand{\FVSS}{\ensuremath{\mathcal{F}_{\mathsf{VSS}}}}
\newcommand{\FTriples}{\ensuremath{\mathcal{F}_{\mathsf{Triples}}}}
\newcommand{\Fun}{\ensuremath{\mathcal{F}}}

\newenvironment{myitemize}
{\begin{list}{$\bullet$}{ 
\itemindent=-0.1in
\itemsep=0.0in
\parsep=0.0in
\topsep=0.0in
\partopsep=0.0in}}{\end{list}}
\newcounter{itemcount}

\newenvironment{myenumerate}
{\setcounter{itemcount}{0}\begin{list}
{\arabic{itemcount}.}{\usecounter{itemcount} \itemindent=-0.2cm
\itemsep=0.0in
\parsep=0.0in
\topsep=5pt
\partopsep=0.0in}}{\end{list}}

\newenvironment{mydescription}
{\setcounter{itemcount}{0}\begin{list}
{\arabic{itemcount}.}{\usecounter{itemcount} \itemindent=-0.5cm
\itemsep=0.0in
\parsep=0.0in
\topsep=5pt
\partopsep=0.0in}}{\end{list}}

%% file: intro.tex
\section{Introduction}
\label{sec:intro}
Secure {\it multi-party computation} (MPC) \cite{Yao82,GMW87,BGW88,RB89} is a fundamental problem in secure distributed computing. Consider a set 
 of $n$ mutually-distrusting parties $\PartySet = \{P_1, \ldots, P_n \}$, where a subset of parties can be corrupted
  by 
  a {\it computationally-unbounded malicious} (Byzantine) adversary $\Adv$. 
  Informally, an MPC protocol allows the parties  to securely compute any function $f$
  of their private inputs, by keeping their respective inputs private. The most popular way of characterizing $\Adv$ is  through
  a {\it threshold}, by assuming that
   it can corrupt {\it any}  subset of up to $t$ parties. In this setting, MPC with {\it perfect security} (where no error is allowed in the outcome)
    is achievable iff $t < n/3$ \cite{BGW88}, while {\it statistical security} (where a negligible error is allowed)
    is achievable iff $t < n/2$ \cite{RB89}. Hirt and Maurer \cite{HM00} generalized the threshold model by introducing the 
    general-adversary model (also known as the {\it non-threshold} setting). In this setting, $\Adv$ is characterized by an {\it adversary structure} 
    $\AdvStructure = \{Z_1, \ldots, Z_h \} \subset 2^{\PartySet}$, which enumerates all possible subsets of potentially corrupt parties, where $\Adv$ can select any subset of parties $Z^{\star} \in \AdvStructure$ for corruption. 
    Modelling the distrust in the system through $\AdvStructure$ allows for more flexibility (compared to the threshold model), especially when
    $\PartySet$ is not too large. In the general-adversary model, MPC with perfect and statistical security 
    is achievable iff $\AdvStructure$ satisfies the $\Q^{(3)}(\PartySet, \AdvStructure)$ and
    $\Q^{(2)}(\PartySet, \AdvStructure)$  conditions respectively.\footnote{We say that 
     $\AdvStructure$ satisfies the $\Q^{(k)}(\PartySet, \AdvStructure)$ condition \cite{HM00}, if the union of no
     $k$ sets from $\AdvStructure$ covers $\PartySet$.} 
     
     In terms of efficiency, MPC protocols against general adversaries are {\it less efficient} than 
     those against threshold adversaries, by several orders of magnitude. 
     Protocols against threshold adversaries 
     typically incur a communication of
      $n^{\Order(1)}$ bits per multiplication,
      compared to $|\AdvStructure|^{\Order(1)}$ bits per multiplication required against general adversaries.\footnote{The cost of any generic MPC protocol
      is typically dominated by the overhead associated with the multiplication operations in $f$.} 
      Since $|\AdvStructure|$ could be exponentially large in $n$, the {\it exact} exponent
      is very important. For instance, as noted in \cite{HT13}, 
      if $n = 25$, then $|\AdvStructure|$ is expected to be around 
      one million, and a protocol with a communication complexity of 
      $\Order(|\AdvStructure|^2 \cdot \mbox{Poly}(n))$ bits is preferred over a protocol 
      with a communication complexity of 
      $\Order(|\AdvStructure|^3 \cdot \mbox{Poly}(n))$ bits. 
      \vspace*{-0.4cm}
      \paragraph{\bf Our Motivation and Results:}
      All the above results hold in the {\it synchronous} communication setting, where the parties are assumed to be globally synchronized, 
      with strict upper bounds on the message delay. 
      Such strict time-outs are, however, extremely difficult to maintain in 
       real-world networks like the Internet, which are better modelled by the {\it asynchronous} communication setting \cite{CanettiThesis}. Here, no timing assumptions are made and messages can be arbitrarily, but finitely delayed, with
       every message sent  being delivered {\it eventually}. 
       Asynchronous protocols are more complex and less efficient when compared to their synchronous counter-parts, since a {\it slow} (but honest)
       sender party {\it cannot} be distinguished from a {\it corrupt} sender party, who does not send any message. To avoid an endless wait, the parties {\it cannot} afford to wait to receive messages from all the parties, which results in disregarding messages from a subset of potentially honest parties.
        Against threshold adversaries, perfectly-secure and statistically-secure {\it asynchronous} MPC (AMPC) is achievable, iff
       $t< n/4$ \cite{BCG93} and $t<n/3$ \cite{BKR94,ADS20} respectively. By using the {\it player-partitioning} argument \cite{HM00}, these results can
       be generalized to show that against general adversaries, perfect and statistical security require $\AdvStructure$ to satisfy the
       $\Q^{(4)}(\PartySet, \AdvStructure)$ and $\Q^{(3)}(\PartySet, \AdvStructure)$ conditions respectively. 
       
       Compared to synchronous MPC protocols, AMPC protocols are not very well-studied \cite{BCG93,BKR94,BH07,PCR15,CP17}, especially
       against general adversaries.
       While perfectly-secure AMPC against general adversaries has been studied
       in \cite{MSR02,CP20}, to the best of our knowledge, there exists {\it no} statistically-secure AMPC protocol against general adversaries. 
       We design communication efficient AMPC protocols against general adversaries, {\it both} with perfect and statistical security,
       whose efficiency is comparable with the most efficient MPC protocols in the {\it synchronous} communication setting. 
       Our results put in the context of relevant existing results are presented in Table \ref{tab:results}, where
       $\F$ denotes a finite field over which all computations are performed, and $\kappa$ denotes the statistical-security parameter.
       \begin{table}[!ht]
       \centering
       {\footnotesize
       \begin{tabular}{|c|c|c|c|c|c|c|}
       \hline
       \multicolumn{3}{ |c| }{Synchronous} & \multicolumn{3}{ |c| }{Asynchronous}   \\ \hline
       Security & Condition & Bits/Multiplication & Security & Condition & Bits/Multiplication \\ \hline
       \multirow{2}{*}{Perfect} & \multirow{2}{*}{$\Q^{(3)}$} & \multirow{2}{*}{$\Order(|\AdvStructure|^2 \cdot \mbox{Poly}(n, |\F|))$ \cite{HT13}} &
          \multirow{2}{*}{Perfect}  & \multirow{2}{*}{$\Q^{(4)}$} & $\Order(|\AdvStructure|^3 \cdot \mbox{Poly}(n, |\F|))$ \cite{CP20} \\
                      & & &  &  &  $\Order(|\AdvStructure|^2 \cdot \mbox{Poly}(n, |\F|))$ {\bf (our result)} \\ \hline        
       Statistical & $\Q^{(2)}$ & $\Order(|\AdvStructure| \cdot \mbox{Poly}(n, \kappa))$ \cite{HT13} & 
        Statistical & $\Q^{(3)}$ & $\Order(|\AdvStructure| \cdot \mbox{Poly}(n, \kappa))$ {\bf (our result)} \\ \hline
       \end{tabular}
       }
        \caption{\label{tab:results}Communication complexity of different MPC protocols against general adversaries in terms of $|\AdvStructure|$} 
        \end{table}     
      \vspace*{-0.2cm}
      
      Our protocols are in the pre-processing model, where the parties generate random secret-shared multiplication-triples.
       The parties then evaluate $\ckt$ in a secret-shared fashion, where Beaver's method \cite{Bea91} is used to evaluate the multiplication gates using the generated multiplication-triples. Our protocols for the pre-processing phase
             closely follow \cite{HT13}. However, there are several non-trivial challenges
       while adapting these protocols to the asynchronous world. Since our protocols are slightly technical, we refer to Section \ref{sec:perfectoverview}
        for the technical overview of our protocols.

%% file: prelim.tex
\section{Preliminaries, Definitions and Existing Asynchronous Primitives}
\label{prelim}
We assume that the parties in $\PartySet = \{P_1, \ldots, P_n\}$ are connected by pair-wise secure channels. 
   The adversary $\Adv$ is assumed to be {\it malicious} and {\it static}, and decides the set of corrupt parties at the beginning of the protocol execution.       
   Parties not under the control of $\Adv$ are called {\it honest}.
     Given $\PartySet' \subseteq \PartySet$, we say that $\AdvStructure$
 satisfies the $\Q^{(k)}(\PartySet', \AdvStructure)$ condition, if for every $Z_{i_1}, \ldots, Z_{i_k} \in \AdvStructure$, the condition
 $\PartySet' \not \subseteq Z_{i_1} \cup \ldots \cup Z_{i_k}$ holds.
 
 We assume  that the parties want to compute a function $f$, represented by a {\it publicly} known
 arithmetic circuit $\ckt$ over a finite field $\F$ consisting of linear and non-linear gates, with 
  $M$ being number of multiplication gates.
  Without loss of generality, we assume that
 each $P_i \in \PartySet$ has an input $x^{(i)}$ for $f$, and that
 all the parties want to learn the single output 
  $y = f(x^{(1)}, \ldots, x^{(n)})$.
      We follow the asynchronous communication model of
  \cite{BCG93,CanettiThesis}.
   Unlike the previous unconditionally-secure AMPC protocols \cite{BKR94,BH07,PCR15,CP17,CP20},
  we prove the security of our protocols using the UC framework \cite{Can01,Gol,Can20}, whose details are presented in Appendix \ref{app:UC}.

    In our protocols, we use a secret-sharing
  based on the one from \cite{Mau02}, defined with respect to
  a {\it sharing specification} $\ShareSpec$, which is a tuple of subsets of $\PartySet$. 
    A sharing specification $\ShareSpec$ is said to be {\it $\AdvStructure$-private},
   if for every $Z \in \AdvStructure$, there is an $S \in \ShareSpec$, such that $Z \cap S = \emptyset$.
   A sharing specification $\ShareSpec$ 
   satisfies the $\Q^{(k)}(\ShareSpec, \AdvStructure)$ condition if for every $Z_{i_1}, \ldots, Z_{i_k} \in \AdvStructure$ and every
   $S \in \ShareSpec$, the condition 
   $S \not \subseteq Z_{i_1} \cup \ldots \cup Z_{i_k}$ holds.
   In our protocols, we use the sharing specification
  $\ShareSpec = (S_1, \ldots, S_h)
   \defined \{ \PartySet \setminus Z | Z \in \AdvStructure\}$, which guarantees that $\ShareSpec$ is {\it $\AdvStructure$-private}.
   This $\ShareSpec$ satisfies the $\Q^{(3)}(\ShareSpec, \AdvStructure)$ and $\Q^{(2)}(\ShareSpec, \AdvStructure)$ conditions,
   if $\AdvStructure$ satisfies the $\Q^{(4)}(\PartySet, \AdvStructure)$ and $\Q^{(3)}(\PartySet, \AdvStructure)$ conditions respectively.
   
    \begin{definition}[\cite{Mau02,HT13}]
  A value $s \in \F$ is said to be secret-shared with respect to $\ShareSpec = (S_1, \ldots, S_h)$, if 
   there exist shares $s_1, \ldots, s_h$, such that $s = s_1 + \ldots + s_h$ and 
   for $q = 1, \ldots, h$, share $s_q$ is know to every (honest) party in $S_q$.  
  \end{definition}
  \noindent A sharing of $s$ is denoted by $[s]$, where $[s]_q$ denotes the $q^{th}$ share. Note that
  $P_i$ will hold the shares $\{[s]_q\}_{P_i \in S_q}$.
  The above secret-sharing is {\it linear}, as 
  $[c_1s_1 + c_2 s_2] = c_1[s_1] + c_2[s_2]$ for any publicly-known $c_1, c_2 \in \F$.
     \paragraph{\bf Asynchronous Reliable Broadcast (Acast):}  
  Acast allows a designated 
   {\it sender}
  $P_S \in \PartySet$ to identically send a message 
   $m \in \{0, 1 \}^{\ell}$ to all the parties.
    If $P_S$ is {\it honest}, then all honest parties eventually output $m$. If $P_S$ is {\it corrupt} and
    some honest party outputs $m^{\star}$, then every
   other honest party eventually outputs $m^{\star}$.
   The above requirements are formalized by an ideal functionality $\FAcast$, presented in Appendix \ref{app:UC}. 
   In \cite{KF05}, a perfectly-secure Acast protocol is presented with a communication complexity of $\Order(n^2 \ell)$ bits, provided $\AdvStructure$ satisfies the 
   $\Q^{(3)}(\PartySet, \AdvStructure)$ condition. The security of the protocol in \cite{KF05} is not proven in the UC framework. For completeness, we do this in Appendix \ref{app:UC}. 
 \vspace*{-0.3cm}
  \paragraph{\bf Asynchronous Byzantine Agreement (ABA):}
 In an ABA protocol \cite{PeaseJACM80,NancyBook,AW04},
 every party has a private bit and the (honest) parties eventually obtain a common output bit almost-surely with probability $1$,
  where the output bit is the input bit of an honest party, if all honest parties have the same input.\footnote{From \cite{FLP85},
  every {\it deterministic} ABA protocol must have non-terminating runs.
  To circumvent this result, randomized ABA protocols are considered and the best we can hope for from such protocols is that the parties eventually obtain an output with
  probability $1$.} 
  The above requirements are formalized through the functionality $\FABA$, presented in Appendix \ref{app:UC}.
 We assume the existence of a perfectly-secure ABA protocol for $\FABA$ with UC-security (see \cite{K01,KF05} for such protocols if
  $\AdvStructure$ satisfies the $\Q^{(3)}(\PartySet, \AdvStructure)$ condition).
  The number of ABA instances will be {\it independent} of the size of $\ckt$ and so, we do not focus on the exact details.
  \vspace*{-0.3cm}
 \paragraph{\bf Verifiable Secret-Sharing (VSS):}
  A VSS protocol allows a designated {\it dealer} $P_D \in \PartySet$ to verifiably secret-share its input $s \in \F$.
  If $P_D$ is {\it honest}, then the honest parties eventually complete the protocol with $[s]$.
  The verifiability guarantees that if $P_D$ is {\it corrupt} and some honest party completes the protocol, then all honest parties eventually complete the protocol with a secret-sharing
  of some value.
   These requirements are formalized through the functionality $\FVSS$ (Fig \ref{fig:FVSS}). The functionality, upon receiving a vector of shares
   from $P_\D$, distributes the appropriate shares to the respective parties. The dealer's input is defined {\it implicitly} as the sum of provided shares.
     We will use 
  $\FVSS$ in our protocols as follows: $P_D$ on having the input $s$, sends a random vector of shares 
  $(s_1, \ldots, s_h)$ to $\FVSS$ where 
  $s_1 + \ldots + s_h = s$.
  If $P_D$ is {\it honest}, then the view of $\Adv$ will be independent of $s$, if $\ShareSpec$ is {\it $\AdvStructure$-private}. Hence, the
   probability distribution of shares learnt by $\Adv$ will be {\it independent} of the dealer's input.
  \begin{systembox}{$\FVSS$}{The ideal functionality for VSS for session id $\sid$.}{fig:FVSS}
	\justify
$\FVSS$ proceeds as follows for each party $P_i \in \PartySet$ and an adversary $\Sim$, and is parametrized 
 by a sharing specification $\ShareSpec = (S_1, \ldots, S_h)$, adversary structure $\AdvStructure$  and a dealer $P_D$. 
  Let $Z^{\star}$ be the set of corrupt parties.
    \begin{myitemize}
         \item[--] On receiving $(\Dealer, \sid, P_{\D}, (s_1, \ldots, s_h))$ from $P_\D$ (or from $\Sim$ if $P_\D \in Z^{\star}$),
         set $s = \sum_{q = 1, \ldots, h} s_q$ and for $q = 1, \ldots, h$, set $[s]_q = s_q$.
         Generate a request-based delayed output $(\Share, \sid, P_D, \{[s]_q \}_{P_i \in S_q})$ for each
              $P_i \not \in Z^{\star}$.\footnote{If $P_D$ is {\it corrupt}, then
         $\Sim$ may not send any input to $\FVSS$,
          in which case the functionality will not generate any output. 
          See Appendix \ref{app:UC} for the meaning of request-based delayed output in asynchronous ideal world.}        
       \end{myitemize}  
  \end{systembox}
 In \cite{CP20}, a {\it perfectly-secure} VSS protocol $\PiPerVSS$ is presented, provided $\ShareSpec$
  satisfies the $\Q^{(3)}(\ShareSpec, \AdvStructure)$ condition (which holds for our $\ShareSpec$).
   The protocol (after a minor modification) incurs a communication of $\Order(|\AdvStructure| \cdot n^2 \log{|\F|} + n^4 \log{n})$
   bits.
   In \cite{CP20}, the UC-security of $\PiPerVSS$
   was {\it not} shown and for completeness, we do so in Appendix \ref{app:Perfect}.
 \vspace*{-0.3cm}
\paragraph{Default Secret-Sharing:}
The perfectly-secure protocol $\PiPerDefaultShare$ takes a {\it public} input $s \in \F$ and $\ShareSpec = (S_1, \ldots, S_h)$ to
 {\it non-interactively} generate $[s]$, where the parties collectively set $[s]_1 = s$ and $[s]_2 = \ldots = [s]_h = 0$.
  \vspace*{-0.3cm}
 \paragraph{Reconstruction Protocols:} 
  Let the parties hold $[s]$ with respect to some $\ShareSpec = (S_1, \ldots, S_h)$ which satisfies
   the $\Q^{(2)}(\ShareSpec, \AdvStructure)$ condition.  Then, \cite{CP20} presents a 
   perfectly-secure protocol $\PiPerRecShare$ to reconstruct $[s]_q$ for any given $q \in \{1, \ldots, h \}$
   and a perfectly-secure protocol $\PiPerRecShare$ to reconstruct $s$. 
   The protocols incur a communication of $\Order(n^2 \log{|\F|})$
   and $\Order(|\AdvStructure| \cdot n^2 \log{|\F|})$ bits respectively (see Appendix \ref{app:Perfect} for the details).

%% file: perfect.tex
\section{Perfectly-Secure Pre-Processing Phase Protocol with $\Q^{(4)}(\PartySet, \AdvStructure)$ Condition}
  \label{sec:perfectoverview}
 Throughout this section, we assume that $\AdvStructure$ satisfies the 
  $\Q^{(4)}(\PartySet, \AdvStructure)$ condition. We present a perfectly-secure protocol which generates a secret-sharing of $M$ random multiplication-triples, 
  unknown to the adversary. 
    The protocol realizes the
   ideal functionality $\FTriples$ (Fig \ref{fig:FTriples}) which allows the ideal-world adversary to specify the shares
    for each of the output triples on the behalf of corrupt parties. 
    The functionality then ``completes" the sharings of all the triples randomly, while keeping them ``consistent" with the shares specified by the 
     adversary.\footnote{This provision is made because 
     in our pre-processing phase protocol, the real-world adversary will have full control over the shares of the corrupt parties corresponding to the
    random multiplication-triples generated in the protocol.}
    \vspace*{-0.2cm}
   \begin{systembox}{$\FTriples$}{The ideal functionality for the asynchronous pre-processing phase with session id $\sid$.}{fig:FTriples}
	\justify
$\FTriples$ proceeds as follows, running with the parties $\PartySet$ and an adversary $\Sim$,
 and is parametrized by an adversary-structure $\AdvStructure$ and
  {\it $\AdvStructure$-private} sharing specification  $\ShareSpec = (S_1, \ldots, S_h) = \{ \PartySet \setminus Z | Z \in \AdvStructure\}$.
   Let $Z^{\star}$ denote the set of corrupt parties. 
    \begin{myitemize}
       \item[--] If there exists a set of parties ${\cal A}$ such that $\PartySet \setminus {\cal A} \in \AdvStructure$ and every 
       $P_i \in {\cal A}$ has sent the message $(\triples, \sid, P_i)$, then send $(\triples, \sid, {\cal A})$ to $\Sim$
       and prepare the output as follows.
         \begin{myitemize}
         \item Generate secret-sharing of $M$ random multiplication-triples. To generate
          one such sharing, randomly select $a, b \in \F$, compute $c = ab$ and execute the steps labelled
           {\bf Single Sharing Generation} 
           for $a, b$ and $c$.
             \item Let $\{([a^{(\ell)}], [b^{(\ell)}], [c^{(\ell)}] )\}_{\ell \in \{1, \ldots, M\}}$ be the resultant secret-sharing of the 
           multiplication-triples.
           Send a request-based delayed output
        $(\tripleshares, \sid, \{[a^{(\ell)}]_q, [b^{(\ell)}]_q, [c^{(\ell)}]_q \}_{\ell \in \{1, \ldots, M\}, P_i \in S_q})$  to each $P_i \in \PartySet \setminus Z^{\star}$
        (no need to send the respective shares to the parties in $Z^{\star}$, as $\Sim$ already has the shares of all the corrupt parties).           
    \end{myitemize}
{\bf Single Sharing Generation}: Do the following to generate a secret-sharing  of a given value $s$.
       \begin{myitemize}
         \item Upon receiving $(\shares, \sid, \{s_q \}_{S_q \cap Z^{\star} \neq \emptyset})$ from $\Sim$, 
         randomly select $s_q \in \F$ corresponding to each $S_q \in \ShareSpec$ for which $S_q \cap Z^{\star} = \emptyset$, such that
         $\displaystyle \sum_{S_q \cap Z^{\star} \neq \emptyset} s_q + \sum_{S_q \cap Z^{\star} = \emptyset} s_q = s$ holds.         
         \footnote{$\Sim$ {\it cannot}
         delay sending the shares on the behalf of the corrupt parties indefinitely as, in our real-world protocol,
         the adversary {\it cannot} indefinitely delay the generation of secret-shared multiplication-triples.}
         For $q = 1, \ldots, h$, set $[s]_q = s_q$.
       \end{myitemize}    
    \end{myitemize}
  \end{systembox}
\input{PerfMult}

%% file: PerfMult.tex
 We now present a protocol for securely realizing $\FTriples$. 
  To design the protocol, we need a multiplication protocol which takes as input 
   $\{([a^{(\ell)}], [b^{(\ell)}]) \}_{\ell = 1, \ldots, M}$
 and securely generates $\{ [c^{(\ell)}]\}_{\ell = 1, \ldots, M}$, where $c^{(\ell)} = a^{(\ell)} b^{(\ell)}$, without revealing
  any additional information about $a^{(\ell)}$ and $b^{(\ell)}$.
   For simplicity, we first explain and present
  the protocol assuming $M = 1$, where the inputs are $[a]$ and $[b]$, and the goal is to securely generate 
  a random sharing  $[ab]$. 
  
  Our starting point is the {\it synchronous} multiplication protocol of \cite{HT13,Mau02}. Note that $ab =  \sum_{(p, q) \in \ShareSpec \times \ShareSpec} [a]_p[b]_q$.
   The main idea is that since $S_p \cap S_q \neq \emptyset$, a {\it publicly-known} party from $S_p \cap S_q$
   can be designated to act as a dealer and generate a random sharing of the summand $[a]_p [b]_q$. 
   For efficiency, every designated ``summand-sharing party" can sum up all the summands assigned to it
    and generate a random sharing of the sum instead.
   If {\it no} summand-sharing party behaves {\it maliciously}, 
    then the sum of all secret-shared sums leads to a secret-sharing
   of $ab$.
   
   To deal with maliciously-corrupt summand-sharing parties, \cite{HT13} first designed
   an {\it optimistic} multiplication protocol $\OptMult$, which takes an additional parameter $Z \in \AdvStruct$ and generates a secret-sharing of
    $ab$, {\it provided} $\Adv$ corrupts a set of parties $Z^{\star} \subseteq Z$.
   The idea used in $\OptMult$ is the same as above, except that the summand-sharing parties are now restricted to the subset 
   $\PartySet \setminus Z$. 
       Since the parties will {\it not} be knowing the identity of corrupt parties in $Z^{\star}$, they run
   $\OptMult$ once for each $Z \in \AdvStruct$. This guarantees that at least one of these instances generates a secret-sharing of $ab$.
   By comparing the output sharings generated in all the instances of $\OptMult$, the parties can detect whether any cheating has  occurred. 
   If no cheating is detected, then any of the output sharings can serve as the sharing of $ab$. Else, the parties consider a pair of {\it conflicting} 
   $\OptMult$ instances (whose resultant output sharings are different) and proceed to a {\it cheater-identification} phase. In this phase,
    based on the values shared by the summand-sharing parties in the conflicting $\OptMult$ instances, the parties
   identify at least one corrupt summand-sharing party. This phase {\it necessarily} requires the participation of 
   {\it all} the summand-sharing parties from the conflicting $\OptMult$ instances. Once a corrupt summand-sharing
   party is identified, the parties disregard all output sharings of $\OptMult$ instances involving that party. This process of comparing the output sharings of $\OptMult$ instances 
    and identifying corrupt parties continues, until all the remaining output sharings are for the same value. 
   \vspace*{-0.4cm}
   \paragraph{\bf Challenges in the Asynchronous Setting:} There are two main non-trivial challenges while applying the above ideas in an {\it asynchronous}
   setting. First, in $\OptMult$, a
    potentially {\it corrupt} party may {\it never} share the sum of the summands designated to that party, leading to an indefinite wait.
    To deal with this, we notice that since $\AdvStructure$ satisfies the $\Q^{(4)}(\PartySet, \AdvStructure)$ condition, 
    each $(S_p  \cap S_q) \setminus Z$ contains at least one {\it honest} party. So instead of designating a {\it single} party
    for the summand $[a]_p[b]_q$, {\it each} party in $\PartySet \setminus Z$ shares the sum of {\it all} the summands it is ``capable" of, thus guaranteeing that each
    $[a]_p[b]_q$ is considered for sharing by at least one (honest) party.
        However, care has to be taken to ensure that any summand $[a]_p[b]_q$ is {\it not} shared multiple times (more on this later).  
    
    The second challenge  is that once the parties identify a pair of conflicting $\OptMult$ instances, the potentially {\it corrupt} summand-sharing
     parties from these instances {\it may not} participate in the cheater-identification phase, thus causing the parties to wait indefinitely. To get around this 
     problem, the multiplication protocol proceeds in {\it iterations}, where in each iteration, the parties run an instance of the {\it asynchronous} $\OptMult$ (outlined above)
     for each $Z \in \AdvStructure$, compare the outputs from each instance, and then proceed to the respective cheater-identification phase if the outputs are not the same.
     However, the summand-sharing parties from previous iterations are {\it not} allowed  
     to participate in future iterations until they participate
     in the cheater-identification phase of all the previous iterations. 
     This prevents the {\it corrupt} summand-sharing parties in previous iterations from acting as summand-sharing parties in future iterations until they clear their ``pending tasks", 
     in which case 
     they are caught and discarded for ever.
      We stress that the {\it honest} parties are eventually ``released" to act as summand-sharing parties
     in future iterations. Thus, even if the corrupt summand-sharing parties from previous iterations are ``stuck" forever,
     the parties eventually progress to the next iteration in case the current iteration ``fails". 
      Once the parties reach an iteration where the outputs of all the $\OptMult$ instances are the same, the protocol stops.
     We show that there will be at most $t[tn + 1] + 1$ iterations, where $t$ is the cardinality of the maximum-sized subset in $\AdvStructure$.
     
     Based on the above discussion, we next present protocols $\OptMult, \MultLCE$ and $\MultGCE$. 
     Protocol $\MultLCE$ represents an iteration where the parties run an instance of $\OptMult$ for each $Z \in \AdvStructure$
     and execute a cheater-identification phase if the iteration fails.
     Protocol $\MultGCE$ calls the protocol $\MultLCE$ multiple times, till it reaches a ``successful" instance of
     $\MultLCE$ (where the outputs of all the instances of $\OptMult$ are
      the same).  
     In these protocols, the parties maintain the following {\it dynamic} sets:
     {\bf (a)} $ \Waitlist^{(i)}_{\iter}$: Denotes the {\it wait-listed} parties maintained by $P_i$, corresponding to instance number $\iter$ of $\MultLCE$ in $\MultGCE$;
     {\bf (b)} $\LocalDiscarded^{(i)}_{\iter}$: Denotes the set of parties {\it locally discarded} by $P_i$ during the 
     cheater-identification phase of instance number $\iter$ of $\MultLCE$ in $\MultGCE$; and
     {\bf (c)} $\Discarded$: Denotes the set of parties, {\it globally discarded} by {\it all} (honest) parties across various instances of
     $\MultLCE$ in $\MultGCE$.\footnote{The reason for two different discarded sets is that the various instances of
     cheater-identification corresponding to the failed $\MultLCE$ instances are executed {\it asynchronously}, thus resulting
     in a corrupt party to be identified by different
     honest parties during different iterations.}
     These sets will be maintained such that no honest party is ever included in the sets $\Discarded$ and
     $\LocalDiscarded^{(i)}_{\iter}$ of any honest $P_i$. Moreover, any honest party which is included in 
     $ \Waitlist^{(i)}_{\iter}$ set of any honest $P_i$ is eventually removed from $ \Waitlist^{(i)}_{\iter}$.      
\subsection{Optimistic Multiplication Protocol}
Protocol $\OptMult$ is executed with respect to a given
 $Z \in \AdvStructure$ and iteration number $\iter$.
  Each party in $\PartySet \setminus Z$ tries to act as a summand-sharing party and shares the sum of all the summands
    it is ``capable" of. To avoid ``repetition" of summands, the parties select {\it distinct} summand-sharing parties in hops
    and ``mark" the summands whose sum is shared
    by the selected summand-sharing party in a hop, ensuring that they are not considered 
    in future hops.
       To agree on the summand-sharing party of each hop, the parties execute an instance of
   the {\it agreement on common subset} (ACS) primitive \cite{BCG93}, where one instance of ABA is invoked on the behalf of each
   candidate summand-sharing party. While voting for a candidate party in $\PartySet \setminus Z$ during a hop, 
   the parties ensure that 
   the candidate has indeed secret-shared some sum, and that  
    it {\it was not} 1) selected in an earlier hop; 2) in the waiting list or the list of locally-discarded parties of any previous iteration; 3) in the list of globally-discarded parties. 
\begin{protocolsplitbox}{$\OptMult(\PartySet,\AdvStruct,\SharingSpec,[a],[b],Z,\iter)$}{Optimistic multiplication in $(\FVSS, \FABA)$-hybrid
  for iteration $\iter$ and session id $\sid$,
  assuming $Z$ to be corrupt. The above code is executed by each $P_i$, who implicitly uses the dynamic sets $\Discarded$,
  $ \Waitlist^{(i)}_{\iter'}$ and  $\LocalDiscarded^{(i)}_{\iter'}$ for $\iter' < \iter$
  }{fig:OptMult}
\justify
\begin{myitemize}
\item[--] \textbf{Initialization} 
	\begin{myitemize}
	\item Initialize the set of ordered pair of indices of {\it all} summands : $\Products_{(Z, \iter)} = \{(p, q)\}_{p, q = 1, \ldots, |\SharingSpec|}$.
	\item Initialize the summand indices corresponding to $P_j \in \PartySet \setminus Z$ : 
	$\Products^{(j)}_{(Z, \iter)} = \{(p, q)\}_{P_j \in S_p \cap S_q}$.
	\item Initialize the set of summands-sharing parties : $\Selected_{(Z, \iter)} = \emptyset$.
	 Initialize the hop number $\hop = 1$.
	\end{myitemize}
	\item[--] Do the following till $\Products_{(Z, \iter)} \neq \emptyset$:
	\begin{myitemize}
	\item \textbf{Sharing Summands}: 
	    \begin{myenumerate}
	    \item If $P_i \notin Z$ and $P_i \notin \Selected_{(Z, \iter)}$, then compute $\displaystyle c^{(i)}_{(Z, \iter)} = \sum_{(p, q) \in \Products^{(i)}_{(Z, \iter)}} [a]_p[b]_q$.
	    Randomly select the shares $c^{(i)}_{{(Z, \iter)}_1}, \ldots, c^{(i)}_{{(Z, \iter)}_h}$, such that $c^{(i)}_{{(Z, \iter)}_1} + \ldots + c^{(i)}_{{(Z, \iter)}_h} = c^{(i)}_{(Z, \iter)}$. 
	    Call $\FVSS$ with message $(\Dealer, \sid_{\hop, i, \iter, Z}, (c^{(i)}_{{(Z, \iter)}_1}, \ldots, c^{(i)}_{{(Z, \iter)}_h}))$, where
	     $\sid_{\hop, i, \iter, Z} = \hop || \sid || i || \iter || Z$.\footnote{The notation $\sid_{\hop, i, \iter, Z}$ is used to distinguish among the different calls to 
	     $\FVSS$ and $\FABA$ within each hop.}
	    \item Keep requesting for an output from $\FVSS$ with $\sid_{\hop, j, \iter, Z}$, for $j = 1, \ldots, n$, till an output is received.
	    \end{myenumerate}
	\item \textbf{Selecting Summand-Sharing Party Through ACS}:    
	     \begin{myenumerate}
	       \item For $j = 1, \ldots, n$, send $(\vote, \sid_{\hop, j, \iter, Z}, 1)$ to $\FABA$, if {\it all} the following conditions hold:
	           \begin{myitemize}
		\item[--] $P_j \notin \Discarded$, $P_j \notin Z$ and 
		$P_j \notin \Selected_{(Z, \iter)}$.
		Moreover, $\forall \iter' < \iter,  P_j \notin \Waitlist^{(i)}_{\iter'}$ and $P_j \notin \LocalDiscarded^{(i)}_{\iter'}$.
		\item[--] An output $(\Share, \sid_{\hop, j, \iter, Z}, P_j, \{ [c^{(j)}_{(Z, \iter)}]_q \}_{P_i \in S_q})$ is received from $\FVSS$,
		with 
		$\sid_{\hop, j, \iter, Z}$.
		\end{myitemize}
	     \item For $j = 1, \ldots, n$, request for an output from $\FABA$ with $\sid_{\hop, j, \iter, Z}$,  until an output is received.    
	     \item If $\exists P_j \in \PartySet$ such that $(\decide, \sid_{\hop, j, \iter, Z}, 1)$ is received from $\FABA$ with  $\sid_{\hop, j, \iter, Z}$, 
	     then for each $P_k \in \PartySet$ for which no $\vote$
              message has been sent yet, 
             send
              $(\vote, \sid_{\hop, k ,\iter, Z}, 0)$ to $\FABA$ with $\sid_{\hop, k, \iter, Z}$.
               \item Once an output $(\decide, \sid_{\hop, j, \iter, Z}, v_j)$ is received from $\FABA$ with $\sid_{\hop, j, \iter, Z}$ for all $j \in \{1, \ldots, n \}$, select
               the least indexed $P_j$, such that $v_j = 1$. Then set $\hop = \hop + 1$ and update the following.
                \begin{myitemize}
		\item[--] $\Selected_{(Z, \iter)} = \Selected_{(Z, \iter)} \cup \{P_j\}$.
		 $\Products_{(Z, \iter)} = \Products_{(Z, \iter)} \setminus \Products^{(j)}_{(Z, \iter)}$.
		\item[--] $\forall P_k \in \PartySet \setminus \{Z \cup \Selected_{(Z, \iter)}\}$:
		$\Products^{(k)}_{(Z, \iter)} = \Products^{(k)}_{(Z, \iter)} \setminus \Products^{(j)}_{(Z, \iter)}$.
		\end{myitemize}
	     \end{myenumerate}
	\end{myitemize}
\item[--] $\forall P_j \in \PartySet \setminus \Selected_{(Z, \iter)}$, participate in an instance of $\PiPerDefaultShare$ with public input 
  $c^{(j)}_{(Z, \iter)} = 0$.
\item[--] \textbf{Output} : Let $c_{(Z, \iter)} \defined c^{(1)}_{(Z, \iter)} + \ldots + c^{(n)}_{(Z, \iter)}$. Output
$\{[c^{(1)}_{(Z, \iter)}]_q , \ldots, [c^{(n)}_{(Z, \iter)}]_q, [c_{(Z, \iter)}]_q \}_{P_i \in S_q}$.
\end{myitemize}
\end{protocolsplitbox}
Lemma \ref{lemma:OptMult} is proven in Appendix \ref{app:Perfect}. To handle $M$ pairs of inputs in $\OptMult$, 
  in each hop, every $P_i$ 
  calls $\FVSS$ $M$ times to share $M$ summations. 
  While voting for a candidate summand-sharing party in a hop, the parties check whether it has shared $M$ values. Hence, 
  there will be $\Order(n^2 M)$ calls to $\FVSS$, but {\it only} $\Order(n^2)$ calls to $\FABA$. 
\begin{lemma}
\label{lemma:OptMult}
Let  $\AdvStructure$ satisfy the $\Q^{(4)}(\PartySet, \AdvStructure)$ condition and 
 $\ShareSpec = \{ \PartySet \setminus Z | Z \in \AdvStructure\}$. 
 Consider an arbitrary $Z \in \AdvStructure$ and $\iter$, 
  such that all honest parties participate in the instance $\OptMult(\PartySet,\AdvStruct,\SharingSpec,[a],[b], Z,\iter)$. Then all
  honest parties eventually compute  $[c_{(Z, \iter)}], [c^{(1)}_{(Z, \iter)}], \ldots, [c^{(n)}_{(Z, \iter)}]$
  where $c_{(Z, \iter)} = c^{(1)}_{(Z, \iter)} + \ldots + c^{(n)}_{(Z, \iter)}$, 
  provided no honest party is included in the $\Discarded$ and $\LocalDiscarded^{(i)}_{\iter'}$ sets 
  and each honest party in the $ \Waitlist^{(i)}_{\iter'}$ sets of every honest $P_i$ is eventually removed, 
   for all
  $\iter' < \iter$. If no party in $\PartySet \setminus Z$ acts maliciously, then 
     $c_{(Z, \iter)} = ab$. In the protocol, $\Adv$ does not learn anything additional about $a$ and $b$.
          The protocol makes $\Order(n^2)$ calls to $\FVSS$ and $\FABA$.
\end{lemma}
\subsection{Multiplication Protocol with Cheater Identification}
Protocol $\MultLCE$ with cheater identification (Fig \ref{fig:MultLCE}) takes as inputs an iteration number
 $\iter$ and $([a], [b])$. If {\it no} party behaves maliciously, then the protocol securely outputs $[ab]$.  
  In the protocol, parties execute an instance of $\OptMult$ for each $Z \in \AdvStructure$ and compare the outputs. 
   Since at least one of the $\OptMult$ instances is guaranteed to output $[ab]$, if all the outputs are same, then no cheating has occurred.
   Otherwise, the parties identify a pair of conflicting
   $\OptMult$ instances with different outputs, executed with respect to $Z$ and $Z'$. Let $\Selected_{(Z, \iter)}$ and 
   $\Selected_{(Z', \iter)}$ be the summand-sharing parties in the conflicting $\OptMult$ instances.
   The parties next proceed to a cheater-identification phase to identify at least one corrupt party in $\Selected_{(Z, \iter)} \cup \Selected_{(Z', \iter)}$. 
   
   Each $P_j \in \Selected_{(Z, \iter)}$ is made to share the sum of the summands from its summand-list 
   overlapping
   with the summand-list of each $P_k \in \Selected_{(Z', \iter)}$ and vice-versa.   
    Next, these ``partitions" are compared, 
    based on which
   at least one corrupt party in $\Selected_{(Z, \iter)} \cup \Selected_{(Z', \iter)}$ is guaranteed to be identified provided {\it all} the parties in 
   $\Selected_{(Z, \iter)} \cup \Selected_{(Z', \iter)}$ secret-share the required partitions.
   The cheater-identification phase will be ``stuck" if the {\it corrupt} parties in  $\Selected_{(Z, \iter)} \cup \Selected_{(Z', \iter)}$
    do not participate. To prevent such corrupt parties from causing future instances of
    $\MultLCE$ to fail, the parties wait-list all the parties in $\Selected_{(Z, \iter)} \cup \Selected_{(Z', \iter)}$.  A party
     is then ``released" only after it has shared all the required values as part of the cheater-identification phase.
    Every honest party is eventually released from the waiting-list. 
    This wait-listing guarantees that corrupt parties will be barred from acting as summand-sharing parties
     as part of the $\OptMult$ instances
    of future invocations of $\MultLCE$, until they participate in the cheater-identification phase of previous failed instances of $\MultLCE$.
     Since the cheater-identification phase is executed asynchronously, 
     each party maintains its own set of {\it locally-discarded} parties, where corrupt parties are included as and when they are identified.  
\begin{protocolsplitbox}{$\MultLCE(\PartySet, \AdvStruct, \SharingSpec, [a], [b], \iter)$}{Code for $P_i$ for multiplication with cheater identification
 for iteration $\iter$ and session id $\sid$, in the $\FVSS$-hybrid}{fig:MultLCE}
 \justify
\begin{myitemize}
\item[--] \textbf{Initialization}: Initialize $\Waitlist^{(i)}_{\iter} = \LocalDiscarded^{(i)}_{\iter} = \emptyset$ and 
  $\flag^{(i)}_{\iter} = \bot$. Fix some (publicly-known) $Z' \in \AdvStruct$.
  \item[--] \textbf{Running Optimistic Multiplication and Checking Pair-wise Differences}:
         \begin{myitemize}
         \item For each $Z \in \AdvStruct$, participate in the instance 
          $\OptMult(\PartySet, \AdvStruct, \SharingSpec,[a], [b], Z, \iter)$ with session id $\sid$. Let
          $\{[c^{(1)}_{(Z, \iter)}]_q ,\allowbreak  \ldots, [c^{(n)}_{(Z, \iter)}]_q, [c_{(Z, \iter)}]_q \}_{P_i \in S_q}$ be the output
          obtained. Moreover, let $\Selected_{(Z, \iter)}$ be set of summand-sharing parties and for 
           each $P_j \in \Selected_{(Z, \iter)}$, let $\Products^{(j)}_{(Z, \iter)}$ be the set of ordered pairs of indices corresponding to the summands
            whose sum has been shared by $P_j$, during this
          instance of
          $\OptMult$.
          \item  Corresponding to
           every $Z \in \AdvStruct$, participate in an instance of $\PiPerRec$ to reconstruct 
            $c_{(Z, \iter)} - c_{(Z', \iter)}$.         
         \end{myitemize}
    \item[--] \textbf{Output in Case of Success}: If $\forall Z \in \AdvStruct$, $c_{(Z, \iter)} - c_{(Z', \iter)} = 0$, then 
    set $\flag^{(i)}_{\iter} = 0$ and output
     $\{[c_{(Z', \iter)}]_q\}_{P_i \in S_q}$.
     \item[--] \textbf{Waiting-List and Cheater Identification in Case of Failure}: If $\exists Z \in \AdvStruct:  c_{(Z, \iter)} - c_{(Z', \iter)} \neq 0$,  then
     let $Z$ be the first set such that $c_{(Z, \iter)} - c_{(Z', \iter)} \neq 0$. 
     Set the {\it conflicting-sets} to be $Z, Z'$, $\flag^{(i)}_{\iter} = 1$ and proceed as follows. 
          \begin{myitemize}
          \item \textbf{Wait-listing Parties}: Set $\Waitlist^{(i)}_{\iter} = \Selected_{(Z, \iter)} \cup \Selected_{(Z', \iter)}$.
          \item \textbf{Sharing Partition of the Summand-Sums}:
            \begin{myenumerate}
               \item If $P_i \in \Selected_{(Z, \iter)}$, compute 
               $\displaystyle d^{(ij)}_{(Z, \iter)} = \sum_{(p, q) \in \Products^{(i)}_{(Z, \iter)} \cap \Products^{(j)}_{(Z', \iter)}} [a]_p [b]_q$, for
               every $ P_j \in \Selected_{(Z', \iter)}$. 
                  Randomly pick 
                  ${d^{(ij)}_{(Z, \iter)}}_1, \ldots, {d^{(ij)}_{(Z, \iter)}}_h$ such that ${d^{(ij)}_{(Z, \iter)}}_1 + \ldots + {d^{(ij)}_{(Z, \iter)}}_h = d^{(ij)}_{(Z, \iter)}$.   
                  Send $(\Dealer, \sid_{i, j, \iter, Z}, ({d^{(ij)}_{(Z, \iter)}}_1, \ldots, {d^{(ij)}_{(Z, \iter)}}_h)$ to
                   $\FVSS$, where $\sid_{i, j, \iter, Z} = \sid || i || j || \iter || Z$.
                \item If $P_i \in \Selected_{(Z', \iter)}$, compute 
               $\displaystyle e^{(ij)}_{(Z', \iter)} = \sum_{(p, q) \in \Products^{(i)}_{(Z', \iter)} \cap \Products^{(j)}_{(Z, \iter)}} [a]_p [b]_q$, for
               all  $P_j \in \Selected_{(Z, \iter)}$.
                  Randomly pick ${e^{(ij)}_{(Z', \iter)}}_1, \ldots, {e^{(ij)}_{(Z', \iter)}}_h$ which sum up to $e^{(ij)}_{(Z', \iter)}$ and then
                  send $(\Dealer, \allowbreak  \sid_{i, j, \iter, Z'}, ({e^{(ij)}_{(Z', \iter)}}_1, \ldots, {e^{(ij)}_{(Z', \iter)}}_h)$ to $\FVSS$, where $\sid_{i, j, \iter, Z'} = \sid || i || j || \iter || Z'$.
               \item Corresponding to every $P_j \in \Selected_{(Z, \iter)}$ and every $P_k \in \Selected_{(Z', \iter)}$, 
               keep requesting for an output from $\FVSS$ with session id $\sid_{j, k, \iter, Z}$, till an output is obtained.                            
               \item Corresponding to every $P_j \in \Selected_{(Z', \iter)}$ and every $P_k \in \Selected_{(Z, \iter)}$, 
               keep requesting for an output from $\FVSS$ with session id  $\sid_{j, k, \iter, Z'}$, till an output is obtained.                            
            \end{myenumerate}
          \item \textbf{Removing Parties from Wait List}: Set $\Waitlist^{(i)}_{\iter} = \Waitlist^{(i)}_{\iter} \setminus \{P_j \}$, if
          {\it all} the following criteria pertaining to $P_j$ hold:
                \begin{myenumerate}
		\item \textbf{$P_j \in \Selected_{(Z, \iter)}$} : if an output $(\Share, \sid_{j, k, \iter, Z}, P_j, \allowbreak \{ [d^{(jk)}_{(Z, \iter)}]_q \}_{P_i \in S_q})$ is
		 received from $\FVSS$ with session id
		$\sid_{j, k, \iter, Z}$, corresponding to each $P_k \in \Selected_{(Z', \iter)}$,
		\item \textbf{$P_j \in \Selected_{(Z', \iter)}$} : if an output $(\Share, \sid_{j, k,  \iter, Z'}, P_j, \allowbreak \{ [e^{(jk)}_{(Z', \iter)}]_q \}_{P_i \in S_q})$ is
		 received from $\FVSS$ with session id
		$\sid_{j, k, \iter, Z'}$, corresponding to 
	        every $P_k \in \Selected_{(Z, \iter)}$.
		\end{myenumerate}
            \item \textbf{Verifying the Summand-Sum Partitions and Locally Identifying Corrupt Parties}:   
                 \begin{myenumerate}
                 \item For every $P_j \in \Selected_{(Z, \iter)}$, participate in an instance of $\PiPerRec$  to reconstruct the difference value
                 $ c^{(j)}_{(Z, \iter)} - \sum_{P_k \in \Selected_{(Z', \iter)}} d^{(jk)}_{(Z, \iter)}$. 
                   If the difference is not $0$, then set
                    $\LocalDiscarded^{(i)}_{\iter} = \LocalDiscarded^{(i)}_{\iter} \cup \{P_j \}$.                   
                 \item For every $P_j \in \Selected_{(Z', \iter)}$, participate in an instance of $\PiPerRec$ to reconstruct the difference value
                $ c^{(j)}_{(Z', \iter)} - \sum_{P_k \in \Selected_{(Z, \iter)}} e^{(jk)}_{(Z', \iter)}$. 
                   If the difference is not $0$, then set
                    $\LocalDiscarded^{(i)}_{\iter} = \LocalDiscarded^{(i)}_{\iter} \cup \{P_j \}$. 
                 \item For each ordered pair $(P_j, P_k)$ where $P_j \in \Selected_{(Z, \iter)}$ and $P_k \in \Selected_{(Z', \iter)}$, 
                 participate in an instance of $\PiPerRec$  to reconstruct  $d^{(jk)}_{(Z, \iter)} - e^{(kj)}_{(Z', \iter)}$. 
                 If the value is not $0$, then do the following:
                     \begin{myenumerate}
                       \item[i.] Participate in instances of $\PiPerRec$ to reconstruct 
                     $d^{(jk)}_{(Z, \iter)}$ and $e^{(kj)}_{(Z', \iter)}$.
                        Participate in instances of $\PiPerRecShare$ to reconstruct
                      $[a]_p$ and $[b]_q$, such that $(p, q) \in \Products^{(j)}_{(Z, \iter)} \cap \Products^{(k)}_{(Z', \iter)}$.
                        \item[ii.] Compare $ \sum_{(p, q) \in \Products^{(j)}_{(Z, \iter)} \cap \Products^{(k)}_{(Z', \iter)}} [a]_p [b]_q$ with 
                        $d^{(jk)}_{(Z, \iter)}$ and $e^{(kj)}_{(Z', \iter)}$ and identify the corrupt party $P_c \in \{P_j, P_k \}$. 
                        Set $\LocalDiscarded^{(i)}_{\iter} = \LocalDiscarded^{(i)}_{\iter} \cup \{P_c \}$. 
                     \end{myenumerate}                 
                 \end{myenumerate}                
          \end{myitemize}       
  \end{myitemize}
  \end{protocolsplitbox}
  Lemma \ref{lemma:MultLCE} is proved in Appendix \ref{app:Perfect}.
  In the Lemma, we say that an instance of $\MultLCE$ is {\it successful}, if $c_{(Z, \curr)} - c_{(Z', \curr)} = 0$ for all
  $Z \in \AdvStructure$ with respect to the publicly-known $Z' \in \AdvStructure$ fixed in the protocol, else the instance {\it fails}.   
  The modifications to $\MultLCE$ for handling $M$ pairs of inputs are simple
  (see Appendix \ref{app:Perfect}), requiring
   $\Order(M \cdot |\AdvStructure| \cdot n^2)$ calls to $\FVSS$, $\Order(|\AdvStructure| \cdot n^2)$ calls to $\FABA$
   and a communication of $\Order((M \cdot |\AdvStructure|^2 \cdot n^2 + |\AdvStructure| \cdot n^4) \log{|\F|})$ bits. 
 \begin{lemma}
 \label{lemma:MultLCE}
 Let $\AdvStructure$ satisfy the $\Q^{(4)}(\PartySet, \AdvStructure)$ condition and let 
  all honest parties participate in $\MultLCE(\PartySet, \allowbreak \AdvStruct, \SharingSpec, [a], \allowbreak [b], \iter)$.
   Then, $\Adv$ does not learn any additional information about $a$ and $b$. Moreover, the following hold.
 \begin{myitemize}
 \item[--] The instance  will eventually be deemed to succeed or fail by the honest parties, where for a successful
  instance, the parties output a sharing of $ab$.
 \item[--] If the instance is not successful,
   then the honest parties will agree on a pair $Z, Z' \in \AdvStructure$ such that 
  $c_{(Z, \iter)} - c_{(Z', \iter)} \neq 0$. Moreover, all honest parties present in the 
  $ \Waitlist^{(i)}_{\iter}$ set of any honest party $P_i$ will eventually be removed and no honest party
    is ever included in the $\LocalDiscarded^{(i)}_{\iter}$ set
   of any honest $P_i$.
   Furthermore, there will be a pair of parties $P_j, P_k$ from
  $\Selected_{(Z, \iter)} \cup \Selected_{(Z', \iter)}$, with at least one of them being maliciously-corrupt, such that if both
  $P_j$ and $P_k$ are removed from the set $\Waitlist^{(h)}_{\iter}$ of any honest party $P_h$, then eventually
   the corrupt party(ies) among $P_j, P_k$ will be included
        in the set $\LocalDiscarded^{(i)}_{\iter}$ of every honest $P_i$.        
      \item[--] The protocol needs $\Order(|\AdvStructure| n^2)$ calls to $\FVSS$ and $\FABA$ and communicates
       $\Order((|\AdvStructure|^2  n^2  + |\AdvStructure| n^4) \log{|\F|})$ bits. 
 \end{myitemize}
 \end{lemma}
  \subsection{Multiplication Protocol}
Protocol $\MultGCE$ (Fig \ref{fig:MultGCE}) takes $([a], [b])$
 and securely generates $[ab]$.
  The protocol proceeds in iterations, where in each iteration, an instance of $\MultLCE$ is invoked. 
   If the iteration is successful, then the parties take the output of the corresponding $\MultLCE$ instance. Else, they proceed to the next iteration, with the cheater-identification phase of failed $\MultLCE$ instances running in the background. Let $t$ be the cardinality of maximum-sized subset from $\AdvStructure$.  To upper bound the
   number of failed iterations, the parties run ACS  after every $tn + 1$ failed iterations
   to ``globally" include a new corrupt party in $\Discarded$.   
      This is done through calls to $\FABA$,
   where the parties vote for a candidate corrupt party, based on the $\LocalDiscarded$ sets of {\it all} failed iterations.
   The idea is that during these $tn + 1$ failed iterations, there will be at least one {\it corrupt} party who is eventually included
   in the $\LocalDiscarded$ set of {\it every} honest party. This is because there can be at most 
   $tn$ distinct pairs of ``conflicting-parties" across the $tn + 1$ failed iterations (follows from Lemma \ref{lemma:MultLCE}).
    At least
   one conflicting pair, say $(P_j, P_k)$, is guaranteed to repeat among the $tn + 1$ failed
   instances, with {\it both} $P_j$ and $P_k$ being removed from the previous waiting-lists.
    Thus, the corrupt party(ies) among $P_j, P_k$ is eventually included to the
   $\LocalDiscarded$ sets. 
   There can be at most $t(tn + 1)$ failed iterations after which {\it all} the corrupt parties will be discarded
   and the next iteration is guaranteed to be successful, with only {\it honest} parties acting as the candidate summand-sharing parties in the underlying instances of $\OptMult$.  
\begin{protocolsplitbox}{$\MultGCE(\PartySet, \AdvStruct, \SharingSpec, [a], [b])$}{Multiplication protocol 
 in the $(\FVSS, \FABA)$-hybrid for $\sid$.
 The above code is executed by every party $P_i$}{fig:MultGCE}
 \justify
 \begin{myitemize}
	\item[--] \textbf{Initialization}: Set $t = \max\{ |Z| :  Z \in \AdvStruct 	\}$, initialize $\Discarded = \emptyset$ and $\iter = 1$.
     \item[--] \textbf{Multiplication with Cheater Identification}: Participate in the instance $\MultLCE(\PartySet, \AdvStruct, \SharingSpec,  [a], [b], \iter)$
     with $\sid$.
          \begin{myitemize}
          \item \textbf{Positive Output}: If $\flag^{(i)}_{\iter}$ is set to $0$, then output the shares obtained
          during the $\MultLCE$ instance.
          \item \textbf{Negative Output}: If $\flag^{(i)}_{\iter}$ is set to $1$ during the
           $\MultLCE$ instance, then proceed as follows.
              \begin{myitemize}
              \item \textbf{Identifying a Cheater Party Through ACS}: If $\iter = k \cdot [tn + 1]$ for some  $k \geq 1$, then do the following.
                  \begin{myenumerate}
                      \item Let $\LocalDiscarded^{(i)}_{r}$ be the set of locally-discarded parties for the
                       instance $\MultLCE(\PartySet, \AdvStruct, \SharingSpec,  [a], [b], r)$, for $r = 1, \ldots, \iter$.
                      For $j = 1, \ldots, n$, send $(\vote, \sid_{j, \iter, k}, 1)$ to $\FABA$ where $\sid_{j, \iter, k} = \sid || j || \iter || k$, if for any
                       $r \in \{1, \ldots, \iter \}$, party $P_j$
                       is present in $\LocalDiscarded^{(i)}_{r}$ and $P_j \not \in \Discarded$. 
                        \item For $j = 1, \ldots, n$, keep requesting for an output from $\FABA$ with $\sid_{j, \iter, k}$,  until an output is received.    
		     \item If  $\exists P_j \in \PartySet$ such that
		      $(\decide, \sid_{j, \iter, k}, 1)$ is received from $\FABA$ with  $\sid_{j, \iter, k}$, 
	    		  then for each $P_{\ell} \in \PartySet$, for which no $\vote$
                           message has been sent yet, send
                 $(\vote, \sid_{\ell ,\iter, k}, 0)$ to $\FABA$ with $\sid_{\ell, \iter, k}$.
               \item Once an output $(\decide, \sid_{\ell, \iter, k}, v_{\ell})$ is received from $\FABA$ with $\sid_{\ell, \iter, k}$ for every $\ell \in \{1, \ldots, n \}$, select
               the minimum indexed party $P_j$ from $\PartySet$, such that $v_j = 1$. Then set $\Discarded = \Discarded \cup \{P_j \}$, set 
               $\iter = \iter + 1$ and go to the step labelled \textbf{Multiplication with Cheater Identification}.
                 \end{myenumerate}
              \item Else set $\iter = \iter + 1$ and go to the step  \textbf{Multiplication with Cheater Identification}.
              \end{myitemize}
          
          \end{myitemize}
 
\end{myitemize}
\end{protocolsplitbox}

Lemma \ref{lemma:MultGCE} is proved in Appendix \ref{app:Perfect}. To handle $M$ pairs of inputs,
  the instances of $\MultLCE$ are now executed with $M$ pairs of inputs in each iteration.
    This requires  $\Order(M \cdot |\AdvStructure| \cdot n^5)$ calls to $\FVSS$, $\Order(|\AdvStructure| \cdot n^5)$ calls to $\FABA$
   and a communication of 
    $\Order((M \cdot |\AdvStructure|^2 \cdot n^5 +  |\AdvStructure| \cdot n^7) \log{|\F|})$ bits. 
   \begin{lemma}
   \label{lemma:MultGCE}
    Let $\AdvStructure$ satisfy the $\Q^{(4)}(\PartySet, \AdvStructure)$ condition and let 
    $\ShareSpec = (S_1, \ldots, S_h) = \{ \PartySet \setminus Z | Z \in \AdvStructure\}$.
    Then $\MultGCE$ takes at most $t(tn + 1)$ iterations and all honest parties eventually output a secret-sharing of $[ab]$, where
    $t = \max\{ |Z| :  Z \in \AdvStruct 	\}$.
    In the protocol, $\Adv$ does not learn anything additional about $a$ and $b$.
         The protocol makes $\Order(|\AdvStructure| \cdot n^5)$ calls to $\FVSS$ and $\FABA$ and additionally incurs a communication of
     $\Order(|\AdvStructure|^2 \cdot n^5 \log{|\F|} + |\AdvStructure| \cdot n^7 \log{|\F|})$ bits.
   \end{lemma}
      \subsection{The Pre-Processing Phase Protocol}
   The perfectly-secure pre-processing phase protocol $\PiPerTriples$ is standard. 
    The parties first jointly generate secret-sharing of $M$ random pairs of values, 
     followed by running an instance of $\MultGCE$ to securely compute
       the product of these pairs. Protocol $\MultGCE$ and the proof of Theorem \ref{thm:PiPerTriples} is provided in Appendix \ref{app:Perfect}.
  \begin{theorem}
  \label{thm:PiPerTriples}
  If $\AdvStructure$ satisfies the $\Q^{(4)}(\PartySet, \AdvStructure)$ condition,
  then  $\PiPerTriples$ is a perfectly-secure protocol for securely realizing $\FTriples$ with UC-security in the $(\FVSS, \FABA)$-hybrid model. 
   The protocol makes $\Order(M \cdot |\AdvStructure| \cdot n^5)$ calls to $\FVSS$, $\Order(|\AdvStructure| \cdot n^5)$ calls to $\FABA$
   and incurs a communication of $\Order(M \cdot |\AdvStructure|^2 \cdot n^5 \log{|\F|} + |\AdvStructure| \cdot n^7 \log{|\F|})$ bits. 
  \end{theorem}

%% file: stat.tex
\section{Statistically-Secure Pre-Processing Phase Protocol with $\Q^{(3)}(\PartySet, \AdvStructure)$ Condition}
\label{sec:statisticaloverview}
 We first present an {\it asynchronous information-checking protocol}
  (AICP) with $\Q^{(3)}(\PartySet, \AdvStructure)$ condition.
  \vspace*{-0.3cm}
\input{ICSig}

\input{statVSSv2}

 \input{statMPC}

%% file: ICSig.tex
\subsection{Asynchronous Information Checking Protocol (AICP)}
 An ICP \cite{RB89,CDDHR99} is used for authenticating data in the presence of a 
  {\it computationally-unbounded} adversary. An AICP \cite{CR93,PCR14} extends ICP for the {\it asynchronous} setting. 
  In an AICP, there are {\it four} entities, a {\it signer} $\mathsf{S} \in \PartySet$,
   an  {\it intermediary} $\mathsf{I} \in \PartySet$, a {\it receiver}
  $\mathsf{R} \in \PartySet$ and all the parties in $\PartySet$ acting as {\it verifiers}
   (note that $\mathsf{S}, \mathsf{I}$ and $\mathsf{R}$ also act as verifiers).
    An AICP has two sub-protocols, one for the {\it authentication phase} and one for 
    the {\it revelation phase}. 
    
    In the authentication phase, $\mathsf{S}$ has some private input
    $s \in \F$, which it distributes to $\mathsf{I}$ along with some {\it authentication information}.
    Each verifier is provided with some {\it verification information}, followed by the parties verifying whether $\mathsf{S}$ has distributed consistent information. 
    The data held by $\mathsf{I}$ at the end of this phase is called  {\it $\mathsf{S}$'s IC-Signature on $s$ for intermediary $\mathsf{I}$ and receiver $\mathsf{R}$}, denoted by
    $\ICSig(\mathsf{S}, \mathsf{I}, \mathsf{R}, s)$. 
    Later, during the revelation phase, $\mathsf{I}$ reveals $\ICSig(\mathsf{S}, \mathsf{I}, \mathsf{R}, s)$ to $\mathsf{R}$,
    who ``verifies" it with respect to the verification information provided by the verifiers and decides whether to accept or reject
    $s$. We require the same security guarantees from AICP as expected from digital signatures,
    namely {\it correctness}, {\it unforgeability} and {\it non-repudiation}.
    %
     Additionally, we will need the {\it privacy} property guaranteeing that if 
    $\mathsf{S, I}$ and $\mathsf{R}$ are {\it all} honest, then $\Adv$ does not learn $s$.
    
    Our AICP is a generalization of the AICP of \cite{PCR14}, which was designed against {\it threshold} adversaries. 
    During the authentication phase, $\mathsf{S}$ embeds $s$ in a random $t$-degree {\it signing-polynomial} $F(x)$, where 
    $t$ is the cardinality of maximum-sized subset in $\AdvStruct$, and gives $F(x)$ to $\mathsf{I}$. In addition, each verifier $P_i$ is given a random {\it verification-point}
    $(\alpha_i, v_i)$ on $F(x)$. 
     Later, during the revelation phase, $\mathsf{I}$ is supposed to reveal $F(x)$ to $\mathsf{R}$, while each verifier 
    $P_i$ is supposed to reveal their verification-points to $\mathsf{R}$, who accepts $F(x)$ if it
     is found to be consistent with ``sufficiently many" verification-points. The above idea achieves all the properties of 
    AICP, except the {\it non-repudiation} property, since a potentially {\it corrupt} $\mathsf{S}$ may distribute ``inconsistent" data to $\mathsf{I}$ and the verifiers.
    To deal with this, during the authentication phase, the parties interact in a ``zero-knowledge" fashion to verify the consistency of the distributed information.
    For this, $\mathsf{S}$ additionally distributes a random $t$-degree {\it masking-polynomial} $M(x)$ to $\mathsf{I}$, while
    each verifier $P_i$ is given a point on $M(x)$ at a distinct $\alpha_i$. The parties then publicly check the consistency of the $F(x), M(x)$ polynomials
    and the distributed points, with respect to a random linear combination of these polynomials and points.  
    The linear combiner is randomly selected by $\mathsf{I}$, only when it is confirmed that $\mathsf{S}$ has distributed the verification-points to sufficiently many verifiers
    in a set $\R$, which we call {\it supporting verifiers}. This ensures that $\mathsf{S}$ has no knowledge beforehand about the random combiner while distributing the points to
    $\R$ and hence, any inconsistency among the data distributed by a {\it corrupt} $\mathsf{S}$ will be detected with a high probability.
    
    \begin{protocolsplitbox}{AICP}{The asynchronous information-checking protocol against general adversaries for session id $\sid$ in the $\FAcast$-hybrid}{fig:AICP}
   \centerline{\underline{Protocol $\Auth(\PartySet,\AdvStructure,\mathsf{S},\mathsf{I},\mathsf{R},s)$}} 
   \justify \justify
\begin{myitemize}
\item[--] {\bf Distributing the Polynomials and the Verification Points}: Only $\mathsf{S}$ executes the following steps.
    \begin{myitemize}
    \item Randomly select $t$-degree {\it signing-polynomial} $F(x)$ and {\it masking-polynomial} $M(x)$,
     such that $F(0) = s$, where
     $t = \max\{ |Z| :  Z \in \AdvStruct 	\}$. 
     For $j = 1,\ldots,n$, randomly select $\alpha_j \in \F \setminus \{0\}$, compute
    $v_j = F(\alpha_j), m_j = M(\alpha_j)$.
    \item Send $(\authPoly, \sid, F(x), M(x))$ to $\mathsf{I}$. For $j = 1, \ldots, n$, 
    send $(\authPoint, \sid, (\alpha_j, v_j, m_j))$ to party $P_j$.    
    \end{myitemize}
\item[--] {\bf Confirming Receipt of Verification Points}: Each party $P_i$ (including $\mathsf{S}, \mathsf{I}$ and $\mathsf{R}$) 
 upon receiving $(\authPoint, \sid, \allowbreak (\alpha_i, v_i, m_i))$ from $\mathsf{S}$, sends $(\Received,\sid,i)$ to $\mathsf{I}$. 

\item[--] {\bf Announcing Masked Polynomial and Set of Supporting Verifiers}: 
	\begin{myitemize}
	\item $\mathsf{I}$, upon receiving $(\Received,\sid,j)$ from a set of parties $\R$ where $\PartySet \setminus \R \in \AdvStructure$, 
	 randomly picks $d \in \F \setminus \{0\}$ and sends $(\Sender,\Acast,\sid_\mathsf{I},(d,B(x), \allowbreak \R))$ to $\FAcast$, where $\sid_\mathsf{I} = \sid || \mathsf{I}$
	  and $B(x) \defined dF(x) + M(x)$.
	\item Every party $P_i \in \PartySet$ keeps requesting for output from $\FAcast$ with $\sid_\mathsf{I}$ until an output is received.
	\end{myitemize}
\item[--] {\bf Announcing Validity of Masked Polynomial : } 
	\begin{myitemize}
	\item $\mathsf{S}$, upon receiving an output $(\mathsf{I},\Acast,\sid_\mathsf{I},(d,B(x),\R))$ from $\FAcast$ with $\sid_\mathsf{I}$, checks
	if $B(x)$ is a $t$-degree polynomial, $\PartySet \setminus \R \in \AdvStructure$ and
	$dv_j + m_j = B(\alpha_j)$ holds 
	  for all $P_j \in \R$. If yes, then it sends
	   $(\Sender,\Acast,\sid_\mathsf{S}, \allowbreak \OK)$ to $\FAcast$, where $\sid_\mathsf{S} = \sid || S$. Else, it sends $(\Sender,\Acast,\sid_\mathsf{S},\NOK,s)$ to $\FAcast$.
	\item Every party $P_i \in \PartySet$ keeps requesting for output from $\FAcast$ with $\sid_\mathsf{S}$ until an output is received.
	\end{myitemize}
\item[--] {\bf Deciding Whether Authentication is Successful}: Every party $P_i$ (including $\mathsf{S}, \mathsf{I}$ and $\mathsf{R}$) 
 upon receiving $(\mathsf{I},\Acast, \allowbreak \sid_\mathsf{I},(d,B(x),\R))$ from $\FAcast$ with $\sid_\mathsf{I}$, 
 sets the variable $\authCompleted^{(\sid,i)}_{\mathsf{S},\mathsf{I},\mathsf{R}}$ to $1$ if either of the following holds.
    \begin{myitemize}
        \item $(\Sender,\Acast,\sid_\mathsf{S},\NOK,s)$ is received from $\FAcast$ with $\sid_\mathsf{S}$. In this case, $P_i$ also sets
          $\ICSig(\mathsf{S}, \mathsf{I}, \allowbreak \mathsf{R}, s) = s$. 
        \item $(\Sender,\Acast,\sid_\mathsf{S},\OK)$ is received from $\FAcast$ with $\sid_\mathsf{S}$. Here,
        $P_i$ sets $\ICSig(\mathsf{S}, \mathsf{I}, \mathsf{R}, s) = F(x)$, if $P_i = \mathsf{I}$.\footnote{If $\mathsf{S}$ broadcasts $s$ along with $\NOK$, then 
        $\ICSig$ will be set {\it publicly} to $s$, while if $\mathsf{S}$ broadcasts $\OK$ then {\it only} $\mathsf{I}$ sets $\ICSig$ to $F(x)$.} \\[.1cm]
    \end{myitemize}
    
\end{myitemize}
  \centerline{\underline{Protocol $\Reveal(\PartySet,\AdvStructure,\mathsf{S},\mathsf{I},\mathsf{R},s)$}}
\begin{myitemize}
\item[--] {\bf Revealing Signing Polynomial and Verification Points}:  Each party $P_i$ (including $\mathsf{S}, \mathsf{I}$ and $\mathsf{R}$) 
 does the following, if 
 $\authCompleted^{(\sid,i)}_{\mathsf{S},\mathsf{I},\mathsf{R}}$ is set to $1$ and $\ICSig(\mathsf{S}, \mathsf{I}, \mathsf{R}, s)$ has {\it not} been {\it publicly} set
  during $\Auth$.
    \begin{myitemize}
    \item If $P_i = \mathsf{I}$ then send 
     $(\revealPoly, \sid, F(x))$ to $\mathsf{R}$, where $\ICSig(\mathsf{S}, \mathsf{I}, \mathsf{R}, s)$ has been set to $F(x)$ during $\Auth$.
    \item If $P_i \in \R$, then send $(\revealPoint,\sid,(\alpha_i, v_i, m_i))$ to $\mathsf{R}$. 
    \end{myitemize}
\item[--] {\bf Accepting or Rejecting the IC-Sig}: The following steps are executed only by $\mathsf{R}$, if 
$\authCompleted^{(\sid,i)}_{\mathsf{S},\mathsf{I},\mathsf{R}}$ is set to $1$ by $\mathsf{R}$ 
 during the protocol $\Auth(\PartySet,\AdvStructure,\mathsf{S},\mathsf{I},\mathsf{R},s)$, where $\mathsf{R} = P_i$.
    \begin{myitemize}
       \item[--] If $\mathsf{R}$ has set $\ICSig(\mathsf{S}, \mathsf{I}, \mathsf{R}, s) = s$ during $\Auth$, then output $s$.
         Else, wait till $(\revealPoly, \sid,F(x))$ is received from $\mathsf{I}$, where $F(x)$ is a $t$-degree polynomial. Then proceed as follows.
		\begin{myenumerate}
		\item[1.] If $(\revealPoint,\sid,(\alpha_j, v_j, m_j))$ is received from $P_j \in \R$, then {\it accept} $(\alpha_j, v_j, m_j)$ if either 
		$v_j = F(\alpha_j)$ or $B(\alpha_j) \neq dv_j + m_j$, where $B(x)$ is received from $\FAcast$ with $\sid_\mathsf{I}$, during
		 $\Auth$.
    	\item[2.] Wait till a subset of parties $\R' \subseteq \R$ is found, such that $\R \setminus \R' \in \AdvStruct$
	and for every $P_j \in \R'$, the corresponding revealed point $(\alpha_j, v_j, m_j)$ is accepted. Then output
	$s = F(0)$.
		\end{myenumerate}	
	\end{myitemize}
\end{myitemize}
\end{protocolsplitbox}
 \begin{lemma}
 \label{lemma:AICP}
 Let $\AdvStructure$ satisfy the $\Q^{(3)}(\PartySet, \AdvStructure)$ condition. Then the pair of protocols $(\Auth, \Reveal)$ satisfy the following properties, except with
  probability at most $\errorAICP \defined \frac{nt}{|\F| - 1}$, where
   $t = \max\{ |Z| :  Z \in \AdvStruct 	\}$.
  
    \begin{myitemize}
       \item[--] {\bf Correctness}:  If $\mathsf{S}, \mathsf{I}$ and $\mathsf{R}$ are {\it honest}, then each honest 
   $P_i$ eventually sets $\authCompleted^{(\sid,i)}_{\mathsf{S},\mathsf{I},\mathsf{R}}$ 
   to $1$ during $\Auth$. Moreover, $\mathsf{R}$ eventually outputs $s$ during $\Reveal$.
       \item[--] {\bf Privacy}: If $\mathsf{S}, \mathsf{I}$ and $\mathsf{R}$ are {\it honest}, then the view of adversary remains
      independent of $s$.
     \item[--] {\bf Unforgeability}: If $\mathsf{S}, \mathsf{R}$ are {\it honest}, $\mathsf{I}$ is corrupt
      and if $\mathsf{R}$ outputs $s' \in \F$ during $\Reveal$, then 
      $s' = s$ holds.      
    \item[--] {\bf Non-repudiation}: If $\mathsf{S}$ is {\it corrupt} and $\mathsf{I}, \mathsf{R}$ are {\it honest}
    and if $\mathsf{I}$ has set $\ICSig(\mathsf{S}, \mathsf{I}, \mathsf{R}, s)$ during $\Auth$, then $\mathsf{R}$ 
    eventually outputs 
    $s$ during $\Reveal$.    
    \end{myitemize}
 Protocol $\Auth$ requires a communication of $\Order(n \cdot \log{|\F|})$ bits
 and $\Order(1)$ calls to $\FAcast$ with $\Order(n \cdot \log{|\F|})$-bit messages. Protocol 
  $\Reveal$ requires a communication of $\Order(n \cdot \log{|\F|})$ bits.
 \end{lemma}
  Lemma \ref{lemma:AICP}, is proven in Appendix \ref{app:AICP}. We use the following
   notations for AICP in our statistical VSS protocol.
 \begin{notation}[\bf Notation for Using AICP]
 While using $(\Auth, \Reveal)$, we will say that:
 \begin{myitemize}
 \item[--] ``$P_i$ {\it gives}  $\ICSig(\sid, P_i, P_j, P_k, s)$ {\it to} $P_j$" to mean that $P_i$ acts as $\mathsf{S}$ and invokes an instance of
  the protocol $\Auth$ with session id $\sid$, where $P_j$ and $P_k$ plays the role of $\mathsf{I}$ and $\mathsf{R}$ respectively.
 \item[--]  ``$P_j$ {\it receives} $\ICSig(\sid, P_i, P_j, P_k, s)$ {\it from} $P_i$" to mean that
  $P_j$, as $\mathsf{I}$, has set 
   $\authCompleted^{(\sid,j)}_{P_i,P_j,P_k}$ to $1$ during protocol $\Auth$ with session id $\sid$, where $P_i$ and $P_k$
   plays the role of $\mathsf{S}$ and $\mathsf{R}$ respectively.   
 \item[--] ``$P_j$ {\it reveals} $\ICSig(\sid, P_i, P_j, P_k, s)$ {\it to} $P_k$" to mean $P_j$, as $\mathsf{I}$, invokes an instance
  of $\Reveal$ with session id $\sid$, with $P_i$ and $P_k$ playing the role of $\mathsf{S}$
  and $\mathsf{R}$ respectively.
 \item[--] ``$P_k$ {\it accepts} $\ICSig(\sid, P_i, P_j, P_k, s)$" to mean that
  $P_k$, as $\mathsf{R}$, outputs $s$ during the instance of $\Reveal$ with session id $\sid$,
  invoked by $P_j$ as $\mathsf{I}$, with $P_i$ playing the role of $\mathsf{S}$.
\end{myitemize} 
 \end{notation}

%% file: statVSSv2.tex
\subsection{Statistically-Secure VSS Protocol with  $\Q^{(3)}(\PartySet, \AdvStructure)$ Condition}
The high level idea behind our statistically-secure protocol $\PiStatVSS$ (Figure \ref{fig:VSS_stat}) 
  is similar to that of the {\it perfectly-secure}
 VSS protocol $\PiPerVSS$ (see Fig \ref{fig:PerAVSS} in Appendix \ref{app:PerfectVSS}). In
 $\PiPerVSS$, dealer $P_\D$, on having the shares $(s_1, \ldots, s_h)$, sends
  $s_q$ to the parties in $S_q \in \ShareSpec$, followed by the parties in $S_q$ performing pair-wise consistency tests of their supposedly common shares
   and publicly announcing the results. Based on these results, 
    the parties identify a {\it core} set $\C_q \subseteq S_q$ where
  $S_q \setminus \C_q \in \AdvStructure$, such that all the (honest) parties in $\C_q$ have received the same share $s_q$ from $P_D$.
   Once such a $\C_q$ is identified, then the {\it honest} parties in $\C_q$, forming a ``majority", can ``help" the (honest) parties in $S_q \setminus \C_q$ get this common $s_q$.
    However, since $\AdvStructure$ {\it now} satisfies the $\Q^{(3)}(\PartySet, \AdvStructure)$ condition, 
    $\C_q$ may have {\it only one} honest party. Consequently, the ``majority-based filtering" used by the parties in 
   $S_q \setminus \C_q$ to get $s_q$ will fail.
   %
   %
   
   To deal with the above problem, the parties in $S_q$ issue IC-Signatures during the pair-wise consistency tests of their supposedly common shares.
   The parties then check whether the common share $s_q$ held by the (honest) parties in $\C_q$
   is ``$(P_i, P_j, P_k)$-authenticated" for {\it every} $P_i, P_j \in \C_q$ and {\it every} $P_k \in S_q$; i.e.~$P_j$ holds
   $\ICSig(P_i, P_j, P_k, s_q)$. Now, to help the parties $P_k \in S_q \setminus \C_q$ obtain the common share $s_q$, 
  {\it every} $P_j \in \C_q$ reveals IC-signed $s_q$ to $P_k$, signed by {\it every} $P_i \in \C_q$. Since $\C_q$ is bound to contain at least one
  {\it honest} party, a {\it corrupt} $P_j$ will {\it fail} to forge an {\it honest} $P_i$'s IC-signature on an incorrect $s_q$. 
  On the other hand, an {\it honest} $P_j$ will be {\it able} to eventually reveal the IC-signature of {\it all} the parties in $\C_q$ on the share $s_q$, which is accepted by 
   $P_k$.
\begin{protocolsplitbox}{$\PiStatVSS$}{The statistically-secure VSS protocol for session id $\sid$ for realizing $\FVSS$ in the $\FAcast$-hybrid model}{fig:VSS_stat}
\justify
\begin{myitemize}
\item[--] {\bf Distribution of Shares}: $P_{\D}$, 
   on having input $(s_1, \ldots, s_h)$, sends $(\dist, \sid, q, s_q)$ to all $P_i \in S_q$, for $q = 1, \ldots, h$.
\item[--] {\bf Pairwise Consistency Tests on IC-Signed Values}: For each $S_q \in \ShareSpec$, each $P_i \in S_q$ does the following.
	\begin{myitemize}
	\item Upon receiving $(\dist,\sid,q,s_{qi})$ from $\D$, give $\ICSig(\sid^{(P_\D,q)}_{i,j,k}, P_i, P_j, P_k, s_{qi})$ to every $P_j \in S_q$,
	corresponding to every $P_k \in S_q$,  where $\sid^{(P_\D,q)}_{i,j,k} = \sid || P_\D || q || i || j || k$.
	\item Upon receiving $\ICSig(\sid^{(P_\D,q)}_{j,i,k}, P_j, P_i, P_k, s_{qj})$ from $P_j \in S_q$ corresponding to
	 every party $P_k \in S_q$, if $s_{qi} = s_{qj}$ holds, then
	  send $(\Sender,\Acast,\sid^{(P_\D,q)}_{i,j},\OK_q(i,j))$ to $\FAcast$, where $\sid^{(P_\D,q)}_{i,j} = \sid || P_\D || q || i || j$.
	\end{myitemize}
\item[--] {\bf Constructing Consistency Graph}: For each $S_q \in \ShareSpec$, each $P_i \in \PartySet$ executes the following steps.
	\begin{myitemize}
	\item Initialize a set $\C_q$ to $\emptyset$. Construct an undirected consistency graph $G^{(i)}_q$ with $S_q$ as the vertex set.
	\item For every $P_j, P_k \in S_q$, keep requesting an output from $\FAcast$ with $\sid^{(P_\D,q)}_{j,k}$, until an output is received.
	\item Add the edge $(P_j,P_k)$ to $G^{(i)}_q$ if $(P_j,\Acast,\sid^{(P_\D,q)}_{j,k},\OK_q(j,k))$ and $(P_k,\Acast,\sid^{(P_\D,q)}_{k,j},\OK_q(k,j))$  is received from 
	$\FAcast$ with $\sid^{(P_\D,q)}_{j,k}$ and $\sid^{(P_\D,q)}_{k,j}$ respectively. 
	\end{myitemize}
\item[--] {\bf Identification of Core Sets and Public Announcements}: $P_{\D}$ executes the following steps to compute the core sets.
	\begin{myitemize}
	\item For each $S_q \in \ShareSpec$, check if there exists a subset of parties $\W_q \subseteq S_q$, such that $S_q \setminus \W_q \in \AdvStructure$,
	and the parties in $\W_q$ form a clique in the consistency graph $G^{\D}_q$. If such a $\W_q$ exists, then assign $\C_q \coloneqq \W_q$.
	\item Once $\C_1, \ldots, \C_h$ are computed,
	 send $(\Sender, \Acast, \sid_{P_{\D}}, \allowbreak \{ \C_q \}_{S_q \in \ShareSpec})$, where $\sid_{P_{\D}} = \sid || P_{\D}$.
	\end{myitemize}
\item[--] {\bf Share computation}: Each $P_i \in \PartySet$ executes the following steps.
	\begin{myitemize}
	\item 
	 Keep requesting for output from $\FAcast$ with $\sid_{P_\D}$ until an output is received.
	\item Upon receiving an output $(\Sender, \Acast, \sid_{P_{\D}}, \{ \C_q \}_{S_q \in \ShareSpec})$ from $\FAcast$ with $\sid_{P_{\D}}$, 
	wait until the parties in $\C_q$ form a clique in $G^{(i)}_q$, corresponding to each $S_q \in \ShareSpec$. For $q = 1, \ldots, h$, verify if
	$S_q \setminus \C_q \in \AdvStructure$. If the verification is successful, then proceed to compute the
	shares corresponding to each $S_q$ such that $P_i \in S_q$
	as follows.
	    \begin{myenumerate}
	    \item If $P_i \in \C_q$ then set $[s]_q = s_{qi}$
	    and corresponding to every signer $P_j \in \C_q$,
	     reveal $\ICSig(\sid^{(P_\D,q)}_{j,i,k}, P_j, P_i, P_k, \allowbreak s_{qi})$ to every receiver party $P_k \in S_q \setminus \C_q$.  
		\item If $P_i \notin \C_q$, then wait till $P_i$ finds some $P_j \in \C_q$ such that $P_i$ has accepted 
		$\ICSig(\sid^{(P_\D,q)}_{k,j,i},P_k, P_j, P_i, s_{qj})$ revealed by the intermediary $P_j$, corresponding to every signer $P_k \in \C_q$. 
		Then set $[s]_q = s_{qj}$.		
	    \end{myenumerate}
	\item Upon computing $\{[s]_q\}_{P_i \in S_q}$,  
	output $(\Share, \sid,P_{\D}, \{[s]_q\}_{P_i \in S_q})$.
	\end{myitemize}
\end{myitemize}
\end{protocolsplitbox}

The properties of the protocol $\PiStatVSS$ stated in Theorem \ref{thm:StatVSS} are proven in Appendix \ref{app:StatisticalVSS}.
\begin{theorem}
\label{thm:StatVSS}
 Let $\AdvStructure$ satisfy the $\Q^{(3)}(\PartySet, \AdvStructure)$ condition. 
  Then $\PiStatVSS$ UC-securely realizes $\FVSS$
   in the $\FAcast$-hybrid model, except with error probability $|\AdvStructure| n^3 \errorAICP$, where $\errorAICP \approx \frac{n^2}{|\F|}$. 
   The protocol makes $\Order(|\AdvStructure| \cdot n^3)$ calls to $\FAcast$ with $\Order(n \cdot \log{|\F|})$ bit messages
    and additionally incurs a communication of 
   $\Order(|\AdvStructure| \cdot n^4 \log{|\F|})$ bits. By replacing the calls to $\FAcast$ with protocol $\PiAcast$, the protocol incurs a total communication of
   $\Order(|\AdvStructure| \cdot n^6 \log{|\F|})$ bits.   
\end{theorem}
\subsubsection{Statistically-Secure VSS for Superpolynomial $|\AdvStructure|$}
\label{sec:superpolynomial}
\input{superpolynomial}

%% file: superpolynomial.tex
  The error probability of  $\PiStatVSS$ depends linearly on $|\AdvStructure|$
  (Theorem \ref{thm:StatVSS}), which is problematic for a large sized $\AdvStructure$. We now discuss modifications to the protocols 
  $\Auth/\Reveal$, followed by the modifications in the way they are used in $\PiStatVSS$, so as to ensure that the error probability of 
  $\PiStatVSS$ is only $n^2 \cdot \errorAICP$, {\it irrespective} of the number of invocations of $\PiStatVSS$.
   The idea is to use local ``dispute control" as used in \cite{HT13}, where the parties
  locally discard corrupt parties {\it as and when} they are identified to be cheating during {\it any} instance of $\Auth/\Reveal$. Once
  a party $P_j$ is locally discarded by some $P_i$, then $P_i$ ``behaves" as if $P_j$ has {\it certainly} behaved maliciously in all ``future" instances of 
  $\Auth/\Reveal$, {\it irrespective} of whether this is not the case or not. 
 \vspace*{-0.3cm}
 \paragraph{\bf Modifications in $\Auth$ and $\Reveal$:}   
  Each $P_i$ maintains a list of locally-discarded parties $\LocalDiscarded^{(i)}$, which it keeps on populating across {\it all} the invoked
   instances of $\Auth$ and $\Reveal$. In any instance of $\Auth$, if $P_i \in \R$ receives an $\OK$ message from $\mathsf{S}$ even though 
   $B(\alpha_i) \neq dv_i + m_i$ holds, then $P_i$ adds $\mathsf{S}$ to $\LocalDiscarded^{(i)}$. 
   Once $P_i$ adds $\mathsf{S}$ to $\LocalDiscarded^{(i)}$, then in any future instance of $\Reveal$
   involving the signer $\mathsf{S}$, party $P_i$, if present in the corresponding $\R$ set, sends a special ``dummy" point to the corresponding receiver $\mathsf{R}$,
    {\it instead} of the verification-point received from $\mathsf{S}$, and this dummy point is {\it always accepted} by $\mathsf{R}$.
    This ensures that once the verifier $P_i$ catches a {\it corrupt} $\mathsf{S}$ trying to break the {\it non-repudiation} property 
     by distributing inconsistent verification-point to $P_i$, then in any future instance of
    AICP involving $\mathsf{S}$, if $P_i$ is added to the corresponding $\R$ set, its verification-point
    will {\it always} be accepted.  
    
   Similarly, if in any instance of $\Reveal$ where $P_i$ is the {\it receiver}, $P_i$ is sure that it has {\it not accepted} the verification-point of some
   {\it honest} verifier belonging to $\R$, then $P_i$ includes the corresponding intermediary $\mathsf{I}$ to $\LocalDiscarded^{(i)}$.
   To check this, in $\Reveal$, $P_i$ now additionally checks if there exists a set of verifiers $\R'' \subseteq \R$, where $\R \setminus \R'' \in \AdvStruct$, such that
   the verification-points received from {\it all} the parties in $\R''$ are {\it not accepted}.  Once $P_i$ adds 
   $\mathsf{I}$ to $\LocalDiscarded^{(i)}$, in any future instance of $\Reveal$ involving $\mathsf{I}$ as intermediary and $P_i$ as the receiver,
   $P_i$ {\it rejects} the IC-signature revealed by $\mathsf{I}$. 
    This ensures that once
   $P_i$ as a receiver catches $\mathsf{I}$ trying to break the {\it unforgeability} property, then from then onwards,
   $\mathsf{I}$ cannot do so in any other instance of $\Reveal$ involving $P_i$ as the receiver.
    \vspace*{-0.3cm}
   \paragraph{\bf Modifications in $\PiStatVSS$:}     
  Party $P_i$ now broadcasts a {\it single} $\OK(i, j)$ message for $P_j$,
 only after receiving the corresponding signature from 
  {\it all} the instances of $\Auth$ involving $P_j$ as the {\it signer} and $P_i$ as the {\it intermediary}, followed by pair-wise consistency tests. 
  Consequently, $P_\D$ now finds a {\it common} core set $\C$ across all the sets $S_1, \ldots, S_h$, where $S_q \setminus \C \in \AdvStructure$ for each
  $S_q$, and where the parties in $\C$ constitute a clique. 
  Moreover, each verifier now waits for {\it all} instances of $\Auth$ between a signer $\mathsf{S}$ and an intermediary $\mathsf{I}$ in 
  $\C$ to complete (by checking if the corresponding $\authCompleted$ variables are all set to $1$),
   before participating in {\it any} instance of $\Reveal$. 
   
   The above modification ensures that if a {\it corrupt} signer in $\C$ gives any {\it verifier} an 
    inconsistent verification-point during {\it any} instance of $\Auth$, it will be caught and locally discarded, except with probability $\errorAICP$. 
    By considering all possibilities for a {\it corrupt} signer and an {\it honest} verifier,
     it follows that except with probability at most $n^2 \cdot \errorAICP$, the verification-points
      of all {\it honest} verifiers will be accepted by every {\it honest} receiver during all the instances of $\Reveal$ in any instance of $\PiStatVSS$.
    On the other hand, if any {\it corrupt} intermediary in $\C$ tries to forge a signature on the behalf of an {\it honest} party in $\C$,
     then except with probability $\errorAICP$, it will be {\it discarded} by an honest receiver $\mathsf{R}$. From then on, 
     $\mathsf{R}$ will always reject any signature revealed by the same intermediary. 
     Hence, by considering all possibilities for a {\it corrupt} intermediary and an {\it honest} receiver, except with 
     probability $n^2 \cdot \errorAICP$, no {\it corrupt} intermediary will be able to forge a signature 
     to any {\it honest} receiver in any instance of $\PiStatVSS$.    
     
     Based on the above discussion, we state the following lemma.
\begin{lemma}
\label{lemma:modifiedSVSS}
The modified $\PiStatVSS$ has error probability of $n^2 \cdot \errorAICP$, independent of the number of invocations.
\end{lemma}

%% file: statMPC.tex
\subsection{Statistically-Secure Protocol for $\FTriples$ in the $(\FVSS, \FABA)$-Hybrid}
Our  {\it statistically-secure} protocol $\PiStatTriples$ for realizing $\FTriples$ with $\Q^{(3)}(\PartySet, \AdvStructure)$
 condition mostly follows \cite{HT13}. Here, we discuss the high level ideas and refer to 
 Appendix \ref{app:Statistical} for formal details and proofs. To explain the idea at a high-level, we consider the case when $M=1$ multiplication-triple is generated through $\PiStatTriples$.
   The modifications to generate $M$ multiplication-triples are straight forward.
   Protocol $\PiStatTriples$ is almost the same as $\PiPerTriples$, except that we now use a {\it statistically-secure} multiplication protocol. 
   \vspace*{-0.4cm}
   \paragraph{Basic Multiplication Protocol:}
     Our starting point is the basic multiplication protocol of \cite{HT13} in the {\it synchronous} setting. 
     The protocol takes $[a], [b]$, along with a set of {\it globally-discarded} parties $\Discarded$
     which are {\it guaranteed} to be corrupt, and outputs $[c]$. In the protocol,  each summand $[a]_p[b]_q$ is assigned to a {\it publicly-known} designated party from $\PartySet \setminus \Discarded$. Every designated summand-sharing party then secret-shares the sum of all the assigned summands, based
      on which the parties compute $[c]$. If no summand-sharing party behaves maliciously, then $c = ab$ holds.
      
      Similar to $\OptMult$, the main challenge while running the above protocol
       in the {\it asynchronous} setting is that a
         corrupt summand-sharing party may {\it never} share the sum of the assigned summands. To deal with this issue, similar to what was done for
   $\OptMult$, 
   we ask {\it each} party in $\PartySet \setminus \Discarded$ to share the sum of all possible
    summands it is capable of, while ensuring that no summand is shared twice.
    The idea here is that since $\AdvStructure$ satisfies the $\Q^{(3)}(\PartySet, \AdvStructure)$ condition,
     for every summand $[a]_p [b]_q$, the set $(S_p \cap S_q) \setminus \Discarded$ is guaranteed
    to contain at least one {\it honest} party who will share $[a]_p[b]_q$.
    Based on this above idea, we design a protocol $\BasicMult$
    which is executed with respect to a set $\Discarded$, and an iteration number $\iter$. 
    Looking ahead, it will be guaranteed that no honest party is ever included in $\Discarded$. The protocol is similar to
   $\OptMult$, except that it {\it does not} take any subset $Z \in \AdvStructure$ as input. 
     \vspace*{-0.4cm}
  \paragraph{Detectable Random-Triple Generation:}
 Based on $\BasicMult$, we design a protocol $\RandMultLCE$, which takes as input an iteration number $\iter$ and an existing set of {\it corrupt} parties
   $\Discarded$. If no party in $\PartySet \setminus \Discarded$ behaves maliciously, then the protocol outputs
  a random secret-shared multiplication-triple $[a_\iter], [b_\iter], [c_\iter]$. Else, except with probability $\frac{1}{|\F|}$, 
  the parties update $\Discarded$ by identifying at least one {\it new} corrupt party among
  $\PartySet \setminus \Discarded$. In the protocol, 
   the parties first generate secret-sharing of random values $a_\iter,b_\iter,b'_\iter$ and $r_\iter$.
   Two instances of $\BasicMult$ with inputs $[a_\iter], [b_\iter]$ and $[a_\iter], [b'_\iter]$ are run
    to obtain $[c_\iter]$ and $[c'_\iter]$ respectively. The parties then reconstruct the ``challenge" $r_\iter$ and {\it publicly} check
    if  $[a_\iter](r_\iter[b_\iter] + [b'_\iter]) = (r_\iter[c_\iter] + [c'_\iter])$ holds, which should be the case if {\it no} cheating has occurred during the instances of $\BasicMult$.
    If the condition holds, then the parties output $[a_\iter], [b_\iter], [c_\iter]$, which is guaranteed to be a multiplication-triple, except with probability $\frac{1}{|\F|}$.
    Otherwise, the parties proceed to identify at least one new corrupt party by reconstructing
    $[a_\iter], [b_\iter], [b'_\iter], [c_\iter], [c'_\iter]$ and the sum of the summands shared by the various summand-sharing parties during the instances of
    $\BasicMult$. 
      \vspace*{-0.4cm}
    \paragraph{The Statistically-Secure Pre-Processing Phase Protocol:}
Protocol $\PiStatTriples$ proceeds in iterations, where in each iteration an instance of $\RandMultLCE$ is invoked, 
 which either succeeds or fails. In case of success, the parties output
  the returned multiplication-triple, else, they continue to the next iteration.
   As a new
    corrupt party is discarded in each failed iteration, the protocol eventually outputs a multiplication-triple.   
\begin{theorem}
\label{thm:PiStatTriples}
Let  $\AdvStructure$ satisfy the $\Q^{(3)}(\PartySet, \AdvStructure)$ condition. 
 Then  $\PiStatTriples$ UC-securely realizes $\FTriples$ in the $(\FVSS, \FABA)$-hybrid model, except with
  error probability of at most $\frac{n}{|\F|}$.
   The protocol makes $\Order(n^3 \cdot M)$ calls to $\FVSS$ and $\Order(n^3)$ calls to $\FABA$,
   and  additionally communicates $\Order((M \cdot |\AdvStruct| \cdot n^3 + |\AdvStructure| \cdot n^4) \log{|\F|})$ bits. 
   
   By replacing the calls to $\FVSS$ with protocol $\PiStatVSS$ (along with the modifications discussed in Section \ref{sec:superpolynomial}), protocol  
   $\PiStatTriples$ UC-securely realizes $\FTriples$ in the $\FABA$-hybrid model, except with
  error probability $n^2 \cdot \errorAICP$. The protocol makes $\Order(n^3)$ calls to $\FABA$
  and incurs a communication of
  $\Order(M \cdot |\AdvStruct| \cdot n^9 \log{|\F|})$ bits.
\end{theorem}

%% file: mpc.tex
\section{MPC Protocols in the Pre-Processing Model}
\label{sec:mpc}
   The MPC protocol $\PiMPC$ in the pre-processing model is standard. 
    The parties first generate secret-shared random multiplication-triples through
     $\FTriples$. Each party then randomly secret-shares its input for $\ckt$ through $\FVSS$. 
       To avoid an indefinite wait, 
   the parties agree on a common subset of parties, whose
    inputs are eventually secret-shared, through ACS. 
   The parties then jointly evaluate each gate in $\ckt$ in a secret-shared fashion by generating a secret-sharing of 
   the gate-output from a secret-sharing of the gate-input(s). Linear gates are evaluated non-interactively due to the linearity of secret-sharing. 
   To evaluate multiplication gates, the parties deploy Beaver's method \cite{Bea91}, using the secret-shared
   multiplication-triples generated by $\FTriples$. 
   Finally, the parties publicly reconstruct the  secret-shared function output.
   Protocol $\PiMPC$ and the proof of Theorem \ref{thm:AMPC} are presented in Appendix \ref{app:MPC}. 
 \begin{theorem}
\label{thm:AMPC}
 Protocol $\PiMPC$ UC-securely realizes the functionality $\Functionality$ for securely computing  $f$ (see Fig \ref{fig:FAMPC} in Appendix \ref{app:UC}) 
   with perfect security in the $(\FTriples, \FVSS, \FABA)$-hybrid model, in the 
   presence of a static malicious adversary characterized by an adversary-structure $\AdvStructure$ satisfying the $\Q^{(3)}(\PartySet, \AdvStructure)$ condition.
   The protocol makes one call to $\FTriples$ and $\Order(n)$ calls to $\FVSS$ and $\FABA$ and additionally incurs a communication of 
   $\Order(M \cdot |\AdvStructure| \cdot n^2 \log{|\F|})$ bits, where $M$ is the number of multiplication gates in the circuit $\ckt$ representing $f$.
   \end{theorem}
   If we replace the calls to $\FTriples$ and $\FVSS$ with perfectly-secure protocol $\PiPerTriples$ and $\PiPerVSS$ respectively, then protocol
   $\PiMPC$ achieves perfect security in the $\FABA$-hybrid. On the other hand, 
    replacing the calls to $\FTriples$ and $\FVSS$ in $\PiMPC$ with  $\PiStatTriples$ and $\PiStatVSS$ respectively
    leads to statistical-security. To bound the error probability of the statistically-secure protocol by $2^{-\kappa}$,
     we select a finite field $\F$ such that $|\F| > n^4 2^{\kappa}$. Based on the above discussion, we get the following corollaries of 
   Theorem \ref{thm:AMPC}.
   \begin{corollary}
   \label{cor:Perfect}
   If $\AdvStructure$ satisfies the $\Q^{(4)}(\PartySet, \AdvStructure)$ condition, then 
    $\PiMPC$ UC-securely realizes $\Functionality$ in the $\FABA$-hybrid model with perfect security.
    The protocol makes $\Order(|\AdvStructure| \cdot n^5)$ calls to $\FABA$
    and incurs a communication of $\Order(M \cdot (|\AdvStructure|^2 \cdot n^7 \log{|\F|} + |\AdvStructure| \cdot n^9 \log{n}))$ bits, where
    $M$ is the number of multiplication gates in $\ckt$.
   \end{corollary}
    \begin{corollary}
   \label{cor:Statistical}
   If $\AdvStructure$ satisfies the $\Q^{(3)}(\PartySet, \AdvStructure)$ condition, then 
    $\PiMPC$ UC-securely realizes $\Functionality$ in the $\FABA$-hybrid model with statistical security.
    If $|\F| > n^4 2^{\kappa}$ for a given statistical-security parameter $\kappa$, then the error probability of the protocol is at most $2^{-\kappa}$.
     The protocol makes $\Order(n^3)$ calls to $\FABA$
   and incurs a communication of
  $\Order(M \cdot |\AdvStruct| \cdot n^9 \log{|\F|})$ bits, where $M$ is the 
   number of multiplication gates in $\ckt$.
   \end{corollary}  

%% file: AppUC.tex
\section{The Asynchronous Universal Composability (UC) Framework and Various Asynchronous Functionalities}
\label{app:UC}
In this section, we discuss the asynchronous UC framework followed in this paper. The discussion is based on the description of the framework against
 threshold adversaries as provided in \cite{Coh16} (which is further based on \cite{KMTZ13,CGHZ16}). We adapt the framework for the case of general adversaries.
  Informally, the security of a protocol is argued by ``comparing" the capabilities of the adversary in two separate
  worlds.
  In the {\it real-world}, the parties exchange messages among themselves, computed as per
  a given protocol. In the {\it ideal-world}, the parties do not interact with {\it each other}, but with a {\it trusted} third-party
  (an ideal functionality), which enables the parties to obtain the result of the computation based on the inputs provided by
  the parties. Informally, a protocol is considered to be secure if whatever an adversary can do in the real protocol
   can be also done in the ideal-world. \\~\\
    \noindent {\bf The Asynchronous Real-World}: An execution of a protocol $\Pi$ in the real-world 
  consists of $n$ {\it interactive Turing machines} (ITMs) representing the parties in $\PartySet$.
  Additionally, there is an ITM for representing the adversary $\Adv$.
     Each ITM is initialized with its random coins and possible inputs. Additionally, $\Adv$ may have some auxiliary input $z$. 
     Following the convention of \cite{CanettiThesis}, 
     the protocol operates {\it asynchronously} by a sequence of {\it activations}, where at each point, a single ITM is active.
       Once activated, a party can perform some local computation, write on its output tape, or send messages to other parties. 
       On the other hand, if the adversary is 
       activated,  it can send messages on the behalf of corrupt parties.       
       The protocol execution is complete once all honest parties obtain their respective outputs.
       We let $\REAL_{\Pi, \Adv(z), Z^{\star}}(\vec{x})$ 
        denote the random variable consisting of the output of the honest parties and the view of the adversary 
        $\Adv$ during the execution of a protocol $\Pi$. Here, $\Adv$ controls parties in $Z^{\star}$ during the execution of protocol $\Pi$ with inputs $\vec{x} = (x^{(1)}, \ldots, x^{(n)})$ for the parties (where party $P_i$ has input $x^{(i)}$), and auxiliary input $z$ for $\Adv$. \\~\\
   \noindent {\bf The Asynchronous Ideal-World}:
   A protocol in the ideal-world 
  consists of $n$ {\it dummy} parties $P_1, \ldots, P_n$, an ideal-world adversary $\Sim$ (also called {\it simulator}) 
  and an ideal functionality $\Functionality$. 
  We consider {\it static} corruptions such that the set of corrupt parties $Z^{\star}$ is fixed at the beginning of the computation
   and is known to $\Sim$. The functionality $\Functionality$ receives the inputs from the respective dummy parties, 
   performs the desired computation $f$ on the received inputs, and sends the outputs to the respective parties. 
   The ideal-world adversary {\it does not} see and {\it cannot} delay the communication between the parties and 
   $\Functionality$. However, it can communicate with $\Functionality$ on the behalf of corrupt parties.
   
   Since $\Functionality$ models the desired behaviour of a real-world protocol which is {\it asynchronous}, 
    ideal functionalities must consider some inherent limitations to model the asynchronous communication model with eventual delivery.
     For example, in a real-world protocol, the adversary  can decide when each honest party learns the output since 
     it has full control over message scheduling. To model the notion of time in the ideal-world, 
     \cite{KMTZ13} uses the concept of {\it number of activations}. 
      Namely, once $\Functionality$ has computed the output for some party,
       it {\it does not} ask ``permission"  from $\Sim$ to deliver it to the respective party.
       Instead, the corresponding party must ``request" $\Functionality$ for the output, which 
        can be done only  when the concerned party is active.  
        Moreover, the adversary can ``instruct" $\Functionality$ to delay the output for each party by ignoring the corresponding 
        requests, but only for a polynomial number of activations.  If a party is activated 
        sufficiently many times, the party will eventually receive the output from $\Functionality$
        and hence, ideal computation eventually completes. That is, each honest party
        eventually obtains its desired output. As in \cite{Coh16}, 
        we use the term ``{\it $\Functionality$ sends a request-based delayed output to} $P_i$",
        to  describe the above interaction between the $\Functionality, \Sim$ and $P_i$.
        
        Another limitation is that in a real-world AMPC protocol, the (honest) parties {\it cannot} afford for all the parties to provide their input for the computation to avoid an
        endless wait, as the corrupt parties may decide not to provide their inputs. Hence, {\it every}
        AMPC protocol suffers from {\it input deprivation}, where inputs of a subset of potentially honest parties (which is decided by the choice of
        adversarial message scheduling) may get ignored during computation. Consequently, once a ``core set" of parties $\CoreSet$ provide their inputs for
        the computation, where $\PartySet \setminus \CoreSet \in \AdvStructure$, the parties have to start computing the function by assuming some default input for
        the left-over parties.
        To model this in the ideal-world, $\Sim$ is given the provision
        to decide the set $\CoreSet$ of parties whose inputs should be taken into consideration by $\Functionality$. We stress that 
        $\Sim$ {\it cannot} delay sending $\CoreSet$ to $\Functionality$ indefinitely. This is because in the real-world
        protocol, $\Adv$ {\it cannot} prevent the honest parties from providing their
        inputs indefinitely.        
        The formal description of
        $\Functionality$ is available in Fig \ref{fig:FAMPC}.
        
  \begin{systembox}{$\Functionality$}{The ideal functionality for asynchronous secure multi-party computation for session id $\sid$.}{fig:FAMPC}
	\justify
$\Functionality$ proceeds as follows, running with the parties $\PartySet = \{P_1, \ldots, P_n \}$ and an adversary $\Sim$, and is parametrized by an $n$-party function
 $f: \F^n \rightarrow \F$ and an adversary structure $\AdvStructure \subseteq 2^{\PartySet}$.
   \begin{enumerate}
   	  \item For each party $P_i \in \PartySet$, initialize an input value $x^{(i)} = \bot$. 
      \item Upon receiving a message $(\inp, \sid, v)$ from some $P_i \in \PartySet$ (or from $\Sim$ if $P_i$ is {\it corrupt}), do the following:
        \begin{itemize}
        \item Ignore the message if output has already been  computed;
        \item Else, set $x^{(i)} = v$ and send $(\inp, \sid, P_i)$ to $\Sim$.\footnote{If $P_i$ is corrupt, then no need to send $(\inp, \sid, P_i)$ to $\Sim$ as the input has been provided
        by $\Sim$ only.}
        \end{itemize}
      \item Upon receiving a message $(\coreset, \sid, \CoreSet)$ from $\Sim$, do the following:\footnote{$\Sim$ cannot delay sending $\CoreSet$ indefinitely; see the discussion
      before the description of the functionality.}
        \begin{itemize}
        \item Ignore the message if $(\PartySet \setminus \CoreSet) \not \in \AdvStructure$ or if output has already been computed;
        \item Else, record $\CoreSet$ and set $x^{(i)} = 0$ for every $P_i \not \in \CoreSet$.\footnote{It is possible that for some $P_i \not \in \CoreSet$,
         the input has been set to a value
        different from $0$ during step 1 and $x^{(i)}$ is now reset to $0$. This models the scenario that in the real-world
         protocol, even if $P_i$ is able to provide its input,
        $P_i$'s inclusion to $\CoreSet$ finally depends upon message scheduling, which is under adversarial control.}
        \end{itemize}
      \item If  $\CoreSet$ has been recorded and the value $x^{(i)}$ has been set to a value different from $\bot$ for every $P_i \in \CoreSet$, then
      compute $y \defined f(x^{(1)}, \ldots, x^{(n)})$ and generate a request-based delayed output $(\out, \sid, (\CoreSet, y))$ for every $P_i \in \PartySet$.
   \end{enumerate}
\end{systembox}

  Similar to the real-world, we let $\IDEAL_{\Functionality, \Sim(z), Z^{\star}}(\vec{x})$ 
        denote the random variable consisting of the output of the honest parties and the view of the adversary 
        $\Sim$, controlling the parties in $Z^{\star}$, 
        with the parties having inputs $\vec{x} = (x^{(1)}, \ldots, x^{(n)})$ (where party $P_i$ has input $x_i$), and auxiliary input $z$ for $\Sim$.
        
        We say that a real-world asynchronous protocol $\Pi$ {\it securely realizes $\Functionality$ with
         perfectly-security} if and only if  for every real-world adversary $\Adv$, there exists an ideal-world adversary  
  $\Sim$ whose running time is polynomial in the running time of $\Adv$, such that for every possible $Z^{\star}$,
   every possible $\vec{x} \in \F^n$ 
   and every possible $z \in \{0, 1 \}^{\star}$, it holds that the random variables 
    \[ \Big \{ \REAL_{\Pi, \Adv(z), Z^{\star}}(\vec{x}) \Big \}  \quad  \mbox{ and }  \quad  \Big \{ \IDEAL_{\Functionality, \Sim(z), Z^{\star}}(\vec{x}) \Big \}\]
   are identically distributed. That is, the random variables are perfectly-indistinguishable.
   
      For statistically-secure AMPC, the parties and adversaries are parameterized with a statistical-security parameter $\kappa$,
   and the above random variables (which are viewed as ensembles, parameterized by $\kappa$) 
   are required to be statistically-indistinguishable. That is, their statistical-distance should be a negligible
   function in $\kappa$.
   
   \paragraph {\bf The Universal-Composability (UC) Framework:}
 While the real-world / ideal-world based security paradigm is used to define the security of a protocol in the ``stand-alone" setting, 
  the more powerful UC framework \cite{Can01,Can20} is used to define
  the security of a protocol when multiple instances of the protocol might be running in parallel, possibly along with other 
  protocols. Informally, the security in the UC-framework is still argued by comparing the real-world and the ideal-world. However, 
  now, in both worlds, the computation takes place in the presence of an additional interactive 
   process (modeled as an ITM) called the {\it environment} and denoted by $\Env$.
   Roughly speaking, $\Env$ models the ``external environment" in which
   protocol execution takes place. The interaction between $\Env$ and the various entities takes place as follows
   in the two worlds.     
   
   In the real-world, the environment gives inputs to the honest parties, receives their outputs, and can communicate with the adversary
   at any point during the execution. During the protocol execution, the environment gets activated first. Once activated, the environment
   can either activate one of the parties by providing some input, or activate $\Adv$ by sending it a message. Once a party
   completes its operations upon getting activated, the control is returned to the environment. Once $\Adv$ gets activated, 
   it can communicate with the environment (apart from sending the messages to the honest parties).
   The environment also fully controls the corrupt parties that send all the messages they receive to $\Env$,
    and follow the orders of $\Env$.
    The protocol execution is completed once $\Env$ stops activating
   other parties, and outputs a single bit.
   
   In the ideal-model, the environment $\Env$ gives inputs to the (dummy) honest parties, receives their outputs, and can communicate with
   $\Sim$ at any point during the execution. The dummy parties act  as channels between $\Env$ and the functionality. That is, they
   send the inputs received from $\Env$ to functionality and transfer the output they receive from the functionality
   to $\Env$. The activation sequence in this world is similar to the one in the real-world. 
   The protocol execution is completed once $\Env$ stops activating
   other parties and outputs a single bit. 
   
   A protocol is said to be UC-secure with {\it perfect-security}, if for every real-world adversary $\Adv$ there exists a simulator $\Sim$, 
   such that for any environment $\Env$, the 
   environment cannot distinguish the real-world from the ideal-world.
   On the other hand, the protocol is said to be UC-secure with {\it statistical-security}, if the 
   environment cannot distinguish the real-world from the ideal-world, except with a probability which
    is a negligible function in the statistical-security parameter $\kappa$.  
\paragraph{\bf The Hybrid Model:} In a $\AcceptedParties$-hybrid model, a protocol execution proceeds as in the real-world. However, the parties have access to
 an ideal functionality $\AcceptedParties$ for some specific task. During the protocol execution, the parties
  communicate with $\AcceptedParties$ as in the ideal-world. The UC framework guarantees
  that an ideal functionality in a hybrid model can be replaced with a protocol that UC-securely realizes $\AcceptedParties$. This is specifically due to the
  following composition theorem from \cite{Can01,Can20}.
  \begin{theorem}[\cite{Can01,Can20}]
  Let $\Pi$ be a protocol that UC-securely realizes some functionality $\Fun$ in the $\AcceptedParties$-hybrid model and let $\rho$ be a protocol that UC-securely
   realizes $\AcceptedParties$. Moreover, let $\Pi^{\rho}$ denote the protocol that is obtained from $\Pi$ by replacing every ideal call to $\AcceptedParties$
   with the protocol $\rho$. Then $\Pi^{\rho}$ UC-securely realizes $\Fun$
   in the model where the parties do not have access to the functionality $\AcceptedParties$.
  \end{theorem}
  \subsection{The Asynchronous Reliable Broadcast (Acast) Functionality and the Protocol}
      The ideal functionality $\FAcast$ capturing the requirements for asynchronous reliable broadcast  is presented in Fig \ref{fig:FAcast}. The functionality,
    upon receiving $m$ from the sender $P_S$, performs a request-based delayed delivery of $m$ to all the parties.
     Notice that if
    $P_S$ is {\it corrupt}, then the functionality {\it may not} receive any message for delivery, in which case parties obtain no output.
    This models the fact that in any real-world Acast protocol, a potentially {\it corrupt} $P_S$ {\it may not} invoke the protocol.

 \begin{systembox}{$\FAcast$}{The ideal functionality for asynchronous reliable broadcast for session id $\sid$.}{fig:FAcast}
	\justify
$\FAcast$ proceeds as follows, running with the parties $\PartySet = \{P_1, \ldots, P_n \}$ and an adversary $\Sim$, and is parametrized 
 by an adversary structure $\AdvStructure \subseteq 2^{\PartySet}$. Let $Z^{\star}$ denote the set of corrupt parties, where
  $Z^{\star}  \in \AdvStructure$.
  \begin{myitemize}
  \item Upon receiving $(\Sender, \Acast, \sid, m)$ from $P_S \in \PartySet$ (or from $\Sim$ if $P_S \in Z^{\star}$), do the following:
     \begin{itemize}
     \item[--] Send $(P_S, \Acast, \sid, m)$ to $\Sim$;\footnote{If $P_S \in Z^{\star}$, then no need to send $(P_S, \Acast, \sid, m)$ to $\Sim$, as 
     in this case $m$ is received from $\Sim$ itself.}     
     \item[--] Send a request-based delayed output
    $(P_S, \Acast, \sid, m)$ to each $P_i \in \PartySet \setminus Z^{\star}$ (no need to send $m$ to the parties in $Z^{\star}$, as $\Sim$ gets $m$ on their behalf).
     \end{itemize}
  \end{myitemize}
\end{systembox}

We next recall the Acast protocol of \cite{KF05} and present it in Fig \ref{fig:A-cast}.
\begin{protocolsplitbox}{\PiAcast$(P_S, m)$}{The perfectly-secure Acast protocol for realizing $\FAcast$}{fig:A-cast}
	\justify
\begin{myitemize}
     \item[--] Code for the Sender $P_S$ (with input $m \in \{0, 1 \}^{\ell}$):
    \begin{myitemize}
    \item Send the message ($\inp, \sid, m$) to all the parties in $\PartySet$.
    \end{myitemize}
    \item[--] Code for each party $P_i \in \PartySet$ (including $P_S$):
    \begin{myenumerate}
    \item If the message ($\inp, \sid, m$) is received from $P_S$, then send the message ($\echo, \sid, m$)
        to all the parties in $\PartySet$.
    \item If the message ($\echo, \sid, m'$) is received from a set of parties 
        $\PartySet \setminus Z$ for some $Z \in \AdvStructure$, then send the message ($\Ready, \sid, m'$) to all the parties.
    \item If the message ($\Ready, \sid, m'$) is received from a set of parties $C$
     where $C \not \in \AdvStructure$, then send the message $(\Ready, \sid, m')$ to all
        the parties in $\PartySet$.
    \item If ($\Ready, \sid, m'$) is received from a set of parties 
         $\PartySet \setminus Z$ for some $Z \in \AdvStructure$, then output $m'$.
    \end{myenumerate} 
\end{myitemize}
\end{protocolsplitbox}

The properties of the protocol $\PiAcast$ are stated in Theorem \ref{thm:Acast}.
\begin{theorem}
\label{thm:Acast}
If $\AdvStructure$ satisfies 
 the $\Q^{(3)}(\PartySet, \AdvStructure)$ condition, then protocol $\PiAcast$ UC-securely realizes $\FAcast$ with perfect security.
 The protocol incurs a communication of $\Order(n^2\ell)$ bits, where $P_S$ has an input of size $\ell$ bits.
\end{theorem}
\begin{proof}
The communication complexity simply follows from the fact that in the protocol, each party needs to send $m$ to every other party.
 For security, consider an arbitrary adversary $\Adv$ attacking $\PiAcast$ by corrupting a set of parties $Z^{\star} \in \AdvStructure$,
  and let $\Env$ be an arbitrary environment. 
  We present a simulator $\SimAcast$  such  that for any set of corrupt parties $Z^{\star} \in \AdvStructure$, 
   the output of the honest parties and the view of the adversary in an execution of $\PiAcast$ with $\Adv$
    is distributed identically to the output of the honest parties and the view of the adversary
      in an execution with $\SimAcast$ involving $\FAcast$ in the ideal world. 
    This further implies that $\Env$ cannot distinguish between the two executions.
       The simulator constructs virtual real-world honest parties and invokes $\Adv$. The simulator simulates the environment
   and the honest parties towards $\Adv$ as follows: in order to simulate $\Env$, the simulator $\SimAcast$ forwards every message it receives from 
   $\Env$ to $\Adv$ and vice-versa. To simulate the execution of honest parties, we consider two cases, depending upon whether
   $P_S$ is corrupt or not. \\~\\
   \noindent {\bf Case I: $P_S$ is honest.} 
  In this case, $\SimAcast$ first interacts with $\FAcast$ and receives
  the output $m$ from the functionality.  
    The simulator then plays the role of $P_S$ with input $m$, as well as the role of the
  honest parties, and 
  interacts with $\Adv$ as per the steps of $\PiAcast$.
  
  It is easy to see that that view of $\Adv$ is identical in the real execution and simulated execution. This is because
  only $P_S$ has input in the protocol and in the simulated execution, $\SimAcast$ plays the role of $P_S$ as per $\PiAcast$ after
  learning the input of $P_S$ from $\FAcast$. Next, conditioned on the view of $\Adv$, we show that the outputs of the honest parties are
  identical in both the executions. So consider an arbitrary $\View$ of $\Adv$.  
    Conditioned on $\View$, all honest parties eventually obtain a  request-based delayed output  $m$
    in the simulated execution, where $m$ is the input of $P_S$ as per $\View$.
    We show that even in the  real execution, all honest parties eventually output $m$. 
  This is because all honest parties eventually  complete steps $1-4$ in the protocol, even if the corrupt parties do not send their messages, as
  the messages of the {\it honest} parties $\PartySet \setminus Z^{\star}$ are eventually selected for delivery
  and $\PartySet \setminus Z^{\star} \not \in \AdvStructure$; the latter holds, as otherwise $\AdvStructure$
  does not satisfy the $\Q^{(2)}(\PartySet, \AdvStructure)$ condition. 
  $\Adv$ may send 
    $\echo$ and $\Ready$    
    messages for $m'$, where $m' \neq m$, on the behalf of corrupt parties. But since $Z^{\star} \in \AdvStructure$ and since 
    $\AdvStructure$ satisfies the $\Q^{(2)}(\PartySet, \AdvStructure)$ condition, it follows that 
    no honest party ever generates a $\Ready$ message for $m'$,  neither in step $2$, nor in step $3$.  
   Thus the output of the honest parties is {\it identically} distributed in both the worlds.
      Consequently, in this case, we conclude that $\Big \{\REAL_{\PiAcast, \Adv(z), \Env}(m) \Big \}_{m \in \{0, 1\}^{\ell}, z \in \{0, 1 \}^{\star}}
     \equiv \Big \{\IDEAL_{\FAcast, \SimAcast(z), \Env}(m) \Big\}_{m \in \{0, 1\}^{\ell}, z \in \{0, 1 \}^{\star}}$ holds, where
     $\equiv$ denotes {\it perfect indistinguishability}.  \\~\\~\\
    \noindent {\bf Case II: $P_S$ is corrupt.} In this case, $\SimAcast$ first plays the role of the honest parties and interacts with
   $\Adv$, as per $\PiAcast$. If in this execution, $\SimAcast$ finds that some {\it honest} party, say
  $P_h$, outputs $m^{\star}$, then $\SimAcast$ interacts with $\FAcast$ 
   by sending $m^{\star}$ as the input to $\FAcast$, on the behalf of $P_S$.
   Else, $\SimAcast$ does not provide any input to $\FAcast$ on the behalf of $P_S$.
   
   It is easy to see that the view of $\Adv$ is identically distributed in the real and the simulated execution.
   This is because only $P_S$, which is under the control of $\Adv$, has an input in the protocol, 
   and $\SimAcast$ plays the role of the honest parties exactly as per the protocol $\PiAcast$.
   We next show that conditioned on the view of $\Adv$, the outputs of the honest parties are identically distributed
   in both the executions.
   
   Consider an arbitrary view $\View$ of $\Adv$, during an execution of $\PiAcast$.
   If according to $\View$, no honest party obtains an output during the execution of $\PiAcast$, then the honest parties do not obtain any output
   in the simulated execution as well. This is because in this case, $\SimAcast$ does not provide any input on the behalf of $P_S$
   to $\FAcast$. On the other hand, consider the case when according to $\View$, some {\it honest} party $P_h$ outputs $m^{\star}$.
   In this case, in the simulated execution, all honest parties eventually obtain an output $m^{\star}$, since $\SimAcast$ provides 
   $m^{\star}$ as the input to $\FAcast$ on the behalf of $P_S$. We next show that even in the real execution, all honest parties
   eventually obtain the output $m^{\star}$.
   
    Since $P_h$ obtained the output $m^{\star}$, it received $\Ready$ messages for $m^{\star}$ during step $4$ of the protocol
    from a set of parties $\PartySet \setminus Z$, for some $Z \in \AdvStructure$. 
       Let $\Hon$ be the set of {\it honest} parties whose $\Ready$ messages are received by $P_h$ during step $4$. It is easy to see that
   $\Hon \not \in \AdvStructure$, as otherwise, $\AdvStructure$ does not satisfy
   the  $\Q^{(3)}(\PartySet, \AdvStructure)$ condition.
    The $\Ready$ messages of the parties in $\Hon$ are eventually delivered to every honest party
   and hence, {\it each} honest party (including $P_h$) eventually executes step $3$ and sends a
   $\Ready$ message for $m^{\star}$.
   It follows that the $\Ready$ messages of {\it all} honest parties $\PartySet \setminus Z^{\star}$
    are eventually delivered
   to every honest party (irrespective of whether $\Adv$ sends all the required
   messages), guaranteeing that all honest parties eventually obtain {\it some} output. To complete the proof,
   we show that this output is the same as $m^{\star}$.
   
    For contradiction, let
   $P_{h'} \neq P_h$ be an honest party who outputs
   $m^{\star \star} \neq m^{\star}$. This implies that
   $P_{h'}$ received $\Ready$ message for $m^{\star \star}$ from at least
    one {\it honest} party.
   From the protocol steps, it follows that an honest party generates a $\Ready$
   message for some potential $m$, only if it either receives $\echo$ messages for $m$ during step 2 from a set of parties $\PartySet \setminus Z$ for some $Z \in \AdvStructure$,
   or $\Ready$ messages for $m$ from a set of parties $C \not \in \AdvStructure$ during step 3.
   So, in order that a subset of parties $\PartySet \setminus Z$ for some $Z \in \AdvStructure$ eventually generates
   $\Ready$ messages for some potential $m$ during step 4, it must be
   the case that some {\it honest} party has received $\echo$ messages for $m$ during step 1 from a set of parties $\PartySet \setminus Z'$ for some
   $Z' \in \AdvStructure$ and has generated a
   $\Ready$ message for $m$.
   
    Since $P_h$ received the $\Ready$ message for $m^{\star}$ from at least one honest party, it must be the case that 
   some honest party has received $\echo$ messages for $m^{\star}$ from a set of parties $\PartySet \setminus Z_1$ for some $Z_1 \in \AdvStructure$.
   Similarly, since $P_{h'}$ received the $\Ready$ message for $m^{\star \star}$ from at least one honest party, it must be the case that 
   some honest party has received $\echo$ messages for $m^{\star \star}$ from a set of parties $\PartySet \setminus Z_2$ for some $Z_2 \in \AdvStructure$.
   Let ${\cal T} = (\PartySet \setminus Z_1) \cap (\PartySet \setminus Z_2)$. Since $\AdvStructure$
   satisfies the $\Q^{(3)}(\PartySet, \AdvStructure)$ condition, it follows
   that $\AdvStructure$
   satisfies the $\Q^{(1)}({\cal T}, \AdvStructure)$ condition and hence ${\cal T}$ is guaranteed to have at least one {\it honest} party. 
   This further implies that there exists some honest party who generated an $\echo$ message for $m^{\star}$ as well
   as $m^{\star \star}$ during step 1, which is impossible. This is because an honest party executes step 1 at most once
   and hence, generates an $\echo$ message at most once.
   Consequently, 
  $\Big \{\REAL_{\PiAcast, \Adv(z), \Env}(m) \Big \}_{m \in \{0, 1\}^{\ell}, z \in \{0, 1 \}^{\star}}
     \equiv \Big \{\IDEAL_{\FAcast, \SimAcast(z), \Env}(m) \Big\}_{m \in \{0, 1\}^{\ell}, z \in \{0, 1 \}^{\star}}$ holds. 
\end{proof}

 \subsection{Asynchronous Byzantine Agreement (ABA)}
 In a {\it synchronous} BA protocol, each party participates with an input bit to obtain an output bit. The protocol guarantees the following three properties.
 \begin{myitemize}
 \item[--] {\it Agreement}: The output bit of all honest parties is the same.
 \item[--] {\it Validity}: If all honest parties have the same input bit, then this will be the common output bit. 
 \item[--] {\it Termination}: All honest parties eventually complete the protocol.
 \end{myitemize}
 In an ABA protocol, the above requirements are slightly weakened, since all (honest) parties {\it may not} be able to provide their inputs to the
 protocol, as waiting for all the inputs may turn out to be an endless wait.
 Hence the decision is taken based on the inputs of a subset of parties $\CoreSet$, where $\PartySet \setminus \CoreSet \in \AdvStructure$.
  Moreover, since the adversary can control the schedule of message delivery,
 it has full control in deciding the set $\CoreSet$. 
 
 The formal specification of an ideal ABA functionality is presented in Fig \ref{fig:FABA}, which is obtained by generalizing the corresponding ideal functionality 
  against {\it threshold} adversaries, as presented in 
  \cite{CGHZ16}. Intuitively, it can be considered as
  a special case of the ideal AMPC functionality (see Fig \ref{fig:FAMPC}), which looks at the set of inputs provided by the
  set of parties in $\CoreSet$, where $\CoreSet$ is decided by the ideal-world adversary. If the input bits provided by all the honest parties in
  $\CoreSet$ are the same, then it is set as the output bit. Else, the output bit is set to be the input bit provided by some corrupt party
  in $\CoreSet$ (for example, the first corrupt party in $\CoreSet$ according to lexicographic ordering).
  In the functionality, the inputs bits provided by various parties are considered to be the ``votes" of the respective parties.
  \begin{systembox}{$\FABA$}{The ideal functionality for asynchronous Byzantine agreement for session id $\sid$.}{fig:FABA}
	\justify
$\FABA$ proceeds as follows, running with the parties $\PartySet = \{P_1, \ldots, P_n \}$ and an adversary $\Sim$, and is parametrized 
 by an adversary-structure $\AdvStructure \subseteq 2^{\PartySet}$. Let $Z^{\star}$ denote the set of corrupt parties, where
  $Z^{\star} \in \AdvStructure$ and let $\Hon = \PartySet \setminus Z^{\star}$. For each party $P_i$, initialize an input value $x^{(i)} = \bot$.
   \begin{enumerate}
       \item Upon receiving a message $(\vote, \sid, b)$ from some $P_i \in \PartySet$ (or from $\Sim$ if $P_i$ is {\it corrupt}) where $b \in \{0, 1 \}$, do the following:
        \begin{myitemize}
        \item Ignore the message if output has been already computed;
        \item Else, set $x^{(i)} = b$ and send $(\vote, \sid, P_i, b)$ to $\Sim$.\footnote{If $P_i \in Z^{\star}$, then no need to send 
        $(\vote, \sid, P_i, b)$ to $\Sim$ as the input has been provided
        by $\Sim$ only.}
        \end{myitemize}
     \item Upon receiving a message $(\coreset, \sid, \CoreSet)$ from $\Sim$, do the following:\footnote{As in the case
      of the AMPC functionality $\Functionality$, $\Sim$ cannot delay sending $\CoreSet$ indefinitely.}
        \begin{myitemize}
        \item Ignore the message if $(\PartySet \setminus \CoreSet) \not \in \AdvStructure$ or if output has been already computed;
        \item Else, record $\CoreSet$.
        \end{myitemize}
      \item If the set $\CoreSet$ has been recorded and the value $x^{(i)}$ has been set to a value different from $\bot$ for every $P_i \in \CoreSet$, then
      compute the output $y$ as follows and generate a request-based delayed output $(\decide, \sid, (\CoreSet, y))$ for every $P_i \in \PartySet$.
        \begin{myitemize}
        \item If $x^{(i)} = b$ holds for all $P_i \in (\Hon \cap \CoreSet)$, then set $y = b$.
        \item Else, set $y = x^{(i)}$, where $P_i$ is the party with the smallest index in $\CoreSet \cap Z^{\star}$.
        \end{myitemize}
   \end{enumerate}
\end{systembox}

%% file: AppPerfectAMPCv2.tex
\section{Properties of the Perfectly-Secure PreProcessing Phase}
\label{app:Perfect}
In this section, we prove the security properties of all the perfectly-secure subprotocols, followed by the perfectly-secure preprocessing phase.
 We first start with the perfectly-secure VSS.
\subsection{Asynchronous VSS Protocol}
\label{app:PerfectVSS}
In this section, we recall the perfectly-secure VSS protocol $\PiPerVSS$ from \cite{CP20}. The protocol is designed with respect to an adversary structure
 $\AdvStructure$ and a sharing specification $\ShareSpec = (S_1, \ldots, S_h) \defined \{ \PartySet \setminus Z | Z \in \AdvStructure\}$, such that 
  $\AdvStructure$
   satisfies the
   $\Q^{(4)}(\PartySet, \AdvStructure)$ condition (this automatically implies that $\ShareSpec$ satisfies the 
    $\Q^{(3)}(\ShareSpec, \AdvStructure)$
   condition). The input for the dealer $P_{\D}$ in the protocol is a vector of shares
    $(s_1, \ldots, s_h)$, the goal being to ensure that the parties output a secret-sharing of $s \defined s_1 + \ldots + s_h$, such that 
    $[s]_q = s_q$, for each $S_q \in \ShareSpec$. 
    The protocol guarantees that even if $P_{\D}$ is {\it corrupt}, if some honest party completes the protocol, then every honest party eventually completes the protocol such
    that there exists some value which has been secret-shared by $P_{\D}$.
   
   The high level idea of the protocol is as follows: the dealer gives the share
    $s_q$ to all the parties in the set $S_q \in \ShareSpec$. To verify whether the dealer has distributed the {\it same} share to all the parties in $S_q$, the 
    parties in $S_q$ perform pairwise consistency tests of the supposedly common share and {\it publicly} announce the result. 
    Next, the parties check if there exists a subset of ``core" parties $\C_q$, where $S_q \setminus \C_q \in \AdvStructure$, who have positively confirmed the pairwise
    consistency of their supposedly common share. Such a subset $\C_q$ is guaranteed for an {\it honest} dealer, as the set of honest parties in $S_q$ always 
    constitutes a candidate set for $\C_q$. To ensure that all honest parties have the same version of the core sets $\C_1, \ldots, \C_h$, the
    dealer is assigned the task of identifying these sets based on the results of the pairwise consistency tests, and making them public. 
    Once the core sets are identified and verified, it is guaranteed that the
        dealer has distributed some common share  to all {\it honest} parties within $\C_q$.
    The next goal is to ensure that even the honest parties in $S_q \setminus \C_q$ get this common share, which is required as per the semantics of our
    secret-sharing.
    For this, the (honest) parties in $S_q \setminus \C_q$ ``filter" out the supposedly common shares received during the pairwise consistency tests
    and ensure that they obtain the common share held by the honest parties in $\C_q$. 
    Protocol $\PiPerVSS$ is formally presented in Fig \ref{fig:PerAVSS}.
\begin{protocolsplitbox}{$\PiPerVSS$}{The perfectly-secure VSS protocol for realizing $\FVSS$ in the $\FAcast$-hybrid model. 
 The above steps are executed by every $P_i \in \PartySet$}{fig:PerAVSS}
\justify

\begin{myitemize}
\item {\bf Distribution of Shares by $P_{\D}$} : If $P_i$ is the dealer $P_{\D}$, then execute the following steps.
	\begin{myenumerate}
	\item On having the shares
	 $s_1, \ldots, s_h \in \F$, send $(\dist, \sid, P_{\D}, q, [s]_q)$ to all the parties $P_i \in S_q$, 
	 corresponding to each $S_q \in \ShareSpec$, where $s \defined s_1 + \ldots + s_h$ and $[s]_q = s_q$.
	\end{myenumerate} 
\item {\bf Pairwise Consistency Tests and Public Announcement of Results} : For each $S_q \in \ShareSpec$, if $P_i \in S_q$, then execute the following steps.
	\begin{myenumerate}
	\item Upon receiving $(\dist, \sid, P_{\D}, q, s_{qi})$ from $\D$, send $(\test, \sid, P_{\D}, q, s_{qi})$ to every party $P_j \in S_q$.
	\item Upon receiving $(\test, \sid, P_{\D}, q, s_{qj})$ from $P_j \in S_q$,
	send $(\Sender, \Acast, \sid^{(P_{\D}, q)}_{ij}, \OK_q(i, j))$ to $\FAcast$ if 
	$s_{qi} = s_{qj}$, where $\sid^{(P_{\D}, q)}_{ij} = \sid || P_{\D} || q || i || j$.\footnote{The notation $\sid^{(P_{\D}, q)}_{ij}$ is used here to distinguish among the different calls to
	$\FAcast$ within the session $\sid$.}
	\end{myenumerate}
\item {\bf Constructing Consistency Graph} : For each $S_q \in \ShareSpec$, execute the following steps.
	\begin{myenumerate}
	\item Initialize $\C_q$ to $\emptyset$. Construct an undirected consistency graph $G^{(i)}_q$ with $S_q$ as the vertex set.
         \item For every ordered pair of parties $(P_j, P_k)$ where $P_j, P_k \in S_q$, 
          keep requesting for an output from $\FAcast$ with $\sid^{(P_{\D}, q)}_{jk}$, till an output is received.
	\item Add the edge $(P_j, P_k)$ to $G^{(i)}_q$ if
	outputs $(P_j, \Acast, \sid^{(P_{\D}, q)}_{jk}, \OK_q(j, k))$ and 
	$(P_k, \Acast, \sid^{(P_{\D}, q)}_{kj}, \allowbreak \OK_q(k, j))$ are received from $\FAcast$ with $\sid^{(P_{\D}, q)}_{jk}$ and 
	$\FAcast$ with $\sid^{(P_{\D}, q)}_{kj}$ respectively.
	\end{myenumerate}
\item {\bf Identification of Core Sets and Public Announcements} :  If $P_i$ is the dealer $P_{\D}$, then execute the following steps.
	\begin{myenumerate}
	\item For each $S_q \in \ShareSpec$, check if there exists a subset of parties $\W_q \subseteq S_q$, such that $S_q \setminus \W_q \in \AdvStructure$
	and the parties in $\W_q$ form a clique in the consistency graph $G^{\D}_q$. If such a $\W_q$ exists, then assign $\C_q \coloneqq \W_q$.
	\item Upon computing the sets $\C_1, \ldots, \C_h$,
	 send $(\Sender, \Acast, \sid_{P_{\D}}, \allowbreak \{ \C_q \}_{S_q \in \ShareSpec})$ to $\FAcast$, where $\sid_{P_{\D}} = \sid || P_{\D}$.
	\end{myenumerate}
\item {\bf Share Computation} : Execute the following steps.
	\begin{myenumerate}
	\item[1.] For each $S_q \in \ShareSpec$ such that $P_i \in S_q$, initialize $[s]_q $ to $\bot$.
	\item[2.] Keep requesting for an output from $\FAcast$ with $\sid_{P_\D}$ until an output is received.
	\item[3.] Upon receiving an output $(\Sender, \Acast, \sid_{P_{\D}}, \{ \C_q \}_{S_q \in \ShareSpec})$ from $\FAcast$ with $\sid_{P_{\D}}$, 
	wait until the parties in $\C_q$ form a clique in $G^{(i)}_q$, for $q = 1, \ldots, h$. Then, verify if
	$S_q \setminus \C_q \in \AdvStructure$, for each $q = 1, \ldots, h$. 
	If the verification is successful, then proceed to compute the output
	as follows.
	    \begin{myenumerate}
	    \item[i.] Corresponding to each $\C_q$ such that $P_i \in \C_q$, set $[s]_q \coloneqq s_{qi} $.
	    \item[ii.] Corresponding to each $\C_q$ such that $P_i \not \in \C_q$, set $[s]_q \coloneqq s_{q}$, where $(\test, \sid, P_{\D}, q, s_{q})$ 
	     is received from a set of parties $\C'_q$
	    such that $\C_q \setminus \C'_q \in \AdvStructure$. 
	    \end{myenumerate}
	\item[4.] Once $[s]_q \neq \bot$ for each $S_q \in \ShareSpec$ such that $P_i \in S_q$, output $(\Share, \sid,P_{\D}, \{[s]_q\}_{P_i \in S_q})$.
	\end{myenumerate}
\end{myitemize}
\end{protocolsplitbox}

We next prove the security of the protocol $\PiPerVSS$.
\begin{theorem}
\label{thm:PerVSSOld}
 Consider a static malicious adversary $\Adv$ characterized by an adversary-structure $\AdvStructure$, satisfying the $\Q^{(4)}(\PartySet, \AdvStructure)$ condition
  and let  $\ShareSpec = (S_1, \ldots, S_h) \defined \{ \PartySet \setminus Z | Z \in \AdvStructure\}$ be the sharing specification.\footnote{Hence $\ShareSpec$
   satisfies the $\Q^{(3)}(\ShareSpec, \AdvStructure)$ condition.} 
  Then protocol $\PiPerVSS$ UC-securely realizes the functionality $\FVSS$
   with perfect security in the $\FAcast$-hybrid model, in the 
   presence of $\Adv$.
\end{theorem}
\begin{proof}
 Let $\Adv$ be an arbitrary adversary corrupting a set of parties
  $Z^{\star} \in \AdvStructure$. Let $\Env$ be an arbitrary environment. We show the existence of a simulator $\SimPVSS$, such that for any
 $Z^{\star} \in \AdvStructure$,
  the outputs of the honest parties and the view of the adversary in the 
   protocol $\PiPerVSS$ is indistinguishable from the outputs of the honest parties and the view of the adversary in an execution in the ideal world involving 
  $\SimPVSS$ and $\FVSS$. The steps of the simulator will be different depending on whether the dealer is corrupt of honest.

If the dealer is {\it honest}, then the simulator interacts with $\FVSS$ and receives the shares of the corrupt parties corresponding to the sets $S_q \in \ShareSpec$ 
 which they are part of. With these shares, the simulator then plays the role of the dealer as well as the honest parties, as per the steps of $\PiPerVSS$, and interacts with
  $\Adv$. The simulator also plays the role of $\FAcast$. If $\Adv$ queries $\FAcast$ for the result of any pairwise consistency test involving an honest party, the simulator provides the appropriate result. In addition, the simulator records the result of any test involving corrupt parties which $\Adv$  sends to $\FAcast$. 
   Based on the results of these pairwise consistency tests, the simulator finds the core sets for each $S_q$ and sends these to $\Adv$ upon request.

If the dealer is {\it corrupt}, the simulator plays the role of honest parties and interacts with $\Adv$, as per the
 steps of  $\PiPerVSS$. This involves recording shares which $\Adv$ distributes to any
  honest party (on the behalf of the dealer),
   as well as performing pairwise consistency tests on their behalf. If $\Adv$ sends core sets for each $S_q \in \ShareSpec$ as input to $\FAcast$, then the simulator checks if these are {\it valid}, and accordingly, sends the shares held by {\it honest} parties in these core sets as the input shares to $\FVSS$ on the behalf
   of the dealer. The simulator is presented in Figure \ref{fig:SimPVSS}.

\begin{simulatorsplitbox}{$\SimPVSS$}{Simulator for the protocol $\PiPerVSS$ where $\Adv$ corrupts the parties in set $Z^{\star} \in \AdvStructure$}{fig:SimPVSS}
	\justify
$\SimPVSS$ constructs virtual real-world honest parties and invokes the real-world adversary $\Adv$. The simulator simulates the view of
 $\Adv$, namely its communication with $\Env$, the messages sent by the honest parties and the interaction with $\FAcast$. 
  In order to simulate $\Env$, the simulator $\SimPVSS$ forwards every message it receives from 
   $\Env$ to $\Adv$ and vice-versa.  The simulator then simulates the various phases of the protocol as follows, depending upon whether the dealer is honest or corrupt. \\~\\
 \centerline{\underline{\bf Simulation When $P_{\D}$ is Honest}} 
\noindent {\bf Interaction with $\FVSS$}: The simulator interacts with the functionality $\FVSS$ and receives a request based delayed output
  $(\Share,\sid,P_{\D},\{[s]_q\}_{S_q \cap Z^{\star} \neq \emptyset})$, on the behalf of the parties in $Z^{\star}$. \\[.2cm]
\noindent {\bf Distribution of Shares by $P_{\D}$}: On the behalf of the dealer, the simulator sends $(\dist, \sid, P_{\D}, q, [s]_q)$ to $\Adv$, corresponding to 
 every $P_i \in Z^{\star} \cap S_q$. \\[.2cm]
 \noindent {\bf Pairwise Consistency Tests}: For each $S_q \in \SharingSpec$ such that $S_q \cap Z^{\star} \neq \emptyset$,
  corresponding to each $P_i \in S_q \cap Z^{\star}$, the simulator does the following.
 	\begin{myitemize}
 	\item[--]  On the behalf of every party $P_j \in S_q \setminus Z^{\star}$, send $(\test, \sid, P_{\D}, q, s_{qj})$ to $\Adv$, where $s_{qj} = [s]_q$.
 	\item[--] If $\Adv$ sends $(\test, \sid, P_{\D}, q, s_{qi})$ on the behalf of $P_i$ to any $P_j \in S_q$, then record it.
 	\end{myitemize}
 \noindent {\bf Announcing Results of Consistency Tests}: 
\begin{myitemize}
\item[--] If for any $S_q \in \SharingSpec$, $\Adv$ requests an output from $\FAcast$ with $\sid^{(P_{\D}, q)}_{ij}$ corresponding to parties
 $P_i \in S_q \setminus Z^{\star}$ and $P_j \in S_q$, then the simulator provides output on the behalf of $\FAcast$ as follows.
 	\begin{myitemize}
 	\item If $P_j \in S_q \setminus Z^{\star}$, then send the output $(P_i, \Acast, \sid^{(P_{\D}, q)}_{ij}, \OK_q(i,j))$.
 	\item If $P_j \in (S_q \cap Z^{\star})$, then send the output $(P_i, \Acast ,\sid^{(P_{\D}, q)}_{ij}, \OK_q(i,j))$, if 
	the message $(\test, \sid, P_{\D}, \allowbreak q, s_{qj})$ has been recorded 
	         on the behalf of $P_j$ for party $P_i$ and 
	       $s_{qj} = [s]_q$ holds.	         	       
	\end{myitemize} 
\item[--] If for any $S_q \in \SharingSpec$ and any $P_i \in S_q \cap Z^{\star}$, 
 $\Adv$ sends $(P_i, \Acast, \sid^{(P_{\D}, q)}_{ij}, \OK_q(i,j))$ to $\FAcast$
 with $\sid^{(P_{\D}, q)}_{ij}$ on the behalf of 
  $P_i$ for any $P_j \in S_q$, then the simulator records it. 
  Moreover, if $\Adv$ requests for an output from $\FAcast$ with $\sid^{(P_{\D}, q)}_{ij}$, then 
  the simulator sends the output $(P_i, \Acast ,\sid^{(P_{\D}, q)}_{ij}, \OK_q(i,j))$ on the behalf of $\FAcast$.
\end{myitemize}
\noindent {\bf Construction of Core Sets and Public Announcement}: 
\begin{myitemize}
\item[--] For each $S_q \in \ShareSpec$, the simulator plays the role of $P_{\D}$ and adds the edge $(P_i,P_j)$ to the graph $G^{\D}_q$ over the vertex set $S_q$, if the following hold.
	\begin{myitemize}
	\item $P_i,P_j \in S_q$.
	\item One of the following is true.
		\begin{myitemize}
		\item $P_i,P_j \in S_q \setminus Z^{\star}$.
		\item If $P_i \in S_q \cap Z^{\star}$ and $P_j \in S_q \setminus Z^{\star}$, then the simulator has recorded 
		$(P_i, \Acast, \sid^{(P_{\D}, q)}_{ij}, \OK_q(i,j))$ sent by $\Adv$ on the behalf of $P_i$ to $\FAcast$ with $\sid^{(P_{\D}, q)}_{ij}$,
		 and recorded $(\test, \sid, P_{\D}, q, s_{qi})$ on the behalf of $P_i$ for $P_j$ such that $s_{qi} = [s]_q$.
		\item If $P_i, P_j \in S_q \cap Z^{\star}$, then the simulator has recorded $(P_i, \Acast, \sid^{(q)}_{ij}, \OK_q(i,j))$
		 and $(P_j, \Acast,  \sid^{(q)}_{ji}, \allowbreak \OK_q(j,i))$ sent by $\Adv$ on behalf $P_i$ and $P_j$ to 
		 $\FAcast$ with $\sid^{(P_{\D}, q)}_{ij}$ and $\sid^{(P_{\D}, q)}_{ji}$ respectively.
		\end{myitemize}
	\end{myitemize}
\item[--] For each $S_q \in \ShareSpec$, the simulator finds the set $\C_q$ which forms a clique in $G^{\D}_q$, such that $S_q \setminus \C_q \in \AdvStructure$.
  When $\Adv$ requests output from $\FAcast$ with $\sid_{P_{\D}}$, 
  the simulator sends the output $(\Sender, \Acast, \sid_{P_{\D}}, \{\C_q\}_{S_q \in \ShareSpec})$ on the behalf of $\FAcast$.\\~\\
\end{myitemize}
 \centerline{\underline{\bf Simulation When $P_{\D}$ is Corrupt}} 
  In this case, the simulator $\SimPVSS$ interacts with $\Adv$ during the various phases of $\PiPerVSS$ as follows. \\[.2cm]
\noindent {\bf Distribution of shares by $P_{\D}$}: For $q = 1, \ldots, h$, if $\Adv$ sends $(\dist, \sid, P_{\D}, q, v)$ 
 on the behalf of $P_{\D}$ to any party $P_i \in S_q \setminus Z^{\star}$, then the simulator records it and sets $s_{qi}$ to be $v$. \\[.2cm]
 \noindent {\bf Pairwise Consistency Tests}: For each $S_q \in \SharingSpec$ such that 
  $S_q \cap Z^{\star} \neq \emptyset$, corresponding to each party $P_i \in S_q \cap Z^{\star}$ and each 
  $P_j \in S_q \setminus Z^{\star}$, the simulator does the following.
 	\begin{myitemize}
 	\item[--] If $s_{qj}$ has been set to some value, then send $(\test,\sid, P_{\D}, q, s_{qj})$ to $\Adv$ on the behalf of $P_j$.
 	\item[--] If $\Adv$ sends $(\test, \sid, P_{\D}, q, s_{qi})$ on the behalf of $P_i$ to $P_j$, then record it. 
 	\end{myitemize}
 \noindent {\bf Announcing Results of Consistency Tests}: 
\begin{myitemize}
\item[--] If for any $S_q \in \SharingSpec$, $\Adv$ requests an output from $\FAcast$ with $\sid^{(P_{\D}, q)}_{ij}$ corresponding to
 parties $P_i \in S_q \setminus Z^{\star}$ and $P_j \in S_q$, then the simulator provides the 
  output on the behalf of $\FAcast$ as follows, if $s_{qi}$ has been set to some value.
 	\begin{myitemize}
 	\item If $P_j \in S_q \setminus Z^{\star}$, then send the output $(P_i,\Acast,\sid^{(P_{\D}, q)}_{ij},\OK_q(i,j))$, if
	$s_{qj}$ has been set to some value and $s_{qi} = s_{qj}$ holds.
 	\item If $P_j \in S_q \cap Z^{\star}$, then send the output $(P_i,\Acast,\sid^{(P_{\D}, q)}_{ij},\OK_q(i,j))$, 
	 if $(\test, \sid, P_{\D}, q, s_{qj})$ sent by $\Adv$ on the behalf of $P_j$ to $P_i$ has been recorded and $s_{qj} = s_{qi}$ holds.
 	\end{myitemize} 
\item[--]  If for any $S_q \in \SharingSpec$ and any $P_i \in S_q \cap Z^{\star}$, 
 $\Adv$ sends $(P_i, \Acast, \sid^{(P_{\D}, q)}_{ij}, \OK_q(i,j))$ to $\FAcast$
 with $\sid^{(P_{\D}, q)}_{ij}$ on the behalf of 
  $P_i$ for any $P_j \in S_q$, then the simulator records it. 
  Moreover, if $\Adv$ requests for an output from $\FAcast$ with $\sid^{(P_{\D}, q)}_{ij}$, then 
  the simulator sends the output $(P_i, \Acast ,\sid^{(P_{\D}, q)}_{ij}, \OK_q(i,j))$ on the behalf of $\FAcast$.
  \end{myitemize}
\noindent {\bf Construction of Core Sets}: For each $S_q \in \ShareSpec$, the simulator plays the role of 
 the honest parties $P_i \in S_q \setminus Z^{\star}$ and adds the edge $(P_j,P_k)$ to the graph $G^{(i)}_q$ over vertex set $S_q$, if the following hold.
\begin{myitemize}
	\item[--] $P_j, P_k \in S_q$.
	\item[--] One of the following is true.
		\begin{myitemize}
		\item If $P_j, P_k \in S_q \setminus Z^{\star}$, then the simulator has set $s_{qj}$ and $s_{qk}$ to some values, such that $s_{qj} = s_{qk}$.
		\item If $P_j \in S_q \cap Z^{\star}$ and $P_k \in S_q \setminus Z^{\star}$, then the simulator has recorded 
		$(P_j, \Acast, \sid^{(P_{\D}, q)}_{jk}, \OK_q(j, k))$ sent by $\Adv$ on the behalf of $P_j$ to $\FAcast$ with $\sid^{(P_{\D}, q)}_{jk}$,
		and recorded $(\test, \sid, P_{\D}, q, s_{qj})$ on the behalf of $P_j$ for $P_k$ and has set $s_{qk}$ to a value such that $s_{qj} = s_{qk}$.
		\item If $P_j, P_k \in S_q \cap Z^{\star}$, then the simulator has recorded 
		$(P_j, \Acast, \sid^{(P_{\D}, q)}_{jk}, \OK_q(j, k))$ and $(P_k, \Acast,  \sid^{(P_{\D}, q)}_{kj}, \allowbreak \OK_q(k, j))$ sent by $\Adv$ on behalf of
		$P_j$ and $P_k$ respectively to $\FAcast$ with $\sid^{(P_{\D}, q)}_{jk}$ and $\FAcast$ with $\sid^{(P_{\D}, q)}_{kj}$. 		
		\end{myitemize}
\end{myitemize}
\noindent {\bf Verification of Core Sets and Interaction with $\FVSS$}:
      \begin{myitemize}
        \item If  $\Adv$ sends $(\Sender, \Acast, \sid_{P_{\D}}, \{\C_q\}_{S_q \in \ShareSpec})$ to $\FAcast$ with $ \sid_{P_{\D}}$ on the behalf of
  $P_{\D}$, then the simulator records it. Moreover, if $\Adv$ requests an output from $\FAcast$ with $ \sid_{P_{\D}}$, then on the behalf
  of $\FAcast$, the simulator sends the output $(P_{\D}, \Acast , \sid_{P_{\D}}, \{\C_q\}_{S_q \in \ShareSpec})$.
       \item If simulator has recorded the sets  $ \{\C_q\}_{S_q \in \ShareSpec}$, then it plays 
     the role of the honest parties and verifies if $\C_1, \ldots, \C_h$ are valid  by checking
     if each $S_q \setminus \C_q \in \AdvStructure$ and if each $\C_q$ constitutes a clique in the graph $G^{(i)}_q$ of every party $P_i \in \PartySet \setminus Z^{\star}$.
    If $\C_1, \ldots, \C_h$ are valid, then the simulator sends $(\Share, \sid, P_{\D}, \{s_q\}_{S_q \in \ShareSpec})$ to $\FVSS$, where $s_q$ is set to 
     $s_{qi}$ corresponding to any $P_i \in \C_q \setminus Z^{\star}$. 
     \end{myitemize}
     
\end{simulatorsplitbox}   

We now prove a series of claims which will help us prove the theorem. We start with an {\it honest}
 $P_{\D}$.
\begin{claim}
\label{claim:PVSSHonestDPrivacy}
If $P_{\D}$ is honest, then the view of $\Adv$ in the simulated execution of $\PiPerVSS$ with $\SimPVSS$ is identically distributed to the view of $\Adv$ in the real execution of $\PiPerVSS$ involving honest parties.
\end{claim}
\begin{proof}
Let $\ShareSpec^{\star} \defined
  \{S_q \in \ShareSpec \mid S_q \cap Z^{\star} \neq \emptyset \}$. Then the 
   view of $\Adv$ during the various executions consists of the following.
\begin{myenumerate}
\item[--] {\bf The shares $\{[s]_q\}_{S_q \in \ShareSpec^{\star}}$ distributed by $P_{\D}$}: In the real execution, $\Adv$ receives $[s]_q$ from $P_{\D}$ for each $S_q \in \ShareSpec^{\star}$. In the simulated execution, the simulator provides this to $\Adv$ on behalf of $P_{\D}$. Clearly, the distribution of the shares is identical in both the
 executions.
\item[--] {\bf Corresponding to every $S_q \in \ShareSpec^{\star}$, 
  messages $(\test, \sid, P_{\D}, q, s_{qj})$ received from party $P_j \in S_q \setminus Z^{\star}$,
    as part of pairwise consistency tests, where $s_{qj} = [s]_q$}: While each $P_j$ sends this to $\Adv$ in the real execution, the simulator sends this on the behalf of $P_j$ in the simulated execution.
    Clearly, the distribution of the messages is identical in both the executions.    
\item[--] {\bf For every $S_q \in \ShareSpec$ and every $P_i,P_j \in S_q$, the outputs $\OK_q(P_i, \Acast, \sid^{(P_{\D}, q)}_{ij}, \OK_q(i,j))$
 of the pairwise consistency tests, received as output from $\FAcast$ with $\sid^{(P_{\D}, q)}_{ij}$}: 
 To compare the distribution of these messages in the two executions,
  we consider the following cases, considering an arbitrary $S_q \in \ShareSpec$ and arbitrary $P_i, P_j \in S_q$.
	\begin{myitemize}
	\item[--] $P_i,P_j \in S_q \setminus Z^{\star}$: In both the executions, 
	 $\Adv$ receives $\OK_q(P_i, \Acast, \sid^{(P_{\D}, q)}_{ij}, \OK_q(i,j))$ as the output from $\FAcast$ with $\sid^{(P_{\D}, q)}_{ij}$.
	\item[--] $P_i \in S_q \setminus Z^{\star}, P_j \in (S_q \cap Z^{\star}) $:  In both the 
	 executions, $\Adv$ receives $\OK_q(P_i, \Acast, \sid^{(P_{\D}, q)}_{ij}, \allowbreak \OK_q(i,j))$ as the output from $\FAcast$ with 
	 $\sid^{(P_{\D}, q)}_{ij}$ if and only if $\Adv$ sent $(\test, \sid, P_{\D}, q, s_{qj})$ on the behalf of $P_j$ to $P_i$ such that $s_{qj} = [s]_q$ holds.
	\item[--] $P_i \in (S_q \cap Z^{\star}) $:  In both the executions, 
	$\Adv$ receives $\OK_q(P_i, \Acast, \sid^{(q)}_{ij}, \OK_q(i,j))$ if and only if 
	$\Adv$ on the behalf of $P_i$ has sent $(P_i, \Acast, \sid^{(P_{\D}, q)}_{ij}, \OK_q(i,j))$ to $\FAcast$
	with $\sid^{(P_{\D}, q)}_{ij}$ for $P_j$.
	\end{myitemize}
   Clearly, irrespective of the case, the distribution of the $\OK_q$ messages is identical in both the executions.
\item[--] {\bf The core sets $\{\C_q\}_{S_q \in \ShareSpec}$}: In both the 
 executions, the sets $\C_q$ are determined based on the $\OK_q$ messages delivered to $P_{\D}$. So the distribution of these sets is also identical. 
\end{myenumerate}
\end{proof}

We next claim that if the dealer is {\it honest}, then conditioned on the view of the adversary $\Adv$ (which is identically distributed in both the executions, as per the previous claim),
  the outputs of the honest parties are identically distributed in both the executions.
\begin{claim}
\label{claim:PVSSHonestDCorrectness}
If $P_{\D}$ is honest, then conditioned on the view of $\Adv$, the outputs of the honest parties during the execution of $\PiPerVSS$ involving $\Adv$ has the
   same distribution as the outputs of the honest parties in the ideal-world involving $\SimPVSS$ and $\FVSS$.
\end{claim}
\begin{proof}
Let $P_{\D}$ be honest and let $\View$ be an arbitrary view of $\Adv$. Moreover, let $\{s_q \}_{S_q \cap Z^{\star} \neq \emptyset}$ be the shares
 of the corrupt parties, as per $\View$. Furthermore, let $\{s_q \}_{S_q \cap Z^{\star} = \emptyset}$ be the shares used by $P_{\D}$
  in the simulated execution, corresponding to the set $S_q \in \ShareSpec$, such that $S_q \cap Z^{\star} = \emptyset$.
   Let $s \defined \displaystyle \sum_{S_q \cap Z^{\star} \neq \emptyset} s_q + \sum_{S_q \cap Z^{\star} = \emptyset} s_q$.
   Then in the simulated execution, each {\it honest} party $P_i$ obtains the output $\{[s]_q\}_{P_i \in S_q}$ from $\FVSS$, where
   $[s]_q = s_q$. We now show that $P_i$ eventually obtains the
    output $\{[s]_q\}_{P_i \in S_q}$ in the real execution as well, if $P_{\D}$'s inputs in the protocol $\PiPerVSS$ are 
    $\{s_q \}_{S_q \in \ShareSpec}$.
    
    Since $P_{\D}$ is {\it honest}, it sends the share $s_q$ to {\it all} the parties in the set $S_q$, which is eventually delivered. Now consider 
    an {\it arbitrary} $S_q \in \ShareSpec$.
     During the pairwise consistency tests, each {\it honest} $P_k \in S_q$ will eventually send $s_{qk} = s_q$ to 
      {\it all} the parties in $S_q$. Consequently, every {\it honest} $P_j \in S_q$ will eventually broadcast the message $\OK_q(j, k)$,
       corresponding to every {\it honest} $P_k \in S_q$. This is because $s_{qj} = s_{qk} = s_q$ will hold.
       These $\OK_q(j, k)$ messages are eventually received by every honest party, including $P_\D$.       
      This implies that the parties in $S_q \setminus Z^{\star}$ will eventually form a clique in the graph $G^{(i)}_q$
      of every {\it honest} $P_i$. This further implies that $P_{\D}$ will eventually find a set $\C_q$ where $S_q \setminus \C_q \in \AdvStructure$
      and where $\C_q$ constitutes a clique in the consistency graph of every honest party. 
      This is because the set $S_q \setminus Z^{\star}$ is guaranteed to eventually constitute a clique.
            Hence $P_{\D}$ eventually broadcasts the sets $\{ \C_q\}_{S_q \in \ShareSpec}$, which are eventually delivered to every honest
      party. Moreover, the verification of these sets will eventually be successful for every honest party.
      
      Next, consider an arbitrary {\it honest} $P_i \in S_q$. If $P_i \in \C_q$, then it has already received the share $s_q$ from $P_{\D}$
      and $s_{qi} = s_q$ holds. Hence, $P_i$ sets $[s]_q$ to $s_q$. So consider the case when $P_i \not \in \C_q$. In this case, 
      $P_i$ sets $[s]_q$ based on the supposedly common values $s_{qj}$ received from the parties
      $P_j \in S_q$ as part of pairwise consistency tests.
       Specifically, $P_i$ checks for a subset of parties $\C'_q \subseteq \C_q$, where $\C_q \setminus \C'_q \in \AdvStructure$, such that
      every party $P_j \in \C'_q$ has sent the {\it same} $s_{qj}$ value to $P_i$ as part of the pairwise consistency test.
      If $P_i$ finds such a set $\C'_q$, then it sets $[s]_q$ to the common $s_{qj}$.
      To complete the proof, we need to show that $P_i$ will eventually find such a set $\C'_q$,
      and if such a set $\C'_q$ is found by $P_i$, then the common $s_{qj}$ is the same as $s_q$.
      
      Assuming that $P_i$ eventually finds such a $\C'_q$, the proof that the common $s_{qj}$ is the same as $s_q$
      follows from the fact that $\C'_q$ is guaranteed to contain at least one {\it honest} party from $\C_q$, who would have received the share
      $s_{qj} = s_q$ from $P_{\D}$ and sent to $P_i$ as part of the pairwise consistency test. This is because $\AdvStructure$
   satisfies the $\Q^{(4)}(\PartySet, \AdvStructure)$ condition.
      Also, since the $\Q^{(4)}(\PartySet, \AdvStructure)$ condition is satisfied, the set of {\it honest} parties in $\C_q$, namely the parties in 
      $\C_q \setminus Z^{\star}$, always constitute a candidate $\C'_q$ set. This is because every party $P_j \in \C_q \setminus Z^{\star}$
      would have sent $s_{qj} = s_q$ to every party in $S_q$ during the pairwise consistency test, and these values are eventually delivered. 
\end{proof}

We next prove certain claims with respect to a {\it corrupt} dealer. The first claim is that the view of $\Adv$ in this
case is also identically distributed in both the real as well as simulated execution. This is simply
because in this case, the {\it honest} parties have {\it no} inputs
  and the simulator simply plays the role of the honest parties {\it exactly}
 as per the steps of the protocol $\PiPerVSS$ in the simulated execution.

\begin{claim}
\label{claim:PVSSCorruptDPrivacy}
If $P_{\D}$ is corrupt, then 
the view of $\Adv$ in the simulated execution of $\PiPerVSS$ with $\SimPVSS$ is identically distributed as the view of $\Adv$ in the real execution of
$\PiPerVSS$ involving honest parties.
\end{claim}
\begin{proof}
The proof follows from the fact that if $P_{\D}$ is {\it corrupt}, then $\SimPVSS$ participates in a full execution of the protocol
 $\PiPerVSS$, by playing the role of the honest parties as per the steps of $\PiPerVSS$.
 Hence, there is a one-to-one correspondence between simulated
  executions and real executions.
\end{proof}

We finally claim that if the dealer is {\it corrupt}, then conditioned on the view of the adversary (which is identical in both
the executions as per the last claim), the outputs of the honest parties are identically distributed in 
 both the executions.
\begin{claim}
\label{claim:PVSSCorruptDCorrectness}
If $\D$ is corrupt, then conditioned on the view of $\Adv$, the output of the honest parties during the execution of $\PiPerVSS$ involving $\Adv$ has the  same distribution as the output of the honest parties in the ideal-world involving $\SimPVSS$ and $\FVSS$.
\end{claim}
\begin{proof}
Let $P_{\D}$ be {\it corrupt} and let 
 $\View$ be an arbitrary view of $\Adv$. We note that whether valid
   core sets $\{\C_q\}_{S_q \in \ShareSpec}$ have been generated during the corresponding execution of $\PiPerVSS$ or not can be found out from $\View$. We now consider the following cases.
\begin{myitemize}
\item[--] {\it No core sets $\{\C_q\}_{S_q \in \ShareSpec}$ are generated as per $\View$}: In this case, the honest parties do not obtain any output in either execution.
 This is because in the real execution of $\PiPerVSS$, the honest parties compute their output only when they get valid core sets $\{\C_q\}_{S_q \in \ShareSpec}$
 from $P_{\D}$'s broadcast. If this is not the case, then in the simulated execution, 
  the simulator $\SimPVSS$ does not provide any input to $\FVSS$ on behalf of $P_\D$; hence, $\FVSS$ does not produce any output for the honest parties.
\item[--] {\it Core sets $\{\C_q\}_{S_q \in \ShareSpec}$ generated as per $\View$ are invalid}:
 Again, in this case, the honest parties do not obtain any output in either execution.
  This is because in the real execution of $\PiPerVSS$, even if the sets 
   $\{\C_q\}_{S_q \in \ShareSpec}$ are received from $P_{\D}$'s broadcast,
    the honest parties compute their output only when each set $\C_q$ is found to be {\it valid} with respect to the verifications
    performed by the honest parties in their own consistency graphs.
    If these verifications fail (implying that the core sets are invalid), then in the simulated execution, 
    the simulator $\SimPVSS$ does not provide any input to $\FVSS$ on behalf of $P_\D$, implying that
     $\FVSS$ does not produce any output for the honest parties.
\item[--] {\it Valid core sets $\{\C_q\}_{S_q \in \ShareSpec}$ are generated as per $\View$}:
  We first note that in this case, $P_{\D}$ has distributed some common share, say $s_q$, determined by $\View$, 
  to all the parties in $\C_q \setminus Z^{\star}$ during the real execution of $\PiPerVSS$. 
  This is because all the parties in $\C_q \setminus Z^{\star}$ are {\it honest}, and form a clique
  in the consistency graph of the honest parties. Hence, 
  each $P_j, P_k \in \C_q \setminus Z^{\star}$ has broadcasted the messages $\OK_q(j, k)$ and $\OK_q(k, j)$
  after checking that $s_{qj} = s_{qk}$ holds, where $s_{qj}$ and $s_{qk}$ are the shares
  received from $P_{\D}$ by $P_j$ and $P_k$ respectively. 
  
  We next show that in the real execution of $\PiPerVSS$, every party in $S_q \setminus Z^{\star}$, eventually sets
  $[s]_q = s_q$. While this is true for the parties in $\C_q \setminus Z^{\star}$, we consider an arbitrary party
  $P_i \in S_q \setminus (Z^{\star} \cup \C_q)$. From the protocol steps, 
   $P_i$ checks for a subset of parties $\C'_q \subseteq \C_q$ where $\C_q \setminus \C'_q \in \AdvStructure$, such that
      every party $P_j \in \C'_q$ has sent the {\it same} $s_{qj}$ value to $P_i$ as part of the pairwise consistency test.
      If $P_i$ finds such a set $\C'_q$, then it sets $[s]_q$ to the common $s_{qj}$.
      We next argue that $P_i$ will eventually find such a set $\C'_q$
      and if such a set $\C'_q$ is found by $P_i$, then the common $s_{qj}$ is the same as $s_q$.
      The proof for this is exactly the {\it same}, as for Claim \ref{claim:PVSSHonestDCorrectness}. 
      
      Thus, in the real execution, every honest party $P_i$ eventually outputs $\{[s]_q = s_q \}_{P_i \in S_q}$.
      From the steps of $\SimPVSS$, the simulator sends the shares $\{s_q \}_{S_q \in \ShareSpec}$ to $\FVSS$ on the behalf of
      $P_{\D}$ in the simulated execution. Consequently, in the ideal world, $\FVSS$ will eventually deliver
      the shares $\{[s]_q = s_q \}_{P_i \in S_q}$ to every honest $P_i$.
      Hence, the outputs of the honest parties are identical in both the worlds.
\end{myitemize}
\end{proof}
The proof of the theorem now follows
  from Claims \ref{claim:PVSSHonestDPrivacy}-\ref{claim:PVSSCorruptDCorrectness}.
\end{proof}
\paragraph{\bf Reducing the Broadcast Complexity of the Protocol $\PiPerVSS$:}
Protocol $\PiPerVSS$ as presented in \cite{CP20} has a {\it broadcast complexity}, which is proportional to the size of $\AdvStructure$. More specifically, 
 in the protocol, $P_\D$ needs to compute a core set $\C_q$ corresponding to each $S_q \in \ShareSpec$. For finding these $\C_q$ sets, every pair of (honest)
 parties need to broadcast an $\OK_q$ message for each other by calling $\FAcast$. This results in the number of bits being broadcasted, proportional to $|\ShareSpec|$, where
  $|\ShareSpec| = |\AdvStructure|$ in our case. A small modification to the protocol can make the broadcast complexity, {\it independent} of $|\AdvStructure|$.
   The idea is to let every party broadcast a {\it single} $\OK$ message for every other party, if the pairwise consistency test with that party is successful across {\it all}
   the sets $S_q$ to which both the parties belong. In a more detail, party $P_i$ sends an $\OK(i, j)$ message to $\FAcast$, only after checking whether
   $s_{qi} = s_{qj}$ holds corresponding to {\it every} $S_q \in \ShareSpec$, such that $P_j \in S_q$ holds. Consequently, $P_\D$ now checks for the presence
   of a {\it single} core set $\C$ where for $q = 1, \ldots, |\ShareSpec|$, the conditions $\C \subseteq S_q$ and $S_q \setminus \C \in \AdvStructure$ hold. Upon finding such a
   $\C$ the dealer broadcasts it by sending it to $\FAcast$. Note that such a $\C$ is eventually obtained for an {\it honest} $P_\D$. This is because the set
   of parties $(S_1 \setminus Z^{\star}) \cap \ldots \cap (S_q \setminus Z^{\star})$ constitutes a candidate $\C$ for an honest $P_\D$, where
   $Z^{\star}$ is the set of {\it corrupt} parties. The rest of the protocol steps remain the same. With these modifications, the communication complexity of the 
   protocol $\PiPerVSS$ is computed as follows: the dealer needs to send the share $s^{(q)}$ to all the parties in $S_q$, and every party in $S_q$ has to send
 the received share to every other party in $S_q$ during pairwise consistency tests. This incurs a communication of $\Order(|\AdvStructure| \cdot n^2 \log{|\F|} )$
  bits, since each $|S_q| = \Order(n)$ and each share $s^{(q)}$ can be represented by $\log{|\F|}$ bits. There will be total $\Order(n^2)$ $\OK$ messages broadcasted,
  where each message can be represented by $\Order(\log n)$ bits, since it represents the index of $2$ parties.
  Moreover, $P_\D$ will broadcast a single core set $\C$ of size $\Order(n \log n)$ bits.
  Based on this discussion, we next state the following theorem for $\PiPerVSS$.
  
\begin{theorem}
\label{thm:PerVSS}
 Consider a static malicious adversary $\Adv$ characterized by an adversary-structure $\AdvStructure$, satisfying the $\Q^{(4)}(\PartySet, \AdvStructure)$ condition
  and let  $\ShareSpec = (S_1, \ldots, S_h) \defined \{ \PartySet \setminus Z | Z \in \AdvStructure\}$ be the sharing specification.\footnote{Hence $\ShareSpec$
   satisfies the $\Q^{(3)}(\ShareSpec, \AdvStructure)$ condition.} 
  Then protocol $\PiPerVSS$ UC-securely realizes the functionality $\FVSS$
   with perfect security in the $\FAcast$-hybrid model, in the 
   presence of $\Adv$.
   The protocol makes $\Order(n^2)$ calls to $\FAcast$ with $\Order(\log{n})$ bit messages, one call to $\FAcast$ with $\Order(n \log{n})$ bit message
    and additionally incurs a communication of 
   $\Order(|\AdvStructure| \cdot n^2 \log{|\F|})$ bits.
   
   By replacing the calls to $\FAcast$ with protocol $\PiAcast$, the protocol incurs a total communication of $\Order(|\AdvStructure| \cdot n^2 \log{|\F|} + n^4 \log{n})$
   bits.
\end{theorem}

\subsection{Asynchronous Reconstruction Protocols}
Let $s$ be a value which is secret-shared with respect to some sharing specification $\ShareSpec = (S_1, \ldots, S_h)$, such that $\ShareSpec$ satisfies the
 $\Q^{(2)}(\ShareSpec, \AdvStructure)$ condition. We first present the protocol 
  $\PiPerRecShare$, which allows {\it all} parties in $\PartySet$ to reconstruct a {\it single} share $[s]_q$, corresponding to a
  designated set $S_q \in \ShareSpec$. In the protocol, every party in $S_q$ sends the share $[s]_q$ to all the parties outside $S_q$, who then
  ``filter" out the potentially incorrect versions of $[s]_q$ and output $[s]_q$. 
     Protocol $\PiPerRecShare$ is formally presented in Figure \ref{fig:PerRecShare}.
\begin{protocolsplitbox}{$\PiPerRecShare(q)$}{Perfectly-secure reconstruction protocol for session id $\sid$ to publicly
  reconstruct the share $[s]_q$ corresponding to $S_q \in \ShareSpec$. 
   The public inputs are $\PartySet$, $\AdvStructure$ and $\ShareSpec$. 
    The above steps are executed by every $P_i \in \PartySet$}{fig:PerRecShare}
\justify
\begin{myitemize}
\item {\bf Sending Share to All Parties}: If $P_i \in S_q$, then execute the following steps.
	\begin{myenumerate}
	\item On having $[s]_q$, send $(\Share,\sid,q,[s]_q)$ to all the parties in $\PartySet \setminus S_q$.
	\end{myenumerate} 
\item {\bf Computing Output}: Based on the following conditions, execute the corresponding steps.
	\begin{myenumerate}
	\item {\bf $P_i \in S_q$}: Output $[s]_q$. 
	\item {\bf $P_i \notin S_q$}: Upon receiving $(\Share, \sid, q, v)$ from a set of parties $S'_q \subseteq S_q$ such that
	$S_q \setminus S'_q \in \AdvStructure$, output $[s]_q = v$.
	\end{myenumerate} 
\end{myitemize}
\end{protocolsplitbox}

\begin{lemma}
\label{lemma:PerRecShare}
Let $\AdvStructure$ be an adversary structure and let $\ShareSpec = (S_1, \ldots, S_h)$ be a sharing specification, such that 
 $\ShareSpec$ satisfies the
 $\Q^{(2)}(\ShareSpec, \AdvStructure)$ condition. Moreover, let $s$ be a value, which is secret-shared as per $\ShareSpec$.
  Then for any $q \in \{1, \ldots, h \}$ and any adversary $\Adv$ corrupting a set of parties $Z^{\star} \in \AdvStructure$,
   all honest parties eventually output the share $[s]_q$ in the protocol $\PiPerRecShare$.
   The protocol incurs a communication of $\Order(n^2 \log{|\F|})$ bits.   
\end{lemma}
\begin{proof}
Consider an arbitrary {\it honest} party $P_i \in \PartySet$. We consider two cases.
\begin{itemize}
\item[--] $P_i \in S_q$: In this case, $P_i$ outputs $[s]_q$.
\item[--] $P_i \notin S_q$: In this case, $P_i$ waits for a subset of parties $S'_q \subseteq S_q$ where $S_q \setminus S'_q \in \AdvStructure$, such that
      every party $P_j \in S'_q$ has sent the {\it same} share $v$ to $P_i$.
      If $P_i$ finds such a set $S'_q$, then it outputs $v$.
      To complete the proof, we need to show that $P_i$ will eventually find such a set $S'_q$
      and if such a set $S'_q$ is found by $P_i$, then the common $v$ is the same as $[s]_q$.
      
      Assuming that $P_i$ eventually finds such a common $S'_q$, the proof that the common $v$ is the same as $[s]_q$
      follows from the fact that $S'_q$ is guaranteed to contain at least one {\it honest} party from $S_q$, who would have sent the share
      $[s]_q$ to $P_i$. This is because the $\Q^{(2)}(\ShareSpec, \AdvStructure)$ condition is satisfied.
      Also, since the $\Q^{(2)}(\ShareSpec, \AdvStructure)$ condition is satisfied, the set of {\it honest} parties in $S_q$, namely the parties in 
      $S_q \setminus Z^{\star}$, always constitute a candidate $S'_q$ set. This is because every party $P_j \in S_q \setminus Z^{\star}$
      would have sent $[s]_q$ to $P_i$ and these values are eventually delivered to $P_i$. 
\end{itemize}
The communication complexity follows from the protocol steps.
\end{proof}

We now present the protocol $\PiPerRec$ (Fig \ref{fig:PerRec}), which allows {\it all} parties in $\PartySet$ to reconstruct a secret shared value $s$.
  The idea is to run an instance of $\PiPerRecShare$ for each $S_q \in \ShareSpec$, and to sum up the shares obtained as the
   output from each instantiation. 
\begin{protocolsplitbox}{$\PiPerRec$}{Perfectly-secure reconstruction protocol for session id $\sid$ to reconstruct a shared value $s$. The public inputs of the protocol are $\PartySet, \ShareSpec$
  and $\AdvStructure$. The above steps are executed by every $P_i \in \PartySet$}{fig:PerRec}
\justify

\begin{myitemize}
\item {\bf Reconstructing Shares}: For each $S_q \in \ShareSpec$, participate
 in an instance $\PiPerRecShare(q)$ with $\sid$ to obtain the output $[s]_q$.
\item {\bf Output Computation}: Output $s = \displaystyle \sum_{S_q \in \ShareSpec} [s]_q$.
\end{myitemize}
\end{protocolsplitbox}

The properties of the protocol $\PiPerRec$ are stated in Lemma \ref{lemma:PerRec}, which follow from the protocol steps and Lemma \ref{lemma:PerRecShare}.
\begin{lemma}
\label{lemma:PerRec}
Let $\AdvStructure$ be an adversary structure and let $\ShareSpec = (S_1, \ldots, S_h)$ be a sharing specification, such that 
 $\ShareSpec$ satisfies the
 $\Q^{(2)}(\ShareSpec, \AdvStructure)$ condition.
  Moreover, let $s$ be a value which is secret-shared as per $\ShareSpec$.
  Then for every adversary $\Adv$ corrupting a set of parties $Z^{\star} \in \AdvStructure$,
   all honest parties eventually output $s$ in the protocol $\PiPerRecShare$.
   The protocol incurs a communication of $\Order(|\ShareSpec| \cdot n^2 \log{|\F|})$ bits, which is 
   $\Order(|\AdvStructure| \cdot n^2 \log{|\F|})$ bits if $|\ShareSpec| = |\AdvStructure|$.
\end{lemma}
\subsection{Properties of the Optimistic Multiplication Protocol $\OptMult$}
In this section, we formally prove the properties of the protocol $\OptMult$ (Fig \ref{fig:OptMult}). 
 While proving these properties, we will assume that  $\AdvStructure$ satisfies the $\Q^{(4)}(\PartySet, \AdvStructure)$ condition.
  This further implies that the sharing specification $\ShareSpec = (S_1, \ldots, S_h) \defined \{ \PartySet \setminus Z | Z \in \AdvStructure\}$
  satisfies the $\Q^{(3)}(\ShareSpec, \AdvStructure)$ condition. 
  Moreover, while proving these properties, we also assume that for every $\iter$, 
  no honest party is ever included in the set $\Discarded$ and 
  all honest parties are eventually removed from the $ \Waitlist^{(i)}_{\iter'}, \LocalDiscarded^{(i)}_{\iter'}$ sets of every honest $P_i$ for every
  $\iter' < \iter$. Note that these conditions are guaranteed in the protocols $\MultLCE$ and $\MultGCE$ (where these sets are constructed and managed),
   where  $\OptMult$ is used as a subprotocol. 
  \begin{claim}
  \label{claim:OptMultAdvCondition}
 For every $Z \in \AdvStructure$ and every ordered pair $(p, q) \in \{1, \ldots, h \} \times \{1, \ldots, h \}$,
  the set $(S_p \cap S_q) \setminus Z$ contains at least one honest party.
  \end{claim}
  \begin{proof}
  From the definition of the sharing specification $\ShareSpec$, we have
   $S_p = \PartySet \setminus Z_p$ and $S_q = \PartySet \setminus Z_q$, where $Z_p, Z_q \in \AdvStructure$. 
   Let $Z^{\star} \in \AdvStructure$ be the set of corrupt parties during the protocol $\OptMult$.
   If $(S_p \cap S_q) \setminus Z$ {\it does not} contain any honest party, then it implies that
   $((S_p \cap S_q) \setminus Z) \subseteq Z^{\star}$.
  This further implies that   $\PartySet \subseteq Z_p \cup Z_q \cup Z \cup Z^{\star}$, implying that 
   $\AdvStructure$ {\it does not} satisfy the $\Q^{(4)}(\PartySet, \AdvStructure)$ condition, which is a contradiction.
  \end{proof}
\begin{claim}
  \label{claim:OptMultACS}
 For every $Z \in \AdvStructure$, if all honest parties participate during the hop number $\hop$ in the protocol $\OptMult$, then
 all honest parties eventually obtain a common summand-sharing party, say $P_j$, for this hop, such that the honest parties
  will eventually hold $[c^{(j)}_{(Z, \iter)}]$. Moreover, party $P_j$ will be distinct from the summand-sharing party
   selected for any hop number $\hop' < \hop$.
  \end{claim}
  \begin{proof}
  Since all honest parties participate in hop number $\hop$, it follows that 
  $\Products_{(Z, \iter)} \neq \emptyset$ at the beginning of hop number $\hop$. This implies that there exists at least one ordered pair $(p, q) \in \Products_{(Z, \iter)}$.
   From Claim \ref{claim:OptMultAdvCondition}, 
    there exists at least one {\it honest} party in $(S_p \cap S_q) \setminus Z$, say $P_k$, who will have both the shares 
  $[a]_p$ as well as $[b_q]$ (and hence the summand $[a]_p [b]_q$). We also note that $P_k$ {\it would not} have been selected as the common summand-sharing 
  party in any previous $\hop' < \hop$, as otherwise $P_k$ would have already included the summand $[a]_p[b]_q$ in the sum $c^{(k)}_{(Z, \iter)}$ shared
  by $P_k$ during hop $\hop'$, implying that $(p, q) \not \in \Products_{(Z, \iter)}$. Now, during the hop number $\hop$, party $P_k$
  will randomly secret-share the sum $c^{(k)}_{(Z, \iter)}$ by making a
  call to $\FVSS$  and every honest $P_i$ will eventually receive an output
   $(\Share, \sid_{\hop, k, \iter, Z}, P_k, \{ [c^{(k)}_{(Z, \iter)}]_q \}_{P_i \in S_q})$  from $\FVSS$
    with $\sid_{\hop, k, \iter, Z}$. 
 Moreover, $P_k$ {\it will not} be present in the set $\Discarded$
  and if $P_k$ is present in the sets $ \Waitlist^{(i)}_{\iter'}, \LocalDiscarded^{(i)}_{\iter'}$ of any honest $P_i$ for any
  $\iter' < \iter$, then it is eventually removed from these sets.
  
  We next claim that during the hop number $\hop$, there will at least one instance of $\FABA$ corresponding to which {\it all} honest parties
  eventually receive the output $1$. For this, we consider two possible cases:
  \begin{myitemize}
  \item[--] {\it At least one honest party participates with input $0$ in the $\FABA$ instance corresponding to $P_k$}: Let 
  $P_i$ be an {\it honest} party, who sends $(\vote, \sid_{\hop, k, \iter, Z}, 0)$ to $\FABA$ with 
  $\sid_{\hop, k, \iter, Z}$. Then from the steps of $\OptMult$, it follows that there exists some $P_j \in \PartySet$, such that
  $P_i$ has received
    $(\decide, \sid_{\hop, j, \iter, Z}, 1)$ as the output from $\FABA$ with  $\sid_{\hop, j, \iter, Z}$. Hence, {\it every} honest party
    will eventually receive the output $(\decide, \sid_{\hop, j, \iter, Z}, 1)$ as the output from $\FABA$ with  $\sid_{\hop, j, \iter, Z}$.
  \item[--] {\it No honest party participates with input $0$ in the $\FABA$ instance corresponding to $P_k$}:
  In this case, {\it every} honest party will eventually send $(\vote, \sid_{\hop, k, \iter, Z}, 1)$ to $\FABA$ with 
  $\sid_{\hop, k, \iter, Z}$ and eventually receives the output $(\decide, \sid_{\hop, k, \iter, Z}, 1)$ from $\FABA$. 
  \end{myitemize}
  Now, based on the above claim, we can further claim that 
   all honest parties will eventually participate with some input in all the
  $n$ instances of $\FABA$ invoked during the hop number $\hop$
  and hence, all the $n$ instances of $\FABA$ during the hop number $\hop$ will eventually produce an output.
  Since the summand-sharing party for hop number $\hop$   
  corresponds to the least indexed $\FABA$ instance in which all the honest parties obtain $1$ as the output,
   it follows that eventually the honest parties will select a summand-sharing party. Moreover, this summand-sharing party will be common,
  as it is based on the outcome of $\FABA$ instances.
  
  Let $P_j$ be the summand-sharing party for the hop number $\hop$.
   We next show that the honest parties will eventually hold $[c^{(j)}_{(Z, \iter)}]$.
   For this, we note that since $P_j$ has been selected as the summand-sharing party, {\it at least} one {\it honest} party, say
   $P_i$, must have sent $(\vote, \sid_{\hop, j, \iter, Z}, 1)$ to $\FABA$ with 
  $\sid_{\hop, j, \iter, Z}$. If not, then $\FABA$ with $\sid_{\hop, j, \iter, Z}$ will never produce the output
  $(\decide, \sid_{\hop, j, \iter, Z}, 1)$ and hence $P_j$ will not be the summand-sharing party for the hop number $\hop$.
  Now since $P_i$ sent $(\vote, \sid_{\hop, j, \iter, Z}, 1)$ to $\FABA$, it follows that
  $P_i$ has received an output $(\Share, \sid_{\hop, j, \iter, Z}, P_j, \{ [c^{(j)}_{(Z, \iter)}]_q \}_{P_i \in S_q})$  from $\FVSS$
    with $\sid_{\hop, j, \iter, Z}$. This implies that $P_j$ must have sent the message 
    $(\Dealer, \sid_{\hop, j, \iter, Z}, (c^{(j)}_{{(\iter, Z)}_1}, \ldots, c^{(j)}_{{(\iter, Z)}_h}))$ to
    $\FVSS$ with $\sid_{\hop, j, \iter, Z}$. 
    Consequently, every honest party will eventually receive their respective outputs from 
     $\FVSS$ with $\sid_{\hop, j, \iter, Z}$ and hence, the honest parties will eventually hold
     $[c^{(j)}_{(Z, \iter)}]$.
     
     Finally, to complete the proof of the claim, we need to show that party $P_j$ is different from the summand-sharing parties selected
     during the hops $1, \ldots, \hop - 1$. If $P_j$ has been selected as a summand-sharing party for any hop number $\hop' < \hop$, then
     no honest party ever sends $(\vote, \sid_{\hop, j, \iter, Z}, 1)$ to $\FABA$ with 
  $\sid_{\hop, j, \iter, Z}$. Consequently, $\FABA$ with $\sid_{\hop, j, \iter, Z}$ will never send the output 
  $(\decide, \sid_{\hop, j, \iter, Z}, 1)$ to any honest party and hence $P_j$ will not be selected as the summand-sharing party
  for hop number $\hop$, which is a contradiction.
  \end{proof}
\begin{claim}
  \label{claim:OptMultComplexity}
  In protocol $\OptMult$, all honest parties eventually obtain an output.
   The protocol makes $\Order(n^2)$ calls to $\FVSS$ and $\FABA$.
  \end{claim}
  \begin{proof}
  From Claim \ref{claim:OptMultAdvCondition} and \ref{claim:OptMultACS}, it follows that the number of hops in the protocol 
   is $\Order(n)$, as in each hop a new summand-sharing party is selected and if all honest parties are included
   in the set of summand-sharing parties $\Selected_{(Z, \iter)}$, then $\Products_{(Z, \iter)}$ becomes $\emptyset$.
   The proof now follows from the fact that in each hop, there are $\Order(n)$ calls to $\FVSS$ and $\FABA$. 
  \end{proof}
  \begin{claim}
  \label{claim:OptMultCorrectness}
  In protocol $\OptMult$, if no party in $\PartySet \setminus Z$ behaves maliciously, then for each $P_i \in \Selected_{(Z, \iter)}$, the condition
   $c^{(i)}_{(Z, \iter)} = \sum_{(p, q) \in \Products^{(i)}_{(Z, \iter)}} [a]_p[b]_q$ holds and
   $c_{(Z, \iter)} = ab$. 
  \end{claim}
  \begin{proof}
  From the protocol steps, it follows that $\Selected_{(Z, \iter)} \cap Z = \emptyset$, as no honest part ever votes for any party from $Z$ as a candidate
  summand-sharing party during any hop in the protocol.  
  Now since $\Selected_{(Z, \iter)} \subseteq (\PartySet \setminus Z)$, 
  if no party in $\PartySet \setminus Z$ behaves maliciously, then it implies that every party 
  $P_i \in \Selected_{(Z, \iter)}$ behaves honestly and secret-shares $c^{(i)}_{(Z, \iter)}$ by calling $\FVSS$, where
  $c^{(i)}_{(Z, \iter)} = \sum_{(p, q) \in \Products^{(i)}_{(Z, \iter)}} [a]_p[b]_q$.
  Moreover, from the protocol steps, it follows that for every $P_j, P_k \in \Selected_{(Z, \iter)}$:
  \[ \displaystyle \Products^{(j)}_{(Z, \iter)} \cap \Products^{(k)}_{(Z, \iter)} = \emptyset.  \]
  To prove this, suppose $P_j$ and $P_k$ are included in $\Selected_{(Z, \iter)}$ during hop number $\hop_j$ and $\hop_k$ respectively, where
  without loss of generality, $\hop_j < \hop_k$. 
  Then from the protocol steps, during $\hop_j$, the parties would set
  $\Products^{(k)}_{(Z, \iter)} = \Products^{(k)}_{(Z, \iter)} \setminus \Products^{(j)}_{(Z, \iter)}$. This ensures that
  during $\hop_k$, there exists no ordered pair $(p, q) \in \{1, \ldots, |\ShareSpec| \} \times \{1, \ldots, |\ShareSpec| \}$, such that
  $(p, q) \in \Products^{(j)}_{(Z, \iter)} \cap \Products^{(k)}_{(Z, \iter)}$.
  
  Since all the parties $P_i \in \Selected_{(Z, \iter)}$ have behaved honestly, 
   from the protocol steps, it also follows that :
  \[ \displaystyle \bigcup_{P_i \in \Selected_{(Z, \iter)}} \Products^{(i)}_{(Z, \iter)} = \{ (p, q) \}_{p, q = 1, \ldots, |\ShareSpec|}. \]
  Finally, from the protocol steps, it follows that 
   $\forall P_j \in \PartySet \setminus \Selected_{(Z, \iter)}$, the condition
  $c^{(j)}_{(Z, \iter)} = 0$ holds. Now since $c_{(Z, \iter)} =  c^{(1)}_{(Z, \iter)} + \ldots + c^{(n)}_{(Z, \iter)}$, it follows that
  if no party in $\PartySet \setminus Z$ behaves maliciously, then $c_{(Z, \iter)} = ab$ holds. 
  \end{proof}
\begin{claim}
\label{claim:OptMultPrivacy}
In $\OptMult$, $\Adv$ does not learn any additional information about $a$ and $b$.
\end{claim}
\begin{proof}
Let $Z^{\star} \in \AdvStructure$ be the set of corrupt parties. To prove the claim, we argue that in the protocol, $\Adv$ does not learn any additional information about
 the shares $\{[a]_p, [b]_p\}_{S_p \cap Z^{\star} = \emptyset}$. For this, consider an arbitrary summand $[a]_p [b]_q$ where $S_p \cap Z^{\star} = \emptyset$ and where
 $q \in \{1, \ldots, h \}$. Clearly, the summand $[a]_p [b]_q$ will not be available with any party in $Z^{\star}$.
 Let $P_j$ be the party from $\Selected_{(Z, \iter)}$, such that $(p, q) \in \Products^{(j)}_{(Z, \iter)}$; i.e.~the summand
 $[a]_p[b]_q$ is included by $P_j$ while computing the summand-sum $c^{(j)}_{(Z, \iter)}$. 
 Clearly $P_j$ is {\it honest}, since $P_j \not \in Z^{\star}$.
 In the protocol, party $P_j$ randomly secret-shares the summand-sum $c^{(j)}_{(Z, \iter)}$, by supplying a random vector of shares
 for $c^{(j)}_{(Z, \iter)}$ to the corresponding $\FVSS$. Now, since $\ShareSpec$ is {\it $\AdvStructure$-private},
  it follows that the shares $\{ [c^{(j)}_{(Z, \iter)}]_r \}_{S_r \cap Z^{\star} \neq \emptyset}$ learnt by $\Adv$ in the protocol
  will be independent of the summand $[a]_p[b]_q$ and hence, independent of $[a]_p$. Using a similar argument, we can conclude
  that the shares learnt by $\Adv$ in the protocol will be independent of the summands $[a]_q [b]_p$ (and hence independent of $[b]_p$),
  where $S_p \cap Z^{\star} = \emptyset$ and where
 $q \in \{1, \ldots, h \}$. 
\end{proof}

\noindent {\bf Lemma \ref{lemma:OptMult}.}
{\it Let  $\AdvStructure$ satisfy the $\Q^{(4)}(\PartySet, \AdvStructure)$ condition and let 
 $\ShareSpec = (S_1, \ldots, S_h) = \{ \PartySet \setminus Z | Z \in \AdvStructure\}$. 
 Consider an arbitrary $Z \in \AdvStructure$ and $\iter$, 
  such that all honest parties participate in the instance $\OptMult(\PartySet,\AdvStruct,\SharingSpec,[a],[b], Z,\iter)$. Then all
  honest parties eventually compute  $[c_{(Z, \iter)}]$ and $([c^{(1)}_{(Z, \iter)}], \ldots, [c^{(n)}_{(Z, \iter)}])$
  where $c_{(Z, \iter)} = c^{(1)}_{(Z, \iter)} + \ldots + c^{(n)}_{(Z, \iter)}$, 
  provided no honest party is ever included in the $\Discarded$ and $\LocalDiscarded^{(i)}_{\iter'}$ sets 
  and every honest party in the $ \Waitlist^{(i)}_{\iter'}$ sets of every honest $P_i$ is eventually removed, 
   for every
  $\iter' < \iter$. If no party in $\PartySet \setminus Z$ behaves maliciously, then 
     $c_{(Z, \iter)} = ab$ holds. In the protocol, $\Adv$ does not learn any additional information about $a$ and $b$.
          The protocol makes $\Order(n^2)$ calls to $\FVSS$ and $\FABA$
}.
\begin{proof}
The proof follows from Claims \ref{claim:OptMultAdvCondition}-\ref{claim:OptMultPrivacy}.
\end{proof}

We end this section by claiming an important property about the protocol $\OptMult$, which will be useful later when we analyze the properties of the
 protocol $\MultGCE$ where $\OptMult$ is used as a sub-protocol.
 \begin{claim}
 \label{claim:OptMultFuture}
 For every $Z \in \AdvStructure$ and every $\iter$, all the following hold for every $P_j \in \Selected_{(Z, \iter)}$ during the instance
   $\OptMult(\PartySet,\AdvStruct,\SharingSpec,[a],[b], Z,\iter)$.
   \begin{myitemize}
     \item[--] There exists at least one honest party $P_i$, such that
      $P_j$ will not be present in the $\Waitlist^{(i)}_{\iter'}$ and $\LocalDiscarded^{(i)}_{\iter'}$ sets of $P_i$ for any $\iter' < \iter$.
     \item[--] $P_j$ will not be present in the set $\Discarded$.
   \end{myitemize}
 \end{claim}
 \begin{proof}
 Consider an arbitrary $P_j \in \Selected_{(Z, \iter)}$, such that $P_j$ is included in $\Selected_{(Z, \iter)}$ during the hop number $\hop$
  in the instance $\OptMult(\PartySet,\AdvStruct,\SharingSpec,[a],[b], Z,\iter)$.
  We prove the first part of the claim through a contradiction.
   Let $\Hon$ be the set of {\it honest} parties and
  for every $P_i \in \Hon$, let there exist some $\iter' < \iter$, such that either $P_j \in \Waitlist^{(i)}_{\iter'}$ or 
   $P_j \in \LocalDiscarded^{(i)}_{\iter'}$. This implies that during 
  hop number $\hop$, no $P_i \in \Hon$ will send $(\vote, \sid_{\hop, j, \iter, Z}, 1)$ to $\FABA$ with 
  $\sid_{\hop, j, \iter, Z}$. Consequently, $\FABA$ with  $\sid_{\hop, j, \iter, Z}$ will never return the output
  $(\decide, \sid_{\hop, j, \iter, Z}, 1)$ for any honest party and hence, $P_j$ will not be selected as the summand-sharing party
  for hop number $\hop$, which is a contradiction.
  
  The second part of the claim also follows using a similar argument as above. Namely, if 
  $P_j$ is present in the set $\Discarded$, then
  no $P_i \in \Hon$ will send $(\vote, \sid_{\hop, j, \iter, Z}, 1)$ to $\FABA$ with 
  $\sid_{\hop, j, \iter, Z}$ and consequently, 
  $P_j$ will not be selected as the summand-sharing party
  for hop number $\hop$, which is a contradiction.
 \end{proof}
\subsection{Properties of the Multiplication Protocol $\MultLCE$ with Cheater Identification}
In this section, we formally prove the properties of the protocol $\MultLCE$ (see Fig \ref{fig:MultLCE} for the formal description of the protocol). 
 While proving these properties, we will assume that  $\AdvStructure$ satisfies the $\Q^{(4)}(\PartySet, \AdvStructure)$ condition.
  This further implies that the sharing specification $\ShareSpec = (S_1, \ldots, S_h) \defined \{ \PartySet \setminus Z | Z \in \AdvStructure\}$
  satisfies the $\Q^{(3)}(\ShareSpec, \AdvStructure)$ condition. Moreover, we will also assume that {\it no} honest party is ever included
  in the set $\Discarded$, which will be guaranteed in the protocol $\MultGCE$ where the set $\Discarded$ is constructed and managed, and where
  $\MultLCE$ is used as a sub-protocol. 
  
  We first give the definition of a {\it successful} $\MultLCE$ instance, which will be used
  throughout this section and the next.
  
   \begin{definition}[{\bf Successful $\MultLCE$ Instance}]
  \label{def:MultGCE}
  For an instance $\MultLCE(\PartySet, \AdvStruct, \SharingSpec, [a], [b], \iter)$, we define the following.
  \begin{myitemize}
  \item[--] The instance is called {\it successful} if and only if 
   for every $Z \in \AdvStruct$, the value $c_{(Z, \iter)} - c_{(Z', \iter)} = 0$, where $Z' \in \AdvStructure$ is the fixed
   set used in the protocol.  
   \item[--] If the instance is not successful, then the sets $Z, Z'$ are called the {\it conflicting-sets} for the instance,
   if $Z$ is the smallest indexed set from $\AdvStructure$ such that $c_{(Z, \iter)} - c_{(Z', \iter)} \neq 0$.
  \end{myitemize}
  \end{definition} 
  
   We first show that any instance of $\MultLCE$ will be eventually found to be either a success or a failure by the honest parties. 
    \begin{claim}
   \label{claim:MultLCETermination}
For every $\iter$, any instance $\MultLCE(\PartySet, \AdvStruct, \SharingSpec, [a], [b], \iter)$
 will eventually be deemed to either succeed or fail by the honest parties, 
  provided no honest party is ever included in the $\Discarded$ and $\LocalDiscarded^{(i)}_{\iter'}$ sets, 
  and all honest parties are eventually removed from the $ \Waitlist^{(i)}_{\iter'}$ sets of every honest $P_i$ for every
  $\iter' < \iter$.
   Moreover, for a 
  successful
  instance, the parties output a sharing of $ab$. If the instance is not successful, then 
   the parties identify the  {\it conflicting-sets} $Z, Z'$ for the instance.
  \end{claim}
  \begin{proof}
  Let $Z^{\star} \in \AdvStructure$ be the set of corrupt parties. If the lemma conditions hold, then
   it follows from 
    Lemma \ref {lemma:OptMult}, that corresponding to every $Z \in \AdvStructure$, the instance $\OptMult(\PartySet, \AdvStruct, \SharingSpec,[a], [b], Z, \iter)$ 
 eventually completes with honest parties obtaining the outputs $[c^{(1)}_{(Z, \iter)}] , \ldots, [c^{(n)}_{(Z, \iter)}], [c_{(Z, \iter)}]$, where
 $c_{(Z, \iter)} = c^{(1)}_{(Z, \iter)} + \ldots + c^{(n)}_{(Z, \iter)}$. Moreover, in the $\OptMult$ instance corresponding to $Z^{\star}$,
 the output $c_{(Z^{\star}, \iter)}$ will be the same as $ab$, since all the parties in $\PartySet \setminus Z^{\star}$ will be honest.
 
 Since $\ShareSpec$ satisfies the $\Q^{(3)}(\ShareSpec, \AdvStructure)$ condition, it follows that with respect to the fixed
  $Z' \in \AdvStruct$, the honest parties will eventually reconstruct the difference $c_{(Z, \iter)} - c_{(Z', \iter)}$, corresponding to {\it every}
  $Z \in \AdvStructure$. Now there are two possibilities. If all the differences $c_{(Z, \iter)} - c_{(Z', \iter)}$ turn out to be $0$, then 
  the $\MultLCE$ instance will be considered to be successful by the honest parties and the honest parties will output
  $[c_{(Z', \iter)}]$, which is bound to be the same as $ab$. This is because $c_{(Z', \iter)} - c_{(Z^{\star}, \iter)} = 0$ and hence
  $c_{(Z', \iter)} = c_{(Z^{\star}, \iter)} = ab$ holds. The other possibility is that all the differences are {\it not} zero, in which case
  the instance $\MultLCE$ will not be considered successful by the honest parties. Moreover, in this case, the parties will set
  $(Z, Z')$ as the conflicting-sets for the instance, where $Z$ is the smallest indexed set from $\AdvStructure$ such that $c_{(Z, \iter)} - c_{(Z', \iter)} \neq 0$.
  \end{proof}

    We next prove a series of claims regarding any $\MultLCE$ instance which is {\it not} successful. 
    We begin by showing that 
    if any instance of $\MultLCE$ is {\it not} successful, then every honest party in eventually removed from the waiting-list
  of the honest parties for that instance. Moreover, no honest party will be ever included in the $\LocalDiscarded$ set of any honest party for that instance.
   \begin{claim}
   \label{claim:MultLCEHonestPartyParticipation}
   For every $\iter$, if the instance $\MultLCE(\PartySet, \AdvStruct, \SharingSpec, [a], [b], \iter)$ is not successful, then 
   every honest party from the set $\Selected_{(Z, \iter)} \cup \Selected_{(Z', \iter)}$ is eventually removed from the waiting set
      $\Waitlist^{(i)}_{\iter}$ of every honest party $P_i$. Moreover, no honest party is ever included
      in the $\LocalDiscarded^{(i)}_{\iter}$ set of any honest party $P_i$.
      \end{claim}
  \begin{proof}
  Suppose that the instance $\MultLCE(\PartySet, \AdvStruct, \SharingSpec, [a], [b], \iter)$ is not successful. This implies that the parties identify
  a pair of conflicting-sets $(Z, Z')$, such that $c_{(Z, \iter)} - c_{(Z', \iter)} \neq 0$. From the protocol steps, every honest party
  $P_i$ initializes $\Waitlist^{(i)}_{\iter}$ to  $\Selected_{(Z, \iter)} \cup \Selected_{(Z', \iter)}$. Let $P_j$ be an arbitrary 
  {\it honest} party belonging to $\Selected_{(Z, \iter)} \cup \Selected_{(Z', \iter)}$. From the protocol steps, 
  party $P_j$ secret-shares all the required values $d^{(jk)}_{(Z, \iter)}, e^{(jk)}_{(Z', \iter)}$ by calling appropriate instances of $\FVSS$
  and eventually these values are secret-shared, with every honest $P_i$ receiving the appropriate shares
  from corresponding $\FVSS$ instances. Consequently, $P_j$ will eventually be removed from the set $\Waitlist^{(i)}_{\iter}$.
    Moreover, since $P_j$ is an {\it honest} party, the $d^{(jk)}_{(Z, \iter)}, e^{(jk)}_{(Z', \iter)}$ values shared
  by $P_j$ will be correct and consequently, the conditions for including $P_j$ to the $\LocalDiscarded^{(i)}_{\iter}$ 
  set of any honest party $P_i$ will fail. That is, if $P_j \in \Selected_{(Z, \iter)}$, then the parties will find that 
                  $\displaystyle c^{(j)}_{(Z, \iter)} - \sum_{P_k \in \Selected_{(Z', \iter)}} d^{(jk)}_{(Z, \iter)} = 0$.
   On the other hand, if $P_j \in \Selected_{(Z', \iter)}$, then the parties will find that
                $\displaystyle c^{(j)}_{(Z', \iter)} - \sum_{P_k \in \Selected_{(Z, \iter)}} e^{(jk)}_{(Z', \iter)} = 0$.
   Moreover, if there exists any $P_k \in \Selected_{(Z, \iter)} \cup \Selected_{(Z', \iter)}$ such that
   either $d^{(jk)}_{(Z, \iter)} \neq e^{(kj)}_{(Z', \iter)}$ or $d^{(kj)}_{(Z, \iter)} \neq e^{(jk)}_{(Z', \iter)}$, 
   then after reconstructing the values shared by $P_j$ and the shares held by $P_j$, the parties
   will find that $P_j$ has behaved honestly and hence, $P_j$ will not be included to the  
    $\LocalDiscarded^{(i)}_{\iter}$ set of any honest $P_i$. 
  \end{proof}

We next give the definition of a {\it conflicting-pair} of parties, which is defined based on the partitions of the summand-sum shared by the summand-sharing parties.   
\begin{definition}[{\bf Conflicting-Pair of Parties}]
\label{MultLCEConflictingPair}
Let $\MultLCE(\PartySet, \AdvStruct, \SharingSpec, [a], [b], \iter)$ be an instance of $\MultLCE$ which  is not successful and let $Z, Z'$ be the corresponding conflicting-sets for the instance. A pair of parties $(P_j,P_k)$ is said to be a 
 {\it conflicting-pair} of parties for this $\MultLCE$ instance if all the following hold:
   \begin{myitemize}
   \item[--] $P_j \in \Selected_{(Z,\iter)}, P_k \in \Selected_{(Z',\iter)}$;
   \item[--] $d^{(jk)}_{(Z, \iter)} \neq e^{(kj)}_{(Z', \iter)}$.   
   \end{myitemize}
\end{definition}
We next show that if an instance of $\MultLCE$ is not successful,
 then certain conditions hold with respect to the summand-sums and the respective partitions shared by the 
  summand-sharing parties during the underlying instances of $\OptMult$ and the 
  cheater-identification
   phase of the $\MultLCE$ instance.
\begin{claim}
\label{MultLCEFailureReasons}
Let $\MultLCE(\PartySet, \AdvStruct, \SharingSpec, [a], [b], \iter)$ be an instance of $\MultLCE$ which is not successful and let $Z, Z'$ be the corresponding conflicting-sets for the instance. Moreover, let $Z^{\star}$ be the set of corrupt parties. 
 Then, one of the following must hold true for some $P_j \in Z^{\star}$.
\begin{myenumerate}
      \item[i.] $P_j \in \Selected_{(Z, \iter)}$ and  
       $\displaystyle c^{(j)}_{(Z, \iter)} - \sum_{P_k \in \Selected_{(Z', \iter)}} d^{(jk)}_{(Z, \iter)} \neq 0$.
       \item[ii.] $P_j \in \Selected_{(Z', \iter)}$ and 
               $\displaystyle c^{(j)}_{(Z', \iter)} - \sum_{P_k \in \Selected_{(Z, \iter)}} e^{(jk)}_{(Z', \iter)} \neq 0$.
        \item[iii.] There is some $P_k \in \Selected_{(Z, \iter)} \cup \Selected_{(Z', \iter)}$ such that 
        either $(P_j,P_k)$ or $(P_k,P_j)$ constitutes a conflicting-pair of parties.
\end{myenumerate} 
\end{claim}
\begin{proof}
Since the instance of $\MultLCE$ is not successful and $Z, Z'$ constitute
 conflicting-sets, it follows that
   $c_{(Z,\iter)} \neq c_{(Z',\iter)}$. Assume for the sake of contradiction that the none of the conditions
   in the claim is true. Then, we can infer the following.
\begin{align*} 
\displaystyle c_{(Z,\iter)} &= \sum_{P_j \in \Selected_{(Z,\iter)}}c^{(j)}_{(Z,\iter)} \\
                                        &= \sum_{P_j \in \Selected_{(Z,\iter)}} \sum_{P_k \in \Selected_{(Z',\iter)}} d^{(jk)}_{(Z,\iter)} \\                                     
                                        &= \sum_{P_j \in \Selected_{(Z,\iter)}} \sum_{P_k \in \Selected_{(Z',\iter)}} e^{(kj)}_{(Z',\iter)} \\
                                         &= \sum_{P_k \in \Selected_{(Z',\iter)}} \sum_{P_j \in \Selected_{(Z,\iter)}} e^{(kj)}_{(Z',\iter)} \\
                                         &= \sum_{P_k \in \Selected_{(Z',\iter)}} c^{(k)}_{(Z',\iter)} \\ 
                                        &= c_{(Z',\iter)},
\end{align*}
where the first equation follows from the definition of $c_{(Z,\iter)}$, the 
 second equation holds because, as per our assumption, $\displaystyle c^{(j)}_{(Z, \iter)} - \sum_{P_k \in \Selected_{(Z', \iter)}} d^{(jk)}_{(Z, \iter)} = 0$ for every
  $P_j \in \Selected_{(Z, \iter)}$, 
 the third equation holds because, as per our assumption, there is {\it no} conflicting-pair of parties, the fifth equation holds because as per our assumption  
 $\displaystyle c^{(k)}_{(Z', \iter)} - \sum_{P_j \in \Selected_{(Z, \iter)}} e^{(kj)}_{(Z', \iter)} = 0$ for every
  $P_k \in \Selected_{(Z', \iter)}$ and the last equation follows from the definition of $c_{(Z',\iter)}$. However,  $c_{(Z,\iter)} = c_{(Z',\iter)}$ is a contradiction.
\end{proof}

We next define a {\it characteristic function} with respect to the partitions of the summands-sum shared by the summand-sharing parties, 
  to ``characterize" instances of $\MultLCE$ which are {\it not} successful. Looking ahead, this will be helpful to upper bound the number of 
  failed $\MultLCE$ instances in the protocol $\MultGCE$.  
\begin{definition}[{\bf Characteristic Function}]
\label{MultCICharacterization}
Let $\MultLCE(\PartySet, \AdvStruct, \SharingSpec, [a], [b], \iter)$ be an instance of $\MultLCE$ 
 which is not successful and let $Z, Z'$ be the corresponding conflicting-sets for the instance. Then 
  the {\it characteristic function} $f_{\Char}$ for this instance is defined as follows. 
\begin{myitemize}
\item[--] If there is some 
  $P_j \in \Selected_{(Z,\iter)}$ such that $\displaystyle c^{(j)}_{(Z, \iter)} - \sum_{P_k \in \Selected_{(Z', \iter)}} d^{(jk)}_{(Z, \iter)} \neq 0$, then $f_{\Char}(\iter) \defined (P_j, P_k)$,
   where $P_k$ is the smallest-indexed party from $\PartySet \setminus \{P_j \}$.\footnote{If there are multiple parties $P_j$ satisfying this condition, then we consider
   the $P_j$ with the smallest index.}
\item[--] Else, if there is some 
 $P_j \in \Selected_{(Z',\iter)}$ such that $\displaystyle c^{(j)}_{(Z', \iter)} - \sum_{P_k \in \Selected_{(Z, \iter)}} e^{(jk)}_{(Z', \iter)} \neq 0$, then $f_{\Char}(\iter) = (P_k, P_j)$,
  where $P_k$ is the smallest-indexed party from $\PartySet \setminus \{P_j \}$.\footnote{If there are multiple parties $P_j$ satisfying this condition, then we consider
   the $P_j$ with the smallest index.}
\item[--] Else, $f_{\Char}(\iter) \defined (P_j, P_k)$, where $(P_j,P_k)$ is a conflicting-pair of parties, corresponding to the $\MultLCE$ instance.\footnote{If there are multiple
 conflicting-pairs, then we consider the one having parties with the smallest indices.}
\end{myitemize}
\end{definition}
Before we proceed, we would like to stress that if $f_{\Char}$ is defined either with respect to the first or the second condition, then party $P_k$ in the pair
 $(P_j, P_k)$ serves as
   a ``dummy" party. This is just for notational convenience to ensure uniformity so that $f_{\Char}$ is {\it always} a pair of parties irrespective of whether
  it is defined with respect to the first, second or third condition.
 
From the definition, it is easy to see that if $f_\Char(\iter) = (P_j, P_k)$, then at least one party among $P_j, P_k$ is 
 {\it maliciously-corrupt}. We next claim that the characteristic function is well defined.
\begin{claim}
\label{MultCIWellDefinedCharacterization}
Let $\MultLCE(\PartySet, \AdvStruct, \SharingSpec, [a], [b], \iter)$ be an instance of $\MultLCE$ 
 which is not successful and let $Z, Z'$ be the corresponding conflicting-sets for the instance. Then $f_\Char(\iter)$ is well defined.
\end{claim}
\begin{proof}
Proof follows from Claim \ref{MultLCEFailureReasons}.
\end{proof}

We next prove an important property by showing that if $f_\Char(\iter) = (P_j, P_k)$ for some instance
  of $\MultLCE$ which is not successful, and if both $P_j$ and $P_k$
   have been removed from the waiting-list of some honest party for that instance, 
   then the corrupt party(ies) among $P_j, P_k$ will eventually be discarded by {\it all} honest parties.
  \begin{claim}
  \label{claim:MultLCECheaterIdentification}
  Let $\MultLCE(\PartySet, \AdvStruct, \SharingSpec, [a], [b], \iter)$ be an instance of $\MultLCE$  which  is not successful and let
   $Z, Z'$ be the corresponding conflicting-sets for the instance. Moreover. let $f_\Char(\iter) = (P_j, P_k)$. If both
  $P_j$ and $P_k$ are removed from the set $\Waitlist^{(h)}_{\iter}$ of {\it any} honest party $P_h$, then
   the corrupt party(ies) among $P_j, P_k$ will eventually be included
        in the set $\LocalDiscarded^{(i)}_{\iter}$ of {\it every} honest $P_i$.        
  \end{claim} 
  \begin{proof}
  Let $f_\Char(\iter) = (P_j, P_k)$, where without loss of generality,
   $P_j \in \Selected_{(Z,\iter)}$ and $P_k \in \Selected_{(Z',\iter)}$.  
   From the definition of characteristic function (Def \ref{MultCICharacterization}), one of the following holds for $P_j$ and $P_k$:
     \begin{myitemize}
     \item[--] {\it $(P_j, P_k)$ constitutes a conflicting-pair}: In this case, $d^{(jk)}_{(Z,\iter)} \neq e^{(kj)}_{(Z',\iter)}$. 
     Since the {\it honest} $P_h$ has removed both $P_j$ and $P_k$ from $\Waitlist^{(h)}_\iter$, from the protocol steps, 
      the outputs $(\Share, \sid_{j, k, \iter, Z},P_j, \allowbreak \{[d^{(jk)}_{(Z,\iter)}]_q\}_{P_h \in S_q})$ and 
      $(\Share, \sid_{k, j, \iter, Z'}, P_k,\{[e^{(kj)}_{(Z',\iter)}]_q\}_{P_h \in S_q})$ have been obtained by $P_h$ from
       $\FVSS$ with $\sid_{j, k, \iter, Z}$ and $\sid_{k, j, \iter, Z'}$ respectively.
       Consequently, each honest party will eventually receive its respective share corresponding to 
       $[d^{(jk)}_{(Z,\iter)}]$ and $[e^{(kj)}_{(Z',\iter)}]$ from the corresponding $\FVSS$ instances. 
       Hence, each honest  party will be able to locally compute its share of $d^{(jk)}_{(Z,\iter)} - e^{(kj)}_{(Z',\iter)}$ and
        participate in the instance of $\PiPerRec$ to reconstruct the difference. Since $\ShareSpec$ satisfies the $\Q^{(3)}(\ShareSpec, \AdvStructure)$
        condition, all honest parties will eventually reconstruct $d^{(jk)}_{(Z,\iter)} - e^{(kj)}_{(Z',\iter)}$ and find that the difference is not $0$.
        Consequently, the honest parties will 
        participate in appropriate instances of $\PiPerRec$ to reconstruct the values
         $d^{(jk)}_{(Z,\iter)}$, $e^{(kj)}_{(Z',\iter)}$, and instances of $\PiPerRecShare$ to reconstruct
          the shares $[a]_p$ and $[b]_q$, such that 
         $(p,q) \in \Products^{(j)}_{(Z,\iter)} \cap \Products^{(k)}_{(Z',\iter)}$. Now, 
         either $d^{(jk)}_{(Z,\iter)}$ or $e^{(kj)}_{(Z',\iter)}$ will not be equal to $\displaystyle \sum_{(p, q) \in \Products^{(j)}_{(Z, \iter)} \cap \Products^{(k)}_{(Z', \iter)}} [a]_p [b]_q$,
         as otherwise $d^{(jk)}_{(Z,\iter)} = e^{(kj)}_{(Z',\iter)}$ will hold, which is a contradiction. 
         Consequently, every honest party $P_i$ will eventually add the corrupt party(ies) among $P_j, P_k$ to $\LocalDiscarded^{(i)}_{\iter}$.
       \item[--] {\it The condition $\displaystyle c^{(j)}_{(Z, \iter)} - \sum_{P_k \in \Selected_{(Z', \iter)}} d^{(jk)}_{(Z, \iter)} \neq 0$ holds}: 
       Since the {\it honest} $P_h$ has removed $P_j$  from $\Waitlist^{(h)}_\iter$, then from the protocol steps, corresponding to every
       $P_k \in \Selected_{(Z', \iter)}$, party $P_h$ has received the output 
       $(\Share, \sid_{j, k, \iter, Z},P_j,  \{[d^{(jk)}_{(Z,\iter)}]_q\}_{P_h \in S_q})$ from  $\FVSS$ with $\sid_{j, k, \iter, Z}$.
       Consequently, for every $P_k \in \Selected_{(Z', \iter)}$,
       all honest parties eventually receive their respective shares corresponding to $[d^{(jk)}_{(Z,\iter)}]$ from the respective $\FVSS$ instances. 
      In the protocol, all honest parties participate in an instance of $\PiPerRec$ with their respective shares corresponding to 
       $\displaystyle [c^{(j)}_{(Z, \iter)}] - \sum_{P_k \in \Selected_{(Z', \iter)}} [d^{(jk)}_{(Z, \iter)}]$ to reconstruct the difference 
       $\displaystyle c^{(j)}_{(Z, \iter)} - \sum_{P_k \in \Selected_{(Z', \iter)}} d^{(jk)}_{(Z, \iter)}$. Now since the difference is not $0$,
       each honest $P_i$ will eventually include the corrupt $P_j$ to $\LocalDiscarded^{(i)}_{\iter}$.     
       \item[--] {\it The condition $\displaystyle c^{(k)}_{(Z', \iter)} - \sum_{P_j \in \Selected_{(Z, \iter)}} e^{(kj)}_{(Z', \iter)} \neq 0$ holds}: 
       This case is symmetric to the previous case and using a similar argument as above, we can conclude that 
         each honest $P_i$ will eventually include the corrupt $P_k$ to $\LocalDiscarded^{(i)}_{\iter}$.    
     \end{myitemize}        
  \end{proof}
  
  We next claim that the adversary does not learn anything additional about $a$ and $b$ in the protocol. 
  \begin{claim}
  \label{claim:MultLCEPrivacy}
  In $\MultLCE$, $\Adv$ does not not learn any additional information about $a$ and $b$.
  \end{claim}
  \begin{proof}
  From Claim \ref{claim:OptMultPrivacy}, $\Adv$ does not learn any additional information about $a$ and $b$ from the instances of $\OptMult$ executed in
   $\MultLCE$. Corresponding to every $Z \in \AdvStructure$, $\Adv$ learns the difference  $c_{(Z, \iter)} - c_{(Z', \iter)}$ which are all $0$,  
    unless the adversary cheats. In case of cheating, the reconstructed differences $c_{(Z, \iter)} - c_{(Z', \iter)}$
    depend completely upon
    the inputs of the adversary and hence learning these differences does not add anything additional about $a$ and $b$ to the adversary's view.
    Next, corresponding to every {\it honest} $P_j \in \Selected_{(Z, \iter)} \cup \Selected_{(Z', \iter)}$, the shares corresponding to 
    $d^{(jk)}_{(Z,\iter)}$ or $e^{(kj)}_{(Z',\iter)}$ learnt by $\Adv$ will be distributed uniformly, since $\ShareSpec$ is $\AdvStructure$-private
    and hence, these shares do not add anything additional about $a$ and $b$ to the adversary's view. 
    Moreover, for every {\it honest} $P_j \in \Selected_{(Z, \iter)} \cup \Selected_{(Z', \iter)}$, 
    $\Adv$ will know beforehand that the differences
     $\displaystyle c^{(j)}_{(Z, \iter)} - \sum_{P_k \in \Selected_{(Z', \iter)}} d^{(jk)}_{(Z, \iter)}$ as well as
     $\displaystyle c^{(j)}_{(Z', \iter)} - \sum_{P_k \in \Selected_{(Z, \iter)}} e^{(jk)}_{(Z', \iter)}$ will be $0$ and hence, learning these differences does not add anything
     additional about $a$ and $b$ to adversary's view.
     On the other hand, for every {\it corrupt} $P_j \in \Selected_{(Z, \iter)} \cup \Selected_{(Z', \iter)}$, the above differences
     completely depend upon the adversary's inputs and hence, reveal no additional information. Finally, if for any ordered pair of parties $(P_j, P_k)$, 
     the condition $d^{(jk)}_{(Z, \iter)} \neq e^{(kj)}_{(Z', \iter)}$ holds, then at least one among $P_j$ and $P_k$ is {\it corrupt}.
     Consequently, the shares $[a]_p$ and $[b]_q$ where $(p, q) \in \Products^{(j)}_{(Z, \iter)} \cap \Products^{(k)}_{(Z', \iter)}$
     reconstructed in this case are already known to the adversary, and do not add anything new to the view of the adversary regarding $a$
     and $b$. 
  \end{proof}

  \begin{claim}
  \label{claim:MultLCECommunication}
  Protocol $\MultLCE$  needs $\Order(|\AdvStructure| \cdot n^2)$ calls to $\FVSS$ and $\FABA$ and incurs an additional communication of 
      $\Order(|\AdvStructure|^2 \cdot n^2 \log{|\F|} + |\AdvStructure| \cdot n^4 \log{|\F|})$ bits. 
  \end{claim}
  \begin{proof}
  In the protocol, corresponding to each $Z \in \AdvStructure$, an instance of $\OptMult$ is executed. From Lemma \ref{lemma:OptMult}, this will require
  $\Order(|\AdvStructure| \cdot n^2)$ calls to $\FVSS$ and $\FABA$. There are $\Order(|\AdvStructure|)$ instances of $\PiPerRec$ to reconstruct
  $\Order(|\AdvStructure|)$ difference values for checking whether the instance is successful or not, incurring a communication of 
  $\Order(|\AdvStructure|^2 \cdot n^2 \log{|\F|})$ bits. If the instance is not successful, then there are $\Order(n^2)$ calls to $\FVSS$ to share
  various summand-sum partitions. To check whether the correct partitions are shared, $\Order(n^2)$ values need to be publicly reconstructed through
   these many instances of $\PiPerRec$, which incurs a communication of $\Order(|\AdvStructure| \cdot n^4 \log{|\F|})$ bits.
  \end{proof}

  The proof of Lemma \ref{lemma:MultLCE} now follows from Claims \ref{claim:MultLCETermination}-\ref{claim:MultLCECommunication}.\\~\\
 \noindent {\bf Lemma \ref{lemma:MultLCE}.} 
 {\it 
  Let $\AdvStructure$ satisfy the $\Q^{(4)}(\PartySet, \AdvStructure)$ condition and let 
  $\ShareSpec = (S_1, \ldots, S_h) = \{ \PartySet \setminus Z | Z \in \AdvStructure\}$. Moreover, 
   let all honest parties participate in the instance $\MultLCE(\PartySet, \AdvStruct, \SharingSpec, [a], [b], \iter)$.
   Then the following hold.
 \begin{myitemize}
 \item[--] The instance  will eventually be deemed to succeed or fail by the honest parties, where for a successful
  instance, the parties output a sharing of $ab$.
 \item[--] If the instance is not successful,
   then the honest parties will agree on a pair $Z, Z' \in \AdvStructure$ such that 
  $c_{(Z, \iter)} - c_{(Z', \iter)} \neq 0$. Moreover, all honest parties present in the 
  $ \Waitlist^{(i)}_{\iter}$ set of any honest party $P_i$ will eventually be removed and no honest party
    is ever included in the $\LocalDiscarded^{(i)}_{\iter}$ set
   of any honest $P_i$.
   Furthermore, there will be a pair of parties $P_j, P_k$ from
  $\Selected_{(Z, \iter)} \cup \Selected_{(Z', \iter)}$, with at least one of them being maliciously-corrupt, such that if both
  $P_j$ and $P_k$ are removed from the set $\Waitlist^{(h)}_{\iter}$ of any honest party $P_h$, then eventually
   the corrupt party(ies) among $P_j, P_k$ will be included
        in the set $\LocalDiscarded^{(i)}_{\iter}$ of every honest $P_i$.        
    \item[--] In the protocol, $\Adv$ does not learn any additional information $a$ and $b$.
      \item[--] The protocol needs $\Order(|\AdvStructure| \cdot n^2)$ calls to $\FVSS$ and $\FABA$ and incurs an additional communication of 
      $\Order(|\AdvStructure|^2 \cdot n^2 \log{|\F|} + |\AdvStructure| \cdot n^4 \log{|\F|})$ bits.
   \end{myitemize}   
   }
 \paragraph{$\MultLCE$ for Inputs $\{([a^{(\ell)}], [b^{(\ell)}] ) \}_{\ell = 1, \ldots, M}$:} The modifications to the protocol $\MultLCE$ for handling
  $M$ pairs of secret-shared inputs is simple. The
   parties now run
  instances of $\OptMult$ handling $M$ pairs of inputs. Corresponding to every pair $(Z, Z')$, the parties reconstruct $M$ differences. If
   any of these differences is non-zero, the parties 
   focus on the smallest-indexed $([a^{(\ell)}], [b^{(\ell)}])$ such that $c^{(\ell)}_{(Z, \curr)} - c^{(\ell)}_{(Z', \curr)} \neq 0$. The 
   parties then proceed to the cheater identification phase with respect to $(Z, Z')$ and $([a^{(\ell)}], [b^{(\ell)}])$. 
   The protocol will require $\Order(M \cdot |\AdvStructure| \cdot n^2)$ calls to $\FVSS$, $\Order(|\AdvStructure| \cdot n^2)$ calls to $\FABA$
   and additionally communicates $\Order(M \cdot |\AdvStructure|^2 \cdot n^2 \log{|\F|} + |\AdvStructure| \cdot n^4 \log{|\F|})$ bits. 
\subsection{Properties of the Multiplication Protocol $\MultGCE$}
In this section, we formally prove the properties of the protocol $\MultGCE$ (see Fig \ref{fig:MultGCE} for the formal description of the protocol). 
 While proving these properties, we will assume that  $\AdvStructure$ satisfies the $\Q^{(4)}(\PartySet, \AdvStructure)$ condition.
  This further implies that the sharing specification $\ShareSpec = (S_1, \ldots, S_h) \defined \{ \PartySet \setminus Z | Z \in \AdvStructure\}$
  satisfies the $\Q^{(3)}(\ShareSpec, \AdvStructure)$ condition. 
 
We begin with the definition of a {\it successful iteration} in protocol $\MultGCE$.
\begin{definition}[{\bf Successful Iteration}]
\label{def:SuccessfulIteration}
In protocol $\MultGCE$, an iteration $\iter$ is called {\it successful}, if every honest $P_i$ sets $\flag^{(i)}_{\iter} = 0$ during the 
 corresponding instance $\MultLCE(\PartySet, \AdvStruct, \SharingSpec,  [a], [b],  \iter)$ of $\MultLCE$.
\end{definition}
  We next claim that during each iteration of the protocol $\MultGCE$, the honest parties will know whether the iteration is successful or not.
  \begin{claim}
  \label{claim:MultGCEIteration}
  For any $\iter$, if all honest parties participate in iteration number $\iter$ of the protocol $\MultGCE$ and if no honest party is ever included
   in the set $\Discarded$, then all honest parties will eventually
  agree on whether the iteration is successful or not.  
  \end{claim}
  \begin{proof}
  We prove the claim through induction on $\iter$. The statement is obviously true for $\iter = 1$, since
   during the instance $\MultLCE(\PartySet, \AdvStruct, \SharingSpec,  [a], [b], 1)$, all honest parties $P_i$
   will eventually set $\flag^{(i)}_{1}$ to a common value from $\{0, 1 \}$ (follows from Lemma \ref{lemma:MultLCE}).
   Assume that the statement is true for $\iter = r$.
  Now consider $\iter = r + 1$ and let all honest parties participate in iteration number $r+1$ by invoking the instance 
  $\MultLCE(\PartySet, \AdvStruct, \SharingSpec,  [a], [b], r + 1)$. From the protocol steps, since the honest parties participate in iteration number
  $r + 1$, it implies that none of the previous $r$ iterations were successful. From Lemma \ref{lemma:MultLCE}, all honest parties from 
  the sets $\Waitlist^{(i)}_1, \ldots, \Waitlist^{(i)}_r$ will eventually be removed for every {\it honest} $P_i$. 
  Moreover, no honest party will ever be included in the sets $\LocalDiscarded^{(i)}_{1}, \ldots, \LocalDiscarded^{(i)}_{r}$.
  Furthermore, as per the lemma condition, no honest party is ever included  in the set $\Discarded$.
  It now follows from Claim \ref{claim:MultLCETermination} and Lemma \ref{lemma:MultLCE}
  that during the instance $\MultLCE(\PartySet, \AdvStruct, \SharingSpec,  [a], [b], r + 1)$,
  all honest parties $P_i$ will eventually set $\flag^{(i)}_{r + 1}$ to a common value from $\{0, 1 \}$  and learn whether the iteration is successful or not.
  \end{proof}
  
  We next claim that if any iteration of $\MultGCE$ is successful, then honest parties output $[ab]$ in that iteration.
    \begin{claim}
  \label{claim:MultGCECorrectness}
If the iteration number $\iter$ in $\MultGCE$ is successful, then honest parties output $[ab]$ during iteration number $\iter$.
  \end{claim}
 \begin{proof}
 Let iteration number $\iter$ in $\MultGCE$ be successful. This implies that every honest $P_i$ sets $\flag_{\iter}^{(i)} = 0$
  during the corresponding instance $\MultLCE(\PartySet, \AdvStruct, \SharingSpec,  [a], [b],  \iter)$ of $\MultLCE$
  and hence this instance of $\MultLCE$ is successful. The proof now follows from Lemma \ref{lemma:MultLCE}.
 \end{proof} 
  
  We next prove that after every $tn + 1$ consecutive unsuccessful iterations of $\MultGCE$, a new corrupt party is globally discarded.
  \begin{claim}
  \label{claim:MultGCECheaterIdentification}
  Let $t \defined \max\{ |Z| :  Z \in \AdvStruct 	\}$.
    In $\MultGCE$, for every $k \geq 1$, if none of the iterations $(k - 1)(tn + 1) + 1, \ldots, k(tn + 1)$ is successful, then eventually, a new corrupt
   party is included in the set $\Discarded$.
  \end{claim}
  \begin{proof}
  Let $Z^{\star} \in \AdvStructure$ be the set of corrupt parties during the execution of $\MultGCE$.
   We prove the claim through strong induction over $k$. 
    \paragraph{\bf Base case: $k = 1$.} We first note that from the protocol steps, the condition $\Discarded = \emptyset$ holds for each of the iterations
    $1, \ldots, tn + 1$, during the corresponding instance of $\MultLCE$ in these iterations. Consequently, 
    from Claim \ref{claim:MultGCEIteration}, for the iterations $1, \ldots, tn + 1$, the honest parties agree on whether the iteration is successful or not.
    Let none of the iterations $1, \ldots, tn + 1$ be successful.
   This implies that for $\iter = 1, \ldots, tn + 1$, none of the instances  $\MultLCE(\PartySet, \AdvStruct, \SharingSpec,  [a], [b],  \iter)$ of $\MultLCE$
   is successful. This further implies that for every $\iter \in \{1, \ldots, tn + 1 \}$, there exists a well-defined {\it unordered} pair of parties $(P_j, P_k)$, such that
   $f_\Char(\iter) = (P_j, P_k)$, with at least one among $P_j, P_k$ being {\it maliciously-corrupt}
    (follows from Claim \ref{MultCIWellDefinedCharacterization}). Let ${\cal C}$ denote the set of all pairs of ``characteristic parties" for the first
    $tn + 1$ instances of $\MultLCE$. That is,
    \[ {\cal C} \defined \{(P_j, P_k): f_\Char(\iter) = (P_j, P_k) \; \mbox{ and } \iter \in \{1, \ldots, tn + 1 \}.  \]
    It then follows that $|{\cal C}| \leq tn$. This is because $|Z^{\star}| \leq t$, implying that there can be at most $tn$ distinct (unordered) pairs of 
    parties, where at least one of the parties in the pair is corrupt. Since the cardinality of ${\cal C}$ is smaller than the number of failed $\MultLCE$ instances,
   from the pigeonhole principle, we can conclude that
    there exist at least two iterations $r, r' \in \{1, \ldots, tn + 1 \}$ where $r < r'$, such that $f_\Char(r) = f_\Char(r') =  (P_j, P_k)$.
    
    Now, let us focus on the failed instances $\MultLCE(\PartySet, \AdvStruct, \SharingSpec,  [a], [b], r)$ and $\MultLCE(\PartySet, \AdvStruct, \SharingSpec,  [a], [b], r')$,
    corresponding to iteration number $r$ and $r'$ respectively in $\MultGCE$.
        Let $\Waitlist^{(i)}_{r}$ and $\LocalDiscarded^{(i)}_{r}$ be the dynamic sets maintained by every party $P_i$ during the instance 
    $\MultLCE(\PartySet, \AdvStruct, \allowbreak \SharingSpec,  [a], [b], r)$. Note that at the time of initializing 
    $\Waitlist^{(i)}_{r}$, both $P_j$ as well as $P_k$ will be present in $\Waitlist^{(i)}_{r}$ (this follows from the protocol steps of $\MultLCE$).
    Let $Z, Z' \in \AdvStructure$ be the {\it conflicting-sets} for the failed instance $\MultLCE(\PartySet, \AdvStruct, \SharingSpec,  [a], [b], r')$.
    From the definition of characteristic function $f_{\Char}$, it follows that $P_j, P_k \in \Selected_{(Z, r')} \cup \Selected_{(Z', r')}$.
    Hence, $P_j$ (resp.~$P_k$) is selected as a summand-sharing party in at least one of instances
    $\OptMult(\PartySet, \AdvStruct, \SharingSpec,[a], [b], Z, r')$ or $\OptMult(\PartySet, \AdvStruct, \SharingSpec,[a], [b], Z', r')$.
    This further implies that there exists at least one honest party, say $P_h$, such that {\it both} $P_j$ as well as $P_k$ are removed
    by $P_h$ from the set $\Waitlist^{(h)}_{r}$. This is because if both $P_j$ as well as $P_k$ are still present in the $\Waitlist^{(i)}_{r}$ set of 
    {\it all} honest parties during the instances $\OptMult(\PartySet, \AdvStruct, \SharingSpec,[a], [b], Z, r')$ and 
    $\OptMult(\PartySet, \AdvStruct, \SharingSpec,[a], [b], Z', r')$, then neither
    $P_j$ not $P_k$ will be selected as a summand-sharing party and hence
    $P_j, P_k \not \in \Selected_{(Z, r')} \cup \Selected_{(Z', r')}$ (follows from Claim \ref{claim:OptMultFuture}), which is a contradiction.  
    Now, if both $P_j$ and $P_k$ are removed from $\Waitlist^{(h)}_{r}$, then from Claim \ref{claim:MultLCECheaterIdentification}, 
    the corrupt party(ies) among $P_j, P_k$ will be eventually included in the $\LocalDiscarded^{(i)}_{r}$ set of {\it every} honest $P_i$.
    For simplicity and without loss of generality, let $P_j$ be the corrupt party among $P_j, P_k$.
    
    In the protocol $\MultGCE$, once the parties find that iteration number $tn + 1$ has failed, they run an instance of ACS to identify a cheating party across
    the first $tn + 1$ failed instances, where
    the parties vote for candidate cheating parties based on the contents of their local $\LocalDiscarded$ sets. To complete the proof for the base case, we need to show 
    that ACS will eventually output a common corrupt party for all the honest parties.  The proof for this is similar to that of Claim \ref{claim:OptMultACS}.
    Namely, as argued above, the corrupt party $P_j$ from the pair $(P_j, P_k)$ above will be eventually included in the $\LocalDiscarded^{(i)}_{r}$ set of {\it every} honest $P_i$.
    We first show that there will be at least one instance of $\FABA$, corresponding to which {\it all} honest parties
  eventually receive the output $1$. For this, we consider two possible cases:
  \begin{myitemize}
  \item[--] {\it At least one honest party participates with input $0$ in the $\FABA$ instance corresponding to $P_j$}: Let 
  $P_i$ be an {\it honest} party, who sends $(\vote, \sid_{j, tn + 1, 1}, 0)$ to $\FABA$ with 
  $\sid_{j, tn + 1, 1}$. Then from the steps of $\MultGCE$, it follows that there exists some $P_{\ell} \in \PartySet$, such that
  $P_i$ has received
    $(\decide, \sid_{\ell, tn + 1, 1}, 1)$ as the output from $\FABA$ with  $\sid_{\ell, tn + 1, 1}$. Hence, {\it every} honest party
    will eventually receive the output $(\decide, \sid_{\ell, tn + 1, 1}, 1)$ as the output from $\FABA$ with  $\sid_{\ell, tn + 1, 1}$.
  \item[--] {\it No honest party participates with input $0$ in the $\FABA$ instance corresponding to $P_j$}:
  In this case, {\it every} honest party will eventually send $(\vote, \sid_{j, tn + 1, 1}, 1)$ to $\FABA$ with 
  $\sid_{j, tn + 1, 1}$. This is because $P_j$ will be eventually included in the $\LocalDiscarded^{(i)}_{r}$ set of {\it every} honest $P_i$.
  Consequently, every honest party eventually receives the output $(\decide, \sid_{j, tn + 1, 1}, 1)$ from $\FABA$. 
  \end{myitemize}
  Now based on the above argument, we can further infer that 
   all honest parties will eventually participate with some input in all the
  $n$ instances of $\FABA$ invoked after the first $tn + 1$ failed iterations
  and hence, all the $n$ instances of $\FABA$  will eventually produce an output.
  Let $P_m$ be the smallest indexed party such that $\FABA$ with 
  $\sid_{m, tn + 1, 1}$ has returned the output $(\decide, \sid_{m, tn + 1, 1}, 1)$.
  Hence, all honest parties eventually include $P_m$ to $\Discarded$. 
  
  Finally, it is easy to see that $P_m \in Z^{\star}$. This is because if $P_m \not \in Z^{\star}$, then $P_m$ is {\it honest}.
  From Claim \ref{claim:MultLCEHonestPartyParticipation} it follows that $P_m$ will {\it not} be included in the
  $\LocalDiscarded^{(i)}_{\iter}$ of any honest $P_i$ for any $\iter \in \{1, \ldots, tn + 1 \}$. 
  Consequently, {\it no} honest $P_i$ will ever send 
   $(\vote, \sid_{m, tn + 1, 1}, 1)$ to $\FABA$ with 
  $\sid_{m, tn + 1, 1}$. Hence, $\FABA$ with 
  $\sid_{m, tn + 1, 1}$ will never return the output $(\decide, \sid_{m, tn + 1, 1}, 1)$, which is a contradiction.
  This completes the proof for the base case.
  \paragraph{\bf Inductive Step:} Let the statement be 
   true for $k = 1, \ldots, k'$. Now consider the case when $k = k' + 1$.
   Let $\Discarded_1, \ldots, \Discarded_{k'}$ be the set of discarded cheating parties 
   after the iteration number $tn + 1, \ldots, k'(tn + 1)$
    respectively.\footnote{Recall that in the protocol, ACS is executed after every block of $tn + 1$ failed iterations and $\Discarded$ gets updated
   through ACS. In the context of the given scenario, the parties would have run ACS after iteration numbers $tn + 1, 2(tn + 1), \ldots, (k' - 1)(tn + 1)$ and $k'(tn + 1)$
   to update the set $\Discarded$. The set $\Discarded_{k'}$ denotes the updated $\Discarded$ set after the $k'^{th}$ ACS execution.}
   From the inductive hypothesis, $\Discarded_1 \subset \Discarded_2 \subset \ldots \subset \Discarded_{k'}$
   and no honest party is present in $\Discarded_{k'}$. Consequently, from the protocol steps and from 
   Claim \ref{claim:MultGCEIteration}, for the iterations $k'(tn + 1) + 1, \ldots, (k' + 1)(tn + 1)$, the honest parties agree on whether the iteration is successful or not.
   Let none of the iterations $k'(tn + 1) + 1, \ldots, (k' + 1)(tn + 1)$ be successful.
   This implies that for $\iter = k'(tn + 1) + 1, \ldots, (k' + 1)(tn + 1)$, 
   none of the instances  $\MultLCE(\PartySet, \AdvStruct, \SharingSpec,  [a], [b],  \iter)$ of $\MultLCE$
   is successful. In the protocol, once the parties find that the iteration number $(k' + 1)(tn + 1)$ is
   not successful, they proceed to select a common cheating party through ACS. Let $\LocalDiscarded^{(i)}_{\iter}, \Waitlist^{(i)}_{\iter}$ be the dynamic sets maintained
   by each party $P_i$ across the iterations $1, \ldots, (k' + 1)(tn + 1)$. 
   
   We first note that none of the parties from $\Discarded_{k'}$ will be selected as a summand-sharing party in any of the
   underlying $\OptMult(\PartySet, \AdvStruct, \SharingSpec,[a], [b], Z, \iter)$ instances, for any $\iter \in \{ k'(tn + 1) + 1, \ldots, (k' + 1)(tn + 1)\}$
   and any $Z \in \AdvStructure$ (this follows from Claim \ref{claim:OptMultFuture}). We also note that 
   there will be at least one party from $Z^{\star}$, which is not present in $\Discarded_{k'}$; i.e.~$\Discarded_{k'} \subset Z^{\star}$.
   If not, then {\it only} honest parties will be selected as summand-sharing parties in all the underlying instances of $\OptMult$
   during the iteration number $k'(tn + 1) + 1$ and hence, the iteration number
   $k'(tn + 1) + 1$ in $\MultGCE$ would be successful, which is a contradiction. Since the iteration number
   $k'(tn + 1) + 1, \ldots, (k' + 1)(tn + 1)$ constitutes $tn + 1$ failed iterations, by applying the same pigeonhole-principle based argument as applied for the base
   case, we can infer that there exists a pair of iterations $r, r' \in \{ k'(tn + 1) + 1, \ldots, (k' + 1)(tn + 1)\}$ where $r < r'$, such that
   $f_{\Char}(r) = f_{\Char}(r') = (P_j, P_k)$, with at least one among $P_j$ and $P_k$ being maliciously-corrupt.
   Moreover, the corrupt party(ies) among $P_j, P_k$ will be from the set $Z^{\star} \setminus \Discarded_{k'}$, since the parties from $\Discarded_{k'}$
   will {\it not} be selected as a summand-sharing party during the iteration number $r$ and $r'$. Next, following the same argument as used for the base case,
   we can infer that the corrupt party(ies) among $P_j$ and $P_k$ will be eventually included in the $\LocalDiscarded^{(i)}_{r}$ set of every {\it honest}
   $P_i$. This will further imply all the $n$ instances of $\FABA$ with $\sid_{1, (k'+1)(tn + 1), (k'+1)}, \ldots, \sid_{n, (k'+1)(tn + 1), (k'+1)}$
   will eventually return an output for all the honest parties, such that at least one of the $\FABA$ instances with $\sid_{\ell, (k'+1)(tn + 1), (k' + 1)}$
   corresponding to the party $P_{\ell}$ will return an output $(\decide, \sid_{\ell, (k' + 1)(tn + 1), (k' + 1)}, 1)$. Let $P_m$ be the smallest indexed party corresponding to 
   which the $\FABA$ instance with  $\sid_{m, (k'+1)(tn + 1), (k' + 1)}$ returns the output $(\decide, \sid_{m, (k' + 1)(tn + 1), (k' + 1)}, 1)$.
   Hence the honest parties will update $\Discarded$ to $\Discarded_{k'} \cup \{P_m \}$.
   To complete the proof, we need to show that $P_m \not \in \Discarded_{k'}$ and $P_m \in Z^{\star}$.
   The former follows from the fact that if $P_m \in \Discarded_{k'}$, then it implies that 
    then no honest party ever sends
   $(\vote, \sid_{m, (k' + 1)(tn + 1), (k' + 1)}, 1)$ to $\FABA$ with 
   $\sid_{m, (k' + 1)(tn + 1), (k' + 1)}$ and consequently , $\FABA$ with 
   $\sid_{m, (k' + 1)(tn + 1), (k' + 1)}$ will never return the output $(\decide, \sid_{m, (k' + 1)(tn + 1), (k' + 1)}, 1)$.
   On the other hand, $P_m \in Z^{\star}$ follows using a similar argument as used for the base case.      
  \end{proof}
  
  An immediate corollary of Claim \ref{claim:MultGCECheaterIdentification} is that there can be at most $t(tn + 1)$ consecutive failed iterations in the protocol 
  $\MultGCE$.
  \begin{corollary}
  \label{corollary:MultGCE}
  In protocol $\MultGCE$, there can be at most $t(tn + 1)$ consecutive failed iterations, where 
  $t \defined \max\{ |Z| :  Z \in \AdvStruct 	\}$.
  \end{corollary}
  
  We next claim that it will take at most $t(tn + 1) + 1$ iterations in the protocol $\MultGCE$ to guarantee that there is at least one successful 
  iteration.
  \begin{claim}
  \label{claim:MultGCETermination}
  In protocol $\MultGCE$, it will take at most $t(tn + 1) + 1$ iterations to ensure that one of these iterations is successful.
  \end{claim}
  \begin{proof}
  Follows easily from Claim \ref{claim:MultGCEIteration} and Corollary \ref{corollary:MultGCE}.
  \end{proof}
  We next claim that the adversary does not learn anything additional about $a$ and $b$ in the protocol.
  \begin{claim}
  \label{claim:MultGCEPrivacy}
  In protocol $\MultGCE$, $\Adv$ does not learn anything additional about $a$ and $b$.
  \end{claim}
  \begin{proof}
  Follows directly from the fact that in every iteration of $\MultGCE$, $\Adv$ does not learn anything additional about $a$ and $b$, 
   which in turn follows from Claim \ref{claim:MultLCEPrivacy}.
  \end{proof}

  Lemma \ref{lemma:MultGCE} now follows from Claim \ref{claim:MultGCETermination}, Claim \ref{claim:MultGCECorrectness} and 
  Claim \ref{claim:MultGCEPrivacy}, where the communication
  complexity follows from the communication complexity of $\MultLCE$ and 
   the fact that there are $t(tn + 1) + 1 = \Order(n^3)$ instances of $\MultLCE$ executed inside the protocol $\MultGCE$. \\~\\
  \noindent {\bf Lemma \ref{lemma:MultGCE}.}
  {\it   Let $\AdvStructure$ satisfy the $\Q^{(4)}(\PartySet, \AdvStructure)$ condition and let 
    $\ShareSpec = (S_1, \ldots, S_h) = \{ \PartySet \setminus Z | Z \in \AdvStructure\}$.
    Then $\MultGCE$ will take at most $t(tn + 1)$ iterations and all honest parties eventually output a secret-sharing of $[ab]$, where
    $t = \max\{ |Z| :  Z \in \AdvStruct 	\}$.
    In the protocol, adversary does not learn anything additional about $a$ and $b$.
         The protocol makes $\Order(|\AdvStructure| \cdot n^5)$ calls to $\FVSS$ and $\FABA$ and additionally incurs a communication of
     $\Order(|\AdvStructure|^2 \cdot n^5 \log{|\F|} + |\AdvStructure| \cdot n^7 \log{|\F|})$ bits.
} 

  \subsection{Perfectly-Secure Pre-Processing Phase Protocol $\PiPerTriples$ and Its Properties}
 Protocol $\PiPerTriples$ for securely realizing $\FTriples$ with $M = 1$ in the $(\FVSS, \FABA)$-hybrid model is presented in
  Fig \ref{fig:PerTriples}.
  
  \begin{protocolsplitbox}{$\PiPerTriples(\PartySet, \AdvStruct, \SharingSpec$)}{A perfectly-secure protocol 
  to securely realize $\FTriples$ with $M = 1$ in $(\FVSS, \FABA)$-hybrid model for session id $\sid$.
 The above code is executed by every party $P_i$}{fig:PerTriples}
 \justify
 \begin{myitemize}
 \item[--] \textbf{Stage I: Generating a Secret-Sharing of Random Pair of Values}.
    \begin{myitemize}
       \item \textbf{Sharing Random Pairs of Values}:
         \begin{myenumerate}
          \item[1.] Randomly select  $a^{(i)}, b^{(i)} \in \F$ and shares $(a^{(i)}_1, \ldots, a^{(i)}_h)$ and 
	    $(b^{(i)}_1, \ldots, b^{(i)}_h)$, such that $a^{(i)}_1 + \ldots + a^{(i)}_h = a^{(i)}$ and $b^{(i)}_1 + \ldots + b^{(i)}_h = b^{(i)}$.
	     Call $\FVSS$ with $(\Dealer, \sid_{i, 1}, (a^{(i)}_1, \ldots, a^{(i)}_h))$ and
	     $\FVSS$ with $(\Dealer, \sid_{i, 2}, (b^{(i)}_1, \ldots, b^{(i)}_h))$ for $\sid_{i, 1}$ and $\sid_{i, 2}$,
	      where $\sid_{i, 1} = \sid || i || 1$ and 
	     $\sid_{i, 2} = \sid || i || 2$.     
	  \item[2.] For $j = 1, \ldots, n$, keep requesting for an output from $\FVSS$ with $\sid_{j, 1}$ and $\sid_{j, 2}$, 
	      till an output is received.
          \end{myenumerate}
         \item {\bf Selecting a Common Subset of Parties Through ACS} 
       \begin{mydescription}
             \item[1.]  If $(\Share, \sid_{j, 1}, P_j, \{[a^{(j)}]_q \}_{P_i \in S_q})$ 
             and $(\Share, \sid_{j, 2}, P_j, \{[b^{(j)}]_q \}_{P_i \in S_q})$ are received from $\FVSS$ with $\sid_{i, 1}$ and $\sid_{i, 2}$ respectively, then send              
              $(\vote, \sid_j, 1)$ to $\FABA$, 
              where $\sid_j \defined \sid || j$.
             \item [2.] For $j = 1, \ldots, n$, request for output from $\FABA$ with $\sid_j$, till an output is received.
             \item [3.] If there exists a subset of parties $\GlobalProv_i$ such that $\PartySet \setminus \GlobalProv_i \in \AdvStructure$
             and $(\decide, \sid_j, 1)$ is received from $\FABA$ with $\sid_j$ corresponding to every $P_j \in \GlobalProv_i$, 
             then send
              $(\vote, \sid_j, 0)$ to $\FABA$ with $\sid_j$
              corresponding to every $P_j$, for which no message has been sent yet.
             \item [4.] Once $(\decide, \sid_j, v_j)$ is received from $\FABA$ for $j = 1, \ldots, n$, set
                  $\CoreSet = \{P_j: v_j = 1\}$.
            \item [5.] 
            Let $a \defined \displaystyle \sum_{P_j \in \CoreSet} a^{(j)}, b \defined \displaystyle \sum_{P_j \in \CoreSet} b^{(j)}$.
            Locally compute the shares $\{[a]_q \}_{P_i \in S_q}$ and $\{[b]_q \}_{P_i \in S_q}$.
       \end{mydescription}
    \end{myitemize}
 \item[--] \textbf{Stage II: Generating the Product}.
      \begin{myitemize}
      \item Participate in the instance $\MultGCE(\PartySet, \AdvStruct, \SharingSpec, [a], [b])$ with $\sid$ 
       and compute $\{[c]_q \}_{P_i \in S_q}$.      
      Output $\{[a]_q, [b]_q, [c]_q \}_{P_i \in S_q}$.
      \end{myitemize}
 
 \end{myitemize}
 \end{protocolsplitbox}

\paragraph{\bf Protocol $\PiPerTriples$ for Generating $L$ Multiplication-Triples:}
 The modifications in $\PiPerTriples$ to generate $M$ multiplication-triples are straight forward.
 During the first stage, each party 
  secret-shares
  $M$ pairs of values, by calling $\FVSS$ $2M$ number of times. While running ACS, a party votes ``positively" for party $P_j$, only
  after receiving an output from {\it all} the $2M$ instances of $\FVSS$ corresponding to $P_j$. During the second stage, the instance of $\MultGCE$ will now take
  $M$ pairs of secret-shared inputs. 

We now prove the security of the protocol $\PiPerTriples$ in the $(\FVSS, \FABA)$-hybrid model.
  While proving these properties, we will assume that  $\AdvStructure$ satisfies the $\Q^{(4)}(\PartySet, \AdvStructure)$ condition.
  This further implies that the sharing specification $\ShareSpec = (S_1, \ldots, S_h) \defined \{ \PartySet \setminus Z | Z \in \AdvStructure\}$
  satisfies the $\Q^{(3)}(\ShareSpec, \AdvStructure)$ condition. 
  \\~\\
   \noindent {\bf Theorem \ref{thm:PiPerTriples}.}
  {\it If $\AdvStructure$ satisfies the $\Q^{(4)}(\PartySet, \AdvStructure)$ condition,
  then  $\PiPerTriples$ is a perfectly-secure protocol for securely realizing $\FTriples$ with UC-security in the $(\FVSS, \FABA)$-hybrid model. 
   The protocol makes $\Order(M \cdot |\AdvStructure| \cdot n^5)$ calls to $\FVSS$, $\Order(|\AdvStructure| \cdot n^5)$ calls to $\FABA$
   and additionally incurs a communication of $\Order(M \cdot |\AdvStructure|^2 \cdot n^5 \log{|\F|} + |\AdvStructure| \cdot n^7 \log{|\F|})$ bits. 
  }
  \begin{proof}
  The communication complexity and the number of calls to $\FVSS$ and $\FABA$ follow from the protocol steps and the communication
  complexity of the protocol $\MultGCE$. So we next prove the security. For ease of explanation, we consider the case where
  only one multiplication-triple is generated in $\PiPerTriples$; i.e.~
   $M = 1$. The proof can be easily modified for any general $M$.
   
   Let $\Adv$ be an arbitrary adversary, attacking the protocol $\PiPerTriples$ by corrupting a set of parties
  $Z^{\star} \in \AdvStructure$, and let $\Env$ be an arbitrary environment. We show the existence of a simulator $\SimPerTriples$ (Fig \ref{fig:SimPerTriples}),
   such that for any
 $Z^{\star} \in \AdvStructure$,
  the outputs of the honest parties and the view of the adversary in the 
   protocol $\PiPerTriples$ is indistinguishable from the outputs of the honest parties and the view of the adversary in an execution in the ideal world involving 
  $\SimPerTriples$ and $\FTriples$. 
  
  The high level idea of the simulator is as follows. Throughout the simulation,
  the simulator itself performs the role of the ideal functionality $\FVSS$ and
   $\FABA$ whenever required. During the first stage,
   whenever $\Adv$ sends a pair of vector of shares to $\FVSS$  on the behalf of a corrupt party, the simulator records
   these vectors.
    On the other hand, for the honest parties, the simulator picks pairs of random values and random shares for those values,
    and distributes the appropriate shares to the corrupt parties, as per $\FVSS$. During ACS, to select the common subset of parties, the simulator
   itself performs the role of $\FABA$ and simulates the honest parties as per the steps of the protocol and $\FABA$. This allows the simulator learn the 
   common subset of parties $\CoreSet$. Notice that 
   the secret-sharing of the pairs of values shared by {\it all} the parties in $\CoreSet$ will be available with the simulator. While the secret-sharing
   of pairs of the honest parties in
   $\CoreSet$ are selected by the simulator itself, for every {\it corrupt} party $P_j$ which is added to $\CoreSet$, at least one honest party $P_i$ should 
    participate with input $1$ in the corresponding call to
   $\FABA$. This implies that the honest party $P_i$ must have received some shares from $\FVSS$ corresponding to the vector of shares
    which $P_j$ sent to
   $\FVSS$. Since in the simulation, the role of $\FVSS$ is played by the simulator itself, it implies that the vector of shares used by 
   $P_j$ will be known to the simulator.
   
   Once the simulator learns $\CoreSet$ and the secret-sharing of the pairs of values shared by the parties
    in $\CoreSet$, during the second stage, the simulator  simulates the rest of the interaction between the honest parties
   and $\Adv$ as per the protocol steps, by itself playing the role of the honest parties. Moreover, in the underlying instances of
   $\OptMult, \MultLCE$ and $\MultGCE$, the simulator itself performs the role of $\FVSS$ and $\FABA$ whenever required. 
   Notice that simulator will be knowing the values which should be shared by the respective parties through $\FVSS$ during the underlying instances
   of $\OptMult$ and $\MultLCE$. This is because these values are completely determined by the secret-sharing of the pairs of values shared by the parties
    in $\CoreSet$, which will be known to the simulator. Consequently, in the simulated execution, the simulator will be knowing which instances
    of $\MultLCE$ are successful and which iterations of $\MultGCE$ are successful. 
    Once the simulated execution is over, 
    the simulator learns the shares of the corrupt parties corresponding to the
   output multiplication-triple in the simulated execution. The simulator then communicates these shares on the behalf of the corrupt parties during its interaction with
   $\FTriples$. This ensures that the shares of the corrupt parties remain the same in both the worlds.
   
   In the steps of the simulator, to distinguish between the values used by the various parties during the real execution and simulated execution, the variables in the simulated execution
   will be used with a \;$\widetilde{}$ symbol.   
\begin{simulatorsplitbox}{$\SimPerTriples$}{Simulator for the protocol $\PiPerTriples$ with $M = 1$
  where $\Adv$ corrupts the parties in set $Z^{\star} \in \AdvStructure$}{fig:SimPerTriples}
	\justify
$\SimPerTriples$ constructs virtual real-world honest parties and invokes the real-world adversary $\Adv$. The simulator simulates the view of
 $\Adv$, namely its communication with $\Env$, the messages sent by the honest parties, and the interaction with $\FVSS$ and $\FABA$. 
  In order to simulate $\Env$, the simulator $\SimPerTriples$ forwards every message it receives from 
   $\Env$ to $\Adv$ and vice-versa.  The simulator then simulates the various stages of the protocol as follows. 
\begin{myitemize}
 \item[--] \textbf{Stage I: Generating a Secret-Sharing of a Random Pair of Values}.
    \begin{myitemize}
       \item \textbf{Sharing Random Pairs of Values}:
          \begin{myitemize}
          \item The simulator simulates the operations of the honest parties during this phase by picking random random pairs of values and random vector of shares
          for those values on their behalf. Namely, when $\Adv$ 
          requests for output from $\FVSS$  with $\sid_{j, 1}$ and $\sid_{j, 2}$ for any $P_j \not \in Z^{\star}$, 
          the simulator picks random values  $\sima^{(j)}, \simb^{(j)} \in \F$ and random shares $(\sima^{(j)}_1, \ldots, \sima^{(j)}_h)$ and 
	    $(\simb^{(j)}_1, \ldots, \simb^{(j)}_h)$, such that $\sima^{(j)}_1 + \ldots + \sima^{(j)}_h = \sima^{(j)}$ and $\simb^{(j)}_1 + \ldots + \simb^{(j)}_h = \simb^{(j)}$.
	    The simulator then sets $[\sima^{(j)}]_q = \sima^{(j)}_q$ and 
             $[\simb^{(j)}]_q = \simb^{(j)}_q$ for $q = 1, \ldots, h$, and 
	    responds to $\Adv$ with the output
          $(\Share, \sid_{j, 1}, P_j, \{[\sima^{(j)}]_q \}_{S_q \cap Z^{\star} \neq \emptyset})$ 
             and $(\Share, \sid_{j, 2}, P_j, \{[\simb^{(j)}]_q \}_{S_q \cap Z^{\star} \neq \emptyset})$ on the behalf of $\FVSS$ with  
             $\sid_{j, 1}$ and $\sid_{j, 2}$ respectively.
           \item Whenever $\Adv$ sends $(\Dealer, \sid_{i, 1}, (\sima^{(i)}_1, \ldots, \sima^{(i)}_h))$ 
           and $(\Dealer, \sid_{i, 2}, (\simb^{(i)}_1, \ldots, \simb^{(i)}_h))$
           to $\FVSS$ with $\sid_{i, 1}$ and $\sid_{i, 2}$ respectively on the behalf of any $P_i \in Z^{\star}$, the simulator sets
           $[\sima^{(i)}]_q = \sima^{(i)}_q$ and $[\simb^{(i)}]_q = \simb^{(i)}_q$ for $q = 1, \ldots, h$, where
           $\sima^{(i)} \defined \sima^{(i)}_1 + \ldots + \sima^{(i)}_h$ and $\simb^{(i)} \defined \simb^{(i)}_1 + \ldots + \simb^{(i)}_h$.
          \end{myitemize}
      \item {\bf Selecting a Common Subset of Parties (ACS)}: 
      The simulator simulates the interface to $\FABA$ for $\Adv$ by itself
      performing the role of $\FABA$ and playing the role of the honest parties, as per the steps of the protocol. 
      When the first honest party completes this phase during the simulated execution, $\SimPerTriples$ learns the set $\CoreSet$. 
      It then sets $\sima \defined \displaystyle \sum_{P_j \in \CoreSet} \sima^{(j)}, \simb \defined \displaystyle \sum_{P_j \in \CoreSet} \simb^{(j)}$
      and computes $[\sima] =  \displaystyle \sum_{P_j \in \CoreSet} [\sima^{(j)}], [\simb] = \displaystyle \sum_{P_j \in \CoreSet} [\simb^{(j)}]$.      
    \end{myitemize}
    \item[--] \textbf{Stage II: Generating the Product}. The simulator plays the role of the honest parties as per the protocol
    and interacts with $\Adv$ for the instance $\MultGCE(\PartySet, \AdvStruct, \SharingSpec, [\sima], [\simb])$, where during the instance,
    the simulator uses the shares $\{[\sima]_q, [\simb]_q \}_{P_i \in S_q}$ on the behalf of every $P_i \not \in Z^{\star}$.
    Moreover, during this instance of $\MultGCE$, the simulator simulates the interface to $\FABA$ for $\Adv$ during the underlying instances of
    $\OptMult$ and during cheater identification, 
     by itself
    performing the role of $\FABA$, as per the steps of the protocol. 
    Furthermore, during the underlying instances of $\OptMult$ and $\MultLCE$, whenever required, the simulator itself plays the role $\FVSS$.    
    \item[--] \textbf{Interaction with $\FTriples$}: Let $\{ [\simc]_q \}_{S_q \cap Z^{\star} \neq \emptyset}$
     be the output shares of 
    the parties in $Z^{\star}$, during the instance $\MultGCE(\PartySet, \AdvStruct, \allowbreak \SharingSpec, [\sima], [\simb])$. The simulator sends 
     $(\shares, \sid, \{[\sima]_q, [\simb]_q, [\simc]_q \}_{S_q \cap Z^{\star} \neq \emptyset})$ to $\FTriples$, on the behalf of the parties in $Z^{\star}$.
 
\end{myitemize}    
\end{simulatorsplitbox}

  We now prove a series of claims, which will help us to finally prove the theorem. 
    We first claim that in
  any execution of $\PiPerTriples$, a set $\CoreSet$ is eventually generated.
  \begin{claim}
  \label{claim:PerTriplesTermination}
  In any execution of $\PiPerTriples$, a common
   set $\CoreSet$ is eventually generated where $\PartySet \setminus \CoreSet \in \AdvStructure$,
    such that for every $P_j \in \CoreSet$, there exists a pair of values held by $P_j$,
   which are eventually secret-shared.
  \end{claim}
  \begin{proof}
  We first show that there always exists some set $Z \in \AdvStructure$ such that in the $\FABA$ instances corresponding to every party
   in $\PartySet \setminus Z$, all honest parties eventually obtain an output $1$. For this, we consider the following two cases.
  \begin{myitemize}
  \item[--] {\it If some {\it honest} party $P_i$ has participated with $\vote$ input $0$ in any instance of $\FABA$ during step 3 of the ACS phase}: 
   this implies that there exists a subset $\GlobalProv_i$ for $P_i$ where $\PartySet \setminus \GlobalProv_i \in \AdvStructure$, such that
   $P_i$ receives the output $(\decide, \sid_j, 1)$ from $\FABA$ with $\sid_j$, corresponding to every $P_j \in \GlobalProv_i$. 
   Consequently, every honest party will eventually receive the same outputs from these $\FABA$ instances.
    Since $\PartySet \setminus \GlobalProv_i \in \AdvStructure$, we get that
   there exists some set $Z \in \AdvStructure$ such that the $\FABA$ instances  corresponding to every party
   in $\PartySet \setminus Z$ responded with output $1$, which is what
  we wanted to show.
   \item[--] {\it No honest party has participated with $\vote$ input $0$ in any instance of $\FABA$}:
    In the protocol, each party $P_j \not \in Z^{\star}$
    sends its vector of shares to $\FVSS$ with
    $\sid_{j, 1}$ and $\sid_{j, 2}$ and every 
    {\it honest} party eventually receives its respective shares from these vectors as the output from the corresponding instances of $\FVSS$.
     Hence, corresponding to
    every $P_j \not \in Z^{\star}$, {\it all} honest parties eventually participate with input $(\vote, \sid_j, 1)$ during the 
    instance of $\FABA$ with $\sid_j$, and this $\FABA$ instance will eventually respond with output $(\decide, \sid_j, 1)$.
    Since $Z^{\star} \in \AdvStructure$, it then follows that even in this case,  
     there exists some set $Z \in \AdvStructure$ such that the $\FABA$ instances  corresponding to every party
   in $\PartySet \setminus Z$ responded with output $1$.
  \end{myitemize}
  We next show that all honest parties eventually receive an output from {\it all} the instances of $\FABA$. Since we have shown
  there exists some set $Z \in \AdvStructure$ such that the $\FABA$ instances  corresponding to every party
   in $\PartySet \setminus Z$ eventually returns the output $1$, it thus follows that all honest parties eventually participate with some $\vote$ inputs in the remaining $\FABA$
  instances and hence will eventually obtain some output from these $\FABA$ instances as well. Since the set 
  $\CoreSet$ corresponds to the $\FABA$ instances in which the honest parties obtain $1$ as the output, it thus follows that eventually,
  the honest parties obtain some $\CoreSet$ where $\PartySet \setminus \CoreSet \in \AdvStructure$.
   Moreover, the set $\CoreSet$ will be common, as it is based
  on the outcome of $\FABA$ instances.
  
  Now consider an arbitrary $P_j \in \CoreSet$. This implies that the parties obtain $1$ as the output from the
  $j^{th}$ instance of $\FABA$. This further implies that at least one {\it honest} party $P_i$ participated in this 
  $\FABA$ instance with $\vote$ input $1$.    This is possible only if $P_i$ received its respective shares from
   the instances of $\FVSS$ with $\sid_{j, 1}$ and $\sid_{j, 2}$,
    further implying that $P_j$ has provided some vector of shares
    $(a^{(j)}_1, \ldots, a^{(j)}_h)$ and $(b^{(i)}_1, \ldots, b^{(i)}_h)$    as inputs to these $\FVSS$ instances.
     It now follows easily that eventually, all honest parties will have their respective shares corresponding to the
    vectors of shares provided by $P_j$, implying that the honest parties will eventually hold $[a^{(j)}]$ and
    $[b^{(j)}]$, where $a^{(j)} \defined a^{(j)}_1 +  \ldots + a^{(j)}_h$ and  $b^{(j)} \defined b^{(j)}_1 +  \ldots + b^{(j)}_h$.
  \end{proof}
We next show that the view generated by $\SimPerTriples$ for $\Adv$ is identically distributed to $\Adv$'s view during the real execution of $\PiPerTriples$.
\begin{claim}
\label{claim:PerTriplesPrivacy}
The view of $\Adv$ in the simulated execution with $\SimPerTriples$ is identically distributed as the view of $\Adv$ in the real execution of
 $\PiPerTriples$.
\end{claim}
\begin{proof}
We first note that in the real-world (during the real execution of $\PiPerTriples$), the view of $\Adv$ consists of the following:
 \begin{mydescription}
 \item[(1)] The vector of shares $(a^{(j)}_1, \ldots, a^{(j)}_h)$ and $(b^{(i)}_1, \ldots, b^{(i)}_h)$   (if any) for
  $\FVSS$ with $\sid_{j, 1}$ and $\sid_{j, 2}$ respectively, corresponding to $P_j \in Z^{\star}$.
 \item[(2)] Shares  $\{[a^{(j)}]_q, [b^{(j)}]_q \}_{S_q \cap Z^{\star} \neq \emptyset}$, corresponding to $P_j \not \in Z^{\star}$.
 \item[(3)] Inputs of the various parties during the $\FABA$ instances as part of ACS and the outputs from the $\FABA$ instances.
 \item[(4)] The view generated for $\Adv$ during the instance of $\MultGCE$.
 \end{mydescription}
The vectors of shares in $(1)$ are the inputs of $\Adv$ and hence they are identically distributed in both the real as well as simulated
 execution of $\PiPerTriples$, so let us fix these vectors. 
  In the real execution, every $P_j \not \in Z^{\star}$ picks its pair of values randomly and the 
  vectors of shares for $\FVSS$, corresponding to these values, uniformly at random.
  On the other hand, in the simulated execution, the simulator picks the pair of values and their shares randomly 
   on the behalf of $P_j$. Now since the sharing specification
   $\ShareSpec = (S_1, \ldots, S_h) \defined \{ \PartySet \setminus Z | Z \in \AdvStructure\}$ is {\it $\AdvStructure$-private}, 
     it follows that the distribution of the shares in $(2)$ is identical in both the real, as well as the simulated execution.
   Specifically, conditioned on the shares in $(2)$, the underlying pairs of values shared by the parties $P_j \not \in Z^{\star}$
   are uniformly distributed.
  Since the partial view of $\Adv$ containing $(1)$ and $(2)$ are identically distributed, let us fix them.
  
  Now conditioned on $(1)$ and $(2)$, it is easy to see that the partial view of $\Adv$ consisting of $(3)$ is identically distributed in both the executions.
 This is because the outputs of the $\FABA$ instances are determined {\it deterministically} based on the inputs provided by the various parties
 in these $\FABA$ instances. Furthermore, the inputs of the parties in these $\FABA$ instances depend upon the order in which various parties
 receive outputs from various $\FVSS$ instances, which is completely determined by $\Adv$, since message scheduling is under the control of
 $\Adv$. Since in the simulated execution, the simulator provides the interface to various instances of $\FABA$ to $\Adv$
  in exactly the same way as $\FABA$ would have been accessed by $\Adv$ in the real execution, it follows
   that the partial view of $\Adv$ containing $(1), (2)$ and $(3)$ is identically distributed in both the executions and so let us fix this.
 This also fixes the set $\CoreSet$, which according to Claim \ref{claim:PerTriplesTermination}, is guaranteed to be generated. 
 
 Let $[a]$ and $[b]$ be the secret-sharing held by the honest parties after stage I, conditioned on the view of $\Adv$ in $(1), (2)$ and $(3)$. 
  Note that in the simulated execution, the simulator will be knowing the complete sharing $[a]$ and $[b]$. 
   This is because $[a]$ and $[b]$ are computed {\it deterministically} based on the secret-sharing of the pairs of the values shared by the
   parties in $\CoreSet$, all of which will be completely known to the simulator in the simulated execution. 
   To complete the proof of the claim, we need to show that the partial view of $\Adv$ consisting of $(4)$
   is identically distributed in both the executions (conditioned on $(1), (2)$ and $(3)$). 
   However, this follows from the privacy of $\OptMult, \MultLCE$ and $\MultGCE$
   (Claims \ref{claim:OptMultPrivacy}, \ref{claim:MultLCEPrivacy} and \ref{claim:MultGCEPrivacy}) 
   and the fact that in the simulated execution, simulator plays the role of the honest parties
   during the instance of $\MultGCE$ exactly as per the steps of $\MultGCE$, where the simulator will be completely knowing the 
   shares of both $[a]$ and $[b]$ corresponding to both the honest as well as corrupt parties. Consequently, it will be knowing the shares
   with which different parties have to participate in the underlying instances of $\OptMult$ and $\MultLCE$.
   Moreover, in the simulated execution, the simulator honestly plays the role of $\FVSS$ and $\FABA$ in these $\OptMult$
   and $\MultLCE$ instances. This 
   guarantees that the view of $\Adv$ during the real execution of the $\MultGCE$ instance
   is exactly the same as the view of $\Adv$ during the simulated execution of $\MultGCE$.
\end{proof}
Finally, we show that conditioned on the view of $\Adv$, the outputs of the honest parties are identically distributed in both the worlds.
\begin{claim}
\label{claim:PerTriplesOutputDistribution}
Conditioned on the view of $\Adv$, the output of the honest parties are identically distributed in the real execution of $\PiPerTriples$
  involving $\Adv$, as well as in the ideal execution
  involving $\SimPerTriples$ and $\FTriples$.
\end{claim}
\begin{proof}
 Let $\view$ be an arbitrary view of $\Adv$, and let $\{([a^{(j)}], [b^{(j)}])\}_{P_j \in \CoreSet}$
 be the secret-sharing of the pairs of values as per $\view$, shared by the parties in $\CoreSet$.
  Note that $\CoreSet$ is bound to have at least one honest party. This is because $\PartySet \setminus \CoreSet \in \AdvStructure$
   and if $\CoreSet \subseteq Z^{\star}$, then it implies that $\AdvStructure$ {\it does not} satisfy
   the $\Q^{(2)}(\PartySet, \AdvStructure)$ condition, which is a contradiction.
  From the proof of Claim \ref{claim:PerTriplesPrivacy}, it follows that
  corresponding to every {\it honest} $P_j \in \CoreSet$, the pairs $(a^{(j)}, b^{(j)})$
  are uniformly distributed conditioned on the shares of these pairs, as determined by $\view$. 
  Let us fix these pairs. 
  
  We show that in the real-world, the 
  honest parties eventually output $([a], [b], [c])$, where conditioned on $\view$, 
  the triple $(a, b, c)$ is a uniformly random multiplication-triple over $\F$.
   From the protocol steps, the parties set 
  $[a] \defined \displaystyle \sum_{P_j \in \CoreSet} [a^{(j)}], [b] \defined \displaystyle \sum_{P_j \in \CoreSet} [b^{(j)}]$.
  Since corresponding to every $P_j \in \CoreSet$, the honest parties eventually hold 
  $[a^{(j)}]$ and $[b^{(j)}]$ (follows from Claim \ref{claim:PerTriplesTermination}), it follows that the honest parties eventually hold
  $[a]$ and $[b]$. Moreover, since $[c]$ is computed as the output of the instance
  $\MultGCE(\PartySet, \AdvStruct, \SharingSpec, [a], [b])$, it follows from Lemma \ref{lemma:MultGCE} that the honest parties will eventually hold
  $[c]$, where $c = ab$. We next show that conditioned on $\view$, the multiplication-triple $(a, b, c)$ is uniformly distributed over $\F$.
   However, this follows from the fact that there exists a one-to-one correspondence between the random pairs shared by the honest
   parties in $\CoreSet$ and $(a, b)$. 
     Namely, from the view point of $\Adv$, for every candidate pair $(a^{(j)}, b^{(j)})$
     shared by the {\it honest} parties $P_j \in \CoreSet$,
    there exists a unique $(a, b)$ which is consistent with $\view$. Since the pairs
     shared by the {\it honest} parties $P_j$ are uniformly distributed and independent of
    $\view$, it follows that $(a, b)$ is also uniformly distributed.
    Since $c = ab$ holds, it follows that $(a, b, c)$ is uniformly distributed. 

To complete the proof, we now show that conditioned on the shares $\{([a]_q, [b]_q, [c]_q) \}_{S_q \cap Z^{\star} \neq \emptyset}$ (which are determined by
    $\view$), the honest
    parties output a secret-sharing
     of some random multiplication-triple in the ideal-world which is consistent with
   the shares  $\{([a]_q, [b]_q, [c]_q) \}_{S_q \cap Z^{\star} \neq \emptyset}$.
    However, this simply follows from the fact that in the ideal-world, the simulator
    $\SimPerTriples$ sends the shares $\{([a]_q, [b]_q, [c]_q) \}_{S_q \cap Z^{\star} \neq \emptyset}$
     to $\FTriples$ on the behalf of the parties in $Z^{\star}$.
    As an output, $\FTriples$ generates a random secret-sharing of some random multiplication-triple consistent with the shares
    provided by $\SimPerTriples$. 
\end{proof}
  The proof of Theorem \ref{thm:PiPerTriples} now follows from Claim \ref{claim:PerTriplesPrivacy} and Claim \ref{claim:PerTriplesOutputDistribution}. 
  \end{proof}

%% file: AppStatAMPCv2.tex
\section{Properties of the Statistically-Secure Pre-Processing Phase}
\label{app:Statistical}
In this section, we prove the security properties of all the statistically-secure subprotocols, followed by the
 statistically-secure preprocessing phase. We first start with the AICP.
\subsection{Properties of Our AICP}
\label{app:AICP}
In this section, we formally prove the properties of our AICP. While proving these properties, we assume that $\AdvStruct$ satisfies the $\Q^{(3)}(\PartySet,\AdvStruct)$ condition. We first show that when $\mathsf{S}$, $\mathsf{I}$ and $\mathsf{R}$ are honest, then all honest parties set the local bit indicating that the authentication has completed to $1$. Furthermore, $\mathsf{R}$ will accept the signature revealed by $\mathsf{I}$.

\begin{claim}[{\bf Correctness}]
\label{claim:AICPCorrectness}
If $\mathsf{S}, \mathsf{I}$ and $\mathsf{R}$ are {\it honest}, then each honest $P_i$ 
 eventually sets $\authCompleted^{(\sid,i)}_{\mathsf{S},\mathsf{I},\mathsf{R}}$ to $1$ during $\Auth$. Moreover, $\mathsf{R}$ eventually
  outputs $s$ during $\Reveal$.  
\end{claim}
\begin{proof}
Let $\mathsf{S},\mathsf{I}$ and $\mathsf{R}$ be honest and let $\Hon$ be the set of {\it honest} parties among $\PartySet$. 
 Moreover, let $Z^{\star} \in \AdvStructure$ be the set of {\it corrupt} parties.
 We first show that each honest party $P_i$
eventually sets $\authCompleted^{(\sid,i)}_{\mathsf{S},\mathsf{I},\mathsf{R}}$ to $1$ during $\Auth$. During $\Auth$, $\mathsf{S}$ chooses 
 the signing-polynomial $F(x)$ such that $s = F(0)$ holds. $\mathsf{S}$ will then
 send the signing-polynomial $F(x)$ and masking-polynomial $M(x)$  to $\mathsf{I}$, and the 
 corresponding verification-point $(\alpha_i, v_i, m_i)$ to each verifier $P_i$, such that $v_i = F(\alpha_i)$ and $m_i = M(\alpha_i)$ holds.
 Consequently, each verifier in $\Hon$ 
  will eventually receive its verification-point and indicates this to $\mathsf{I}$. Since $\PartySet \setminus \Hon = Z^{\star} \in \AdvStructure$,
  it follows that $\mathsf{I}$ will eventually find a set $\R$, such that $\PartySet \setminus \R \in \AdvStructure$, and where 
   each verifier in $\R$ has indicated to $\mathsf{I}$ that it has received its verification-point. Consequently,
    $\mathsf{I}$ will compute $B(x) = dF(x) + M(x)$, and broadcast $(d, B(x), \R)$, which is eventually delivered to every honest party, including 
     $\mathsf{S}$.
     Moreover,  $\mathsf{S}$ will find that $B(\alpha_j) = dv_j + m_j$ holds 
     for all the verifiers $P_j \in \R$. Consequently, $\mathsf{S}$ will broadcast an $\OK$ message, which is eventually received by every honest party $P_i$, who then
     sets $\authCompleted^{(\sid,i)}_{\mathsf{S},\mathsf{I},\mathsf{R}}$ to $1$. Moreover, $\mathsf{I}$ will set $\ICSig(\mathsf{S}, \mathsf{I}, \mathsf{R}, s)$ to
     $F(x)$. 
     
     During $\Reveal$, $\mathsf{I}$ will send $F(x)$ to $\mathsf{R}$, and each verifier $P_i \in \Hon \cap \R$ will send its 
     verification points $(\alpha_i,v_i,m_i)$ to $\mathsf{R}$. These points and the polynomial $F(x)$ are eventually received by
     $\mathsf{R}$. Moreover, the condition $v_i = F(\alpha_i)$ will hold true for these points, and consequently these points
      will be {\it accepted}. 
     Since $\R \setminus (\Hon \cap \R) \subseteq Z^{\star} \in \AdvStructure$, it follows that
     $\mathsf{R}$ will eventually find a subset $\R' \subseteq \R$ where $\R \setminus \R' \in \AdvStructure$, such that the points corresponding to
     all the parties in $\R'$ are accepted. This implies that 
      $\mathsf{R}$ will eventually output $s = F(0)$.  
\end{proof}

We next show that when $\mathsf{S},\mathsf{I}$ and $\mathsf{R}$ are {\it honest}, then the adversary does not learn anything about
  $s$ during either $\Auth$ or $\Reveal$.
\begin{claim}[{\bf Privacy}]
\label{claim:AICPPrivacy}
If $\mathsf{S}, \mathsf{I}$ and $\mathsf{R}$ are {\it honest}, then the view of adversary $\Adv$ throughout $\Auth$ and $\Reveal$ is 
 independent of $s$. 
\end{claim}
\begin{proof}
Let $t = \max\{ |Z| :  Z \in \AdvStruct 	\}$ and let $Z^\star \in \AdvStructure$ be the set of corrupt parties. 
 For simplicity and without loss of generality, let 
  $|Z^\star| = t$. 
  During $\Auth$, the adversary $\Adv$ learns  $t$ verification-points $\{(\alpha_i, v_i, m_i) \}_{P_i \in Z^\star}$. 
  However, since $F(x)$ is a random $t$-degree polynomial with $F(0) = s$,
  the points $\{(\alpha_i, v_i) \}_{P_i \in Z^{\star}}$ are distributed independently of $s$.
  That is, for every candidate $s \in \F$ from the point of view of $\Adv$, there is a corresponding unique
  $t$-degree polynomial $F(x)$, such that $F(\alpha_i) = v_i$ holds corresponding to every $P_i \in Z^{\star}$.

  During $\Auth$, the adversary $\Adv$ also learns $d$ and the blinded-polynomial $B(x) = dF(x) + M(x)$, 
  along with the points $\{(\alpha_i, v_i) \}_{P_i \in Z^{\star}}$.
    However, this does not add any new information about $s$ to the view of the adversary. This is because
  $M(x)$ is a random $t$-degree polynomial. Hence for every candidate $M(x)$ polynomial from the point of view of
  $\Adv$ where $M(\alpha_i) = m_i$ holds for every $P_i \in Z^{\star}$, there is a corresponding unique 
  $t$-degree polynomial $F(x)$, such that $F(\alpha_i) = v_i$ holds corresponding to every $P_i \in Z^{\star}$, and where
  $dF(x) + M(x) = B(x)$. We also note that in $\Auth$, the signer $\mathsf{S}$ does not broadcast $s$, which follows from the 
  Claim \ref{claim:AICPCorrectness}. Finally, $\Adv$ does not learn anything new about $s$ during $\Reveal$, 
  since the verification-points and the signing-polynomial are sent only to $\mathsf{R}$.  
\end{proof}

We next prove the unforgeability property.
\begin{claim}[{\bf Unforgeability}]
\label{claim:AICPUnforgeability}
If $\mathsf{S}, \mathsf{R}$ are {\it honest}, $\mathsf{I}$ is corrupt
      and if $\mathsf{R}$ outputs $s'$ during $\Reveal$, then 
      $s' = s$ holds, except with probability 
      at most $\frac{nt}{|\F| - 1}$.
\end{claim}
\begin{proof}
Let $\Hon$ be the set of honest parties in $\PartySet$ and let $Z^{\star}$ be the set of corrupt parties.
 Since $\mathsf{R}$ outputs $s'$ during $\Reveal$, it implies that during $\Auth$, the variable 
 $\authCompleted^{(\sid,i)}_{\mathsf{S},\mathsf{I},\mathsf{R}}$ is set to $1$ by $\mathsf{R}$, if
 $\mathsf{R} = P_i$. This further implies that $\mathsf{S}$ has broadcasted either an $\OK$ or an $\NOK$ message
 during $\Auth$, which further implies that $\mathsf{I}$ has broadcasted some blinded-polynomial $B(x)$ during $\Auth$.
 Now there are now two possible cases.
 \begin{myitemize}
     \item[--] {\it $\mathsf{S}$ has broadcasted $\NOK$ along with $s$ during $\Auth$}:
       In this case, every honest party including $\mathsf{R}$ would set $\ICSig(\mathsf{S}, \mathsf{I}, \mathsf{R}, s)$ to $s$ during $\Auth$. Moreover, during
       $\Reveal$, the receiver $\mathsf{R}$ outputs $s$. Hence, in this case, $s' = s$ holds with probability $1$.       
    \item[--] {\it $\mathsf{S}$ has broadcasted $\OK$ during $\Auth$}: This implies that during $\Auth$, 
    $\mathsf{I}$ had broadcasted a $t$-degree blinded-polynomial $B(x)$, along with the set $\R$. Furthermore,
    $\mathsf{S}$ has verified that $\PartySet \setminus \R \in \AdvStructure$ and 
    $B(\alpha_i) = dv_i + m_i$ holds for every verifier $P_i \in \R$.
    Now during $\Reveal$, if $\mathsf{I}$ sends $F(x)$ as $\ICSig(\mathsf{S}, \mathsf{I}, \mathsf{R}, s)$ to $\mathsf{R}$, then again
    $s' = s$ holds with probability $1$.     
    So consider the case when $\mathsf{I}$ sends $F'(x)$ as $\ICSig(\mathsf{S}, \mathsf{I}, \mathsf{R}, s)$ to $\mathsf{R}$, where 
    $F'(x)$ is a $t$-degree polynomial such that $F'(x) \neq F(x)$ and where $F'(0) = s'$. In this case, we show that except with probability at most $\frac{nt}{|\F| - 1}$,
    the verification-point of no {\it honest} verifier from $\R$ will get accepted by $\mathsf{R}$ during $\Reveal$, with respect to $F'(x)$.
     Now assuming that this statement is true,
    the proof follows from the fact that in order for $F'(x)$ to be accepted by $\mathsf{R}$, it should accept the verification-point of at least
    one {\it honest} verifier from $\R$ with respect to $F'(x)$. This is because 
    $\mathsf{R}$ should find a subset of verifiers $\R' \subseteq \R$ whose corresponding verification-points are accepted, where $\R \setminus \R' \in \AdvStructure$. So clearly,
    the set of {\it corrupt} verifiers
     in $\R$ {\it cannot} form a candidate $\R'$. This is because since $\AdvStructure$ satisfies the $\Q^{(3)}(\PartySet,\AdvStructure)$ condition,  
     it satisfies the $\Q^{(2)}(\R,\AdvStructure)$ condition as $\PartySet \setminus \R \in \AdvStructure$.
     This further implies that $\AdvStructure$ satisfies the $\Q^{(1)}(\R',\AdvStructure)$ condition as $\R \setminus \R' \in \AdvStructure$. 
     Hence, any candidate for $\R'$ must contain at least one honest party from $\R$.

    Consider an arbitrary verifier $P_i \in \Hon \cap \R$ from which $\mathsf{R}$ receives the verification-point $(\alpha_i,v_i,m_i)$ during $\Reveal$. 
      This point can be {\it accepted} only if either of the following holds.
	\begin{myitemize}
	\item $v_i = F'(\alpha_i)$:
	 This is possible with probability at most $\frac{t}{|\F| - 1}$. 
	 This is because $F'(x)$ and $F(x)$, being distinct $t$-degree polynomials can have  at most $t$ points in common, and since the evaluation-point $\alpha_i$ corresponding to $P_i$, being randomly selected from $\F - \{0\}$, will not be known to
	 $\mathsf{I}$.
	\item $dv_i + m_i \neq B(\alpha_i)$: This is impossible, as otherwise $\mathsf{S}$ would have broadcasted $s$ and $\NOK$ during 
	$\Auth$, which is a contradiction.	
	\end{myitemize}
     Now as there could be up to $n - 1$ {\it honest} verifiers in $\R$, it follows from the union bound that except with probability at most $\frac{nt}{|\F| - 1}$,
     the polynomial $F'(x)$ will not be accepted.     
 \end{myitemize}
\end{proof}

We next prove the non-repudiation property.
\begin{claim}[{\bf Non-Repudiation}]
\label{claim:AICPNonRepudiation}
If $\mathsf{S}$ is {\it corrupt} and $\mathsf{I}, \mathsf{R}$ are {\it honest}
    and if $\mathsf{I}$ has set $\ICSig(\mathsf{S}, \mathsf{I}, \mathsf{R}, s)$ during $\Auth$, then $\mathsf{R}$ 
   eventually outputs 
    $s$ during $\Reveal$, except with probability at most $\frac{n}{|\F| - 1}$.    
\end{claim}
\begin{proof}
Let $\Hon$ be the set of honest parties in $\PartySet$ and $Z^{\star} \in \AdvStructure$ be the set of corrupt parties.
 Since  $\mathsf{I}$ has set $\ICSig(\mathsf{S}, \mathsf{I}, \mathsf{R}, s)$ during $\Auth$, it implies that 
 that it has set the variable $\authCompleted^{(\sid, i)}_{\mathsf{S},\mathsf{I},\mathsf{R}}$ to $1$ during $\Auth$, if 
 $\mathsf{I} = P_i$. This further implies that $\mathsf{I}$ has broadcasted a blinded-polynomial 
 $B(x)$, the linear combiner $d$ and the set $\R$, where $B(x) = dF(x) + M(x)$ and where
  $F(x)$ and $M(x)$ are the signing and masking polynomials received by $\mathsf{I}$ from $\mathsf{S}$.  
   Moreover, $\mathsf{S}$ has broadcasted either an $\OK$ message or an $\NOK$ message. Consequently, all {\it honest}
   parties $P_j$, including $\mathsf{R}$, eventually set $\authCompleted^{(\sid, j)}_{\mathsf{S},\mathsf{I},\mathsf{R}}$ to $1$.
   Now there are two possible cases.   
   \begin{myitemize}
    \item {\it $\mathsf{S}$ has broadcasted $\NOK$, along with $s$ during $\Auth$}: In this case, all {\it honest} parties, including
    $\mathsf{I}$ and $\mathsf{R}$, set $\ICSig(\mathsf{S}, \mathsf{I}, \mathsf{R}, s)$ to $s$ during $\Auth$. 
    Moreover, from the steps of $\Reveal$, $\mathsf{R}$ outputs $s$ during $\Reveal$. Thus, the claim holds in this case with probability $1$.
    \item {\it $\mathsf{S}$ has broadcasted $\OK$ during $\Auth$}: In this case, 
    $\mathsf{I}$ sets $\ICSig(\mathsf{S}, \mathsf{I}, \mathsf{R}, s)$ to $F(x)$, where
    $F(0) = s$. 
    During $\Reveal$, $\mathsf{I}$ sends
    $F(x)$ to $\mathsf{R}$. Moreover, every verifier $P_i \in \Hon \cap \R$ eventually sends its verification-point 
    $(\alpha_i,v_i,m_i)$ to $\mathsf{R}$. We next show that except with probability at most $\frac{n}{|\F|-1}$, 
    all these verification-points are  accepted by $\mathsf{R}$. 
    Now assuming that this statement is true, the proof follows from the fact that 
      $\Hon \cap \R = \R \setminus Z^{\star}$. Consequently, $\mathsf{R}$ eventually accepts the verification-points from a subset of
      parties $\R' \subseteq \R$ where $\R \setminus \R' \in \AdvStructure$ and outputs
      $s$.
      
     Consider an arbitrary verifier $P_i \in \Hon \cap \R$ whose verification-point $(\alpha_i, v_i, m_i)$ is received by $\mathsf{R}$
     during $\Reveal$. Now there are two possible cases, depending upon the relationship that holds between $F(\alpha_i)$ and $v_i$
      during $\Auth$.
     \begin{myitemize}
	\item[--] {\it $v_i = F(\alpha_i)$ holds}: In this case, according to the protocol steps of $\Reveal$, the point $(\alpha_i, v_i, m_i)$ is {\it accepted} by $\mathsf{R}$.
	\item[--] {\it $v_i \neq F(\alpha_i)$ holds}: In this case, we claim that except with probability at most $\frac{1}{|\F| - 1}$, the condition 
	$dv_i + m_i \neq B(\alpha_i)$ will hold, implying that the point $(\alpha_i, v_i, m_i)$ is {\it accepted} by $\mathsf{R}$.
	This is because the {\it only} way $dv_i + m_i = B(\alpha_i)$ holds is when $\mathsf{S}$ distributes $(\alpha_i, v_i, m_i)$ to $P_i$
	where $v_i \neq F(\alpha_i)$ and $m_i \neq M(\alpha_i)$ holds, and $\mathsf{I}$ selects $d = (M(\alpha_i) - m_i)\cdot(v_i - F(\alpha_i))^{-1}$.
	However, $\mathsf{S}$ will {\it not} be knowing the random $d$ from $\F \setminus \{ 0\}$ which $\mathsf{I}$ picks while distributing
	$F(x), M(x)$ to $\mathsf{I}$, and $(\alpha_i, v_i, m_i)$ to $P_i$. 
	\end{myitemize}
    Now, as there can be up to $n - 1$ {\it honest} verifiers in $\R$, from the union bound, it follows that except with probability at most 
    $\frac{n}{|\F| - 1}$, the verification-point of {\it all} honest verifiers in $\R$ are accepted by $\mathsf{R}$.
   \end{myitemize}  
\end{proof}
We next derive the communication complexity of $\Auth$ and $\Reveal$.
\begin{claim}
\label{claim:AICPCommunication}
 Protocol $\Auth$ incurs a communication of $\Order(n \cdot \log{|\F|})$ bits
 and $\Order(1)$ calls to $\FAcast$ with $\Order(n \cdot \log{|\F|})$-bit messages. Protocol 
  $\Reveal$ requires a communication of $\Order(n \cdot \log{|\F|})$ bits.
\end{claim}
\begin{proof}
During $\Auth$, signer $\mathsf{S}$ sends $t$-degree polynomials
  $F(x)$ and $M(x)$ to $\mathsf{I}$, and verification-points to each verifier. This requires a communication of $\Order(n \cdot \log{|\F|})$ bits.
  Intermediary $\mathsf{I}$ needs to broadcast $B(x), d$ and the set $\R$, which requires one call to $\FAcast$ with a message of size
  $\Order(n \cdot \log{|\F|})$ bits. Moreover, $\mathsf{S}$ may need to broadcast $s$, 
  which requires one call to $\FAcast$ with a message of size
  $\Order(\log{|\F|})$ bits. 
   During $\Reveal$, $\mathsf{I}$ may send $F(x)$ to $\mathsf{R}$, and each verifier may send its verification-point to $\mathsf{R}$.
    This will require a communication of $\Order(n \cdot \log{|\F|})$ bits.
\end{proof}

Lemma \ref{lemma:AICP} now follows from Claims \ref{claim:AICPCorrectness}-\ref{claim:AICPCommunication}. \\~\\
\noindent {\bf Lemma \ref{lemma:AICP}.}
{\it Let $\AdvStructure$ satisfy the $\Q^{(3)}(\PartySet, \AdvStructure)$ condition. Then the pair of protocols $(\Auth, \Reveal)$ satisfies the following properties, except with
  probability at most $\errorAICP \defined \frac{nt}{|\F|-1}$, where
   $t = \max\{ |Z| :  Z \in \AdvStruct 	\}$.
    \begin{myitemize}
       \item[--] {\bf Correctness}:  If $\mathsf{S}, \mathsf{I}$ and $\mathsf{R}$ are {\it honest}, then each honest 
   $P_i$ eventually sets $\authCompleted^{(\sid,i)}_{\mathsf{S},\mathsf{I},\mathsf{R}}$ 
   to $1$ during $\Auth$. Moreover, $\mathsf{R}$ eventually outputs $s$ during $\Reveal$.
       \item[--] {\bf Privacy}: If $\mathsf{S}, \mathsf{I}$ and $\mathsf{R}$ are {\it honest}, then the view of adversary remains
      independent of $s$.
     \item[--] {\bf Unforgeability}: If $\mathsf{S}, \mathsf{R}$ are {\it honest}, $\mathsf{I}$ is corrupt
      and if $\mathsf{R}$ outputs $s' \in \F$ during $\Reveal$, then 
      $s' = s$ holds.      
    \item[--] {\bf Non-repudiation}: If $\mathsf{S}$ is {\it corrupt} and $\mathsf{I}, \mathsf{R}$ are {\it honest}
    and if $\mathsf{I}$ has set $\ICSig(\mathsf{S}, \mathsf{I}, \mathsf{R}, s)$ during $\Auth$, then $\mathsf{R}$ 
    eventually outputs 
    $s$ during $\Reveal$.    
    \end{myitemize}
 Protocol $\Auth$ incurs a communication of $\Order(n \cdot \log{|\F|})$ bits
 and $\Order(1)$ calls to $\FAcast$ with $\Order(n \cdot \log{|\F|})$-bit messages. Protocol 
  $\Reveal$ requires a communication of $\Order(n \cdot \log{|\F|})$ bits.}
\subsection{Statistical VSS Protocol}
\label{app:StatisticalVSS}
In this section, we prove the properties of $\PiStatVSS$ (see Fig \ref{fig:VSS_stat} for the formal description of the protocol)
 stated in Theorem \ref{thm:StatVSS}. Throughout the section, we assume that $\AdvStruct$ satisfies the $\Q^{(3)}(\PartySet,\AdvStruct)$ condition, 
  implying that $\ShareSpec = (S_1, \ldots, S_h) \defined \{ \PartySet \setminus Z | Z \in \AdvStructure\}$ satisfies the $\Q^{(2)}(\ShareSpec,\AdvStruct)$ condition. \\ \\
{\bf Theorem \ref{thm:StatVSS}}
 {\it  Let $\AdvStructure$ satisfy the $\Q^{(3)}(\PartySet, \AdvStructure)$ condition
  and let  $\ShareSpec = (S_1, \ldots, S_h) = \{ \PartySet \setminus Z | Z \in \AdvStructure\}$. 
  Then protocol $\PiStatVSS$ UC-securely computes $\FVSS$
   in the $\FAcast$-hybrid model, except with an error probability of at most $|\AdvStructure| \cdot n^3 \cdot \errorAICP$, where $\errorAICP \approx \frac{n^2}{|\F|}$. 
   The protocol makes $\Order(|\AdvStructure| \cdot n^3)$ calls to $\FAcast$ with $\Order(n \cdot \log{|\F|})$ bit messages
    and additionally incurs a communication of 
   $\Order(|\AdvStructure| \cdot n^4 \log{|\F|})$ bits. By replacing the calls to $\FAcast$ with protocol $\PiAcast$, the protocol incurs a total communication of
   $\Order(|\AdvStructure| \cdot n^6 \log{|\F|})$ bits.   
}
\begin{proof}
In the protocol, the dealer needs to send the share $s_q$ to all the parties in $S_q$, and this requires a communication of 
 $\Order(|\AdvStruct| \cdot n \log{|\F|})$ bits. An instance of $\Auth$ and $\Reveal$ is executed with respect to every ordered triplet of parties
  $P_i, P_j, P_k \in S_q$, leading to $\Order(|\AdvStructure| \cdot n^3)$ instances of $\Auth$ and $\Reveal$ being executed. 
    The communication complexity now follows from the communication complexity of $\Auth$ and $\Reveal$ (Claim \ref{claim:AICPCommunication})
    and from the communication complexity of the protocol $\PiAcast$ (Theorem \ref{thm:Acast}).

We next prove the security of the protocol. Let $\Adv$ be an arbitrary adversary, attacking the protocol $\PiStatVSS$ by corrupting a set of parties
  $Z^{\star} \in \AdvStructure$, and let $\Env$ be an arbitrary environment. We show the existence of a simulator $\SimSVSS$, such that for any
 $Z^{\star} \in \AdvStructure$,
  the outputs of the honest parties and the view of the adversary in the 
   protocol $\PiStatVSS$ is indistinguishable from the outputs of the honest parties and the view of the adversary in an execution in the ideal world involving 
  $\SimSVSS$ and $\FVSS$, except with probability at most $|\AdvStructure| \cdot n^3 \cdot \errorAICP$, where $\errorAICP \approx \frac{n^2}{|\F|}$
    (see Lemma \ref{lemma:AICP}). The simulator is very 
     similar to the simulator $\SimPVSS$ for the protocol $\PiPerVSS$ (see Fig \ref{fig:SimPVSS} in Appendix \ref{app:PerfectVSS}),
      except that the simulator now has to simulate giving and accepting signatures on the behalf of honest parties, as part of pairwise consistency checks.
       In addition, for each $S_q \in \ShareSpec$, the simulator has to simulate revealing signatures to the corrupt parties in
        $S_q \setminus \C_q$ on the behalf of the honest parties in $\C_q$. The simulator is formally presented in Figure \ref{fig:SimSVSS}.

\begin{simulatorsplitbox}{$\SimSVSS$}{Simulator for the protocol $\PiStatVSS$ where $\Adv$ corrupts the parties in set $Z^{\star}  \in \AdvStructure$}{fig:SimSVSS}
	\justify
$\SimSVSS$ constructs virtual real-world honest parties and invokes the real-world adversary $\Adv$. The simulator simulates the view of
 $\Adv$, namely its communication with $\Env$, the messages sent by the honest parties and the interaction with $\FAcast$. 
  In order to simulate $\Env$, the simulator $\SimPVSS$ forwards every message it receives from 
   $\Env$ to $\Adv$ and vice-versa.  The simulator then simulates the various phases of the protocol as follows, depending upon whether the dealer is honest or corrupt. \\~\\
 \centerline{\underline{\bf Simulation When $P_{\D}$ is Honest}} 
\noindent {\bf Interaction with $\FVSS$}: the simulator interacts with the functionality $\FVSS$ and receives a request based delayed output
  $(\Share,\sid,P_{\D},\{[s]_q\}_{S_q \cap Z^{\star} \neq \emptyset})$, on the behalf of the parties in $Z^{\star}$.\\[.2cm]
\noindent {\bf Distribution of Shares}: On the behalf of the dealer, the simulator sends $(\dist, \sid, P_{\D}, q, [s]_q)$ to $\Adv$, corresponding to 
 every $P_i \in Z^{\star} \cap S_q$. \\[.2cm]
 \noindent {\bf Pairwise Consistency Tests on IC-Signed Values}: 
 \begin{myitemize}
 \item[--]  For each $S_q \in \SharingSpec$ such that $S_q \cap Z^{\star} \neq \emptyset$,
  corresponding to each $P_i \in S_q \cap Z^{\star}$, the simulator does the following.
 	\begin{myitemize}
 	\item  On the behalf of every party $P_j \in S_q \setminus Z^{\star}$ as a signer
	and every $P_k \in S_q$ as a receiver, perform the role of the signer and the honest verifiers as per the steps of $\Auth$
	and interact with $\Adv$ on the behalf of the honest parties to give 
	 $\ICSig(\sid^{(P_\D,q)}_{j,i,k},P_j, P_i, P_k, s_{qj})$ to $P_i$, where $s_{qj} = [s]_q$.
 	\item On the behalf of every $P_j, P_k \in S_q$ as intermediary and receiver respectively, perform the role of the honest parties
	as per the steps of $\Auth$ and interact with $\Adv$ on the behalf of the honest parties, if $\Adv$
	gives the signature $\ICSig(\sid^{(P_\D,q)}_{i,j,k}, P_i, P_j, P_k, s_{qi})$  to $P_j$ on the behalf of the signer $P_i$. 
	Upon receiving the signature $\ICSig(\sid^{(P_\D,q)}_{i,j,k}, P_i, P_j, P_k, s_{qi})$ from $P_i$, record it.
	 	\end{myitemize}
 \item[--] For each $S_q \in \SharingSpec$ and for every $P_i, P_j \in S_q \setminus Z^{\star}$ the simulator simulates $P_i$ giving
   $\ICSig(\sid^{(P_\D,q)}_{i,j,k},P_i, P_j, P_k,v)$ to $P_j$, corresponding to each $P_k \in S_q$, 
   by playing the role of the honest parties and interacting with $\Adv$ on their behalf,
   as per the steps of $\Auth$, in the respective $\Auth$ instances. Based on the following conditions, 
   the simulator chooses the value $v$ in these instances  as follows.
 	\begin{myitemize}
 	\item {$S_q \cap Z^{\star} \neq \emptyset$}: Choose $v$ to be $[s]_q$. 
 	\item {$S_q \cap Z^{\star} = \emptyset$}: Pick a random element from $\F$ as $v$.
 	\end{myitemize}
 
 \end{myitemize}
 \noindent {\bf Announcing Results of Pairwise Consistency Tests}: 
\begin{myitemize}
\item[--] If for any $S_q \in \SharingSpec$, $\Adv$ requests an output from $\FAcast$ with $\sid^{(P_{\D}, q)}_{i, j}$ corresponding to parties
 $P_i \in S_q \setminus Z^{\star}$ and $P_j \in S_q$, then the simulator provides the output on the behalf of $\FAcast$ as follows.
 	\begin{myitemize}
 	\item If $P_j \in S_q \setminus Z^{\star}$, then send the output $(P_i, \Acast, \sid^{(P_{\D}, q)}_{i, j}, \OK_q(i,j))$.
 	\item If $P_j \in (S_q \cap Z^{\star})$, then send the output $(P_i, \Acast ,\sid^{(P_{\D}, q)}_{i, j}, \OK_q(i,j))$, if 
	$\ICSig(\sid^{(P_\D,q)}_{j,i,k},P_j, P_i, P_k, s_{qj})$ has been recorded 
	         on the behalf of $P_j$ as a signer, corresponding to the intermediary $P_i$ and every $P_k \in S_q$ as a receiver, 
	         such that $s_{qj} = [s]_q$ holds.	         	       
	\end{myitemize} 
\item[--] If for any $S_q \in \SharingSpec$ and any $P_i \in S_q \cap Z^{\star}$, 
 $\Adv$ sends $(P_i, \Acast, \sid^{(P_{\D}, q)}_{i, j}, \OK_q(i,j))$ to $\FAcast$
 with $\sid^{(P_{\D}, q)}_{i, j}$ on the behalf of 
  $P_i$ for any $P_j \in S_q$, then the simulator records it. 
  Moreover, if $\Adv$ requests an output from $\FAcast$ with $\sid^{(P_{\D}, q)}_{i, j}$, then 
  the simulator sends the output $(P_i, \Acast ,\sid^{(P_{\D}, q)}_{i, j}, \OK_q(i,j))$ on the behalf of $\FAcast$.
\end{myitemize}
\noindent {\bf Construction of Core Sets and Public Announcement}: 
\begin{myitemize}
\item[--] For each $S_q \in \ShareSpec$, the simulator plays the role of $P_{\D}$ and adds the 
 edge $(P_i, P_j)$ to the graph $G^{(\D)}_q$ over the vertex set $S_q$, if any one of the following is true.
		\begin{myenumerate}
		\item $P_i,P_j \in S_q \setminus Z^{\star}$.
		\item If $P_i \in S_q \cap Z^{\star}$ and $P_j \in S_q \setminus Z^{\star}$, then the simulator has recorded 
		$(P_i, \Acast, \sid^{(P_{\D}, q)}_{i, j}, \OK_q(i,j))$ sent by $\Adv$ on the behalf of $P_i$ to $\FAcast$ with $\sid^{(P_{\D}, q)}_{i, j}$, and has 
		recorded $\ICSig(\sid^{(P_\D,q)}_{i,j,k},P_i, P_j, P_k, s_{qi})$ on the behalf of $P_i$ as a signer
		and $P_j$ as an intermediary corresponding to every party $P_k \in S_q$ as a receiver, such that $s_{qi} = [s]_q$ holds.
		\item If $P_i, P_j \in S_q \cap Z^{\star}$, then the simulator has recorded $(P_i, \Acast, \sid^{(q)}_{i, j}, \OK_q(i,j))$
		 and $(P_j, \Acast,  \sid^{(q)}_{j, i}, \allowbreak, \OK_q(j,i))$ sent by $\Adv$ on behalf $P_i$ and $P_j$ respectively, to 
		 $\FAcast$ with $\sid^{(P_{\D}, q)}_{i, j}$ and $\FAcast$ with $\sid^{(P_{\D}, q)}_{j, i}$.
		\end{myenumerate}	
\item[--] For each $S_q \in \ShareSpec$, the simulator finds a set $\C_q$ which forms a clique in $G^{\D}_q$, such that $S_q \setminus \C_q \in \AdvStructure$.
  When $\Adv$ requests output from $\FAcast$ with $\sid_{P_{\D}}$, 
  the simulator sends the output $(\Sender, \Acast, \sid_{P_{\D}}, \{\C_q\}_{S_q \in \ShareSpec})$ on the behalf of $\FAcast$.
\end{myitemize}
\noindent {\bf Share Computation}: Once $\C_1, \ldots, \C_q$ are computed, then 
 for each $S_q \in \ShareSpec$, 
  simulator does the following for every $P_i \in (S_q \setminus \C_q) \cap Z^{\star}$ and every $P_j \in \C_q \setminus Z^{\star}$.
	\begin{myitemize}
	\item[--] Simulate the revelation of the signature $\ICSig(\sid^{(P_\D,q)}_{k,j,i},P_k, P_j, P_i, s_{qk})$ to $P_i$
	on the behalf of the intermediary $P_j$ corresponding to  every signer $P_k \in \C_q$, where $s_{qk} = [s]_q$, 
	by playing the role of the honest parties as per $\Reveal$ and interacting with $\Adv$. \\[.2cm] 	
	\end{myitemize}
 \centerline{\underline{\bf Simulation When $P_{\D}$ is Corrupt}} 
  In this case, the simulator $\SimSVSS$ interacts with $\Adv$ during the various phases of $\PiStatVSS$ as follows. \\[.2cm]
\noindent {\bf Distribution of Shares}: For $q = 1, \ldots, h$, if $\Adv$ sends $(\dist, \sid, P_{\D}, q, v)$ 
 on the behalf of $P_{\D}$ to any party $P_i \in S_q \setminus Z^{\star}$, then the simulator records it and sets $s_{qi}$ to $v$. \\[.2cm]
 \noindent {\bf Pairwise Consistency Tests on IC-Signed Values}: 
 \begin{myitemize}
 \item[--] For each $S_q \in \SharingSpec$ such that 
  $S_q \cap Z^{\star} \neq \emptyset$, corresponding to each party $P_i \in S_q \cap Z^{\star}$ and each 
  $P_j \in S_q \setminus Z^{\star}$, the simulator does the following.
 	\begin{myitemize}
 	\item If $s_{qj}$ has been set to some value, then simulate giving
	 $\ICSig(\sid^{(P_\D,q)}_{j,i,k},P_j, P_i, P_k, s_{qj})$ to $\Adv$ on the behalf of $P_j$ as a signer, corresponding to every $P_k \in \PartySet$ as receiver,
	 by playing the role of the honest parties as per the steps of $\Auth$.
 	\item Upon receiving $\ICSig(\sid^{(P_\D,q)}_{i,j,k}, P_i, P_j, P_k, s_{qi})$ from $\Adv$ on the behalf of $P_i$ as a 
	signer, corresponding to $P_j \in S_q$ as an intermediary and $P_k \in S_q$ as a receiver, record $\ICSig(\sid^{(P_\D,q)}_{i,j,k}, P_i, P_j, P_k, s_{qi})$.
 	\end{myitemize}
 \item[--] For each $S_q \in \SharingSpec$ such that 
  $S_q \cap Z^{\star} = \emptyset$, corresponding to each party $P_i, P_j \in S_q$, the simulator does the following.
  	\begin{myitemize}
  	\item Upon setting $s_{qi}$ to some value, simulate $P_i$ giving $\ICSig(\sid^{(P_\D,q)}_{i,j,k},P_i, P_j, P_k,s_{qi})$ to $P_j$, corresponding
	to every receiver $P_k \in S_q$, by playing the role of the honest parties and interacting with $\Adv$ as per the steps of $\Auth$.
  	\end{myitemize}
 \end{myitemize}
 
 \noindent {\bf Announcing Results of Pairwise Consistency Tests}: 
\begin{myitemize}
\item[--] If for any $S_q \in \SharingSpec$, $\Adv$ requests an output from $\FAcast$ with $\sid^{(P_{\D}, q)}_{i, j}$ corresponding to
 parties $P_i \in S_q \setminus Z^{\star}$ and $P_j \in S_q$, then the simulator provides the 
  output on the behalf of $\FAcast$ as follows, if $s_{qi}$ has been set to some value.
 	\begin{myitemize}
 	\item If $P_j \in S_q \setminus Z^{\star}$, then send the output $(P_i,\Acast,\sid^{(P_{\D}, q)}_{i, j},\OK_q(i,j))$, if
	$s_{qj}$ has been set to some value and $s_{qi} = s_{qj}$ holds.
 	\item If $P_j \in S_q \cap Z^{\star}$, then send the output $(P_i, \Acast ,\sid^{(P_{\D}, q)}_{i, j}, \OK_q(i,j))$, if 
	$\ICSig(\sid^{(P_\D,q)}_{j,i,k},P_j, P_i, P_k, s_{qj})$ has been recorded 
	         on the behalf of $P_j$ as a signer for the intermediary $P_i$, corresponding to every $P_k \in S_q$ as a receiver, such that $s_{qj} = s_{qi}$ holds.
 	\end{myitemize} 
\item[--]  If for any $S_q \in \SharingSpec$ and any $P_i \in S_q \cap Z^{\star}$, 
 $\Adv$ sends $(P_i, \Acast, \sid^{(P_{\D}, q)}_{i, j}, \OK_q(i,j))$ to $\FAcast$
 with $\sid^{(P_{\D}, q)}_{i, j}$ on the behalf of 
  $P_i$ for any $P_j \in S_q$, then the simulator records it. 
  Moreover, if $\Adv$ requests for an output from $\FAcast$ with $\sid^{(P_{\D}, q)}_{i, j}$, then 
  the simulator sends the output $(P_i, \Acast ,\sid^{(P_{\D}, q)}_{i, j}, \OK_q(i,j))$ on the behalf of $\FAcast$.
  \end{myitemize}
\noindent {\bf Construction of Core Sets}: For each $S_q \in \ShareSpec$, the simulator plays the role of 
 the honest parties $P_i \in S_q \setminus Z^{\star}$ and adds the edge $(P_j,P_k)$ to the graph $G^{(i)}_q$ over vertex set $S_q$, if any one of the following is true.
		\begin{myitemize}
		\item If $P_j, P_k \in S_q \setminus Z^{\star}$, then the simulator has set $s_{qj}$ and $s_{qk}$ to some values, such that $s_{qj} = s_{qk}$ holds.
		\item If $P_j \in S_q \cap Z^{\star}$ and $P_k \in S_q \setminus Z^{\star}$, then all the following should hold.
		\begin{myitemize}
		\item[--] The simulator has recorded 
		$(P_j, \Acast, \sid^{(P_{\D}, q)}_{j, k}, \OK_q(j,k))$ sent by $\Adv$ on the behalf of $P_j$ to $\FAcast$ with $\sid^{(P_{\D}, q)}_{j, k}$;
		\item[--] The simulator has recorded 
		$\ICSig(\sid^{(P_\D,q)}_{j,k,m},P_j, P_k, P_m, s_{qj})$ on the behalf of $P_j$ as a signer and $P_k$ as an intermediary, corresponding to
		 every receiver  $P_m \in S_q$;
		 \item[--] The simulator has set $s_{qk}$ to a value such that $s_{qj} = s_{qk}$ holds.
		 \end{myitemize}
		\item If $P_j, P_k \in S_q \cap Z^{\star}$, then the simulator has recorded 
		$(P_j, \Acast, \sid^{(P_{\D}, q)}_{j, k}, \OK_q(j, k))$ and $(P_k, \Acast, \sid^{(P_{\D}, q)}_{k, j}, \allowbreak  \OK_q(k, j))$ sent by $\Adv$ on behalf of
		$P_j$ and $P_k$ respectively, to $\FAcast$ with $\sid^{(P_{\D}, q)}_{j, k}$ and $\FAcast$ with $\sid^{(P_{\D}, q)}_{k, j}$. 		
		\end{myitemize}
\noindent {\bf Verification of Core Sets and Interaction with $\FVSS$}:
      \begin{myitemize}
        \item If  $\Adv$ sends $(\Sender, \Acast, \sid_{P_{\D}}, \{\C_q\}_{S_q \in \ShareSpec})$ to $\FAcast$ with $ \sid_{P_{\D}}$ on the behalf of
  $P_{\D}$, then the simulator records it. Moreover, if $\Adv$ requests for an output from $\FAcast$ with $ \sid_{P_{\D}}$, then on the behalf
  of $\FAcast$, the simulator sends the output $(P_{\D}, \Acast , \sid_{P_{\D}}, \{\C_q\}_{S_q \in \ShareSpec})$.
       \item If simulator has recorded the sets  $ \{\C_q\}_{S_q \in \ShareSpec}$, then it plays 
     the role of the honest parties and verifies if for $q = 1, \ldots, h$, the set 
     $\C_q$ is valid with respect to $S_q$, by checking
     if $S_q \setminus \C_q \in \AdvStructure$ and if $\C_q$ constitutes a clique in the graph $G^{(i)}_q$ of every party $P_i \in \PartySet \setminus Z^{\star}$.
    If $\C_1, \ldots, \C_q$ are valid, then the simulator sends $(\Share, \sid, P_{\D}, \{s_q\}_{S_q \in \ShareSpec})$ to $\FVSS$, where $s_q$ is set to 
     $s_{qi}$ corresponding to any $P_i \in \C_q \setminus Z^{\star}$. 
     \end{myitemize}
     
\end{simulatorsplitbox}   

We now prove a series of claims, which helps us to prove the theorem. We start with an {\it honest}
 $P_{\D}$.
\begin{claim}
\label{claim:SVSSHonestDPrivacy}
If $P_{\D}$ is honest, then the view of $\Adv$ in the simulated execution of $\PiStatVSS$ with $\SimPVSS$ is identically distributed 
 to the view of $\Adv$ in the real execution of $\PiStatVSS$ involving honest parties.
\end{claim}
\begin{proof}
Let $\ShareSpec^{\star} \defined
  \{S_q \in \ShareSpec \mid S_q \cap Z^{\star} \neq \emptyset \}$. Then the 
   view of $\Adv$ during the two executions consists of the following.
\begin{myenumerate}
\item[--] {\bf The shares $\{[s]_q\}_{S_q \in \ShareSpec^{\star}}$ distributed by $P_{\D}$}: In the real execution, $\Adv$ receives $[s]_q$ from $P_{\D}$ for each $S_q \in \ShareSpec^{\star}$. In the simulated execution, the simulator provides this to $\Adv$ on behalf of $P_{\D}$. Clearly, the distribution of the shares is identical in both the
 executions.
\item[--] {\bf Corresponding to every $S_q \in \ShareSpec^{\star}$ and every triplet of parties $P_i, P_j, P_k$ where $P_j \in S_q \setminus Z^{\star}$,
  $P_i \in S_q \cap Z^{\star}$ and $P_k \in S_q$, the signature
  $\ICSig(\sid^{(P_\D,q)}_{j,i,k},P_j, P_i, P_k, s_{qj})$ received from $P_j$ 
    as part of pairwise consistency tests}: While $P_j$ sends this to $\Adv$ in the real execution, the simulator sends this on the behalf of $P_j$ in the simulated execution.
    Clearly, the distribution of the messages learnt by $\Adv$ during the corresponding instances of $\Auth$ is identical in both the executions.    
\item[--] {\bf Corresponding to every $S_q \in \ShareSpec$, every pair of parties $P_i, P_j \in S_q \setminus Z^{\star}$ and every $P_k \in S_q$, the view generated 
  when $P_i$ gives $\ICSig(\sid^{(P_\D,q)}_{i,j,k}, P_i, P_j, P_k, v))$ to $P_j$}: We consider the following two cases.
	\begin{myitemize}
	\item $S_q \in \ShareSpec^{\star}$ : In both the real and simulated execution, the value of $v$ is $[s]_q$. Since the simulator simulates the interaction of honest parties with $\Adv$ during the simulated execution, the distribution of messages is identical in both the executions.
	\item $S_q \notin \ShareSpec^{\star}$ : In the simulated execution, the simulator chooses $v$ to be a random element from $\F$, while in the real execution, $v$ is $[s]_q$. However, due to the privacy property of AICP (Claim \ref{claim:AICPPrivacy}), the view of $\Adv$ is independent of the value of $v$ in either of the 
	executions. Hence, the distribution of the messages is identical in both the executions. 
	\end{myitemize}
\item[--] {\bf For every $S_q \in \ShareSpec$ and every $P_i,P_j \in S_q$, the outputs $(P_i, \Acast, \sid^{(P_{\D}, q)}_{i, j}, \OK_q(i,j))$
 of the pairwise consistency tests, received from $\FAcast$ with $\sid^{(P_{\D}, q)}_{i, j}$}: 
 To compare the distribution of these messages in the two executions,
  we consider the following cases, considering an arbitrary $S_q \in \ShareSpec$ and arbitrary $P_i, P_j \in S_q$.
	\begin{myitemize}
	\item[--] $P_i,P_j \in S_q \setminus Z^{\star}$: In both the executions, 
	 $\Adv$ receives $(P_i, \Acast, \sid^{(P_{\D}, q)}_{ij}, \OK_q(i,j))$ as the output from $\FAcast$ with $\sid^{(P_{\D}, q)}_{i, j}$.
	\item[--] $P_i \in S_q \setminus Z^{\star}, P_j \in (S_q \cap Z^{\star}) $:  In both the 
	 executions, $\Adv$ receives $(P_i, \Acast, \sid^{(P_{\D}, q)}_{i, j}, \allowbreak \OK_q(i,j))$ as the output from $\FAcast$ with 
	 $\sid^{(P_{\D}, q)}_{i, j}$ if and only if $\Adv$ gave $\ICSig(\sid^{P_\D,q}_{j,i,k},P_j, P_i, P_k,s_{qj})$ on the behalf of $P_j$ to
	  $P_i$, corresponding to every $P_k \in S_q$, such that $s_{qj} = [s]_q$ holds.
	\item[--] $P_i \in (S_q \cap Z^{\star}) $:  In both the executions, 
	$\Adv$ receives $(P_i, \Acast, \sid^{(q)}_{i, j}, \OK_q(i,j))$ if and only if 
	$\Adv$ on the behalf of $P_i$ has sent $(P_i, \Acast, \sid^{(P_{\D}, q)}_{i, j}, \OK_q(i,j))$ to $\FAcast$
	with $\sid^{(P_{\D}, q)}_{i, j}$ for $P_j$.
	\end{myitemize}
   Clearly, irrespective of the case, the distribution of the $\OK$ messages is identical in both the executions.
\item[--] {\bf The core sets $\{\C_q\}_{S_q \in \ShareSpec}$}: In both the 
 executions, the sets $\C_q$ are determined based on the $\OK_q$ messages delivered to $P_{\D}$. So the distribution of these sets is also identical. 
\end{myenumerate}
\item[--] {\bf Corresponding to every $S_q \in \ShareSpec^{\star}$, for every triplet of parties $P_i, P_j, P_k$ where
 $P_i \in \C_q \setminus Z^{\star}$, $P_j \in (S_q \setminus \C_q) \cap Z^{\star}$ and 
 $P_k \in \C_q$, 
  the signatures $\ICSig(\sid^{P_\D,q}_{k,i,j},P_k, P_i, P_j,s_{qk})$ revealed by party $P_i$ to $P_j$, signed by party $P_k$}: We note that the distribution of core sets $\C_q$ 
   is the same in both the executions. In the real execution, $P_i$, upon receiving $\ICSig(\sid^{P_\D,q}_{k,i,j},P_k, P_i, P_j,s_{qk})$ during $\Auth$,
    checks if $s_{qk} = s_{qi}$ holds, before adding the edge $(P_i,P_k)$ in $G^i_q$. Since $P_\D$ is honest, $s_{qi} = [s]_q$. In the simulated
     execution as well, the simulator reveals $\ICSig(\sid^{P_\D,q}_{k,i,j},P_k, P_i, P_j,s_{qk})$ to $\Adv$, where $s_{qk} = [s]_q$. Hence, the distribution of messages is identical in both executions.
\end{proof}

We next claim that if the dealer is {\it honest}, then conditioned on the view of the adversary $\Adv$ (which is identically distributed in both the executions, as per the previous claim),
  the outputs of the honest parties are identically distributed in both the executions.
\begin{claim}
\label{claim:SVSSHonestDCorrectness}
If $P_{\D}$ is honest, then conditioned on the view of $\Adv$, the output of the honest parties during the execution of $\PiStatVSS$ involving $\Adv$ has the
   same distribution as the output of the honest parties in the ideal-world involving $\SimPVSS$ and $\FVSS$, except with probability
   at most $|\AdvStructure| \cdot n^3 \cdot \errorAICP$, where $\errorAICP \approx \frac{n^2}{|\F|}$.
\end{claim}
\begin{proof}
Let $P_{\D}$ be honest and let $\View$ be an arbitrary view of $\Adv$. Moreover, let $\{s_q \}_{S_q \cap Z^{\star} \neq \emptyset}$ be the shares
 of the corrupt parties, as per $\View$. Furthermore, let $\{s_q \}_{S_q \cap Z^{\star} = \emptyset}$ be the shares used by $P_{\D}$
  in the simulated execution corresponding to the set $S_q \in \ShareSpec$, such that $S_q \cap Z^{\star} = \emptyset$. 
     Let $s \defined \displaystyle \sum_{S_q \cap Z^{\star} \neq \emptyset} s_q + \sum_{S_q \cap Z^{\star} = \emptyset} s_q$.
   Then, in the simulated execution, each {\it honest} party $P_i$ obtains the output $\{[s]_q\}_{P_i \in S_q}$ from $\FVSS$, where
   $[s]_q = s_q$. We now show that except with probability at most $|\AdvStructure| \cdot n^3 \cdot \errorAICP$,
   each honest $P_i$ eventually obtains the
    output $\{[s]_q\}_{P_i \in S_q}$ in the real execution as well, if $P_{\D}$'s inputs in the protocol $\PiStatVSS$ are 
    $\{s_q \}_{S_q \in \ShareSpec}$.
    
    Since $P_{\D}$ is {\it honest}, it sends the share $s_q$ to {\it all} the parties in the set $S_q$, which is eventually delivered. Now consider {\it any} $S_q \in \ShareSpec$.
     During the pairwise consistency tests, each {\it honest} $P_k \in S_q$ will eventually send $\ICSig(\sid^{(P_\D,q)}_{k,j,m}, P_k, P_j, P_m, s_{qk})$ to 
      {\it all} the parties $P_j$ in $S_q$, with respect to every receiver
       $P_m \in \PartySet$, where $s_{qk} = s_q$. Consequently, every {\it honest} $P_j \in S_q$ will eventually broadcast the $\OK_q(j, k)$ message, 
       corresponding to every {\it honest} $P_k \in S_q$. This is because, by the correctness of AICP (Claim \ref{claim:AICPCorrectness}), 
       $P_j$ will receive $s_{qk}$, and $s_{qj} = s_{qk} = s_q$ will hold.
      So, every {\it honest} party (including $P_{\D}$) eventually receives the 
      $\OK_q(j, k)$ messages
          This implies that the parties in $S_q \setminus Z^{\star}$ will eventually form a clique in the graph $G^{(i)}_q$
      of every {\it honest} $P_i$. This further implies that $P_{\D}$ will eventually find a set $\C_q$ where $S_q \setminus \C_q \in \AdvStructure$
      and where $\C_q$ constitutes a clique in the consistency graph of every honest party. 
      This is because the set $S_q \setminus Z^{\star}$ is guaranteed to eventually constitute a clique.
            Hence, $P_{\D}$ eventually broadcasts the sets $\{ \C_q\}_{S_q \in \ShareSpec}$, which are eventually delivered to every honest
      party. Moreover, the verification of these sets will eventually be successful for every honest party.
      
      Next consider an arbitrary $S_q$ and an arbitrary
            {\it honest} $P_i \in S_q$. If $P_i \in \C_q$, then it has already received the share $s_{qi}$ from $P_{\D}$
      and $s_{qi} = s_q$ holds. Hence, $P_i$ sets $[s]_q$ to $s_q$. So consider the case when $P_i \not \in \C_q$. In this case, 
      $P_i$ waits to find some $P_j \in \C_q$ such that $P_i$ accepts the signature
		$\ICSig(\sid^{(P_\D,q)}_{k,j,i},P_k, P_j, P_i, s_{qj})$ from intermediary $P_j$, corresponding to every signer $P_k \in \C_q$
		and upon finding such a $P_j$, party $P_i$ sets $[s]_q$ to $s_{qj}$. We show that except with probability at most $n \cdot \errorAICP$,
		party $P_i$ will eventually find a candidate $P_j$ satisfying the above condition. Moreover,
		if $P_i$ finds a candidate $P_j$ satisfying the above condition, then except with probability at most $n \cdot \errorAICP$,
		the condition $s_{qj} = s_q$ holds. As $P_i$ can have up to $\Order(n)$ candidates for $P_j$, 
		it will follow from the union bound that except with probability at most $n^2 \cdot \errorAICP$, 
		party $P_i$ will eventually compute $[s]_q$. 
		 Now assuming these statements are true, the proof follows from the union bound
		and the fact that $S_q$ can be any set out of $|\AdvStructure|$ subsets in $\ShareSpec$ and for any $S_q$, there could be upto
		$\Order(n)$ honest parties $P_i$ in $S_q \setminus C_q$. We next proceed to prove the above two statements.
		
   Since $\ShareSpec$ satisfies the $\Q^{(2)}(\ShareSpec, \AdvStructure)$ condition and $S_q \setminus \C_q  \in \AdvStructure$, it follows that  		
  $\AdvStructure$ satisfies the $\Q^{(1)}(\C_q,\AdvStructure)$ condition and hence $\C_q$  contains at least one {\it honest} party, say $P_{h}$. 
  Consider any arbitrary $P_k \in \C_q$. From the protocol steps, $P_h$ has broadcasted the $\OK_q(h, k)$ after receiving 
  $\ICSig(\sid^{(P_\D,q)}_{k,h,i},P_k, P_h, P_i, s_{qk})$ from $P_k$ during $\Auth$ and verifying that $s_{qk} = s_{qh}$ holds, where $s_{qh} = s_q$. 
  It then follows from Lemma \ref{lemma:AICP}, that except with probability at most $\errorAICP$, party
  $P_i$ will accept the signature  $\ICSig(\sid^{(P_\D,q)}_{k,h,i},P_k, P_h, P_i, s_{qh})$ revealed by $P_h$.
  Hence, except with probability at most $n \cdot \errorAICP$, party $P_i$ will eventually accept the signature 
   $\ICSig(\sid^{(P_\D,q)}_{k,h,i},P_k, P_h, P_i, s_{qh})$ corresponding to {\it all} $P_k \in \C_q$, revealed by $P_h$. 
   
   Finally, consider an {\it arbitrary} $P_j \in \C_q$, such that $P_i$ has accepted the signature $\ICSig(\sid^{(P_\D,q)}_{k,j,i},P_k, P_j, P_i, s_{qj})$
    corresponding to {\it all} $P_k \in \C_q$ and sets $[s]_q$ to $s_{qj}$. Now one of these signatures corresponds to the signer $P_k = P_h$.
    If $P_j$ is {\it corrupt}, then it follows from Lemma \ref{lemma:AICP}, that except with probability at most $\errorAICP$,
    the condition $s_{qj} = s_{qh}$ holds. As there can be up to $\Order(n)$ honest parties $P_h$ in $\C_q$,
    it follows that $P_j$ will fail to reveal signature of any honest party from $\C_q$ on any $s_{qj} \neq s_q$, except with probability at most
    $n \cdot \errorAICP$. Since there can be up to $\Order(n)$ corrupt parties $P_j \in \C_q$, it then follows from the union bound that
    except with error probability $n^2 \cdot \errorAICP$, no corrupt party from $\C_q$ will be able to forge the signature of any honest party from $\C_q$
    on an incorrect $s_q$. 
\end{proof}

We next prove certain claims with respect to a {\it corrupt} dealer. The first claim is that the view of $\Adv$ in this
case is also identically distributed in both the real as well as simulated execution. This is simply
because in this case, the {\it honest} parties have {\it no} inputs
  and the simulator simply plays the role of the honest parties, {\it exactly}
 as per the steps of the protocol $\PiStatVSS$ in the simulated execution.

\begin{claim}
\label{claim:SVSSCorruptDPrivacy}
If $P_{\D}$ is corrupt, then 
the view of $\Adv$ in the simulated execution of $\PiStatVSS$ with $\SimPVSS$ is identically distributed to the view of $\Adv$ in the real execution of
$\PiStatVSS$ involving honest parties.
\end{claim}
\begin{proof}
The proof follows from the fact that if $P_{\D}$ is {\it corrupt}, then $\SimPVSS$ participates in a full execution of the protocol
 $\PiStatVSS$ by playing the role of the honest parties as per the steps of $\PiStatVSS$.
 Hence, there is a one-to-one correspondence between simulated
  executions and real executions.
\end{proof}

We finally claim that if the dealer is {\it corrupt}, then conditioned on the view of the adversary (which is identical in both
the executions as per the last claim), the outputs of the honest parties are identically distributed in 
 both the executions.
\begin{claim}
\label{claim:SVSSCorruptDCorrectness}
If $\D$ is corrupt, then conditioned on the view of $\Adv$, the output of the honest parties during the execution of $\PiStatVSS$ involving $\Adv$ has the  same distribution as the output of the honest parties in the ideal-world involving $\SimPVSS$ and $\FVSS$, except with probability
   at most  $|\AdvStructure| \cdot n^3 \cdot \errorAICP$, where $\errorAICP \approx \frac{n^2}{|\F|}$.
\end{claim}
\begin{proof}
Let $P_{\D}$ be {\it corrupt} and let 
 $\View$ be an arbitrary view of $\Adv$. We note that it can be found out from $\View$ whether valid
   core sets $\{\C_q\}_{S_q \in \ShareSpec}$ have been generated during the corresponding execution of $\PiStatVSS$ or not. We now consider the following cases.
\begin{myitemize}
\item[--] {\it No core sets $\{\C_q\}_{S_q \in \ShareSpec}$ are generated as per $\View$}: In this case, the honest parties do not obtain any output in either execution.
 This is because in the real execution of $\PiStatVSS$, the honest parties compute their output only when they get valid core sets $\{\C_q\}_{S_q \in \ShareSpec}$
 from $P_{\D}$'s broadcast. If this is not the case, then in the simulated execution, 
  the simulator $\SimPVSS$ does not provide any input to $\FVSS$ on behalf of $P_\D$; hence, $\FVSS$ does not produce any output for the honest parties.
\item[--] {\it Core sets $\{\C_q\}_{S_q \in \ShareSpec}$ generated as per $\View$ are invalid}:
 Again, in this case, the honest parties do not obtain any output in either execution.
  This is because in the real execution of $\PiStatVSS$, even if the sets 
   $\{\C_q\}_{S_q \in \ShareSpec}$ are received from $P_{\D}$'s broadcast,
    the honest parties compute their output only when each $\C_q$ set is found to be {\it valid} with respect to the verifications
    performed by the honest parties in their own consistency graphs.
    If these verifications fail (implying that the core sets are invalid), then in the simulated execution, 
    the simulator $\SimPVSS$ does not provide any input to $\FVSS$ on behalf of $P_\D$, implying that
     $\FVSS$ does not produce any output for the honest parties.
\item[--] {\it Valid core sets $\{\C_q\}_{S_q \in \ShareSpec}$ are generated as per $\View$}:
  We first note that in this case, $P_{\D}$ has distributed some common share, say $s_q$, as determined by $\View$, 
  to all the parties in $\C_q \setminus Z^{\star}$, during the real execution of $\PiStatVSS$. 
  This is because all the parties in $\C_q \setminus Z^{\star}$ are {\it honest}, and form a clique
  in the consistency graph of the honest parties. Hence, 
  each $P_j, P_k \in \C_q \setminus Z^{\star}$ has broadcasted the messages $\OK_q(j, k)$ and $\OK_q(k, j)$
  after checking that $s_{qj} = s_{qk}$ holds, where $s_{qj}$ and $s_{qk}$ are the values 
  received from $P_{\D}$ by $P_j$ and $P_k$ respectively. 
  
  We next show that in the real execution of $\PiStatVSS$, except with probability at most $n^3 \cdot \errorAICP$, all
  {\it honest} parties in $S_q \setminus Z^{\star}$ eventually set
  $[s]_q$ to $s_q$. While this is obviously true for the parties in $\C_q \setminus Z^{\star}$, 
  the proof when 
  $P_i \in S_q \setminus (Z^{\star} \cup \C_q)$ is exactly the {\it same}, as in Claim \ref{claim:SVSSHonestDCorrectness}. 

     Since $|\ShareSpec| = |\AdvStructure|$, it then follows that 
     in the real execution, except with probability at most $n^3 \cdot \errorAICP$, every honest party $P_i$
      eventually outputs $\{[s]_q = s_q \}_{P_i \in S_q}$.
      From the steps of $\SimPVSS$, the simulator sends the shares $\{s_q \}_{S_q \in \ShareSpec}$ to $\FVSS$ on the behalf of
      $P_{\D}$ in the simulated execution. Consequently, in the simulated execution, $\FVSS$ will eventually deliver
      the shares $\{[s]_q = s_q \}_{P_i \in S_q}$ to every honest $\mathsf{I}$.
      Hence, except with probability at most $|\AdvStructure| \cdot n^3 \cdot \errorAICP$, the outputs of the honest parties are identical in both the executions.
\end{myitemize}
\end{proof}
The proof of the theorem now follows
  from Claims \ref{claim:SVSSHonestDPrivacy}-\ref{claim:SVSSCorruptDCorrectness}.
\end{proof}



\subsection{The Basic Multiplication Protocol $\BasicMult$ and Its Properties}
Protocol $\BasicMult$ 
 is presented in Figure \ref{fig:BasicMult}, which is executed with respect to a set $\Discarded$ of globally discarded parties, and an iteration number $\iter$. 
  Looking ahead, it will be guaranteed that no honest party is ever included in $\Discarded$. The protocol is almost the same as
   the protocol $\OptMult$, except that it {\it does not} take any subset $Z \in \AdvStructure$ as input. 
    Consequently, the various dynamic sets and session ids maintained in the protocol {\it will} not be notated with $Z$ (unlike the protocol $\OptMult$).

\begin{protocolsplitbox}{$\BasicMult(\PartySet, \AdvStruct, \SharingSpec, [a], [b], \Discarded, \iter)$}{Non-robust basic multiplication protocol in the $(\FVSS, \FABA)$-hybrid model
  for session id $\sid$. The above code is executed by every party $P_i$}{fig:BasicMult}
\justify
\begin{myitemize}
\item[--] \textbf{Initialization}: 
  Initialize $\Products_\iter = \{(p, q)\}_{p, q = 1, \ldots, |\SharingSpec|}$, $\Selected_\iter = \emptyset$, 
   $\hop = 1$ and corresponding to each $P_j \in \PartySet \setminus \Discarded$, set
	$\Products^{(j)}_\iter = \{(p, q)\}_{P_j \in S_p \cap S_q}$.
	\item[--] Do the following till $\Products_\iter \neq \emptyset$:
	\begin{myitemize}
	\item \textbf{Sharing Summands}: Same as in $\OptMult$, except that 
	$P_i$ randomly secret-shares $\displaystyle c^{(i)}_\iter = \sum_{(p, q) \in \Products^{(i)}_\iter} [a]_p[b]_q$
	 by calling $\FVSS$ with 
	     $\sid_{\hop, i} \defined \sid || \hop || i $, if $P_i \notin \Selected_\iter$.
	\item \textbf{Selecting Summand-Sharing Party Through ACS}:
	 Same as in $\OptMult$, except that $(\vote, \allowbreak \sid_{\hop, j}, 1)$ is sent to $\FABA$
	 with $\sid_{\hop, j}$ corresponding to any $P_j \in \PartySet$, if all the following hold:
	       \begin{myitemize}
		\item[--] $P_j \notin \Discarded$, 
		$P_j \notin \Selected_\iter$ and an 
		output $(\Share, \sid_{\hop, j}, P_j, \{ [c^{(j)}_\iter]_q \}_{P_i \in S_q})$ is received from $\FVSS$ with 
		$\sid_{\hop, j}$,
		 corresponding to the dealer $P_j$.
		\end{myitemize}
	       If $P_j$ is selected as common summand-sharing party for this hop, then update the following.
                \begin{myitemize}
		\item[--] $\Selected_\iter = \Selected_\iter \cup \{P_j\}$.
		\item[--] $\Products_\iter = \Products_\iter \setminus \Products^{(j)}_\iter$.
		\item[--] $\forall P_k \in \PartySet \setminus \{\Discarded \cup \Selected_\iter \}$:
		$\Products^{(k)}_\iter = \Products^{(k)}_\iter \setminus \Products^{(j)}_\iter$.
		\item[--] $\hop = \hop + 1$.
		\end{myitemize}
	\end{myitemize}
\item[--] $\forall P_j \in \PartySet \setminus \Selected_\iter$, participate in an instance of $\PiPerDefaultShare$ with public input 
  $c^{(j)} = 0$.
\item[--] \textbf{Output} : Let $c_\iter \defined c^{(1)}_\iter + \ldots + c^{(n)}_\iter$. Output
$\{[c^{(1)}_\iter]_q , \ldots, [c^{(n)}_\iter]_q, [c_\iter]_q \}_{P_i \in S_q}$.
\end{myitemize}
\end{protocolsplitbox}

We next formally prove the properties of the protocol $\BasicMult$. 
 While proving these properties, we will assume that  $\AdvStructure$ satisfies the $\Q^{(3)}(\PartySet, \AdvStructure)$ condition.
  This further implies that the sharing specification $\ShareSpec = (S_1, \ldots, \allowbreak S_h) \defined \{ \PartySet \setminus Z | Z \in \AdvStructure\}$
  satisfies the $\Q^{(2)}(\ShareSpec, \AdvStructure)$ condition. 
  Moreover, while proving these properties, we assume that no honest party is ever included in the set $\Discarded$. 
   Note that this will be ensured in the protocol $\RandMultLCE$, where $\BasicMult$ is used as a subprotocol. We first show that the intersection of any two sets in $\ShareSpec$ contains at least one honest party {\it outside} $\Discarded$.
  \begin{claim}
  \label{claim:BasicMultAdvCondition}
 For every $Z \in \AdvStructure$ and every ordered pair $(p, q) \in \{1, \ldots, h \} \times \{1, \ldots, h \}$,
  the set $(S_p \cap S_q) \setminus \Discarded$ contains at least one honest party.
  \end{claim}
  \begin{proof}
  From the definition of the sharing specification $\ShareSpec$, we have
   $S_p = \PartySet \setminus Z_p$ and $S_q = \PartySet \setminus Z_q$, where $Z_p, Z_q \in \AdvStructure$. Let $Z^{\star} \in \AdvStructure$ be the set of corrupt parties during the protocol $\BasicMult$. Now, $S_p \cap S_q = (\PartySet \setminus Z_p) \cap (\PartySet \setminus Z_q) = \PartySet \setminus (Z_p \cup Z_q)$. This means that $(S_p \cap S_q) \cup Z_p \cup Z_q = \PartySet$. If $(S_p \cap S_q) \subseteq Z^{\star}$, then $Z^{\star} \cup Z_p \cup Z_q = \PartySet$. This is a violation of the $\Q^{(3)}(\PartySet,\AdvStructure)$ condition, and hence, $S_p \cap S_q$ contains at least one honest party. Since $\Discarded$ contains only corrupt parties, $(S_p \cap S_q) \setminus \Discarded$ contains at least one honest party.
  \end{proof}
  
  We next claim a series of properties related to protocol $\BasicMult$ whose proofs are almost identical to the proof of the corresponding properties
   for protocol $\OptMult$. Hence, we skip the formal proofs.
\begin{claim}
  \label{claim:BasicMultACS}
 For any $\iter$, if all honest parties participate during the hop number $\hop$ in the protocol 
  $\BasicMult(\PartySet, \allowbreak \AdvStruct,\SharingSpec, [a],[b],\iter)$, then
 all honest parties eventually obtain a common summand-sharing party, say $P_j$, for this hop, such that the honest parties
  will eventually hold $[c^{(j)}_\iter]$. Moreover, party $P_j$ will be distinct from the summand-sharing party
   selected for any hop number $\hop' < \hop$.
  \end{claim}
  \begin{proof}
  The proof is identical to that of Claim \ref{claim:OptMultACS}, except that we now use Claim \ref{claim:BasicMultAdvCondition} to argue that
   for every ordered pair $(p, q) \in \Products_\iter$, there exists at least one {\it honest} party in $(S_p \cap S_q) \setminus \Discarded$,
    say $P_k$, who will have both the shares 
  $[a]_p$ as well as $[b_q]$ (and hence the summand $[a]_p [b]_q$). 
  \end{proof}
\begin{claim}
  \label{claim:BasicMultComplexity}
  In protocol $\BasicMult$, all honest parties eventually obtain an output.
   The protocol makes $\Order(n^2)$ calls to $\FVSS$ and $\FABA$.
  \end{claim}
  \begin{proof}
 The proof is similar to that of Claim \ref{claim:OptMultComplexity}.
  \end{proof}
\begin{claim}
  \label{claim:BasicMultPrivacy}
  During protocol $\BasicMult$, $\Adv$ learns nothing about $a$ and $b$.
  \end{claim}
  \begin{proof}
  The proof is similar to that of Claim \ref{claim:OptMultPrivacy}.
  \end{proof}
  \begin{claim}
  \label{claim:BasicMultCorrectness}
  In $\BasicMult$, if no party in $\PartySet \setminus \Discarded$ behaves maliciously, then for each $P_i \in \Selected_\iter$, the condition
   $\displaystyle c^{(i)}= \sum_{(p, q) \in \Products^{(i)}_\iter} [a]_p[b]_q$ holds, which further implies that $c = ab$ holds. 
  \end{claim}
  \begin{proof}
  The proof is similar to that of Claim \ref{claim:OptMultCorrectness}.
  \end{proof}
  Lemma \ref{lemma:BasicMult} now follows from Claims \ref{claim:BasicMultAdvCondition}-\ref{claim:BasicMultCorrectness}.
\begin{lemma}
\label{lemma:BasicMult}
Let  $\AdvStructure$ satisfy the $\Q^{(3)}(\PartySet, \AdvStructure)$ condition and let 
 $\ShareSpec = (S_1, \ldots, S_h) = \{ \PartySet \setminus Z | Z \in \AdvStructure\}$. 
 Consider an arbitrary $\iter$, 
  such that all honest parties participate in the instance $\BasicMult(\PartySet,\AdvStruct,\SharingSpec, \allowbreak [a],[b], \Discarded, \iter)$.
  Then all
  honest parties eventually compute  $[c_{\iter}]$ and $[c^{(1)}_{\iter}], \ldots, [c^{(n)}_{\iter}]$
  where $c_{\iter} = c^{(1)}_{\iter} + \ldots + c^{(n)}_{\iter}$, 
  provided no honest party is ever included in the $\Discarded$. 
  If no party in $\PartySet \setminus \Discarded$ behaves maliciously, then 
  $c_{\iter} = ab$ holds. 
  In the protocol, $\Adv$ does not learn any additional information about $a$ and $b$.
  The protocol makes $\Order(n^2)$ calls to $\FVSS$ and $\FABA$.
\end{lemma}
  We claim another property of $\BasicMult$, which will be useful later while analyzing the properties of 
   $\RandMultLCE$, where $\BasicMult$ is used as a sub-protocol.
 \begin{claim}
 \label{claim:BasicMultFuture}
 For any $\iter$, if $P_j \in \Selected_{\iter}$ during the instance
  $\BasicMult(\PartySet,\AdvStruct,\SharingSpec,  [a],[b], \Discarded, \iter)$, 
  then $P_j \not \in \Discarded$.
 \end{claim}  
 \begin{proof}
 The proof is similar to that of Claim \ref{claim:OptMultFuture}.
 \end{proof}
  
  We finally end this section by discussing the modifications to the protocol $\BasicMult$ for handling $M$ pairs of inputs.
  \paragraph{Protocol $\BasicMult$ for $M$ pairs of inputs:} Protocol $\BasicMult$ can be easily modified if
 executed with input $\{([a^{(\ell)}], [b^{(\ell)}] ) \}_{\ell = 1, \ldots, M}$. The modifications will be along similar lines to those done for
  $\OptMult$. 
  Consequently, there will be $\Order(n^2 M)$ calls to $\FVSS$, but {\it only} $\Order(n^2)$ calls to $\FABA$.
\subsection{Protocol $\RandMultLCE$ for Detectable Random-Triple Generation and Its Properties}
Protocol $\RandMultLCE$ for one triple is formally presented in Fig \ref{fig:StatMultCI}. 
\begin{protocolsplitbox}{$\RandMultLCE(\PartySet, \AdvStruct, \SharingSpec, \Discarded, \iter)$}{Detectable triple generation protocol in the $(\FVSS, \FABA)$-hybrid model}{fig:StatMultCI}
\justify
\begin{myitemize}
 \item[--] \textbf{Generating Secret-Sharing of Random Values}: 
   The parties jointly generate $[a_{\iter}], [b_{\iter}], [b'_{\iter}]$ and $[r_{\iter}]$, where
   $a_{\iter}, b_{\iter}, b'_{\iter}$ and $r_{\iter}$ are random from the view-point of $\Adv$,
   by using a similar procedure as in $\PiPerTriples$. For this, each $P_i \in \PartySet$ acts as a dealer, picks
   random $a^{(i)}_{\iter}, b^{(i)}_{\iter}, b'^{(i)}_{\iter}, r^{(i)}_{\iter}$ from $\F$ and generates
   random $[a^{(i)}_{\iter}], [b^{(i)}_{\iter}], [b'^{(i)}_{\iter}]$ and $[r^{(i)}_{\iter}]$, by making calls to $\FVSS$.
   The parties then agree on a common subset of parties $\CoreSet$ through ACS as in $\PiPerTriples$, such that 
   $\PartySet \setminus \CoreSet \in \AdvStructure$ and
      for each $P_j \in \CoreSet$, the honest parties eventually hold
   $[a^{(j)}_{\iter}], [b^{(j)}_{\iter}], [b'^{(j)}_{\iter}]$ and $[r^{(j)}_{\iter}]$. The parties then set
   \[ [a_\iter] \defined \displaystyle \sum_{P_j \in \CoreSet} [a^{(j)}_\iter], \quad  [b_\iter] \defined \displaystyle \sum_{P_j \in \CoreSet} [b^{(j)}_\iter], \quad 
    [b'_\iter] \defined \displaystyle \sum_{P_j \in \CoreSet} [b'^{(j)}_\iter] \quad \mbox{ and } 
   [r_\iter] \defined \displaystyle \sum_{P_j \in \CoreSet} [r^{(j)}_\iter].\] 
 \item[--] \textbf{Running Multiplication Protocol and Reconstructing the Random Challenge}:
       \begin{myitemize}
       \item The parties participate in instances 
        $\BasicMult(\PartySet, \AdvStruct, \SharingSpec, [a_\iter], [b_\iter], \Discarded, \iter)$
         and $\BasicMult(\PartySet, \AdvStruct, \SharingSpec,  [a_\iter], [b'_\iter], \allowbreak  \Discarded, \iter)$
         to get outputs $\{[c^{(1)}_\iter] , \ldots, [c^{(n)}_\iter], [c_\iter] \}$ and 
         $\{[c'^{(1)}_\iter], \ldots, [c'^{(n)}_\iter], [c'_\iter] \}$ respectively. 
         Let $\Selected_{\iter, c}$ and $\Selected_{\iter, c'}$ be the summand-sharing parties for the two instances respectively. 
         Moreover, for $P_j \in \Selected_{\iter, c}$, let $\Products^{(j)}_{\iter, c}$
          be the set of ordered pairs of indices corresponding to the summands whose sum has been shared by $P_j$ 
          during the instance $\BasicMult(\PartySet, \AdvStruct, \SharingSpec, [a_\iter], [b_\iter], \Discarded, \iter)$.
          Similarly, for $P_j \in \Selected_{\iter, c'}$, let $\Products^{(j)}_{\iter, c'}$
          be the set of ordered pairs of indices corresponding to the summands whose sum has been shared by $P_j$ 
          during the instance $\BasicMult(\PartySet, \AdvStruct, \SharingSpec, [a_\iter], [b'_\iter], \Discarded, \iter)$.         
       \item Once the parties obtain their respective outputs from the instances of $\BasicMult$, they 
       participate in an instance of $\PiPerRec$ with shares corresponding to $[r_\iter]$, to reconstruct $r_\iter$.
       \end{myitemize}
\item[--] {\bf Detecting Errors in Instances of $\BasicMult$:}
	\begin{myitemize}
	\item The parties locally compute $[e_{\iter}] \defined r_\iter[b_\iter] + [b'_\iter]$ and then 
	participate in an instance of $\PiPerRec$ with shares corresponding to $[e_\iter]$, to reconstruct $e_\iter$.		
	\item The parties locally compute $[d_{\iter}] \defined e_\iter [a_\iter] - r_\iter [c_\iter] - [c'_\iter]$ 
	and then 
	participate in an instance of $\PiPerRec$ with shares corresponding to $[d_\iter]$, to reconstruct $d_\iter$.	
	\item {\bf Output Computation in Case of Success:} 
	If $d_\iter = 0$, then every party $P_i \in \PartySet$ sets the Boolean variable 
	$\flag^{(i)}_\iter=0$ and outputs $\{([a_\iter]_q, [b_\iter]_q, [c_\iter]_q)\}_{P_i \in S_q}$.
	\item[--] {\bf Cheater Identification in Case of Failure:} If $d_\iter \neq 0$, then every party $P_i \in \PartySet$ sets the Boolean variable 
	$\flag^{(i)}_\iter = 1$ and proceeds as follows.
	\begin{myitemize}
	\item Participate in appropriate instances of $\PiPerRecShare$ to reconstruct the shares
	$\{ [a_\iter]_q, [b_\iter]_q, \allowbreak [b'_\iter]_q\}_{S_q \in \ShareSpec}$ and appropriate instances of $\PiPerRec$ 
		to reconstruct $c^{(1)}_\iter, \ldots, c^{(n)}_\iter, \allowbreak c'^{(1)}_\iter, \ldots, c'^{(n)}_\iter$.
	\item Set $\displaystyle \Discarded = \Discarded \cup \{P_i\}$, if $P_i \in \Selected_{\iter, c} \cup \Selected_{\iter, c'}$ 
	and the following holds for $P_i$:
	\[ \displaystyle  r_{\iter} \cdot c^{(i)}_\iter + c'^{(i)}_\iter \neq r_{\iter} \cdot \sum_{(p, q) \in \Products^{(i)}_{\iter, c}} [a_\iter]_p[b_\iter]_q 
	 + \sum_{(p, q) \in \Products^{(i)}_{\iter, c'}}[a_\iter]_p [b'_\iter]_q .\]
	\end{myitemize}  
	\end{myitemize}
\end{myitemize}
\end{protocolsplitbox}

We now formally prove the properties of the protocol $\RandMultLCE$. 
 While proving these properties, we will assume that  $\AdvStructure$ satisfies the $\Q^{(3)}(\PartySet, \AdvStructure)$ condition.
  This further implies that $\ShareSpec = (S_1, \ldots, S_h) \defined \{ \PartySet \setminus Z | Z \in \AdvStructure\}$
  satisfies the $\Q^{(2)}(\ShareSpec, \AdvStructure)$ condition.
  
   We first claim that the honest parties eventually compute
    $[a_\iter], [b_\iter], [b'_\iter]$ and $[r_\iter]$
\begin{claim}
\label{claim:RandMultLCERandomACS}
Consider an arbitrary $\iter$, 
  such that all honest parties participate in the instance $\RandMultLCE(\PartySet, \AdvStruct, \SharingSpec, \allowbreak \Discarded, \iter)$, where
   no honest party is present in $\Discarded$. Then the honest parties eventually compute 
   $[a_\iter], [b_\iter], [b'_\iter]$ and $[r_\iter]$.
\end{claim}
\begin{proof}
The proof is similar to the proof of Claim \ref{claim:PerTriplesTermination}. 
\end{proof}
We next claim that all honest parties will eventually agree on whether the instances of $\BasicMult$ in $\RandMultLCE$ has succeeded or failed.
\begin{claim}
\label{claim:RandMultLCETermination}
Consider an arbitrary $\iter$, 
  such that all honest parties participate in the instance $\RandMultLCE(\PartySet, \AdvStruct, \SharingSpec, \allowbreak \Discarded, \iter)$,
   where
   no honest party is present in $\Discarded$.
  Then all honest parties eventually reconstruct a (common) value $d_\iter$.
  Consequently, each honest $P_i$ eventually sets $\flag^{(i)}_\iter$ to either $0$ or $1$. 
\end{claim}
\begin{proof}
From Claim \ref{claim:RandMultLCERandomACS}, the honest parties eventually hold $[a_\iter], [b_\iter], [b'_\iter]$ and $[r_\iter]$.
 From Lemma \ref{lemma:BasicMult}, it follows that the honest parties eventually hold the outputs
 $\{[c^{(1)}_\iter] , \ldots, [c^{(n)}_\iter], [c_\iter] \}$ and 
         $\{[c'^{(1)}_\iter], \ldots, [c'^{(n)}_\iter], [c'_\iter] \}$ from the 
          corresponding instances of $\BasicMult$.  From Lemma \ref{lemma:PerRec},
          the honest parties eventually reconstruct $r_\iter$ from the corresponding instance of $\PiPerRec$. 
           From the linearity property of secret-sharing, it then follows that the honest parties
          eventually hold $[e_\iter]$ and hence eventually reconstruct $e_{\iter}$ 
          from the corresponding instance of $\PiPerRec$. Again, from the linearity property of secret-sharing, it follows that
          the honest parties eventually hold $[d_\iter]$, followed by eventually reconstructing $d_\iter$ from the corresponding instance of $\PiPerRec$. 
          Now based on whether $d_\iter$ is $0$ or not, each {\it honest} $P_i$ eventually sets $\flag^{(i)}_\iter$ to either $0$ or $1$. 
\end{proof}

We next claim that if no party in $\PartySet \setminus \Discarded$ behaves maliciously, then the honest parties eventually hold a secret-shared
 multiplication-triple. 
 
\begin{claim}
\label{claim:RandMultLCEHonestBehaviour}
Consider an arbitrary $\iter$, 
  such that all honest parties participate in the instance $\RandMultLCE(\PartySet, \AdvStruct, \SharingSpec, \allowbreak \Discarded, \iter)$,
   where
   no honest party is present in $\Discarded$.
   If no party in $\PartySet \setminus \Discarded$ behaves maliciously, then $d_\iter = 0$
   and 
    the honest parties eventually hold
   $([a_\iter], [b_\iter], [c_\iter])$, where $c_\iter = a_\iter \cdot b_\iter$ holds.
\end{claim}   
\begin{proof}
If no party in $\PartySet \setminus \Discarded$ behaves maliciously, then from Lemma \ref{lemma:BasicMult}, the honest parties eventually compute
 $[c_\iter]$ and $[c'_\iter]$ from the respective instances of $\BasicMult$, such that $c_\iter = a_\iter \cdot b_\iter$ and 
 $c'_\iter = a_\iter \cdot b'_\iter$ holds. From Claim \ref{claim:RandMultLCETermination}, the honest parties will eventually reconstruct $d_\iter$.
  Moreover, since $c_\iter = a_\iter \cdot b_\iter$ and 
   $c'_\iter = a_\iter \cdot b'_\iter$ holds, the value $d_\iter$ will be $0$ and consequently, the honest parties will output 
  $([a_\iter], [b_\iter], [c_\iter])$.
\end{proof} 
 We next show that if $d_\iter \neq 0$, then the honest parties eventually include at least one new maliciously-corrupt party in the set $\Discarded$.
\begin{claim}
\label{claim:RandMultLCECorruptBehaviour}
Consider an arbitrary $\iter$, 
  such that all honest parties participate in the instance $\RandMultLCE(\PartySet, \AdvStruct, \SharingSpec, \allowbreak \Discarded, \iter)$,
   where
   no honest party is present in $\Discarded$. 
   If $d_\iter \neq 0$, then the honest parties eventually update $\Discarded$ by adding a new maliciously-corrupt
   party in $\Discarded$.
   \end{claim}
   \begin{proof}
   Let $d_\iter \neq 0$ and let $\Selected_\iter$ be the set of summand-sharing parties across the two instances of $\BasicMult$ executed in $\RandMultLCE$; 
   i.e.~ $\Selected_\iter \defined \Selected_{\iter, c} \cup \Selected_{\iter, c'}$. Note that
   there exists no $P_j \in \Selected_\iter$ such that $P_j \in \Discarded$, which follows from Claim \ref{claim:BasicMultFuture}.
   We claim that there exists at least one party $P_j \in \Selected_\iter$, such that corresponding to
    $c^{(j)}_\iter$ and $c'^{(j)}_\iter$, the following holds:
   \[r_\iter \cdot c^{(j)}_\iter + c'^{(j)}_\iter \neq r_\iter \cdot \sum_{(p,q) \in \Products^{(j)}_{\iter, c}} [a_\iter]_p [b_\iter]_q + 
    \sum_{(p, q) \in \Products^{(j)}_{\iter, c'}} [a_\iter]_p [b'_\iter]_q. \]
   Assuming the above holds, the proof now follows from the fact that once the parties reconstruct 
   $d_\iter \neq 0$, they proceed to reconstruct the shares
   $\{ [a_\iter]_q, [b_\iter]_q,  [b'_\iter]_q\}_{S_q \in \ShareSpec}$ through appropriate instances of $\PiPerRecShare$ 
   and the values $c^{(1)}_\iter, \ldots, c^{(n)}_\iter,  c'^{(1)}_\iter, \ldots, c'^{(n)}_\iter$ through 
   appropriate instances of $\PiPerRec$. Upon reconstructing these values, party $P_j$ will be eventually included in the set
   $\Discarded$. Moreover, it is easy to see that $P_j$ is a maliciously-corrupt party, since for every {\it honest}
   $P_j \in \Selected_\iter$, the condition 
   $c^{(j)}_\iter = \displaystyle \sum_{(p,q) \in \Products^{(j)}_{\iter, c}}[a_\iter]_p [b_\iter]_q$
    and $c'^{(j)}_\iter = \sum_{(p,q) \in \Products^{(j)}_{\iter,c'}}[a_\iter]_p [b'_\iter]_q$ holds. 
    
    We prove the above claim through a contradiction. So let the following condition hold for {\it each} $P_j \in \Selected_\iter$:
    \[r_\iter \cdot c^{(j)}_\iter + c'^{(j)}_\iter = r_\iter \cdot \sum_{(p,q) \in \Products^{(j)}_{\iter, c}} [a_\iter]_p [b_\iter]_q + 
    \sum_{(p, q) \in \Products^{(j)}_{\iter, c'}} [a_\iter]_p [b'_\iter]_q. \]
   Next, summing the above equation over all $P_j \in \Selected_\iter$, we get that the following holds:
     \[ \sum_{P_j \in  \Selected_\iter} r_\iter \cdot c^{(j)}_\iter + c'^{(j)}_\iter = \sum_{P_j \in  \Selected_\iter} r_\iter \cdot \sum_{(p,q) \in \Products^{(j)}_{\iter, c}} [a_\iter]_p [b_\iter]_q + 
    \sum_{(p, q) \in \Products^{(j)}_{\iter, c'}} [a_\iter]_p [b'_\iter]_q. \]
   This implies that the following holds:
     \[ r_\iter \cdot  \sum_{P_j \in  \Selected_\iter} c^{(j)}_\iter + c'^{(j)}_\iter =  r_\iter \cdot \sum_{P_j \in  \Selected_\iter}  \sum_{(p,q) \in \Products^{(j)}_{\iter, c}} [a_\iter]_p [b_\iter]_q + 
    \sum_{(p, q) \in \Products^{(j)}_{\iter, c'}} [a_\iter]_p [b'_\iter]_q. \]
    Now based on the way $a_{\iter}, b_{\iter}, b'_{\iter}, c_{\iter}$ and $c'_{\iter}$  are defined, the above implies that the following holds:
    \[r_\iter \cdot c_\iter + c'_\iter = r \cdot a_\iter \cdot b_\iter + a_\iter \cdot b'_\iter \]
   This further implies that 
   \[ r_\iter \cdot c_\iter + c'_\iter = (r_\iter \cdot b_\iter + b'_\iter) \cdot a_\iter \]
   Since in the protocol $e_\iter \defined r_\iter \cdot b_\iter + b'_\iter$, the above implies that
   \[r_\iter \cdot c_\iter + c'_\iter  = e_\iter \cdot a_\iter  \quad \Rightarrow  \; e_\iter \cdot a_\iter - r_\iter \cdot c_\iter - c'_\iter = 0 \quad \Rightarrow \; d_\iter = 0,   \]
   where the last equality follows from the fact that in the protocol, $d_\iter \defined e_\iter \cdot a_\iter - r_\iter \cdot c_\iter - c'_\iter$.
   However $d_\iter = 0$ is a contradiction, since according to the hypothesis of the claim, we are given that
    $d_\iter \neq 0$.
   \end{proof}

We next show that if the honest parties output a secret-shared triple in the protocol, then except with probability $\frac{1}{|\F|}$, the triple is a multiplication-triple.
 Moreover, the triple will be random for the adversary.
 
 \begin{claim}
\label{claim:RandMultLCECorrectness}
Consider an arbitrary $\iter$, 
  such that all honest parties participate in the instance $\RandMultLCE(\PartySet, \AdvStruct, \SharingSpec, \allowbreak \Discarded, \iter)$,
   where
   no honest party is present in $\Discarded$. 
   If $d_\iter = 0$, then the honest parties eventually output 
  $([a_\iter], [b_\iter], [c_\iter])$, where 
   except with probability  $\frac{1}{|\F|}$, the condition $c_\iter = a_\iter \cdot b_\iter$ holds. 
   Moreover, the view of $\Adv$ will be independent of $(a_\iter, b_\iter, c_\iter)$.
\end{claim}
\begin{proof}
Let $d_\iter = 0$. Then from the protocol steps, the honest parties eventually output $([a_\iter], [b_\iter], \allowbreak [c_\iter])$.
  In the protocol $d_{\iter} \defined e_\iter \cdot a_\iter - r_\iter \cdot c_\iter - c'_\iter$, where
  $e_\iter \defined r_\iter \cdot b_\iter + b'_\iter$. Since $d_\iter = 0$ holds, it implies that the honest parties have verified that the following holds:
   \[r_\iter (a_\iter \cdot b_\iter - c_\iter) = (c'_\iter - a_\iter \cdot b'_\iter). \]
  We also note that $r_\iter$ will be a random element from $\F$ and will be unknown to $\Adv$ till it is publicly reconstructed. This simply follows from the fact
  there will be at least one {\it honest} party $P_j$ in the set $\CoreSet$, such that the corresponding value $r^{(j)}_\iter$ shared by $P_j$ will be random from the
  view-point of $\Adv$. We also note that $r_\iter$ will be unknown to $\Adv$, till the outputs for the underlying instances
  of $\BasicMult$ are computed, and the honest parties hold $[c_\iter]$ and $[c'_\iter]$. This is because in the protocol, 
   the honest parties start participating in the instance of $\PiPerRec$ to reconstruct $r_\iter$, only after they obtain their 
   respective shares corresponding to $[c_\iter]$ and $[c'_\iter]$. 
  Now we have the following cases with respect to
 whether any party from $\PartySet \setminus \Discarded$ behaved maliciously during the underlying instances of $\BasicMult$.
 \begin{myitemize}
 \item[--] {\bf Case I: $c_\iter = a_\iter \cdot b_\iter$ and $c'_\iter = a_\iter \cdot b'_\iter$} --- In this case, $(a_\iter, b_\iter, c_\iter)$ is a multiplication-triple.
 \item[--] {\bf Case II: $c_\iter = a_\iter \cdot b_\iter$, but $c'_\iter \neq a_\iter \cdot b'_\iter$} --- This case is never possible, as this will lead to the contradiction
  that $r_\iter (a_\iter \cdot b_\iter - c_\iter) \neq (c'_\iter - a_\iter \cdot b'_\iter)$ holds.
 \item[--] {\bf Case III: $c_\iter \neq a_\iter \cdot b_\iter$, but $c'_\iter = a_\iter \cdot b'_\iter$} --- This case is possible only if $r_\iter = 0$, as otherwise 
   this will lead to the contradiction
  that $r_\iter (a_\iter \cdot b_\iter - c_\iter) \neq (c'_\iter - a_\iter \cdot b'_\iter)$ holds. However, since $r_\iter$ is a random element from $\F$, it implies that this case
  can occur only with probability at most $\frac{1}{|\F|}$.
  \item[--] {\bf Case IV: $c_\iter \neq a_\iter \cdot b_\iter$ as well as $c'_\iter \neq a_\iter \cdot b'_\iter$} --- This case is possible only if 
   $r_\iter = (c'_\iter - a_\iter \cdot b'_\iter) \cdot (a_\iter \cdot b_\iter - c_\iter)^{-1}$, as otherwise 
   this will lead to the contradiction
  that $r_\iter (a_\iter \cdot b_\iter - c_\iter) \neq (c'_\iter - a_\iter \cdot b'_\iter)$ holds. However, since $r_\iter$ is a random element from $\F$, it implies that this case
  can occur only with probability at most $\frac{1}{|\F|}$.
 \end{myitemize}
Hence, we have shown that except with probability at most $\frac{1}{|\F|}$, the triple $(a_\iter, b_\iter, c_\iter)$ is a multiplication-triple. 
 To complete the proof, we need to argue that the view of $\Adv$ in the protocol, will be independent of the triple 
 $(a_\iter, b_\iter, c_\iter)$. For this, we first note that $a_\iter, b_\iter$ and $b'_\iter$ will be random for the adversary. 
 The proof for this is similar to that of Claim \ref{claim:PerTriplesOutputDistribution} and follows from the fact that there will be at least
 one {\it honest} party $P_j$ in $\CoreSet$, such that the corresponding values $a^{(j)}_\iter, b^{(j)}_\iter$ and $b'^{(j)}_\iter$ shared by $P_j$ will be randomly distributed
 for $\Adv$. From Lemma \ref{lemma:BasicMult}, $\Adv$ learns nothing additional 
  about $a_\iter$, $b_\iter$ and $b'_\iter$ during the two instances of $\BasicMult$. 
   While $\Adv$ learns the value of $e_\iter$, since $b'_\iter$ is a uniformly distributed for $\Adv$, 
    for every candidate value of $b'_\iter$ from the view-point of $\Adv$,
    there is a corresponding value of $b_\iter$ consistent with the $e_\iter$ learnt by $\Adv$.
    Hence, learning $e_\iter$ does not add any new information about $(a_\iter, b_\iter, c_\iter)$ to the view of $\Adv$. 
    Moreover, $\Adv$ will be knowing beforehand that $d_\iter$ will be $0$ and hence, learning this value does not change
    the view of $\Adv$ regarding $(a_\iter, b_\iter, c_\iter)$.     
\end{proof}

We next derive the communication complexity of the protocol $\RandMultLCE$.
\begin{claim}
\label{claim:RandMultLCECommunication}
Protocol $\RandMultLCE$ requires $\Order(n^2)$ calls to $\FVSS$ and $\FABA$, and incurs a communication of $\Order(|\AdvStruct| \cdot n^3 \log{|\F|})$ bits.
\end{claim}
\begin{proof}
Follows from the communication complexity of the protocol $\BasicMult$ (Claim \ref{claim:BasicMultComplexity}) and the fact that if $d_\iter \neq 0$, then 
  the parties proceed to publicly reconstruct $\Order(n)$ values through instances of $\PiPerRec$ and publicly reconstruct
  $\Order(|\ShareSpec|)$ number of shares through instances of $\PiPerRecShare$, where
  $|\ShareSpec| = |\AdvStructure|$ for our sharing specification $\ShareSpec$.
\end{proof}

The proof of Lemma \ref{lemma:RandMultLCE} now follows from Claims \ref{claim:RandMultLCERandomACS}-\ref{claim:RandMultLCECommunication}. 
\begin{lemma}
\label{lemma:RandMultLCE}
Let  $\AdvStructure$ satisfy the $\Q^{(3)}(\PartySet, \AdvStructure)$ condition and let 
 $\ShareSpec = (S_1, \ldots, S_h) = \{ \PartySet \setminus Z | Z \in \AdvStructure\}$. 
  Consider an arbitrary $\iter$, 
  such that all honest parties participate in the instance $\RandMultLCE(\PartySet, \AdvStruct, \SharingSpec, \allowbreak \Discarded, \iter)$, where
   no honest party is present in $\Discarded$.  Then
   each honest $P_i$ eventually sets $\flag^{(i)}_\iter$ to either $0$ or $1$. 
  In the former case, the honest parties output 
  $([a_\iter], [b_\iter], [c_\iter])$, such that with probability at least
   of $1 - \frac{1}{|\F|}$, the condition $c_\iter = a_\iter \cdot b_\iter$ holds. Moreover, the view of 
   $\Adv$ will be independent of the triple $(a_\iter, b_\iter, c_\iter)$. 
   In the latter case, the honest parties will eventually include at least one new maliciously-corrupt party 
    $P_j$ to $\Discarded$.
    The protocol makes $\Order(n^2)$ calls to $\FVSS$ and $\FABA$, and incurs a communication of $\Order(|\AdvStruct| \cdot n^3 \log{|\F|})$ bits.
\end{lemma}
\paragraph{Protocol $\RandMultLCE$ for $M$ Triples:}
The extension of the protocol $\RandMultLCE$  for generating $M$ triples is straight forward. 
 The parties first generate $M$ random shared tuples
 $\{([a^{(\ell)}_{\iter}], [b^{(\ell)}_{\iter}], [b'^{(\ell)}_{\iter}])\}_{\ell = 1, \ldots, M}$ and a {\it single}
  random challenge $[r_{\iter}]$.
  The parties then run $2M$ instances of $\BasicMult$ to compute
  $\{([c^{(\ell)}_{\iter}], [c'^{(\ell)}_{\iter}])\}_{\ell = 1, \ldots, M}$, followed
   by probabilistically checking if all the instances of $\BasicMult$
  are executed correctly, by using the {\it same} $r_{\iter}$ for all the instances. 
  If cheating is detected in any of the instances, then the parties proceed further to identify at least one new maliciously-corrupt party
  and update $\Discarded$, as done in $\RandMultLCE$. 
  The protocol makes $\Order(n^2 \cdot M)$ calls to $\FVSS$ and $\Order(n^2)$ calls to $\FABA$, and incurs a communication of
   $\Order((M \cdot |\AdvStruct| \cdot n^2 + |\AdvStructure| \cdot n^3) \log{|\F|})$ bits.
\subsection{Statistically-Secure Protocol $\PiStatTriples$ and Its Properties}
Protocol $\PiStatTriples$ for generating $M = 1$ multiplication-triple is presented in Fig \ref{fig:StatTriples}.
\begin{protocolsplitbox}{$\PiStatTriples(\PartySet, \AdvStruct, \SharingSpec$)}{A statistically-secure protocol 
  for $\FTriples$ with $M = 1$ in $(\FVSS, \FABA)$-hybrid for session id $\sid$}{fig:StatTriples}
 \justify
 \begin{myitemize}
\item[--] \textbf{Initialization}: Parties initialize $\Discarded = \emptyset$
and $\iter = 1$.
	\item[--] {\bf Detectable Triple Generation}:
	 Parties participate in an instance $\RandMultLCE(\PartySet,\AdvStruct,\SharingSpec, \Discarded, \iter)$ with session id 
	 $\sid_\iter \defined \sid || \iter$. Each $P_i \in \PartySet$ then proceeds as follows.
	\begin{myitemize}
	\item {\bf Positive Output}: If $\flag^{(i)}_\iter$ is set to $0$ during the instance 
	$\RandMultLCE(\PartySet,\AdvStruct,\SharingSpec, \Discarded, \iter)$, then 
	 output the shares $\{([a_\iter]_q,[b_\iter]_q,[c_\iter]_q)\}_{P_i \in S_q}$ obtained during the instance of $\RandMultLCE$.
	\item {\bf Negative Output}: If $\flag^{(i)}_\iter$ is set to $1$ during the instance 
	$\RandMultLCE(\PartySet,\AdvStruct,\SharingSpec, \Discarded, \iter)$, 
	then set $\iter = \iter+1$ and go to the step {\bf Detectable Triple Generation}. 
	\end{myitemize}
 \end{myitemize}	
\end{protocolsplitbox}
\paragraph{Protocol $\PiStatTriples$ for Generating $M$ Multiplication-Triples:} The only modification will be
 to call $\RandMultLCE$ for generating $M$ random triples. 

We next prove the security of the protocol $\PiStatTriples$ in the $(\FVSS, \FABA)$-hybrid model.
  While proving these properties, we will assume that  $\AdvStructure$ satisfies the $\Q^{(3)}(\PartySet, \AdvStructure)$ condition.
  This further implies that the sharing specification $\ShareSpec = (S_1, \ldots, S_h) \defined \{ \PartySet \setminus Z | Z \in \AdvStructure\}$
  satisfies the $\Q^{(2)}(\ShareSpec, \AdvStructure)$ condition. 
  \\~\\
   \noindent {\bf Theorem \ref{thm:PiStatTriples}.}
  {\it Let  $\AdvStructure$ satisfy the $\Q^{(3)}(\PartySet, \AdvStructure)$ condition and let 
 $\ShareSpec = (S_1, \ldots, S_h) = \{ \PartySet \setminus Z | Z \in \AdvStructure\}$. 
 Then  $\PiStatTriples$ securely realizes $\FTriples$ with UC-security in the $(\FVSS, \FABA)$-hybrid model, except with
  error probability of at most $\frac{n}{|\F|}$.
   The protocol makes $\Order(M \cdot n^3)$ calls to $\FVSS$ and $\Order(n^3)$ calls to $\FABA$,
   and additionally incurs a communication of $\Order((M \cdot |\AdvStruct| \cdot n^3 + |\AdvStructure| \cdot n^4) \log{|\F|})$ bits. 
  }
  \begin{proof}
  The communication complexity and the number of calls to $\FVSS$ and $\FABA$ simply follows from the communication complexity of $\RandMultLCE$
   and the fact that there might be $\Order(n)$ instances of $\RandMultLCE$ in the protocol.
    This is because from Lemma \ref{lemma:RandMultLCE}, if any instance of $\RandMultLCE$ fails,
    then at least one new {\it corrupt} party is globally discarded and included in $\Discarded$. Once all the corrupt parties are included in $\Discarded$, then 
    from Claim \ref{claim:RandMultLCEHonestBehaviour}, the next instance of $\RandMultLCE$ is bound to give the correct output.
    
         We next prove the security. For the ease of explanation, we consider the case where
  only one multiplication-triple is generated in $\PiStatTriples$; i.e.~$M = 1$. The proof can easily be  modified for any general $M$.
   
   Let $\Adv$ be an arbitrary adversary, attacking the protocol $\PiStatTriples$ by corrupting a set of parties
  $Z^{\star} \in \AdvStructure$, and let $\Env$ be an arbitrary environment. We show the existence of a simulator $\SimStatTriples$ (Fig \ref{fig:SimStatTriples}),
   such that for any
 $Z^{\star} \in \AdvStructure$,
  the outputs of the honest parties and the view of the adversary in the 
   protocol $\PiStatTriples$ is indistinguishable from the outputs of the honest parties and the view of the adversary in an execution in the ideal world involving 
  $\SimStatTriples$ and $\FTriples$, except with probability at most $\frac{n}{|\F|}$.
  
  The high level idea of the simulator is very similar to that of the simulator for the protocol $\PiPerTriples$ (see the proof of Theorem \ref{thm:PiPerTriples}).
  Throughout the simulation,
  the simulator itself performs the role of the ideal functionalities $\FVSS$ and
   $\FABA$ whenever required and performs the role of the honest parties, exactly as per the steps of the protocol. 
   In each iteration, the simulator 
   simulates the actions of honest parties during the underlying instance of $\RandMultLCE$ by playing the role of the honest parties with random inputs.
   Once the simulator finds any iteration of $\RandMultLCE$ to be successful, the simulator learns the secret-sharing of the output triple of that iteration
   and sends the shares of this triple, corresponding to the corrupt parties to $\FTriples$, on the behalf of $\Adv$.

\begin{simulatorsplitbox}{$\SimStatTriples$}{Simulator for the protocol $\PiStatTriples$
  where $\Adv$ corrupts the parties in set $Z^{\star} \in \AdvStructure$}{fig:SimStatTriples}
	\justify
$\SimStatTriples$ constructs virtual real-world honest parties and invokes the real-world adversary $\Adv$. The simulator simulates the view of
 $\Adv$, namely its communication with $\Env$, the messages sent by the honest parties, and the interaction with $\FVSS$ and $\FABA$. 
  In order to simulate $\Env$, the simulator $\SimStatTriples$ forwards every message it receives from 
   $\Env$ to $\Adv$ and vice-versa.  The simulator then simulates the various stages of the protocol as follows. 
\begin{myitemize}
\item[--] \textbf{Initialization}: On behalf of the honest parties, the simulator initializes $\Discarded$ to $\emptyset$ and $\iter$ to $1$.
\item[--] \textbf{Detectable Triple Generation}: The simulator plays the role of the honest parties as per the protocol and interacts with $\Adv$ for an instance
 $\RandMultLCE(\PartySet,\AdvStruct,\SharingSpec, \Discarded, \iter)$. 
  During this instance, the simulator simulates the interface for $\FABA$ and $\FVSS$ for $\Adv$ during the underlying instances of 
  $\BasicMult$, by itself performing the role of $\FABA$ and $\FVSS$. Next, based on whether the instance is successful or not, 
   simulator does the following.
	\begin{myitemize}
	\item {\bf If during the instance $\RandMultLCE(\PartySet,\AdvStruct,\SharingSpec, \Discarded, \iter)$, 
           simulator has set $\flag^{(i)}_\iter=0$, corresponding to any $P_i \not \in Z^{\star}$}: 
           In this case, let $([\sima_\iter], [\simb_\iter], [\simc_\iter])$ be the output of the honest parties from the instance of $\RandMultLCE$. 
           The simulator then sets $\{ [\sima_\iter]_q, [\simb_\iter]_q, [\simc_\iter]_q \}_{S_q \cap Z^{\star} \neq \emptyset}$
           to be the shares corresponding to the parties in $Z^{\star}$ and goes to the step labelled {\bf Interaction with $\FTriples$}.            
	\item {\bf If during the instance $\RandMultLCE(\PartySet,\AdvStruct,\SharingSpec, \Discarded, \iter)$, 
           simulator has set $\flag^{(i)}_\iter=1$, corresponding to any $P_i \not \in Z^{\star}$}: 
           In this case, the simulator sets $\iter = \iter+1$ and goes to step labelled {\bf Detectable Triple Generation}.
	\end{myitemize}
\item[--] {\bf Interaction with $\FTriples$}:
  Let $\{[\widetilde{a}]_q,[\widetilde{b}]_q,[\widetilde{c}]_q\}_{S_q \cap Z^{\star} \neq \emptyset}$ be the shares set by the simulator
   corresponding to the parties in $Z^{\star}$. The simulator sends
    $(\shares,\sid, \{[\widetilde{a}]_q, [\widetilde{b}]_q, [\widetilde{c}]_q\}_{S_q \cap Z^{\star} \neq \emptyset})$ to $\FTriples$,  on the behalf of $\Adv$.
\end{myitemize}    
\end{simulatorsplitbox}

We now prove a series of claims which will help us to finally prove the theorem. We first show that the view generated by $\SimStatTriples$ for $\Adv$ is identically distributed to $\Adv$'s view during the real execution of $\PiStatTriples$.

\begin{claim}
\label{claim:StatTriplesPrivacy}
The view of $\Adv$ in the simulated execution with $\SimPerTriples$ is identically distributed as the view of $\Adv$ in the real execution of
 $\PiStatTriples$.
\end{claim}
\begin{proof}
In both the real as well as simulated execution, the parties run an instance of $\RandMultLCE$ for each iteration $\iter$, where in the simulated execution, the role
 of the honest parties is played by the simulator, including the role of $\FVSS$ and $\FABA$. 
  Now, in either execution, if $\flag^{(i)}_\iter$ is set to $0$ during some iteration $\iter$  corresponding to any {\it honest} $P_i$,
  then from Lemma \ref{lemma:RandMultLCE}, 
  the view of $\Adv$ will be independent of the underlying triple and hence, will be 
   identically distributed in both the executions. Else, in both executions, 
   at least one new corrupt party gets discarded and the parties
    proceed to the next iteration. Hence, the view of $\Adv$ in both executions is identically distributed.
\end{proof}

We now show that conditioned on the view of $\Adv$, the output of honest parties is identically distributed in the real execution of $\PiStatTriples$ involving $\Adv$, as well as in the ideal execution involving $\SimStatTriples$ and $\FTriples$.

\begin{claim}
\label{claim:StatTriplesOutputDistribution}
Conditioned on the view of $\Adv$, the output of the honest parties is identically distributed in the real execution of $\PiStatTriples$ involving $\Adv$ and in the ideal execution involving $\SimStatTriples$ and $\FTriples$, except with probability at most $\frac{n}{|\F|}$.
\end{claim}
\begin{proof}
Consider an arbitrary view $\View$ of $\Adv$, generated as per some execution of $\PiStatTriples$. From Lemma \ref{lemma:RandMultLCE}, 
 in the real execution of $\PiStatTriples$, during each iteration, 
  all honest parties either obtain shares of a random multiplication triple, 
  or discard a {\it new} maliciously-corrupt party. 
  Since $|Z^{\star}| < n$, it will take less than $n$ iterations to discard all the maliciously-corrupt parties. 
  Furthermore, once all parties in $Z^{\star}$ are discarded, from Claim \ref{claim:RandMultLCEHonestBehaviour}, the 
  next instance of $\RandMultLCE$ will output a secret-shared multiplication-triple for the honest parties. 
  Consequently, within $n$ iterations, there will be some iteration $\iter$, such that all honest parties $P_i$ eventually set
  $\flag_\iter^{(i)}$ to $0$ and output a secret-shared triple $([a_{\iter}], [b_{\iter}], [c_{\iter}])$. Moreover, from 
   the union bound, it follows that except with probability at most $\frac{n}{|\F|}$, the triple $(a_\iter, b_\iter, c_\iter)$ will be a multiplication-triple.
   Furthermore, from Lemma \ref{lemma:RandMultLCE}, the triple will be randomly distributed over $\F$. 
   
 To complete the proof, we show that conditioned on the shares $\{([a_{\iter}]_q, [b_{\iter}]_q, [c_{\iter}]_q) \}_{S_q \cap Z^{\star} \neq \emptyset}$ (which are determined by
    $\View$), the honest
    parties output a secret-sharing
     of some random multiplication-triple in the simulated execution, which is consistent with
   the shares  $\{([a_{\iter}]_q, [b_{\iter}]_q, [c_{\iter}]_q) \}_{S_q \cap Z^{\star} \neq \emptyset}$.
    However, this simply follows from the fact that in the simulated execution, 
    $\SimStatTriples$ sends the shares $\{([a_{\iter}]_q, [b_{\iter}]_q, [c_{\iter}]_q) \}_{S_q \cap Z^{\star} \neq \emptyset}$
     to $\FTriples$ on the behalf of the parties in $Z^{\star}$, and as an output, $\FTriples$ generates a random secret-sharing of some random multiplication-triple consistent with these shares. 
\end{proof}

\end{proof}

%% file: AppMPC.tex
\section{MPC Protocol in the Pre-Processing Model}
\label{app:MPC}
  The {\it perfectly-secure} AMPC protocol $\PiMPC$ in the $(\FTriples, \FVSS, \FABA)$-hybrid model is presented in Fig \ref{fig:PerAMPC}. The high level idea behind the protocol
   is already discussed in Section \ref{sec:mpc}. 
  The protocol has a {\it pre-processing} phase where secret-shared random multiplication triples are generated, an
  {\it input} phase where each party verifiably generates a secret-sharing of its input for the function $f$ and a common subset of input-providers
  is selected, and a {\it circuit-evaluation} phase where the circuit is securely evaluated and the function output is publicly reconstructed.
  
   In the protocol, all honest parties may not be reconstructing the function-output at the same ``time" and
   different parties may be at different phases of the protocol, as the protocol is executed asynchronously. Consequently, a party upon reconstructing the function-output,
   {\it cannot} afford to terminate immediately, as its presence and participation might be needed for the completion of various phases of the protocol by other honest parties.
   A standard trick to get around this problem in the AMPC protocols \cite{HNP05,HNP08,Coh16}
    is to have an additional {\it termination phase}, whose code is executed concurrently throughout the protocol to check if a party can ``safely" terminate the protocol with 
   the function output.

\begin{protocolsplitbox}{$\PiMPC$}{The perfectly-secure AMPC protocol in the $(\FTriples, \FVSS, \FABA)$-hybrid model. The public inputs of
 the protocol are $\PartySet, \ckt$ and $\AdvStructure$. The above steps are executed by every 
 $P_i \in \PartySet$}{fig:PerAMPC}
	\justify
Set the sharing specification as $\ShareSpec = (S_1, \ldots, S_h)
 \defined \{ \PartySet \setminus Z | Z \in \AdvStructure\}$, where $\AdvStructure$ is the adversary structure.\footnote{Thus $\ShareSpec$ is
  {\it $\AdvStructure$-private}.}
\begin{center}\underline{\bf Pre-Processing Phase}\end{center}
     \begin{mydescription}
     \item [1.] Send $(\triples, \sid, P_i)$ to the functionality $\FTriples$.
     \item [2.] Request output from $\FTriples$ until an output
       $(\tripleshares, \sid, \{[a^{(\ell)}]_q, [b^{(\ell)}]_q, [c^{(\ell)}]_q) \}_{\ell \in \{1, \ldots, M\}, P_i \in S_q})$ is received from $\FTriples$. 
     \end{mydescription}  
\begin{center}\underline{\bf Input Phase}\end{center}
 Once the output from $\FTriples$ is received, then proceed as follows.  
    \begin{myitemize}
    \item {\bf Secret-sharing of the Inputs and Collecting Shares of Other Inputs}:
      \begin{mydescription}
         \item Upon having the input $x^{(i)}$ for the function $f$, randomly select the shares $x^{(i)}_1, \ldots, x^{(i)}_h \in \F$, subject to the condition that
         $x^{(i)} = x^{(i)}_1 + \ldots + x^{(i)}_h$.
         Send $(\Dealer, \sid_i, P_i, (x^{(i)}_1, \ldots, x^{(i)}_h))$ to $\FVSS$, where 
         $\sid_i \defined \sid || i$.\footnote{The notation $\sid_i$ is used here to distinguish among the $n$ different calls to
         $\FVSS$.}
         \item For $j = 1, \ldots, n$, 
          request for output from $\FVSS$ with $\sid_j$ corresponding to the dealer $P_j$,  until 
         an output is received.          
      \end{mydescription}
    \item {\bf Selecting Common Input-Providers}: 
       \begin{mydescription}
             \item[1.]  If $(\Share, \sid_j, P_j, \{[x^{(j)}]_q \}_{P_i \in S_q})$ is received from $\FVSS$ with $\sid_j$, then  send              
              $(\vote, \sid_j, 1)$ to $\FABA$ with $\sid_j$, 
              where $\sid_j \defined \sid || j$.\footnote{The notation $\sid_j$ is used here to distinguish among the $n$ different calls to
         $\FABA$.}
             \item [2.] For $j = 1, \ldots, n$, keep requesting for output from $\FABA$ with $\sid_j$ until an output is received.
             \item [3.] If there exists a set of parties $\GlobalProv_i$, such that $\PartySet \setminus \GlobalProv_i \in \AdvStructure$
             and $(\decide, \sid_j, 1)$ is received from $\FABA$ with $\sid_j$ corresponding to each $P_j \in \GlobalProv_i$, 
             then send
              $(\vote, \sid_j, 0)$ to every $\FABA$ with $\sid_j$ for which no input has been provided yet.
             \item [4.] Once $(\decide, \sid_j, v_j)$ is received from $\FABA$ with $\sid_j$ for every $j \in \{1, \ldots, n \}$, set
                  $\CoreSet = \{P_j: v_j = 1\}$.
            \item [5.] Wait until $(\Share, \sid_j, P_j, \{[x^{(j)}]_q \}_{P_i \in S_q})$ is received from $\FVSS$ for every $P_j \in \CoreSet$.              
               For every $P_j \not \in \CoreSet$, participate in an instance of the protocol $\PiPerDefaultShare$ with public input $0$ to generate a default secret-sharing of $0$.               
       \end{mydescription}
    \end{myitemize}
\begin{center}\underline{\bf Circuit-Evaluation Phase}\end{center}
  Evaluate each gate $g$ in the circuit
 according to the topological ordering as follows, depending upon the type of $g$.
        \begin{myitemize}
           \item {\bf Addition Gate}: If $g$ is an addition gate with inputs $x, y$ and output $z$, then corresponding to every $S_q$ such that
           $P_i \in S_q$, set $[z]_q = [x]_q + [y]_q$ as the share
           corresponding to $z$. Here $\{[x]_q\}_{P_i \in S_q}$ and $\{[y]_q\}_{P_i \in S_q}$ are $P_i$'s shares 
           corresponding to gate-inputs $x$ and $y$ respectively.
             \item {\bf Multiplication Gate}: If $g$ is the $\ell^{th}$ multiplication gate with inputs $x, y$ and output $z$, where $\ell \in \{1, \ldots, M \}$, then do the following:
               \begin{mydescription}
               \item [1.] Corresponding to every $S_q$ such that $P_i \in S_q$, set 
                  $[d^{(\ell)}]_q \defined [x]_q - [a^{(\ell)}]_q$ and $[e^{(\ell)}]_q \defined [y]_q - [b^{(\ell)}]_q$, where $\{[x]_q\}_{P_i \in S_q}$ and $\{[y]_q\}_{P_i \in S_q}$ are 
                  $P_i$'s shares corresponding to gate-inputs $x$ and $y$ respectively and $\{([a^{(\ell)}]_q, [b^{(\ell)}]_q, [c^{(\ell)}]_q)\}_{P_i \in S_q}$
                  are $P_i$'s shares corresponding to the $\ell^{th}$ multiplication-triple.
               \item [2.] Participate in instances of $\PiPerRec$ with shares $\{[d^{(\ell)}]_q \}_{P_i \in S_q}$ and $\{[e^{(\ell)}]_q \}_{P_i \in S_q}$ to publicly reconstruct 
                 $d^{(\ell)}$ and $e^{(\ell)}$, where $d^{(\ell)} \defined x - a^{(\ell)}$ and $e^{(\ell)} \defined y - b^{(\ell)}$.
               \item [4.] Upon reconstructing $d^{(\ell)}$ and $e^{(\ell)}$, corresponding to every $S_q$ such that $P_i \in S_q$, set
                 $[z]_q \defined d^{(\ell)} \cdot e^{(\ell)} + d^{(\ell)} \cdot [b^{(\ell)}]_q + e^{(\ell)} \cdot [a^{(\ell)}]_q + [c^{(\ell)}]_q$.
                 Set $\{[z]_q \}_{P_i \in S_q}$ as the shares corresponding to $z$.
               \end{mydescription}
          \item {\bf Output Gate}: If $g$ is the output gate with output $y$, then participate in an instance of $\PiPerRec$ with shares $\{[y]_q \}_{P_i \in S_q}$ to publicly
          reconstruct $y$.
       \end{myitemize}
\begin{center}\underline{\bf Termination Phase}\end{center}
Concurrently execute the following steps during the protocol:
   \begin{mydescription}
   \item[1.] If the circuit-output $y$ is computed, then send $(\Ready, \sid, P_i, y)$ to every party in $\PartySet$.
      \item[2.] If the message $(\Ready, \sid, P_j, y)$ is received from a set of parties ${\cal A}$ such that $\AdvStructure$ satisfies  $\Q^{(1)}({\cal A}, \AdvStructure)$ condition, then 
       send $(\Ready, \sid, P_i, y)$  to every party in $\PartySet$.
       \item[3.] If the message $(\Ready, \sid, P_j, y)$ is received from a set of parties $\W$ such that $\PartySet \setminus \W \in \AdvStructure$, 
       then output $y$ and terminate.
   \end{mydescription}
\end{protocolsplitbox}

Intuitively, protocol $\PiMPC$ eventually terminates as the set $\CoreSet$ is eventually decided. This is because even if the corrupt parties do not secret-share their inputs,
 the inputs of all honest parties are eventually secret-shared. Once $\CoreSet$ is decided, the evaluation of each gate will be eventually completed:
  while the addition gates are evaluated non-interactively, the evaluation of multiplication gates requires reconstructing the corresponding masked gate-inputs
  which is eventually completed due to the reconstruction protocols.
  The privacy of the inputs of the honest parties in $\CoreSet$ will be maintained as the sharing specification $\ShareSpec$ is $\AdvStructure$-private.
  Moreover, the inputs of the corrupt parties in $\CoreSet$ will be independent of the inputs of the honest parties in $\CoreSet$, as inputs are secret-shared
  via calls to $\FVSS$. Finally, correctness holds since each gate is evaluated correctly. 
   We next rigorously formalize this intuition by giving a formal security proof
   and show that the protocol $\PiMPC$ is perfectly-secure, if the parties have access to ideal functionalities
   $\FTriples, \FVSS$ and $\FABA$.   \\~\\
\noindent {\bf Theorem \ref{thm:AMPC}.}
  {\it  Protocol $\PiMPC$ UC-securely realizes the functionality $\Functionality$ for securely computing  $f$ (see Fig \ref{fig:FAMPC} in Appendix \ref{app:UC}) 
   with perfect security in the $(\FTriples, \FVSS, \FABA)$-hybrid model, in the 
   presence of a static malicious adversary characterized by an adversary-structure $\AdvStructure$, satisfying the $\Q^{(3)}(\PartySet, \AdvStructure)$ condition.
   The protocol makes one call to $\FTriples$ and $\Order(n)$ calls to $\FVSS$ and $\FABA$ and additionally incurs a communication of 
   $\Order(M \cdot |\AdvStructure| \cdot n^2 \log{|\F|})$ bits, where $M$ is the number of multiplication gates in the circuit $\ckt$ representing $f$.}
\begin{proof}
The communication complexity in the $(\FTriples, \FVSS, \FABA)$-hybrid model follows from the fact that for evaluating each multiplication gate, the 
 parties need to run $2$ instances of the reconstruction protocol $\PiPerRec$. 
 
 For security, let $\Adv$ be an arbitrary real-world adversary  corrupting the set of parties $Z^{\star} \in \AdvStructure$
  and let $\Env$ be an arbitrary environment. We  show  the  existence  of  a  simulator $\SimPerAMPC$, such  that 
   the output of honest parties and the view of the adversary in an execution of the real protocol with $\Adv$
    is identical to the output in an execution with $\SimPerAMPC$ involving $\Functionality$ in the ideal model. 
    This further implies that $\Env$ cannot distinguish between the two executions.
      The steps of the simulator are  given in 
   Fig \ref{fig:SimPerAMPC}. 
   
   The high level idea of the simulator is as follows. During the simulated execution, the simulator itself performs the role of the ideal functionalities $\FTriples, \FVSS$
    and $\FABA$ whenever required. Performing the role of $\FTriples$ allows the simulator to learn the secret-sharing
     of all the multiplication-triples.
   During the input phase, whenever $\Adv$ secret-shares any value through $\FVSS$ 
    on the behalf of a corrupt party, the simulator 
    records this on the behalf of the corrupt party. This allows the simulator to learn the function-input of the corresponding
   corrupt party. On the other hand, for the {\it honest} parties, the simulator picks arbitrary values as their 
   function-inputs and simulates the secret-sharing of those input values using random shares, as per $\FVSS$. 
    To select the common input-providers during the simulated execution, the simulator
   itself performs the role of $\FABA$ and simulates the honest parties as per the steps of the protocol and $\FABA$. This allows the simulator
   to learn the 
   common subset of input-providers $\CoreSet$, which the simulator passes to the functionality $\Functionality$. Notice that 
   the function-inputs for each corrupt party in $\CoreSet$ will be available with the simulator. 
   This is because for every corrupt party $P_j$ which is added to $\CoreSet$, at least one honest party $P_i$ should participate with input $1$ in the corresponding call to
   $\FABA$. This implies that the honest party $P_i$ must have received the shares $P_j$ sent to $\FVSS$ from $\FVSS$. Since in the simulation, the role of $\FVSS$ is played by the simulator, it implies that the full vector of shares provided by 
   $P_j$ to $\FVSS$ will be known to the simulator.
      Hence, along with $\CoreSet$, the simulator can send the corresponding function-inputs of the corrupt parties in $\CoreSet$ to $\Functionality$.
      Upon receiving the function-output, the simulator simulates the steps of the honest parties for the gate evaluations as per the
    protocol. Finally, for the output gate, the simulator arbitrarily computes a secret-sharing of the function-output $y$ received from $\Functionality$, which is consistent
    with the shares which corrupt parties hold for the output-gate sharing. Then, on the behalf of the honest parties, the simulator sends the shares corresponding to
    the above sharing of $y$ during the public reconstruction of $y$. This ensures that in the simulated execution, $\Adv$ learns the function-output $y$. For the termination phase, 
    the simulator sends $y$ on the behalf of honest parties.

\begin{simulatorsplitbox}{$\SimPerAMPC$}{Simulator for the protocol $\PiMPC$ where $\Adv$ corrupts the parties in set
  $Z^{\star} \in \AdvStructure$}{fig:SimPerAMPC}
	\justify
$\SimPerAMPC$ constructs virtual real-world honest parties and invokes the real-world adversary $\Adv$. The simulator simulates the view of
 $\Adv$, namely its communication with $\Env$, the messages sent by the honest parties, and the interaction with various functionalities. 
  In order to simulate $\Env$, the simulator $\SimPerAMPC$ forwards every message it receives from 
   $\Env$ to $\Adv$ and vice-versa.  The simulator then simulates the various phases of the protocol as follows.
\begin{center}\underline{\bf Pre-Processing Phase}\end{center}
{\bf Simulating the call to $\FTriples$}: The simulator simulates the steps of $\FTriples$ by itself playing the role of $\FTriples$.
 Namely, it receives the shares corresponding to the parties in $Z^{\star}$ for each multiplication-triple from $\Adv$ and then
 randomly generates secret-sharing of $M$ random multiplication-triples $\{(\sima^{(\ell)}, \simb^{(\ell)}, \simc^{(\ell)})\}_{\ell = 1, \ldots, M}$
  consistent with the provided shares.
  At the end of simulation of this phase, the simulator will know the entire vector of shares corresponding to the secret-sharing of all
  multiplication-triples.
  \begin{center}\underline{\bf Input Phase}\end{center}
  \begin{myitemize}
   \item The simulator simulates the operations of the honest parties during the input phase, by randomly picking $\simx^{(j)}$ as the input, for
 every $P_j \not \in Z^{\star}$, selecting random shares $\simx^{(j)}_1, \ldots, \simx^{(j)}_h$ such that
  $\simx^{(j)} = \simx^{(j)}_1 + \ldots + \simx^{(j)}_h$, and setting $[\simx^{(j)}]_q = \simx^{(j)}_q$, for $q = 1, \ldots, h$. 
  When $\Adv$ requests output from $\FVSS$ with $\sid_j$ on the behalf of any party $P_i \in Z^{\star}$,
  then the simulator  responds with an output 
  $(\Share, \sid_j, P_j, \{[\simx^{(j)}]_q \}_{P_i \in S_q})$ 
   on the behalf of $\FVSS$.
   \item Whenever $\Adv$ sends $(\Dealer, \sid_i, P_i, (x^{(i)}_1, \ldots, x^{(i)}_h))$ 
      to $\FVSS$ on the behalf of any $P_i \in Z^{\star}$, the simulator
      records the input $x^{(i)} \defined x^{(i)}_1 + \ldots + x^{(i)}_h$ on the behalf of $P_i$ and sets
      $[x^{(i}] = (x^{(i)}_1, \ldots, x^{(i)}_h)$.
   \item When the simulation reaches the ``Selecting Common Input-Providers" stage, the simulator simulates the interface of $\FABA$ to $\Adv$ by itself
   performing the role of $\FABA$. 
   When the first honest party completes the simulated input phase, $\SimPerAMPC$ learns the set $\CoreSet$. 
    \end{myitemize}
 \noindent {\bf Interaction with $\Functionality$}:  Once the simulator learns $\CoreSet$, it sends 
  the input values $x^{(i)}$ that it has recorded on the behalf of each
  $P_i \in (Z^{\star} \cap \CoreSet)$, and the set of input-providers $\CoreSet$ to $\Functionality$. Upon receiving the 
   output $y$ from $\Functionality$, the simulator starts the simulation of circuit-evaluation phase.
\begin{center}\underline{\bf Circuit-Evaluation Phase}\end{center}
The simulator simulates the evaluation of each gate $g$ in the circuit in topological order as follows:
\begin{myitemize}
\item {\bf Addition Gate}: Since this step involves local computation, the simulator does not have to simulate any messages on the behalf of the
honest parties. The simulator locally adds the secret-sharings corresponding to the gate-inputs and obtains the 
  secret-sharing corresponding to the gate-output.
\item {\bf Multiplication Gate}: If $g$ is the $\ell^{th}$ multiplication gate in the circuit, then the simulator takes the complete secret-sharing of the 
 $\ell^{th}$ multiplication triple $(\sima^{(\ell)}, \simb^{(\ell)}, \simc^{(\ell)})$ and computes the messages of the honest parties as per the
  steps of the protocol (by considering the secret-sharing of the above multiplication-triple and the secret-sharing of the gate-inputs),
   and sends them to $\Adv$ on the behalf of the honest parties as part of the instances of $\PiPerRec$ protocol.
     Once the simulation of the circuit-evaluation phase is done, the simulator will know the secret-sharing corresponding to the gate-output. 
\item {\bf Output Gate}: Let $[\simy] = (\simy_1, \ldots, \simy_h)$ be the secret-sharing corresponding to the output gate, available with $\SimPerAMPC$
 during the simulated circuit-evaluation. The simulator then randomly selects shares $\simby_1, \ldots, \simby_h$ such that
 $\simby_1+  \ldots +  \simby_h = y$ and $\simby_q = \simy_q$ corresponding to every $S_q \in \ShareSpec$ where
 $S_q \cap Z^{\star} \neq \emptyset$. Then, as part of the instance of $\PiPerRec$ protocol to reconstruct the function output,
  the simulator sends the shares $\{\simby_q \}_{S_q \in \ShareSpec}$ to $\Adv$ on the behalf of the honest parties.
\end{myitemize}
\begin{center}\underline{\bf Termination Phase}\end{center}
 The simulator sends a $\Ready$ message for $y$ to $\Adv$ on the behalf of $P_i \not \in Z^{\star}$,
  if in the simulated execution, $P_i$ has computed $y$.   
\end{simulatorsplitbox}
    
 We next prove a sequence of claims, which helps us to show that the joint distribution of the honest parties
  and the view of $\Adv$ is identical in both the real, as well as the ideal-world.
    We first claim that in  any execution of $\PiMPC$, a set $\CoreSet$ is eventually generated. This automatically implies that
    the honest parties eventually possess a secret-sharing of $M$ random multiplication-triples generated by $\FTriples$, as well as 
    a secret-sharing of the inputs of the parties in $\CoreSet$.
  \begin{claim}
  \label{claim:PerAMPCTermination}
  In any execution of $\PiMPC$, a set $\CoreSet$ is eventually generated, such that for every $P_j \in \CoreSet$, there exists some $x^{(j)}$ held by $P_j$
   which is eventually secret-shared.
  \end{claim}
  \begin{proof}
  As the proof of this claim is similar to the proof of Claim \ref{claim:PerTriplesTermination}, we skip the formal proof.
  \end{proof}
   
We next show that the view generated by $\SimPerAMPC$ for $\Adv$ is identically distributed to $\Adv$'s view during the real execution of $\PiMPC$.
\begin{claim}
\label{claim:PerAMPCPrivacy}
The view of $\Adv$ in the simulated execution with $\SimPerAMPC$ is identically distributed to the view of $\Adv$ in the real execution of
 $\PiMPC$.
\end{claim}
\begin{proof}
It is easy to see that the view of $\Adv$ during the pre-processing phase is identically distributed in both the executions. This is because in both the executions,
 $\Adv$ receives no messages from the honest parties and the steps of $\FTriples$ are executed by the simulator itself in the simulated execution.
 Namely, in both the executions, $\Adv$'s view consists of the shares of $M$ random multiplication-triples corresponding to the parties in $Z^{\star}$. 
 So, let us fix these shares. Conditioned on these shares,  
  during the input phase, $\Adv$ learns the shares $\{[x^{(j)}]_q \}_{P_j \not \in Z^{\star}, (S_q \cap Z^{\star}) \neq \emptyset}$ during the real execution
  corresponding to the
 parties $P_j \not \in Z^{\star}$. In the simulated execution, it learns the shares 
  $\{[\simx^{(j)}]_q \}_{P_j \not \in Z^{\star}, (S_q \cap Z^{\star}) \neq \emptyset}$. Since the sharing specification $\ShareSpec$
   is $\AdvStructure$-private and the vector of shares $(x^{(j)}_1, \ldots, x^{(j)}_h)$ as well as 
   $(\simx^{(j)}_1, \ldots, \simx^{(j)}_h)$ are randomly chosen, it follows that 
  the distribution of  the shares $\{[x^{(j)}]_q \}_{P_j \not \in Z^{\star}, (S_q \cap Z^{\star}) \neq \emptyset}$
  as well as $\{[\simx^{(j)}]_q \}_{P_j \not \in Z^{\star}, (S_q \cap Z^{\star}) \neq \emptyset}$ is identical and
  independent of both $x^{(j)}$ as well as $\simx^{(j)}$,
    so let us fix these shares. Since the role of $\FABA$ is played by the simulator itself, it follows easily that the view of $\Adv$ during the selection of the set $\CoreSet$
  is identically distributed in both the real as well as the simulated execution.
 
 During the evaluation of linear gates, no communication is involved. During the evaluation of multiplication gates, in the simulated execution,
 the simulator will know the secret-sharing associated with gate-inputs and also the secret-sharing of the associated
 multiplication-triple. Hence, the simulator correctly sends the shares corresponding to the values $d^{(\ell)}$ and $e^{(\ell)}$ as per the protocol
 on the behalf of the honest parties. Moreover, the values $d^{(\ell)}$ and $e^{(\ell)}$ will be randomly distributed for $\Adv$ in both the executions, since
  the underlying multiplication-triple is randomly distributed, conditioned on the shares of the corrupt parties.
  Thus, $\Adv$'s view during the evaluation of multiplications gates is identically distributed
 in both the executions.
  
 For the output gate, the shares received by $\Adv$ in the real execution from the honest parties correspond to a secret-sharing
 of the function-output $y$. From the steps of $\SimPerAMPC$, it is easy to see that the same holds even in the simulated execution,
 as $\SimPerAMPC$ sends to $\Adv$ shares corresponding to a secret-sharing of $y$, which are consistent with the shares
 held by $\Adv$. Hence, $\Adv$'s view is identically distributed in both the executions during the evaluation of output gate.
 Finally, it is easy to see that $\Adv$'s view is identically distributed in both the executions during the termination phase. This is because 
 in both the executions, every honest party who has obtained the function output $y$, sends a $\Ready$ message for $y$.
\end{proof}
 
  We next claim that conditioned on the view of $\Adv$ (which is identically distributed in both the executions from the last claim),
    the output of the honest parties is identically distributed in both the worlds.
 \begin{claim}
\label{claim:PerAMPCOutputDistribution}
Conditioned on the view of $\Adv$, the output of the honest parties is
  identically distributed in the real execution of $\PiMPC$ involving $\Adv$, as well as in the ideal execution
 involving $\SimPerAMPC$ and $\Functionality$.
\end{claim}
 \begin{proof}
 Let $\View$ be an arbitrary view of $\Adv$, and let $\CoreSet$ be the set of input-providers determined by $\View$ 
  (from Claim \ref{claim:PerAMPCTermination}, such a set $\CoreSet$ is bound to exist).
  Moreover, according to $\View$,  for every $P_i \in \CoreSet$, there exists some input $x^{(i)}$ such that the parties hold a 
  secret-sharing of $x^{(i)}$. Furthermore, from Claim \ref{claim:PerAMPCPrivacy},  
   if $P_i \in Z^{\star}$ then the corresponding secret-sharing is included in $\View$. For $P_i \not \in Z^{\star}$, 
   the corresponding $x^{(i)}$ is uniformly distributed conditioned on the shares of $x^{(i)}$ available with $\Adv$ as determined by
   $\View$. Let us fix the $x^{(i)}$ values corresponding to the parties in $\CoreSet$
   and denote the vector of values $x^{(i)}$, where $x^{(i)} = 0$ if $P_i \not \in \CoreSet$, by $\vec{x}$.
   
   It is easy to see that in the ideal-world, the output of the honest parties is $y$, where $y \defined f(\vec{x})$.
   This is because $\SimPerAMPC$ provides the identity of $\CoreSet$ along with the inputs $x^{(i)}$ corresponding to $P_i \in (\CoreSet \cap Z^{\star})$ to
   $\Functionality$. We now show that the honest parties eventually output $y$ even in the real-world. For this, we argue
    that all the values during the circuit-evaluation phase of the protocol are correctly secret-shared. 
    Since the evaluation of linear gates needs only local computation, it follows
 that the output of the linear gates will be correctly secret-shared. During the evaluation of a multiplication gate, the honest parties will hold a 
  secret-sharing of the corresponding $d^{(\ell)}$ and $e^{(\ell)}$ values, as during the pre-processing phase,
  all the multiplication-triples are generated in a secret-shared fashion, since they are computed
 and distributed by $\FTriples$. Since $\ShareSpec$ satisfies the $\Q^{(2)}(\ShareSpec, \AdvStructure)$ condition,
  the honest parties eventually get $d^{(\ell)}$ and $e^{(\ell)}$ through the instances of $\PiPerRec$.
     This automatically implies that the honest parties eventually hold a secret-sharing of $y$ and reconstruct it correctly, as
     $y$ is reconstructed through an instance of $\PiPerRec$. 
     Hence, during the termination phase, every honest party will eventually send a $\Ready$ message for $y$, while the parties in $Z^{\star}$
     may send a $\Ready$ message for $y' \neq y$. Since $Z^{\star} \in \AdvStructure$, it follows that
     no honest party ever sends a $\Ready$ message for $y'$. Hence no honest party ever outputs $y'$, as it will never receive
     the required number of $\Ready$ messages for $y'$.
           Since the $\Ready$ messages of the {\it honest}
     parties for $y$ are eventually delivered to every honest party, it follows that eventually, all honest parties 
     receive sufficiently many $\Ready$ messages to obtain some output, even if the corrupt parties does not send the required messages.
     
     Now let $P_i$ be the {\it first honest} party to terminate the protocol with some output. From the above arguments, the output has to be $y$.
     This implies that $P_i$ receives $\Ready$ messages for $y$ 
    from a set of parties $\PartySet \setminus Z$, for some $Z \in \AdvStructure$. 
       Let $\Hon$ be the set of {\it honest} parties whose $\Ready$ messages are received by $P_i$. It is easy to see that
   $\Hon \not \in \AdvStructure$, as otherwise, $\AdvStructure$ does not satisfy the $\Q^{(3)}(\PartySet, \AdvStructure)$ condition.
    The $\Ready$ messages of the parties in $\Hon$ are eventually delivered to every honest party
   and hence, {\it each} honest party (including $P_i$) eventually executes step $2$ of the termination phase and sends a
   $\Ready$ message for $y$.
   It follows that the $\Ready$ messages of {\it all} honest parties $\PartySet \setminus Z^{\star}$
    are eventually delivered
   to every honest party (irrespective of whether $\Adv$ sends all the required
   messages), guaranteeing that all honest parties eventually obtain the output $y$.
 \end{proof}
The theorem now follows from Claims \ref{claim:PerAMPCTermination}-\ref{claim:PerAMPCOutputDistribution}.
\end{proof}